\pdfoutput=1
\documentclass[review=false]{jfp-epi}
\usepackage{jfp-cup2epi}

\usepackage{hyperref}
\usepackage{cleveref}   %
\usepackage{etoc}       %
\usepackage{caption}

\usepackage{soul}
\sethlcolor{yellow}

\let\originalleft\left
\let\originalright\right
\renewcommand{\left}{\mathopen{}\mathclose\bgroup\originalleft}
\renewcommand{\right}{\aftergroup\egroup\originalright}

\usepackage{anyfontsize} %
\usepackage{amsmath}     %
\usepackage{amssymb}     %
\usepackage{semantic}    %
\usepackage{stmaryrd}   %
\usepackage{mathpartir}  %
\usepackage{braket}      %
\usepackage{thmtools}    %
\usepackage{xspace}
\usepackage{proof}
\usepackage{myyhmath}
\usepackage{rotating}
\usepackage{tikz-cd}
\usepackage{changepage}
\usepackage{listings}
\usepackage{mdframed}
\usepackage{mathtools}
\usepackage{ dsfont }
\usepackage[toc,page]{appendix}

\definecolor{violeta}{HTML}{9000FF}
\definecolor{propcolor}{HTML}{3F7D31}
\definecolor{etcolor}{HTML}{ED720E}

\newenvironment{block}[1][t]
  {\begin{array}[#1]{@{}l@{}}}
  {\end{array}}

\definecolor{lightgray}{gray}{0.90}
\newcommand{\Gbox}[1]{\colorbox{lightgray}{$#1$}}

\newcommand{\gbox}[1]{{\setlength{\fboxsep}{1pt}\colorbox{lightgray}{\tiny$#1$}}}

\mathlig{[]}{\square}
\mathlig{|-s}{\vdash^{\!\forall}}
\mathlig{|-}{\vdash}
\mathlig{|}{\mid} %
\mathlig{@>}{\leadsto} %

\reservestyle{\oblang}{\mathsf}
\oblang{if[if\;],then[\;then\;],else[\;else\;],let[let\;],in[\;in\;]}
\newcommand{\ite}[3]{\<if> #1 \<then> #2 \<else> #3}

\newcommand{\Int}{\mathsf{Int}}

\newcommand{\Bool}{\mathsf{Bool}}

\newcommand{\oblset}[1]{\textsc{#1}}
\newcommand{\Type}{\oblset{Type}}
\newcommand{\Var}{\oblset{Var}}
\newcommand{\Value}{\oblset{Value}}
\newcommand{\Term}{\oblset{Term}}

\newcommand{\fff}{\textsf{false}}
\newcommand{\ttt}{\textsf{true}}
\newcommand{\true}{\mathsf{true}}
\newcommand{\false}{\mathsf{false}}

\newcommand{\dom}{\mathit{dom}}
\newcommand{\cod}{\mathit{cod}}

\newcommand{\meet}{\sqcap}
\newcommand{\union}{\cup}

\newcommand{\red}{\longmapsto}
\newcommand{\nred}{~-->~}
\newcommand{\error}{\textup{\textbf{error}}}

\newcommand{\ie}{\emph{i.e.}\xspace}

\newcommand{\eg}{\emph{e.g.}\xspace}

\newcommand{\dt}{\dyn{t}}

\newcommand{\lett}[0]{\mathsf{let}}

\mathlig{++}{\mathbin{++}}

\newcommand{\GType}{\oblset{GType}}

\newcommand{\gprec}{\sqsubseteq} %

\newcommand{\?}{\textsf{\upshape ?}} %
\newcommand{\consistent}[1]{\widetilde{#1}}
\newcommand{\collecting}[1]{\wideparen{#1}}
\newcommand{\cT}{{\consistent{T}}} %
\newcommand{\clT}{\collecting{T}}%

\newcommand{\ccod}{\consistent{\cod}}
\newcommand{\cdom}{\consistent{\dom}}

\newcommand{\ceq}{\consistent{=}}

\newcommand{\trans}[1]{\circ^{#1}}

\newcommand{\ev}{\varepsilon}

\newcommand{\pr}[1]{\braket{{#1}}}

\newsavebox{\cTsave}
\savebox{\cTsave}{$\cT$}
\newsavebox{\clTsave}
\savebox{\clTsave}{$\clT$}

\usepackage{comment}
\usepackage{amsthm}
\usepackage{array}
\usepackage{mathrsfs} 
\usepackage{graphicx}
\usepackage{tensor}
\usepackage{wrapfig}
\usepackage{BOONDOX-cal}
\usepackage{paralist}
\usepackage{minibox}
\usepackage{longtable}   %

\makeatletter
\ifdefined\newcounteralias
  \let\gp@newcounteralias\newcounteralias
  \renewcommand*{\newcounteralias}[2]{%
    \expandafter\ifx\csname c@#1\expandafter\endcsname\csname c@#2\endcsname
    \else\gp@newcounteralias{#1}{#2}\fi}
\fi
\makeatother

\newcommand{\ExEndSymbol}{$\square$}

\newtheorem{lemma}{Lemma}
\newtheorem{definition}{Definition}
\newtheorem{theorem}{Theorem}
\theoremstyle{definition}
\newtheorem{ExampleNoQED}{Example}

\newenvironment{example}
  {\begin{ExampleNoQED}}
  {\hfill\ExEndSymbol\end{ExampleNoQED}}

\PassOptionsToPackage{names,dvipsnames}{xcolor}

\DeclareMathAlphabet{\mathsf}{OT1}{LibertinusSans-LF}{m}{n}
\SetMathAlphabet{\mathsf}{bold}{OT1}{LibertinusSans-LF}{bx}{n}

\declaretheorem[style=remark,numbered=no]{case}

\usepackage{xspace}
\renewcommand{\ie}{i.e.\@\xspace}
\newcommand{\aka}{a.k.a.\@\xspace}
\newcommand{\wrt}{w.r.t.\@\xspace}
\renewcommand{\eg}{e.g.\@\xspace}

\newenvironment{talign*}
 {\csname align*\endcsname}
 {\endalign}

\usepackage[colorinlistoftodos]{todonotes}

\newcommand{\proofheading}[1]{\paragraph*{\bfseries #1}}
\newenvironment{theoremproof}[1]{\begin{proof}[{\normalfont\bfseries\itshape Proof of #1}]}{\end{proof}}

\newcommand{\store}{s}

\DeclareDocumentCommand{\reftype}{O{\ttt} O{\Int}}{\{\nu: #2~|~ #1\}}

\newcommand{\ad}{~\text{and}~}

\newcommand{\crpsymbol}{\gamma_{p}}
\newcommand{\crtsymbol}{\gamma_{\tau}}
\newcommand{\crtasymbol}{\gamma_{T}}

\DeclareDocumentCommand{\crp}{m}{\crpsymbol (#1)}
\DeclareDocumentCommand{\crt}{m}{\crtsymbol (#1)}
\DeclareDocumentCommand{\crta}{m}{\crtasymbol (#1)}

\makeatletter
\newcommand{\DeclareRevPart}[2]{\@namedef{rev@col@#1}{#2}}

\def\rev@mode{parts}
\def\rev@mode@parts{parts}
\def\rev@mode@single{single}
\def\rev@mode@off{off}
\newcommand{\revmode}[1]{%
  \@ifundefined{rev@mode@#1}%
    {\PackageError{revision}{Unknown \string\revmode\space `#1'}%
                            {Use parts, single or off.}}%
    {\def\rev@mode{#1}}}

\def\rev@singlecolor{blue}

\newcommand{\rev@apply}[1]{%
  \ifx\rev@mode\rev@mode@parts
    \@ifundefined{rev@col@#1}%
      {\PackageError{revision}{Unknown revision key `#1'}%
                              {Declare it with \string\DeclareRevPart.}}%
      {\expandafter\color\expandafter{\csname rev@col@#1\endcsname}}%
  \else
    \ifx\rev@mode\rev@mode@single \color{\rev@singlecolor}%
    \else \ignorespaces          %
    \fi
  \fi}

\newcommand{\rev}[2][mt1]{{\rev@apply{#1}#2}}
\newenvironment{revblock}[1][mt1]{\rev@apply{#1}}{}
\makeatother

\definecolor{revviolet}{HTML}{7A1FA2}
\DeclareRevPart{mt1}{blue}       %
\DeclareRevPart{fo1}{revviolet}  %
\definecolor{revgreen}{HTML}{1B7F3B}
\DeclareRevPart{fo2}{revgreen}   %

\newcommand{\change}[1]{%
  \PackageWarning{revision}{\string\change\space is deprecated; use
    \string\rev[mt1] or \string\rev[fo1] instead}%
  \rev[mt1]{#1}}

\definecolor{leancolor}{HTML}{A0522D}
\newcommand{\leanname}[1]{{\def\_{\textunderscore\allowbreak}\textcolor{leancolor}{\texttt{#1}}}}
\newif\ifmechnotes \mechnotestrue
\newcommand{\mechnote}[1]{\ifmechnotes\leavevmode\hskip-.5em\relax
  \footnote{\raggedright \textsf{Lean} mechanization: #1}\hskip.5em\relax\fi\ignorespaces}

\newcommand{\extends}[1]{{ #1}}

\definecolor{evdcolor}{HTML}{F26035}
\definecolor{evcolor}{HTML}{F282B4}
\definecolor{sourcecolor}{HTML}{0071BC}
\definecolor{targetcolor}{HTML}{F26035}
\definecolor{changecolor}{HTML}{00CC99}
\definecolor{proposalcolor}{HTML}{008000}

\newcommand{\trg}[1]{\ensuremath{\textcolor{targetcolor}{#1}}}
\newcommand{\src}[1]{\ensuremath{\textcolor{sourcecolor}{\mathsf{#1}}}}

\newcommand{\colora}[1]{{\color{red}{#1}}}
\newcommand{\colorb}[1]{{\color{blue}{#1}}}

\DeclareDocumentCommand{\grel}{}{\rel_{\text{AGT}} }
\DeclareDocumentCommand{\ggprec}{}{\gprec_{AGT} }

\DeclareDocumentCommand{\ev}{O{} O{}}{{\color{evcolor}\varepsilon^{#2}_{#1}} }
\DeclareDocumentCommand{\evp}{O{} O{}}{{\color{evcolor}\varepsilon'^{#2}_{#1}} }
\DeclareDocumentCommand{\evpp}{O{} O{}}{{\color{evcolor}\varepsilon''^{#2}_{#1}} }

\NewDocumentCommand{\evbind}{O{1}}{\mathit{bind}_{
  \if 1#1 l \else r \fi
}}

\renewcommand{\nred}{~-->~}

\DeclareDocumentCommand{\enred}{O{\store}}{\overset{#1}{\nred}}
\DeclareDocumentCommand{\ered}{O{\store}}{\overset{#1}{\red}}

\NewDocumentCommand{\ctx}{O{}}{\textcolor{SymPurple}{\mathsf{ctx}}_{#1}}

\newcommand{\erase}[1]{\ensuremath{\left\lfloor#1\right\rfloor}}
\newcommand{\anf}[1]{\ensuremath{\mathsf{mnf}(#1)}}%
\newcommand{\ph}{\ensuremath{\_}}

\renewcommand{\t}{\ensuremath{t}\xspace}

\newcommand{\eq}{%
\mathrel{
 \!
\settowidth{\@tempdima}{--}%
\resizebox{\@tempdima}{\height}{=}%
\!
}
}

\newcommand{\Formula}{\oblset{Formula}\xspace}

\newcommand{\cp}{\widetilde{p}}

\definecolor{StoreGreen}{HTML}{5E720A}
\definecolor{SymPurple}{HTML}{460A72}
\definecolor{EqBlue}{HTML}{003587}

\newcommand{\rtype}{{\mathsf{Real}}}
\newcommand{\btype}{{\mathsf{Bool}}}

\newcommand{\pty}[1][]{\src{\rho_{#1}}}

\newcommand{\ptyp}[1][]{\src{\rho'_{#1}}}
\newcommand{\p}[1][]{\trg{\varrho_{#1}}}

\newcommand{\ps}[1][]{p_{#1}}

\DeclareDocumentCommand{\elaborate}{m m m}{ #1 : {#3} \leadsto \textcolor{targetcolor}{#2}}
\DeclareDocumentCommand{\ela}{m m m}{ #1 : #3 \leadsto #2 }
\DeclareDocumentCommand{\elalign}{m m m}{ #1 : #3 \leadsto& #2 }

\DeclareDocumentCommand{\doubleparent}{m}{\{\hspace{-6pt}\{\hspace{1pt}#1\hspace{1pt}\}\hspace{-6pt}\}}
\DeclareDocumentCommand{\d}{ m O{} }{\doubleparent{#1 ~|~ #2}}
\DeclareDocumentCommand{\ds}{ m }{\doubleparent{#1}}
\DeclareDocumentCommand{\dt}{ m}{\doubleparent{#1}}
\DeclareDocumentCommand{\db}{ m m O{} }{\dctx[#1][\doubleparent{#2 ~|~ #3}]}

\DeclareDocumentCommand{\dr}{ m m O{} }{\dctx[#1][\doubleparent{#2 ~|~ #3}]}

\DeclareDocumentCommand{\eq}{ m }{ = #1 }

\DeclareDocumentCommand{\jtf}{m m}{ {#1} \vdash {#2} }
\DeclareDocumentCommand{\sugarconj}{m m}{ {#1} \raisebox{0.5ex}{$\land$} \hspace{-8.5pt} \land {#2} }

\DeclareDocumentCommand{\egtys}{O{} O{}}{{\color{evcolor}\varepsilon^{#2}_{#1}} }
\DeclareDocumentCommand{\egtysp}{O{} O{}}{{\color{evcolor}\varepsilon'^{#2}_{#1}}}
\DeclareDocumentCommand{\eegtys}{O{} O{}}{{\color{evcolor}\varepsilon'^{#2}_{#1}}}
\DeclareDocumentCommand{\eegtysp}{O{} O{}}{{\color{evcolor}\varepsilon''^{#2}_{#1}}}
\DeclareDocumentCommand{\egtyspp}{O{} O{}}{{\color{evcolor}\varepsilon''^{#2}_{#1}}}

\DeclareDocumentCommand{\gtys}{O{} O{}}{\trg{{\sigma}^{#2}_{#1}}}
\DeclareDocumentCommand{\gtysp}{O{} O{}}{\trg{{\delta}^{#2}_{#1}}}
\DeclareDocumentCommand{\ggtys}{O{} O{}}{\trg{{\sigma'}^{#2}_{#1}}}
\DeclareDocumentCommand{\ggtysp}{O{} O{}}{\trg{{\delta'}^{#2}_{#1}}}
\DeclareDocumentCommand{\gtyspp}{O{} O{}}{\trg{{\sigma''}^{#2}_{#1}}}
\DeclareDocumentCommand{\gtya}{O{} O{}}{\trg{{\mu}^{#2}_{#1}}}
\DeclareDocumentCommand{\gtyb}{O{} O{}}{\trg{{\nu}^{#2}_{#1}}}
\DeclareDocumentCommand{\gtyap}{O{} O{}}{\trg{{\mu'}^{#2}_{#1}}}
\DeclareDocumentCommand{\gtybp}{O{} O{}}{\trg{{\nu'}^{#2}_{#1}}}
\DeclareDocumentCommand{\gtas}{O{} O{}}{{\color{blue} {\mu}^{#2}_{#1}}}
\DeclareDocumentCommand{\gtbs}{O{} O{}}{{\color{blue} {\nu}^{#2}_{#1}}}
\DeclareDocumentCommand{\gta}{O{} O{}}{\trg{{\mu}^{#2}_{#1}}}
\DeclareDocumentCommand{\gtb}{O{} O{}}{\trg{{\nu}^{#2}_{#1}}}

\DeclareDocumentCommand{\sgtys}{O{} O{}}{{\color{sourcecolor}{\sigma}^{#2}_{#1}}}
\DeclareDocumentCommand{\sgtysp}{O{} O{}}{{\color{sourcecolor}{\delta}^{#2}_{#1}}}
\DeclareDocumentCommand{\sggtys}{O{} O{}}{{\color{sourcecolor}{\sigma'}^{#2}_{#1}}}
\DeclareDocumentCommand{\sggtysp}{O{} O{}}{{\color{sourcecolor}{\delta'}^{#2}_{#1}}}
\DeclareDocumentCommand{\sgtyspp}{O{} O{}}{{\color{sourcecolor}{\sigma''}^{#2}_{#1}}}
\DeclareDocumentCommand{\sgtya}{O{} O{}}{{\color{sourcecolor}{\mu}^{#2}_{#1}}}
\DeclareDocumentCommand{\sgtyb}{O{} O{}}{{\color{sourcecolor}{\nu}^{#2}_{#1}}}
\DeclareDocumentCommand{\sgtyap}{O{} O{}}{{\color{sourcecolor}{\mu'}^{#2}_{#1}}}
\DeclareDocumentCommand{\sgtybp}{O{} O{}}{{\color{sourcecolor}{\nu'}^{#2}_{#1}}}

\DeclareDocumentCommand{\bga}{O{}}{\Gamma}

\DeclareDocumentCommand{\ga}{O{} O{}}{\gamma_{#1} {#2} }

\DeclareDocumentCommand{\tys}{O{} O{}}{\tau^{#2}_{#1}}
\DeclareDocumentCommand{\tysp}{O{} O{}}{\tau'^{#2}_{#1}}
\DeclareDocumentCommand{\tyspp}{O{} O{}}{\tau''^{#2}_{#1}}
\DeclareDocumentCommand{\tya}{O{} O{}}{T^{#2}_{#1}}
\DeclareDocumentCommand{\tyap}{O{} O{}}{T'^{#2}_{#1}}
\DeclareDocumentCommand{\tyb}{O{} O{}}{S^{#2}_{#1}}

\DeclareDocumentCommand{\vty}{O{}}{ty({#1})}
\DeclareDocumentCommand{\gtyc}{O{} O{}}{\zeta^{#2}_{#1}}
\DeclareDocumentCommand{\gtyd}{O{} O{}}{{\color{evdcolor} \xi^{#2}_{#1}}}
\DeclareDocumentCommand{\gtye}{O{} O{}}{{\color{evdcolor} \xi'^{#2}_{#1}}}
\DeclareDocumentCommand{\gtyf}{O{} O{}}{{\color{evdcolor} \xi''^{#2}_{#1}}}
\DeclareDocumentCommand{\phii}{O{} O{}}{\Phi^{#1}}
\DeclareDocumentCommand{\phi}{O{} O{}}{\Phi^{#1}_{#2}}
\DeclareDocumentCommand{\pphi}{O{}}{\Phi_{#1}}
\DeclareDocumentCommand{\pphip}{O{}}{\Phi'_{#1}}
\DeclareDocumentCommand{\pphipp}{O{}}{\Phi''_{#1}}
\DeclareDocumentCommand{\phip}{O{} O{}}{\Phi'^{#1}_{#2}}
\DeclareDocumentCommand{\phipp}{O{} O{}}{\Phi''^{#1}_{#2}}
\DeclareDocumentCommand{\rctx}{O{} O{}}{ {#1} \vartriangleright {#2}}
\DeclareDocumentCommand{\rrctx}{O{} O{}}{ {#2} }
\DeclareDocumentCommand{\sctx}{O{} O{}}{}
\DeclareDocumentCommand{\mctx}{O{} O{} O{}}{ {#1} |-{#3}}
\DeclareDocumentCommand{\dctx}{O{} O{}}{ {#1} \triangleright {#2} }
\DeclareDocumentCommand{\phty}{m m}{ #1 \triangleright #2 }
\DeclareDocumentCommand{\mmap}{O{} O{}}{\mathcal{M}^{#2}_{#1}}
\DeclareDocumentCommand{\coupling}{O{} O{} O{} O{}}{\exists {#1}.~{#2} ~ \land
~ {#3} ~ \land ~ {#1} > 0  => {#4}}
\DeclareDocumentCommand{\cp}{O{}}{C {#1}}
\DeclareDocumentCommand{\justify}{O{} O{} O{} O{}}{{#1} |- {#2} {#3} {#4}}
\DeclareDocumentCommand{\dset}{O{} O{}}{\trg{\mathcal{V}^{#2}_{#1}}}

\newcommand{\iSet}{\mathcal{I}}
\newcommand{\jSet}{\mathcal{J}}

\newcommand{\kSet}{\mathcal{K}}

\newcommand{\kSett}{\mathcal{K}}
\newcommand{\kkSett}{\mathcal{K'}}

\newcommand{\asc}[2]{#1 :: #2}
\newcommand{\reoderSym}{\stackrel{\mathclap{\normalfont\mbox{\scriptsize r}}}{=}}
\newcommand{\reoder}[2]{#1 \reoderSym #2}
\newcommand{\initReorder}[2]{#1 \, \mathbin{\begin{turn}{90}$\hspace{-0.2em}=$\end{turn}} #2}

\DeclareDocumentCommand{\lett}{m m m}{\<let> #1 = #2 \<in> #3}
\DeclareDocumentCommand{\letnl}{m m m}{\begin{block}
\<let> #1 = #2 \<in>\\ #3
\end{block}}
\DeclareDocumentCommand{\lettt}{m m m O{\set{\overline{\gtya[i]}}}}{\<let> #1 = #2 \<in>^{#4} #3}
\DeclareDocumentCommand{\add}{m m}{#1 \; {+} \;  #2}

\DeclareDocumentCommand{\spsum}{O{\ps}}{\oplus_{#1}}
\DeclareDocumentCommand{\psum}{O{ \jtf{\phi}{\p[1],\p[2]}}}{\oplus_{#1}}
\DeclareDocumentCommand{\pssum}{O{\pty}}{\src{\oplus_{#1}}}
\DeclareDocumentCommand{\ppsum}{O{\phi} O{\p[1]} O{\p[2]}}{ \trg{\tensor*[_{#2}]{\oplus}{^{#1}_{#3}}} }
\newcommand{\vjtf}[2]{#1 \Vdash #2}
\newcommand{\vjtfl}[2]{#1 \Vdash_{\leftvar} #2}
\newcommand{\vjtfr}[2]{#1 \Vdash_{\rightvar} #2}
\newcommand{\realizesSym}{\rightsquigarrow}
\newcommand{\realizes}[2]{#1 \mathrel{\realizesSym} #2}
\newcommand{\wrealizesSym}{\rightsquigarrow^{\mathsf{w}}}
\newcommand{\wrealizes}[2]{#1 \mathrel{\wrealizesSym} #2}
\newcommand{\rvdash}{\mathrel{\vdash_{\!\realizesSym}}}
\DeclareDocumentCommand{\tpsum}{O{\pty}}{\trg{\oplus_{#1}}}

\DeclareDocumentCommand{\leftvar}{}{\mathcal{l}}
\DeclareDocumentCommand{\rightvar}{}{\mathcal{r}}

\DeclareDocumentCommand{\ssum}{O{}}{\sum\limits_{#1}}

\DeclareDocumentCommand{\j}{O{}}{k_{#1}}
\DeclareDocumentCommand{\jp}{O{}}{k'_{#1}}
\DeclareDocumentCommand{\jpp}{O{}}{k''_{#1}}
\DeclareDocumentCommand{\jj}{O{}}{j_{#1}}
\DeclareDocumentCommand{\jjp}{O{}}{j'_{#1}}
\DeclareDocumentCommand{\jjpp}{O{}}{j''_{#1}}
\DeclareDocumentCommand{\kp}{O{}}{k'_{#1}}
\DeclareDocumentCommand{\kpp}{O{}}{k''_{#1}}
\DeclareDocumentCommand{\kk}{O{}}{k_{#1}}
\DeclareDocumentCommand{\kkp}{O{}}{k'_{#1}}
\DeclareDocumentCommand{\kkpp}{O{}}{k''_{#1}}
\DeclareDocumentCommand{\v}{O{} O{}}{v_{#1}}
\DeclareDocumentCommand{\vp}{O{}}{v'_{#1}}
\DeclareDocumentCommand{\vpp}{O{}}{v''_{#1}}
\DeclareDocumentCommand{\t}{O{}}{t_{#1}}
\DeclareDocumentCommand{\tp}{O{}}{t'_{#1}}
\DeclareDocumentCommand{\tpp}{O{}}{t''_{#1}}
\DeclareDocumentCommand{\m}{O{}}{m_{#1}}
\DeclareDocumentCommand{\mp}{O{}}{m'_{#1}}
\DeclareDocumentCommand{\mpp}{O{}}{m''_{#1}}
\DeclareDocumentCommand{\n}{O{}}{n_{#1}}
\DeclareDocumentCommand{\np}{O{}}{n'_{#1}}
\DeclareDocumentCommand{\npp}{O{}}{n''_{#1}}
\DeclareDocumentCommand{\w}{O{}}{w_{#1}}
\DeclareDocumentCommand{\wp}{O{}}{w'_{#1}}
\DeclareDocumentCommand{\wpp}{O{}}{w''_{#1}}
\DeclareDocumentCommand{\u}{O{}}{u_{#1}}
\DeclareDocumentCommand{\diverge}{O{} O{}}{{#1} ~\textcolor{EqBlue}{\Uparrow } ~{#2}}
\DeclareDocumentCommand{\converge}{O{} O{}}{{#1} ~\textcolor{EqBlue}{\Downarrow} ~ {#2}}
\DeclareDocumentCommand{\elab}{O{}}{ \textcolor{targetcolor}{#1} }
\DeclareDocumentCommand{\jreds}{O{} O{}}{ ~{\Downarrow_{#1}^{#2}}~ }

\DeclareDocumentCommand{\pt}{m O{1}}{#1^{#2}}

\DeclareDocumentCommand{\projv}{m}{ #1.\alpha }
\DeclareDocumentCommand{\projl}{m}{ #1.\leftvar }

\DeclareDocumentCommand{\projr}{m}{ #1.\rightvar }
\DeclareDocumentCommand{\projpi}{m}{ #1.\pi }

\newcommand{\rel}{\sim}
\newcommand{\erel}{\mathbin{\dot\sim}}
\DeclareDocumentCommand{\tagprecr}{m m m m}{ #1 \gprec #2 \vdash_{\rightvar} #3 \eprec #4 }
\DeclareDocumentCommand{\tagprecl}{m m m m}{ #1 \gprec #2 \vdash_{\leftvar} #3 \eprec #4 }
\newcommand{\eprec}{\mathrel{\dot\sqsubseteq}}

\DeclareDocumentCommand{\errort}{m}{ \textup{\textbf{error}}_{#1} }

\mathlig{|-t}{~{\vdash}~}
\mathlig{|-d}{~{\vdash}~}
\mathlig{|-ss}{~{\vdash_s}~}

\DeclareDocumentCommand{\evd}{O{} O{}}{ {\color{evdcolor} \xi^{#2}_{#1}} }
\DeclareDocumentCommand{\evdp}{O{} O{}}{ {\color{evdcolor} \xi'^{#2}_{#1}} }
\DeclareDocumentCommand{\evdpp}{O{} O{}}{ {\color{evdcolor} \xi''^{#2}_{#1}} }

\NewDocumentCommand{\relVV}{O{k} m m}{\mathcal{V}_{#1} \llbracket #2 \gprec #3 \rrbracket}
\NewDocumentCommand{\relV}{O{\denv} m m m m O{\gsmap} O{X} O{Y}}{(#2 ,#3) \in \relVV{#1}{#4}{#5}[#6]{#7}{#8}}
\NewDocumentCommand{\relVU}{O{\denv} m m m O{\gsmap}}{(#2 ,#3) \in \mathcal{V}_{#1}\llbracket #4 \rrbracket}

\NewDocumentCommand{\relTT}{O{k} m m}{\mathcal{T}_{#1} \llbracket #2 \gprec #3 \rrbracket}

\NewDocumentCommand{\relT}{O{\denv} m m m m O{\gsmap}}{(#2 ,#3) \in \mathcal{T}_{#1}\llbracket #4;#5 | #6 \rrbracket}
\NewDocumentCommand{\relTS}{O{\denv} m m m m}{\relT[#1]{#2}{#3}{#4}{#5}[S]}
\NewDocumentCommand{\relTI}{O{\denv} m m m m}{\relT[#1]{#2}{#3}{#4}{#5}[I]}

\NewDocumentCommand{\inatom}{m m m m m O{} O{} O{\gsmapp}}{(#1 ,#2) \in \mathsf{Atom}\llbracket#3;#4|#5\rrbracket}

\newcommand{\functype}[2]{#1 -> #2}

\newcommand{\nreds}[2]{~{\Downarrow^{#1}_{#2}}~}

\DeclareDocumentCommand{\snreds}{O{} m O{}}{~{^{#1}_{#3}{\Downarrow}_{s}^{#2} }~}

\newcommand{\sV}[1][]{\mathscr{V}_{#1}}

\newcommand{\V}[1][]{\trg{\mathscr{V}_{#1}}}

\newcommand{\slang}{SPLC\xspace}
\newcommand{\glang}{GPLC\xspace}
\newcommand{\tlang}{TPLC\xspace}

\newcommand{\plang}{\ensuremath{\lambda_{\oplus}}\xspace}

\NewDocumentCommand{\var}{m m}{\alpha_{#1#2}}
\NewDocumentCommand{\varr}{O{}}{\alpha_{#1}}

\NewDocumentCommand{\cww}{O{}}{\omega_{#1}}
\NewDocumentCommand{\cwwp}{O{}}{\omega'_{#1}}
\NewDocumentCommand{\cwwpp}{O{}}{\omega''_{#1}}
\NewDocumentCommand{\cw}{O{} m m}{\omega_{#1}(#2, #3)}
\NewDocumentCommand{\fcw}{O{} m m}{\omega_{#1#2#3}}
\NewDocumentCommand{\cwp}{O{} m m}{\omega'_{#1}(#2, #3)}
\NewDocumentCommand{\cwpp}{O{} m m}{\omega''_{#1}(#2, #3)}
\NewDocumentCommand{\ccw}{O{} m m}{r_{#1#2#3}}

\newcommand{\supp}{\mathsf{supp}}

\newcommand{\CC}[1][]{\mathscr{C}_{#1}}

\newcommand{\DA}{\mathscr{A}}
\newcommand{\DB}{\mathscr{B}}

\NewDocumentCommand{\cjudg}{m m m m}{#1 |- #2 #3 #4}

\NewDocumentCommand{\cjudgext}{m m m m m m m}{{#1}^{#7} |- {#2}^{#5}  ~{#3}~ {#4}^{#6}}

\newcommand{\liftT}[1]{{\color{black}\lceil} #1 {\color{black}\rceil}}
\newcommand{\liftD}[1]{{\color{black}\lceil} #1 {\color{black}\rceil}}

\NewDocumentCommand{\liftP}{m O{\cww[i]}}{{\lceil} #1 {\rceil_{#2}}}

\DeclareDocumentCommand{\ssub}{m m m}{\mathit{sub}(#1, #2, #3)}

\newcommand{\tm}[1][]{\trg{m_{#1}}}
\newcommand{\tmp}[1][]{\trg{m'_{#1}}}
\newcommand{\tn}[1][]{\trg{n_{#1}}}
\newcommand{\tnp}[1][]{\trg{n'_{#1}}}
\DeclareDocumentCommand{\tv}{O{} O{}}{\trg{v^{#2}_{#1}}}
\DeclareDocumentCommand{\tvp}{O{} O{}}{\trg{v'^{#2}_{#1}}}
\DeclareDocumentCommand{\tw}{O{} O{}}{\trg{w^{#2}_{#1}}}
\newcommand{\twp}[1][]{\trg{w'_{#1}}}
\newcommand{\tu}[1][]{\trg{u_{#1}}}
\newcommand{\tup}[1][]{\trg{u'_{#1}}}
\newcommand{\tx}[1][]{\trg{x_{#1}}}
\newcommand{\ty}[1][]{\trg{y_{#1}}}
\newcommand{\ttb}[1][]{\trg{b_{#1}}}
\newcommand{\ttr}[1][]{\trg{r_{#1}}}

\newcommand{\sm}[1][]{\src{m_{#1}}}

\newcommand{\srcn}[1][]{\src{n_{#1}}}

\newcommand{\sn}[1][]{\src{n_{#1}}}

\newcommand{\sv}[1][]{\src{v_{#1}}}
\newcommand{\svp}[1][]{\src{v'_{#1}}}
\newcommand{\sw}[1][]{\src{w_{#1}}}
\newcommand{\swp}[1][]{\src{w'_{#1}}}
\newcommand{\sx}[1][]{\src{x_{#1}}}
\newcommand{\ssb}[1][]{\src{b_{#1}}}
\newcommand{\ssr}[1][]{\src{r_{#1}}}

\newcommand{\DType}{\oblset{DType}}
\newcommand{\GDType}{\oblset{GDType}}
\newcommand{\GFDType}{\FGDType}

\newcommand{\GProb}{\oblset{GProb}}
\newcommand{\TVar}{\oblset{TVar}}

\newcommand{\FGType}{\oblset{FSType}}

\newcommand{\FGDType}{\oblset{FDType}}
\newcommand{\GFType}{\FGType}

\newcommand{\fgdt}{formula distribution types\xspace}
\newcommand{\sfgdt}{formula distribution type\xspace}
\newcommand{\fgt}{formula simple types\xspace}
\newcommand{\sfgt}{formula simple type\xspace}
\newcommand{\Fgt}{Formula simple types\xspace}
\newcommand{\Fgdt}{Formula distribution types\xspace}

\newcommand{\ft}{formula type\xspace}

\newcommand{\static}[1]{#1}

\newcommand{\satisfiablee}[1]{\mathit{sat}\left(#1\right)}

\newcommand{\FV}{\mathit{FV}}
\newcommand{\TV}[1]{\mathit{TV}\left(#1\right)}

\newcommand{\couplingLiftName}{\mathit{L}}
\DeclareDocumentCommand{\couplingLift}{m m m}{\couplingLiftName_{#3}\left(#1, #2\right)}
\DeclareDocumentCommand{\couplingLiftExt}{m m m m m m}{\couplingLiftName_{#3}^{#6}\left(#1^{#4}, #2^{#5}\right)}

\newcommand{\witnessname}{\mathit{W}}
\DeclareDocumentCommand{\witness}{m m m}{\witnessname_{#3}(#1, #2)}

\revmode{off}
\usepackage{amsmath,nicefrac}
\usepackage{leftidx}

\allowdisplaybreaks

\usepackage{tikz}
\usetikzlibrary{shapes, automata, shadows, arrows, positioning, fit,calc}

\tikzstyle{circ} = [draw, rectangle, rounded corners,  text centered, fill= white!20, text height=0.5em, text width=1.5em]
\tikzstyle{circbigger} = [draw, rectangle, rounded corners,  text centered, fill= white!20, text height=0.5em, text width=2.0em]
\tikzstyle{circinvisible} = [draw=white, circle, fill= white!0, text centered, fill= white!20, text height=1.7em, text width=2.5em, inner sep=0pt]
\tikzstyle{line} = [draw, -latex]
\tikzstyle{arrow} = [thick, ->, >=stealth]
\tikzstyle{process} = [rectangle, minimum width=3cm, minimum height=1cm, text centered, draw=black, fill=orange!30]
\tikzstyle{nodePrec} = [rectangle, minimum width=3.5cm, minimum height=0.7cm, text centered, draw=black, fill=white!30]
\tikzstyle{circc} = [circle, fill, inner sep=1.5pt]

\definecolor{darkblue}{HTML}{011480}
\definecolor{teal}{HTML}{21999D}

\lstdefinelanguage{scala}{
    morekeywords={},
    otherkeywords={=>,!,:=,+},
    sensitive=true,
    breaklines=true,
    breakatwhitespace=true,
    morecomment=[l]{//},
    morecomment=[n]{/*}{*/},
    morestring=[b]",
    morestring=[b]',
    morestring=[b]""",
    escapeinside={(*}{*)}, 
    moredelim=**[is][{\btHL}]{`}{`},
    morekeywords=[5]{def,val,if,else, data, match, with, let, in},
    keywordstyle=[5]{\bfseries\color{darkblue}},
    morekeywords=[6]{isUserActive,login,localCall,externalCall,server1,server2,new,eval,criticalTask, foo},
    keywordstyle=[6]{\color{teal}}, 
    commentstyle=\itshape\color{gray}
  } 


\begin{document}
\etocdepthtag.toc{body}   %

\journaltitle{JFP}
\cpr{Cambridge University Press}
\doival{10.1017/xxxxx}

\lefttitle{A Gradual Probabilistic Lambda Calculus}
\righttitle{Journal of Functional Programming}

\totalpg{\pageref{lastpage01}}
\jnlDoiYr{2025}

\title{A Gradual Probabilistic Lambda Calculus}

\author{Matías Toro}
\orcid{0000-0002-5315-0198} 
\affiliation{%
  \institution{Computer Science Department (DCC), University of Chile, Chile \\ %
  Millennium Institute of Foundational Research on Data (IMFD)}
  \city{Santiago}
  \country{Chile}
  \authoremail{mtoro@dcc.uchile.cl}
}

\author{Federico Olmedo}
\orcid{0000-0003-0217-6483} 
\affiliation{%
  \institution{Computer Science Department (DCC), University of Chile, Chile \\ %
  Millennium Institute of Foundational Research on Data (IMFD)}
  \city{Santiago}
  \country{Chile}
  \authoremail{folmedo@dcc.uchile.cl}
 }

\author{Wenjia Ye}
\orcid{0000-0002-3968-6201} 
\affiliation{%
  \institution{National University of Singapore}
  \city{Singapore}
  \country{Singapore}
  \authoremail{yewenjia@connect.hku.hk}
}

 \renewcommand{\thefootnote}{}
\footnotetext{This work has been partially sponsored by ANID DFG 220011,  ANID FONDECYT Iniciación 11250054, and ANID Millennium Science Initiative Program code ICN17\_002.}
\renewcommand{\thefootnote}{\arabic{footnote}}

\begin{abstract}
\rev[fo1]{Probabilistic programming languages extend ordinary programming languages with the ability to sample values from probability distributions. Type systems can establish invariants of such programs, for instance by assigning to a program a distribution over types that specifies the probability of each kind of outcome.} However, adopting such a type discipline can be cumbersome or overly conservative.

Gradual typing addresses this problem by supporting a smooth transition between static and dynamic checking, and has been successfully applied for languages with different constructs and type abstractions. Nevertheless, its benefits have never been explored in the context of probabilistic languages.

In this work, we present and formalize \glang, a gradual source probabilistic lambda calculus. \glang includes a binary probabilistic choice operator and allows programmers to gradually introduce/remove static type---and probability---annotations. The static semantics of \glang heavily relies on the notion of probabilistic couplings, as required for defining several relations, such as consistency, precision, and consistent transitivity. The dynamic semantics of \glang is given via elaboration to the target language \tlang, which features a distribution-based semantics interpreting programs as probability distributions over final values. Regarding the language metatheory, we establish that \tlang---and therefore also \glang---is \emph{type safe} and satisfies two of the so-called \emph{refined criteria} for gradual languages,  namely, that it is a conservative extension of a fully static variant and that it satisfies \rev[fo1]{(a restricted form of)} the gradual guarantee, behaving smoothly with respect to type precision. \rev[fo1]{The metatheory is mechanized in the \textsf{Lean}~4 proof assistant.}
 
\end{abstract}

\maketitle

\section{Introduction}\label{sec:intro}

In a nutshell, \emph{probabilistic programming languages} are
traditional programming languages that, on top of their regular
constructs, offer the possibility of sampling values from probability
distributions~\citep{DBLP:conf/icse/GordonHNR14,DBLP:journals/corr/abs-1809-10756}. They find applications in a wealth of different areas,
ranging from more traditional application domains such as randomized
algorithms~\citep{Motwani:1995} and cryptography~\citep{Goldwasser:1984}
to more novel application domains such as differential
privacy~\citep{DBLP:journals/fttcs/DworkR14} and machine
learning~~\citep{Ghahramani:2015,DBLP:conf/sigsoft/ClaretRNGB13}. These
latter have led to a remarkable resurgence of probabilistic
programming in the past years, with the development of a growing
number of new probabilistic programming systems~\citep{DBLP:conf/aistats/LeBW17,DBLP:conf/ilp/Pfeffer10,10.5555/3023476.3023503,dippl,DBLP:conf/aplas/Kiselyov16,DBLP:conf/iclr/TranHSB0B17}.

\rev[fo1]{Typing a probabilistic program raises questions of its own, since the
program evaluates not to a single value but to a probability distribution over
values. A typing discipline tailored to such programs carries this semantics
over to the types: with \emph{distribution
types}~\citep{Lago2017ProbabilisticTB,breuvartDalLago:ppdp2018,10.5555/3470152.3470193,dalLagoAl:popl2021},
the type of a program is not a single type but a probability
distribution over types, such as $\ds{\rtype^{\nicefrac{1}{2}},
\btype^{\nicefrac{1}{2}}}$ for a program that returns a number or a Boolean
with equal probability. Distribution types express quantitative specifications
that ordinary types cannot, but a \emph{static} discipline of distribution
types is rigid, since a type cannot leave any probability unspecified, and two
distribution types are equal only if they assign the same probability to every
simple type.
A static type system thus rejects, for instance, a conditional whose branches
have different distribution types, even if the branch that is actually traversed at runtime is compatible with the type that the context of the conditional requires.}

\rev[fo1]{\emph{Gradual typing}~\citep{siekTaha:sfp2006} has addressed this kind of problem by integrating static and dynamic typing in a single language. It does so by incorporating \emph{imprecise}
(\aka \emph{gradual}) types, which represent types only partially known at
compile time. A gradual language type-checks
programs optimistically, based on a notion of \emph{consistency} between
imprecise types, and its runtime is responsible for detecting (and reporting)
any violation of the assumptions made statically. Gradual typing has been
successfully applied to programming languages with diverse constructs and
typing disciplines, including first-class classes~\citep{takikawaAl:oopsla2012}, mutable
references~\citep{siekTaha:sfp2006,hermanAl:hosc10,Siek2015MonotonicRF,toroTanter:scp2020},
effects as primitives~\citep{banadosAl:jfp2016}, tagged and untagged
unions~\citep{toroTanter:sas2017}, labeling
operations (for reasoning about information flow)~\citep{disneyFlanagan:stop2011,fennellThiemann:csf2013,toroAl:toplas2018,azevedo:lics2020}, algebraic data types~\citep{malewskiAl:oopsla2021}, and shape checking for ML~\citep{Hattori23,migeed24}. Whether its
benefits carry over to probabilistic programming languages is, however,
an open question.}

\rev[fo1]{Probabilities change what imprecision, consistency and runtime
checking mean, in at least three respects. First, in a distribution type
imprecision can be introduced in two places: in the types, as in $\ds{\?^{\nicefrac{1}{2}}, \btype^{\nicefrac{1}{2}}}$,
and in the probabilities, as in $\ds{\rtype^{\?}, \btype^{\?}}$. Unknown
probabilities are not independent of one another, since the probabilities of a
distribution must add up to $1$, and they must therefore be tracked
symbolically, as variables subject to constraints. Second, the relations that
drive gradual typing---type equality, consistency and precision---become
relations between distributions, and lifting them from types to distributions
in a principled way is a problem in itself. For instance, consistency must
accept a program of type $\ds{\rtype^{\?}, \btype^{\?}}$ where
$\ds{\rtype^{\nicefrac{1}{2}}, \btype^{\nicefrac{1}{2}}}$ is expected. Third,
the assumptions made statically are claims about probabilities, so the runtime
checks compare distributions rather than individual values. In the previous
example, a violation arises if the program returns a number with probability
$\nicefrac{2}{3}$.}

\rev[fo1]{In this work, we show that the benefits of gradual typing \emph{do} carry
over to probabilistic programming languages. On the one hand, we design, to the
best of our knowledge, the first gradual probabilistic language, addressing the
three problems above; on the other hand, we establish for it the metatheoretic results natural to all gradual languages, some in a restricted form. Our language is a core
calculus, a simply typed lambda calculus with binary probabilistic choice, in
which a distribution type can specify the distribution of outcomes of a component,
such as the uptime of a service or the success probability of an approximate
operation. Concretely, our development unfolds in four steps.}

First, we present \slang, a probabilistic $\lambda$-calculus that
extends ordinary $\lambda$-calculus with a (binary) probabilistic
choice operator and acts as the static end of our gradual language.
It features a big-step semantics relating programs to the probability distribution of final values and to better accommodate the
derivation of the gradual variant, its type system presents some
distinguished features such as the presence of ascriptions, partial
functions $\dom$ and $\cod$ over types, and explicit type
equality. Furthermore, equality over types is semantic (instead of syntactic).

Second, we introduce \glang, our source gradual language, whose derivation from \slang is justified using the Abstracting Gradual Typing (AGT)
methodology, a systematic approach for deriving gradual languages based
on abstract interpretation~\citep{garciaAl:popl2016}. For the
so-derived notions of type consistency and \rev[fo1]{type precision}, we also provide alternative---more amenable to automation---inductive characterizations, based on the notion of probabilistic couplings~\citep{DBLP:journals/corr/abs-1103-4577}. In effect, probabilistic couplings are a fundamental ingredient behind all our technical development.

Notably, \glang allows unknown probabilities not only at the type level, but also at the term level, in probabilistic choices. This yields an increased expressivity and flexibility---characteristic of all gradual languages---and also the opportunity to leverage the language for program refinement\rev[fo1]{: one can leave the probability of a choice unknown and let the language derive a concrete value under which the underlying typing constraints hold, if such a value~exists}.

Third, we define the dynamic semantics of our gradual language by
translating \glang into the target gradual language \tlang. The runtime semantics of \tlang  incorporates the required evidence to confirm or discard the optimistic assumptions made by  \glang type system. In turn, this requires adapting gradual types to encode unknown probabilities through symbolic variables, which are constrained by well-formedness conditions.
        
Finally, to formally validate our overall language design, we establish three fundamental properties of \glang. First, we prove that it is a conservative extension of
\slang. Second, we show that it satisfies type safety. Lastly, we show
that it behaves smoothly with respect to precision, a property known as
the \emph{gradual guarantee}~\citep{siekAl:snapl2015}. \rev[fo1]{At run time, the conservative extension and the gradual guarantee hold in a restricted form, which Section~\ref{sec:target} makes precise. All results are mechanized in the \textsf{Lean}~4 proof assistant and the formal development is publicly available at \url{https://github.com/pleiad/gplc-lean}}.

Altogether, this provides the first steps toward the theoretical foundations of gradual probabilistic programming\rev[fo1]{, and a starting point for gradualizing richer probabilistic languages}.

\subsubsection*{\rev[fo1]{Differences with the conference version}}
\begin{revblock}[fo1]
This article is an extended and corrected version of our paper presented at
OOPSLA 2023~\citep{yeAl:oopsla2023}, which in the remainder we refer to as the \emph{conference
version}. Its novel contributions include: 
\begin{itemize}
\item \emph{A mechanization of the metatheory in \textsf{Lean}~4.} It covers every
  numbered result of Sections~\ref{sec:static}--\ref{sec:target}, together
  with the auxiliary lemmas of the appendix on which their proofs rest.
\item \emph{Pencil-and-paper proofs of the results.} We also include human-readable versions of the proofs in the appendix. 
\item \emph{Corrections to the language design and its metatheory.} Mechanizing
  the metatheory revealed that several results of the conference
  version do not hold as originally stated. We introduce the following main corrections. 
  \begin{itemize}
  \item We changed the concretization of gradual distribution types, so that a
    gradual entry may stand for several static entries but never for none. This
    recovers the equivalence between the AGT-derived consistency and its inductive version.
  \item The type system now manipulates formula types, which track unknown
    probabilities symbolically together with the constraints they satisfy, while
    gradual types remain as user annotations. This recovers the static gradual
    guarantee.
  \item We relax consistency and precision on runtime objects, which, unlike
    annotations, need not match entries of probability zero. This recovers the monotonicity of evidence composition, on which the dynamic gradual guarantee rests.
    \item The dynamic gradual guarantee is stated in a weaker form, for terms that
    both reduce and whose more precise term raises no error.
    \item The dynamic conservative extension is stated through an erasure of
    gradual values to static ones, assuming that both runs exist.
  \end{itemize}
  \item \emph{An improved presentation.} It includes a second motivating example, a
  more thorough discussion of the key design decisions, and additional
  explanations and examples of the central notions of the development.
\end{itemize}
\end{revblock}

\subsubsection*{\rev[fo1]{Article Organization}} The rest of the \rev[fo1]{article} is organized as follows. Section~\ref{sec:overview} discusses the motivation and some key design decisions and challenges behind our probabilistic gradual language. Section~\ref{sec:static}  presents the probabilistic lambda calculus (\slang) acting as the static end of the gradualization. Section~\ref{sec:source} develops the source gradual language (\glang) and Section~\ref{sec:target} the target language (\tlang), together with the metatheory. \rev[fo1]{Section~\ref{sec:discussion} compares the distribution semantics with a sampling-based alternative and discusses its limitations.} Section~\ref{sec:relwork} overviews the related work and Section~\ref{sec:conclusion} concludes. \rev[fo1]{The appendix collects the proofs of the main results.} 

\medskip
\noindent\textit{Note.} This article employs color coding to present information effectively and is best consumed using an electronic device or color printer.

\section{Overview}\label{sec:overview}
We next discuss the motivation behind a \emph{gradual} probabilistic language
through a concrete use case and summarize some key aspects and challenges
behind the design of our gradual probabilistic language.

\subsection{A Gradual Probabilistic Language: Why?}\label{sec:overview-why}

Assume we must develop a web application for a company, in particular,
the login endpoint. To authenticate a user, we must verify that the
user remains active in the company, information that is provided by an
external web service (exposed by a foreign library). As usual, we
support both production and development modes, where in development
mode we replace the external web service with a local function,
conveniently defined for developing and testing purposes. \rev[fo1]{The mode
is not fixed when the program is compiled: the flag \lstinline|prod| is read
from the deployment environment when the application starts.}

Under these requirements, we quickly prototype the following (untyped)
program, written in a \textsf{Scala}-like language:
\begin{lstlisting}
def isUserActive(user) = 
    if (prod) externalCall(user) 
    else localCall(user)

def login(user, pass) = 
    if (isUserActive(user)) /* test password */ 
    else false
\end{lstlisting}

\subsubsection*{Probabilistic modeling}
The company now requires that the \lstinline|login| endpoint have a
$95\%$ uptime (availability). However, after some research, we learn
that the external web service \lstinline|externalCall| has only a
$90\%$ uptime, returning a \lstinline|503 Service Unavailable| error
when down. Written in such an untyped language, the above program is
unable to capture this uptime information, let alone detect the
impossibility to comply with the login requirements.

As a first step to address this problem, we can adopt a typed
language that includes \emph{distribution types}. Loosely speaking,
distribution types represent probability
distributions of ``simpler'' types, and can crisply model uptime information. For instance, 
the return type of \lstinline|externalCall| shall now be represented by
\lstinline[mathescape]|{Bool${}^{\frac{90}{100}}$, Error503${}^{\frac{10}{100}}$}|, and the return type of
\lstinline|login| by
\lstinline[mathescape]|{Bool${}^{\frac{95}{100}}$, Error503${}^{\frac{5}{100}}$}|. Furthermore, we can implement a
\lstinline|localCall| function compatible with the uptime
requirement of the \lstinline|login| endpoint, as follows:
\begin{lstlisting}[mathescape]
def localCall(user:$\;$String): {Bool(*${}^{\frac{95}{100}}$*), Error503(*${}^{\frac{5}{100}}$*)} =  
    true (*$\psum[\frac{95}{100}]$*) error503 
\end{lstlisting}
A program of the form  \lstinline[mathescape]|m $\psum[p]$ n| is known as a \emph{probabilistic
  choice} between  \lstinline[mathescape]|m| and  \lstinline[mathescape]|n|, and behaves like  \lstinline[mathescape]|m| with probability
 \lstinline[mathescape]|p| and like  \lstinline[mathescape]|n| with probability  \lstinline[mathescape]|1-p|.\footnote{\label{fn:choice-simulation}A probabilistic
  choice \lstinline[mathescape]|m $\:\psum[p]\:$ n| can
be readily simulated (in an approximate manner) by all programming languages
that include the commonplace primitive \lstinline|random()|, returning an
(approximately) uniform value in the $[0,1]$ interval. It suffices
to take program \lstinline[mathescape]|if (random() $\;$ <= p) $\ $ m else n|.}

\subsubsection*{Limitations of static typing}
Adopting a static typing for our probabilistic language would be
cumbersome, as it would require inserting type annotations everywhere,
or else extending the language with a type inference mechanism. In either
case, the static typing can be overly conservative, rejecting (at compile time)
programs that may indeed go right at runtime. For example, declaring
the return types of functions  \lstinline|externalCall| and
\lstinline|localCall| as argued above
(\lstinline[mathescape]|{Bool${}^{\frac{90}{100}}$, Error503${}^{\frac{10}{100}}$}| and
\lstinline[mathescape]|{Bool${}^{\frac{95}{100}}$, Error503${}^{\frac{5}{100}}$}|,  respectively),
would render function \lstinline|isUserActive| ill-typed as the two
branches of the conditional in the function body would have different
types. Note that, even though the uptime of the external web service is
incompatible with the uptime requirements of the login endpoint, we
still would like to have a program that is able to execute in development mode
(and in production mode, with minor modifications in typing annotations,
if uptime requirements are reconciled). \rev[fo1]{Since the mode is only known
at runtime, no (simple) static type checker can accept the program when it runs in
development mode and reject it when it runs in production mode. Only a
runtime check can tell the two apart.}

\subsubsection*{Gradual typing to the rescue}
Gradual typing addresses this problem by supporting a smooth transition
between static and dynamic typing, introducing imprecision on
static types via the unknown annotation $\?$.\footnote{A fully untyped
  program is considered to have unknown annotations everywhere.}
Intuitively, an unknown type (resp. probability) $\?$ represents any
type (resp. probability). For example, using gradual (distribution)
types we can partially annotate the program to assert only a subset of
function uptimes:
\begin{lstlisting}[mathescape]
val externalCall: ? -> {Bool(*${}^{\frac{90}{100}}$*), Error503(*${}^{\frac{10}{100}}$*)} = ...

def isUserActive(user:$\;$?): ? = 
    if (prod) externalCall(user) :: ? 
    else localCall(user) :: ?

def login(user:$\;$?, pass:$\;$?): {Bool(*${}^{\frac{95}{100}}$*), Error503(*${}^{\frac{5}{100}}$*)} = 
    if (isUserActive(user)) ...
\end{lstlisting}
To render  \lstinline|isUserActive| well-typed, we also had to ascribe both its conditional branches to the unknown type (written \lstinline|:: ?|),  since the conditional branches have different (fully static) types.

The type checker of a gradual language treats type equality
optimistically, through the notion of \emph{consistency}. Consistency
between gradual types tests the plausibility of equality between any
of the static types they represent. For instance, gradual type
 \lstinline|? -> Bool| is consistent with
 \lstinline|Int -> ?|, written
\lstinline[mathescape]|? -> Bool $\sim$ Int -> ?|, because (during runtime) they can both
represent, \eg, the fully static type  \lstinline|Int -> Bool|.

In view of this
optimistic treatment of equality, the above program is accepted
statically as the unknown type $\?$ is (trivially) consistent with every other type. If 
the application is in development mode, then the \lstinline|login|
endpoint runs successfully. On the contrary, if the application is in
production mode, a runtime error is raised. This is because, even
though  \lstinline[mathescape]|{Bool${}^{\frac{90}{100}}$, Error503${}^{\frac{10}{100}}$} $\sim$ ?| and
\lstinline[mathescape]|? $\sim$ {Bool${}^{\frac{95}{100}}$, Error503${}^{\frac{5}{100}}$}|, 
\lstinline[mathescape]|{Bool${}^{\frac{90}{100}}$, Error503${}^{\frac{10}{100}}$} $\not\sim$ {Bool${}^{\frac{95}{100}}$, Error503${}^{\frac{5}{100}}$}|. 
Said otherwise, consistency is not transitive. Therefore, gradual
languages incorporate  runtime checks to detect any potential violation of the optimistic
assumptions performed statically during type checking.

Finally, note that we can increase the program precision, \eg, changing
the return type of \lstinline|isUserActive| from \lstinline|?|  to 
\lstinline[mathescape]|{Bool${}^{\frac{95}{100}}$, Error503${}^{\frac{5}{100}}$}|, which would make the program
ill-typed, failing thus at compile time. 
Intuitively, this is because the declared return type of \lstinline|isUserActive| would not match the inferred type \lstinline|?| of its body. 

Besides for this enhanced expressivity, one can also employ our \emph{gradual} probabilistic language for program refinement purposes. To illustrate this application, assume that the external service \lstinline|externalCall| is now required to have an uptime of \emph{at least} 95\%. This can be modeled by declaring its return type as \lstinline[mathescape]|{Bool${}^{\frac{95}{100}}$, Bool${}^{?}$, Error503${}^{?}$}|. In contrast to the above example where  \lstinline|?| represented unknown \emph{types}, here, both occurrences of \lstinline|?| represent unknown (possibly different) \emph{probabilities}, which together with \lstinline[mathescape]|$\frac{95}{100}$| must sum up to 1.

Furthermore, assume that the external service originally relied on a single server of 90\% uptime (\lstinline|server1|) to keep track of active users. To reach the desired uptime of (at least) 95\%, the service provider decides to buy a new---very costly---server of 98\% uptime (\lstinline|server2|). A naive implementation of the service would simply dispense with \lstinline|server1| and rely only on \lstinline|server2| to respond queries. However, this would negatively impact on \lstinline|server2| \rev[fo1]{lifetime}, diminishing the return of the performed investment. To avoid this problem and still benefit from \lstinline|server1|, a possible solution is to \emph{probabilistically} choose, upon each query, which server will respond to the query. The fundamental question left to answer is whether this design would result in an overall (expected) uptime of at least 95\%. To answer this question, we can consider the following program:
\begin{lstlisting}[mathescape]
val server1: ? -> {Bool(*${}^{\frac{90}{100}}$*), Error503(*${}^{\frac{10}{100}}$*)} = ...

val server2: ? -> {Bool(*${}^{\frac{98}{100}}$*), Error503(*${}^{\frac{2}{100}}$*)} = ...

def externalCall(user:$\;$?): {Bool(*${}^{\frac{95}{100}}$*), Bool(*${}^{?}$*), Error503(*${}^{?}$*)} = (*\label{ln:excall} *)
    server1(user) (*$\psum[?]$*) server2(user) (*\label{ln:pchoice} *)
\end{lstlisting}
where symbol \lstinline|?| in the probabilistic choice \lstinline[mathescape]|$\psum[?]$| (line \ref{ln:pchoice}) also represents an unknown probability. Our gradual language correctly typechecks this program and after extracting a set of constraints and checking their satisfiability, the language runtime confirms the feasibility of the proposed \lstinline|externalCall| design. 

To explain this in more detail, let us assume that when queried about any user, \lstinline|server1| (resp.\ \lstinline|server2|) responds \lstinline|true| with probability $\tfrac{90}{100}$ (resp.\ $\tfrac{98}{100}$) and \lstinline|error503| with the complementary probability. The instrumentation of the language runtime introduces symbolic variables to represent all unknown probabilities. In this case, say $\cww[\texttt{Bool}]$, $\cww[\texttt{Error503}]$ and $\cww[\texttt{choice}]$ represent the unknown probabilities encoded by \lstinline|?| respectively in \lstinline[mathescape]| Bool${}^{?}$|, \lstinline[mathescape]|Error503${}^{?}$| (line \ref{ln:excall}) and \lstinline[mathescape]|$\psum[?]$| (line \ref{ln:pchoice}). 

When invoking \lstinline|externalCall| with any user, the runtime semantics determines, after some calculations, that the result will be \lstinline|true| with probability $\frac{90}{100} \cww[\texttt{choice}] 
+ \frac{98}{100} (1 - \cww[\texttt{choice}]) = \frac{98}{100} - \frac{8}{100} \cww[\texttt{choice}]$ and 
\lstinline|error503| with probability $\frac{10}{100} \cww[\texttt{choice}] 
+ \frac{2}{100} (1 - \cww[\texttt{choice}]) = \frac{2}{100} + \frac{8}{100} \cww[\texttt{choice}]$, where the symbolic variables are constrained by formulas $\tfrac{95}{100} + \cww[\texttt{Bool}] + \cww[\texttt{Error503}] ~=~1$ (to ensure that the return type of \lstinline|externalCall| is  well-formed), $\frac{95}{100} +  \cww[\texttt{Bool}] = \frac{98}{100} - \frac{8}{100} \cww[\texttt{choice}]$ and $\cww[\texttt{Error503}] = \rev[fo1]{\frac{2}{100}} + \frac{8}{100} \cww[\texttt{choice}]$ (to ensure that the declared return type of \lstinline|externalCall| is consistent with the computed type of its body).\footnote{Formally, the instrumentation of the runtime semantics yields a handful of further constraints, but altogether they are equivalent to the considered subset.} The language runtime will confirm that these constraints are indeed satisfiable, and the invocation to \lstinline|externalCall| will then complete successfully. In effect, any probability $\cww[\texttt{choice}] \leq \frac{37.5}{100}$ yields a valid solution to the constraints, providing a concrete refinement of \lstinline|externalCall| and, more generally, validating its proposed design.

\rev[fo1]{Checking these types adds two things to simply computing the distribution of the underlying model, as an exact inference algorithm does. First, it locates violations: each ascription states a requirement, and the runtime compares the distribution at that point with it, so a violation is reported at the ascription whose requirement is unmet rather than left to be read off the final distribution. Second, it handles unknown probabilities: computing the distribution would require fixing their values, whereas the constraints that the runtime accumulates delimit the values for which every requirement holds, as the bound on $\cww[\texttt{choice}]$ above shows.}

\subsection{A Second Motivating Example}\label{sec:overview-second}
We now discuss a second motivating example inspired by the paradigm of approximate computing.

Suppose we want to develop an evaluator for basic arithmetic expressions, following a traditional, deterministic---fully precise---computational model. In a statically typed language, we might sketch the following code:
\begin{lstlisting}[mathescape]
  data Expr = Num Int | Add Expr Expr | ...
  
  def eval (e: Expr) : Int =
      match e with 
      Num n -> n
      Add e1 e2 -> eval e1 + eval e2
      ...
\end{lstlisting}
Now, consider moving to an approximate computation setting~\citep{DBLP:journals/csur/Mittal16b}. The fundamental idea behind this paradigm is to trade a small amount of accuracy for a significant gain in efficiency, in application domains that naturally tolerate approximate results or where input data are already noisy or imprecise. At the hardware level, this can be realized using approximate arithmetic circuits, whose basic operations can be modeled probabilistically~\citep{DBLP:journals/cacm/CarbinMR16,DBLP:conf/oopsla/BostonSGC15}. For example, suppose the addition operator of our underlying hardware has a success probability of only 99\%. After extending the static language with distribution types, we can model this behavior by declaring the addition operation with the following type:
\begin{lstlisting}[mathescape]
  val + : Int -> Int -> {Int(*${}^{\frac{99}{100}}$*), Error(*${}^{\frac{1}{100}}$*)} =  ...
\end{lstlisting}
Nevertheless, we will not be able to assign any valid (return) type to the \lstinline|eval| function because the probability of producing a successful result depends on the (structural) complexity of the input expression. In contrast, within a gradual language, we can assign \lstinline|eval| the following type:
\begin{lstlisting}[mathescape]
  def eval (e: Expr) : {Int${}^{?}$, Error${}^{?}$} = ...  
     /* same as before */
\end{lstlisting}

Moreover, the gradual language allows expressing partial functions that require reliable arguments. For example, suppose we need to define a \lstinline|criticalTask| that, given an arithmetic expression \lstinline|e|, should proceed only if the value of \lstinline|e| can be computed with a given confidence level, say 98\%. If the underlying task is abstracted by function \lstinline|foo|, we can write:
\begin{lstlisting}[mathescape]
  val foo: Int -> ? =  ...

  def criticalTask (e: Expr) : ? =
      foo (eval(e) :: {Int(*${}^{\frac{98}{100}}$*), Int${}^{?}$, Error${}^{?}$})
\end{lstlisting}
For any well-typed arithmetic expression \lstinline|e|, the program \lstinline|criticalTask(e)| is also well-typed as \lstinline[mathescape]|{Int${}^{?}$, Error${}^{?}$} $\sim$ {Int${}^{\frac{98}{100}}$, Int${}^{?}$, Error${}^{?}$}|. At runtime, the language will invoke \lstinline|foo| if \lstinline|e| is ``simple enough,'' for example a constant, and can be reduced successfully with (at least) a 98\% probability. Otherwise, it will raise a runtime error.

\subsection{Design Decisions and Challenges}\label{sec:overview-design}
When designing our source (\glang) and target (\tlang) gradual
probabilistic languages, we faced several design decisions and challenges.

\subsubsection*{Where to introduce imprecision}
In most traditional gradual typing calculi, imprecision is introduced via the unknown type
$\?$. To gain expressivity and flexibility, in this work we allow
imprecision at the type level as well as the probability level. For
example, given the fully static distribution type
\lstinline[mathescape]|{Bool${}^{\frac{9}{10}}$, Error503${}^{\frac{1}{10}}$}|, we can introduce imprecision either in
probabilities, \eg \lstinline[mathescape]|{Bool${}^{\gbox{\?}}$, Error503${}^{\frac{1}{10}}$}|, in the underlying types, \eg  \lstinline[mathescape]|{${\Gbox{\?}}^{\frac{9}{10}}$,  Error503${}^{\frac{1}{10}}$}|, or in both. Note that there is no
need to introduce the unknown distribution as it can be represented by
the gradual distribution type
\lstinline[mathescape]|{${\?}^{\?}$}|. Interestingly, unknown probabilities are
particularly useful for expressing \emph{probability bounds}. As hinted above, we can use type
\lstinline[mathescape]|{Bool${}^{\frac{95}{100}}$, Bool$^?$, Error503$^?$}| to represent a service with an uptime of \emph{at least}
95\%. On the other hand, type
\lstinline[mathescape]|{Bool$^?$,  Error503${}^{\frac{5}{100}}$, Error503$^?$}| models an uptime of \emph{at most} 95\%.

\subsubsection*{Tracking dependencies of probability annotations}
When dealing with unknown probabilities, as in type
\lstinline[mathescape]|{Bool${}^{\frac{9}{10}}$, Bool${}^{\?}$, Error503${}^{\?}$}|, the gradual language must ensure that
the concrete probabilities they represent induce only
well-defined static distribution types, with a total probability of
$1$. This requirement induces implicit dependencies and gradual
probabilities are thus elaborated to fresh variables ($\cww$)
constrained by formulas, \eg of the form $\frac{9}{10} + \cww[1] + \cww[2] = 1$.

\subsubsection*{Ascribing to distribution types} One of the
fundamental features of \glang is the possibility of \rev[fo1]{ascribing
distribution types to programs}. For example, a program $f  = \asc{
  (\lambda x : \? . x)}{ \{ (\rtype-> \?)^{\frac{1}{2}}, (\? ->
  \btype)^{\frac{1}{2}} \} }$ behaves as a function that takes a
number as argument with probability $\tfrac{1}{2}$, and as a function
that returns a Boolean also with probability $\tfrac{1}{2}$.\footnote{\label{fn:dist-asc}\rev[fo1]{For
a \emph{value}, as here, ascribing a distribution type is equivalent to a
probabilistic choice between simple ascriptions, namely
$(\asc{(\lambda x : \? . x)}{\rtype -> \?}) \spsum[\nicefrac{1}{2}]
(\asc{(\lambda x : \? . x)}{\? -> \btype})$, which has the same type and
reduces to the same distribution. The equivalence does not extend to arbitrary
terms, for three reasons. First, simple ascriptions apply to values only.
Second, a distribution ascription \emph{checks} a claim about the
probabilities, whereas a choice \emph{generates} them; for instance,
$\asc{(1 \spsum[\nicefrac{1}{3}]
\mathit{true})}{\{\rtype^{\nicefrac{1}{2}}, \btype^{\nicefrac{1}{2}}\}}$ is
ill typed. Third, a distribution ascription evaluates its term once and routes
each outcome to a compatible entry of the type, whereas a choice between
ascriptions of the same term selects the ascription independently of the
outcome.}} Reducing
an application to $f$ and correctly propagating the respective type
information is not a trivial task. Intuitively, our approach consists
in ``pushing'' the real argument into each (compatible) type in the
distribution. For instance, the reduction of  program $f \;1$
proceeds, informally, as follows: $$f \;1 ~\mapsto^{*}~
\{\asc{(\lambda x : \? . x)}{(\rtype-> \?)^{\frac{1}{2}}}\;1,
\asc{(\lambda x : \? . x)}{(\? -> \Bool)^{\frac{1}{2}} \;1\} }
~\mapsto^{*}~ \{ 1::\?^{\frac{1}{2}}, \error^{\frac{1}{2}} \}$$ 

\subsubsection*{\rev[fo1]{Why a distribution semantics}}
\begin{revblock}[fo1]
The runtime semantics of \tlang computes the full distribution of values a
program reduces to rather than sampling one execution at a
time. This choice is forced by ascriptions to distribution
types: such an ascription claims something about the \emph{aggregate} probabilities of
the outcomes, and a single execution path cannot refute it. Consider the program
\begin{lstlisting}
let x: ? = 1 (*$\psum[\frac{1}{3}]$*) (2 (*$\psum[\frac{1}{3}]$*) true) in
let y: {Int(*${}^{\frac{1}{3}}$*), Bool(*${}^{\frac{2}{3}}$*)} = x in
y
\end{lstlisting}
where \lstinline|x| reduces to \lstinline|1|, \lstinline|2| and \lstinline|true|
with probabilities $\frac{1}{3}$, $\frac{2}{9}$ and $\frac{4}{9}$, respectively.
Each path, taken on its own, fits the type declared for \lstinline|y|:
$\frac{1}{3} \leq \frac{1}{3}$, $\frac{2}{9} \leq \frac{1}{3}$ and
$\frac{4}{9} \leq \frac{2}{3}$. Only the aggregate reveals the violation: the
integers together have probability $\frac{5}{9}$, not $\frac{1}{3}$. A sampling semantics, which follows one path and its probability, thus
cannot detect this violation, whereas the distribution semantics does.
\end{revblock}

\subsubsection*{\rev[fo1]{Which notion of termination}}
\begin{revblock}[fo1]
Probabilistic programs admit two notions of
termination~\citep{kaminski2019advanced,Kaminski:AI}. A program terminates
\emph{certainly} when there exists a finite bound for the length of all its
executions. It terminates \emph{almost surely} when (the length of its executions may grow arbitrarily large, but) its infinite executions have probability zero, \ie when it
terminates with probability $1$. The standard example of the gap between the
two notions is the program that flips a fair coin until it shows heads and
returns the number of flips. Every number of flips has positive probability (geometrically distributed), and the single execution that
only ever sees tails is infinite, with probability zero.

Both classes of programs can be written in our gradual language. Recursion is not
primitive, but the unknown type allows a function to be applied to itself, as
in the untyped lambda calculus, and the coin-flipping program then reads
\begin{lstlisting}[mathescape]
def flips(self: ?): Int =
    let heads: Bool = true (*$\psum[\frac{1}{2}]$*) false in
    if (heads) 1 else 1 + self(self)

flips(flips)
\end{lstlisting}
The program is well typed, since the type of \lstinline|self| is consistent
with any function type.

The distribution semantics of our gradual language, however, classifies as terminating only the programs that terminate certainly. It
reduces a program by reducing both branches of every probabilistic choice and
combining their distributions, so it assigns a result exactly to the programs
whose executions form a finite tree, and it treats a program that terminates
only almost surely, like the one above, as divergent. The metatheory
inherits this choice: type safety, the dynamic gradual guarantees and the dynamic conservative extension speak only about reductions that exist, so
for programs that terminate almost surely they hold trivially. Extending the checks, and with them the results, to almost-sure termination is left as future work (Section~\ref{sec:discussion-limits}).
\end{revblock}

\subsubsection*{Couplings as a central tool}
Defining some key relations between distribution types is another
technical challenge. For instance, should we consider 
distribution type 
$\{(\rtype -> \?)^\frac{1}{2}, (\? -> \rtype)^{\frac{1}{2}}\}$
consistent with 
$\{(\? -> \btype)^\frac{1}{3}, (\rtype -> \?)^{\frac{2}{3}}\}$?
Is 
$\{ \rtype^{\frac{1}{2}}, \?^{\frac{1}{2}} \}$
more precise than  
$\{ \btype^{\frac{2}{3}}, \?^{\frac{1}{3}} \}$?
To define these (and other) relations over distribution types we heavily rely on the notion of probabilistic coupling\rev[fo1]{~\citep{DBLP:journals/corr/abs-1103-4577}}, which yields a canonical lifting from relations over pairs of sets to probability distributions over the sets. 

\section{\slang: Static Language}\label{sec:static}
In this section, we present \slang, a statically-typed lambda calculus,
extended with a probabilistic choice operator, which is the starting
point---static end---of our gradualization effort. The static semantics
of \slang is based on that of \plang from~\citep{Lago2017ProbabilisticTB},
with two major differences: \slang features a semantic (rather
than syntactic) equality between types and also allows type
ascriptions. As for the dynamic semantics, programs are interpreted as probability distributions over final values.

\subsection{Syntax}

\begin{figure}[t]  
\begin{displaymath}
\begin{array}{r@{\hspace{0.3em}}c@{\hspace{0.8em}}l@{\hspace{2em}}l}
   \multicolumn{4}{c}{r\in \mathbb{R}, \quad b\in \mathbb{B}, \quad x \in
  \text{Var}, \quad \ps \in [0,1], \quad \tys \in \Type,  \quad \tya \in \DType}\\[1.5ex]
\tys       & ::=        & \rtype ~|~ \btype ~|~ \tys -> \tya   &
                                                                   \text{(simple types)} \\
  \tya        & ::=        & \d{\tys[i][\ps[i]]}[i \in \iSet] & \text{(distribution types)} \\[1.5ex]
\m, \n          
        & ::= & v ~|~ v \; w ~|~ \lett{x}{\m}{\n} ~|~ \m \spsum \n                    & \text{(terms)}\\ 
  &     & \asc{m}{\tya} ~|~ \asc{v}{\tys} ~|~ \ite{v}{m}{n} ~|~ \add{v}{w}  & \\
  v,w 
        & ::= & x ~|~ r ~|~ b ~|~ (\lambda x:\tys. \m) & \text{(values)}\\
\end{array}
\end{displaymath}
\caption{Syntax of \slang.}
\label{fig:static-syntax}
\end{figure}

The syntax of \slang is presented in Figure~\ref{fig:static-syntax},
comprising its type and term languages.

\subsubsection*{Type language}
The type language contains two (mutually defined) syntactic
categories: simple types ($\Type$) and distribution types ($\DType$). A \emph{simple
  type}, ranged over by $\tys$, can be the type $\rtype$ of real
numbers, the type $\btype$ of Boolean values, or a function type of
the form $\tys -> \tya$, where $\tya$ is a distribution type.  A
\emph{distribution type}, ranged over by $\tya$, is a multi-set of
pairs comprised of a simple type $\tys$ and a probability $\ps$ in the
interval $[0,1]$. Intuitively, we use
$\d{\tys[i][\ps[i]]}[i \in \iSet]$ to denote a distribution type in which
simple type $\tys[i]$ occurs with probability $\ps[i]$, for each $i$
in the (non-empty and finite) subset $\iSet$ of the natural numbers. For
instance, distribution type
$\dt{\rtype^{\frac{1}{4}}, \rtype^{\frac{1}{4}},
  \btype^{\frac{1}{2}}}$ represents $\rtype$ with probability
$\frac{1}{4}+\frac{1}{4}=\frac{1}{2}$ and $\btype$ with probability
$\frac{1}{2}$. Notationwise, we sometimes omit the index set $\iSet$
and simply write $\ds{\tys[i][\ps[i]]}$. %
Finally, note that distribution types---as the name suggests---represent
\emph{probability distributions} (over simple types) and therefore,
well-typed programs are associated distribution types whose
probabilities sum up to $1$ (this restriction is formally captured by the
notion of \emph{type well-formedness} defined in Section~\ref{sec:SPLC-typeSystem}).

\subsubsection*{Term language}\label{sec:static_type_system}
Terms, ranged over by $\m, \n$, and values, ranged over by $v, w$, are mutually defined. 
A \emph{term} can be a value $v$, an application $v\;w$ between two
values, a let expression $\lett{x}{\m}{\n}$, a probabilistic choice $\m
\spsum \n$,  a term ascription $\asc{m}{\tya}$, a value ascription
$\asc{v}{\tys}$, a conditional $\ite{v}{m}{n}$, or an addition
$\add{v}{w}$ between two values. Note that terms are defined in
\rev[fo1]{monadic normal form~\citep{Moggi1991NotionsOC}, that is, the term bound by a $\<let>$ is an arbitrary computation,} which pushes all the
reasoning about probabilities to the $\<let>$construct. Randomization is introduced through probabilistic
choices: program  $\m \spsum \n$ behaves like (\ie reduces to) $m$ with
probability $p$ and like $n$ with probability $1-p$. Finally, a \emph{value} can be a variable $x$, a real number $r$, a Boolean value $b$, or a lambda abstraction $\lambda x:\tys. \m$.

\subsection{Type System}
\label{sec:SPLC-typeSystem}

\begin{figure}[t]
\begin{flushleft}
  \framebox{$\Gamma |-ss \v : \tys, \quad \Gamma |-ss \m : \tya, \quad \Gamma |-ss \sV : \tya $}
\end{flushleft}
  \def \MathparLineskip {\lineskip=1.4ex}
  \begin{mathpar}
  \inference[(Tr)]{}
  {\Gamma |-ss r : \rtype} \and
  \inference[(Tb)]{}
  {\Gamma |-ss b : \btype} \and
  \inference[(Tx)]{\Gamma(x) = \tys}
  {\Gamma |-ss x : \tys} \and
  \inference[(Tv)]{\Gamma |-ss v : \tys}
  {\Gamma |-ss v : \ds{\tys[][1]} } \and
  \inference[(T$\lambda$)]{\Gamma, x:\tys |-ss \m : \tya & \jtf{}{\tys}}
  {\Gamma |-ss \lambda x:\tys. \m : \tys -> \tya} \and
  \inference[(T$::\tys$)]{\Gamma |-ss v : \tysp &  \tysp =_{s} \tys & \jtf{}{\tys} }
  {\Gamma |-ss \asc{v}{\tys} : \ds{\tys[][1]}}  \and
  \inference[(Tapp)]{
    \Gamma |-ss v : \tys[1]  &  
    \Gamma |-ss w :  \tys[2] & \dom(\tys[1]) =_{s} \tys[2] 
  }
  {\Gamma |-ss v\;w : \cod(\tys[1]) } \and
  \inference[(T$\oplus$)]{
    \Gamma |-ss \m : \tya[1]  &  
    \Gamma |-ss \n : \tya[2] & 
  }
  {\Gamma |-ss { \m} \spsum { \n} :  \ps \cdot \tya[1] + (1- \ps) \cdot \tya[2]} \and %
  \inference[(Tlet)]{
    \Gamma |-ss \m : \d{\tys[i][\ps[i]]}[i \in \iSet]  \\
    \forall i\in\iSet.~\Gamma, x : \tys[i] |-ss \n : \tya[i]
  }
  {\Gamma |-ss \lett{x}{\m}{\n} : \sum_{i \in \iSet} \ps[i] \cdot \tya[i]} \and
  \inference[(T$::\tya$)]{\Gamma |-ss \m : \tyap & \tyap =_{s} \tya & \jtf{}{\tya}
  } 
  {\Gamma |-ss \asc{ \m }{\tya} : \tya}\and
  \inference[(T$+$)]{
    \Gamma |-ss v : \tys[1]   &  \tys[1] =_{s} \rtype \\
    \Gamma |-ss w :  \tys[2]  &  \tys[2] =_{s} \rtype \\
  }
  {\Gamma |-ss \add{v}{w} : \ds{\rtype^1} } \and
  \inference[(Tif)]{
    \Gamma |-ss v : \tys  & \tys =_{s} \btype  \\
    \Gamma |-ss m : \rev[fo2]{\tya[1]} & 
    \Gamma |-ss n : \rev[fo2]{\tya[2]} & \rev[fo2]{\tya[1] =_{s} \tya[2]}
  }
  {\Gamma |-ss \ite{v}{m}{n} : \rev[fo2]{\tya[1]} } \and
  \inference[(V)]{
    \forall i \in \iSet. |-ss \v[i] : \tys[i] 
      }
      {\Gamma |-ss \d{ \pt{\v[i]}[\ps[i]]}[i \in \iSet]  : \d{ \tys[i][\ps[i]] }[i \in \iSet] } 
  \and
  \end{mathpar}\\[1.5ex]
  \begin{tabular}{@{}l@{\hspace{1em}}l@{\hspace{1.5em}}l@{}}
    $\dom:\Type \rightharpoonup \Type$ & $\cod:\Type \rightharpoonup \DType$ &   $\cdot: [0,1] \times\DType \rightarrow \DType$  \\
    $\dom(\tys -> \tya) = \tys $ & $\cod(\tys -> \tya) = \tya $ & $\ps \cdot \d{\tys[i][\ps[i]]}[i \in \iSet] =
    \d{\tys[i][\ps\cdot\ps[i]]}[i \in \iSet] $ \\
    $\dom(\tys)~\text{undef. otherwise} $  &  $ \cod(\tys)~\text{undef. otherwise}$ \\
  \end{tabular}\\[1.5ex] 
  \begin{tabular}{l}
 $+:\DType \times \DType \rightharpoonup \DType$ \\
 $\d{\tys[i][\ps[i]]}[i \in \iSet] + \d{\tys[j][\ps[j]]}[j \in \jSet]
    = \d{\tys[i][\ps[i]]}[i \in \iSet] \union \d{\tys[j][\ps[j]]}[j
    \in \jSet] \quad \text{if~} 
     \ssum[i \in \iSet] \ps[i] + \ssum[j \in \jSet] \ps[j] \leq 1$ 
  \end{tabular} 
  \caption{Type system of \slang.}
  \label{static-typing}
\end{figure}

Figure~\ref{static-typing} presents the type system of \slang. Type
rules are defined using a pair of mutually-defined judgments: one for
values and another for computations. Judgment $\Gamma |-ss
\v : \tys$ (resp. $\Gamma |-ss \m : \tya$)  for values (resp. 
computations) denotes that value $\v$ (resp. term $\m$) has simple
type $\tys$ (resp. distribution type $\tya$) under type environment
$\Gamma$, which maps variables to simple types.

Type rules for values are standard, with only a few rules deserving special
attention. For example, rule (Tv) allows assigning a value of
simple type $\tys$ also distribution type $\ds{\tys[][1]}$ (\eg
program $1$ can be typed as $\ds{\Int^1}$). Also, note that rules
(T$\lambda$), (T$::\tys$) and (T$::\tya$) require all program type
annotations to be well-formed. We say that a distribution type
$\d{\tys[i][\ps[i]]}[i \in \iSet]$ is \emph{well-formed}, written
$|- \d{\tys[i][\ps[i]]}[i \in \iSet]$, if
$\sum_{i \in \iSet}\ps[i] = 1$ and simple type $\tys[i]$ is
well-formed for every $i \in \iSet$. A simple type $\tys$ is
\emph{well-formed}, written $|- \tys$, if it is either a base type ($\rtype$
or $\btype$) or a function type $\tys -> \tya$, where $\tys$ and
$\tya$ are well-formed. \rev[fo2]{An environment $\Gamma$ is
\emph{well-formed}, written $|- \Gamma$, if every simple type it assigns to a
variable is well-formed.}

A particularity of {\slang}'s type system is that it relies on a
semantic---rather than syntactic---notion of type equality
($=_{s}$), as used in rules (T$::\tys$), (Tapp), \rev[fo2]{(Tif)} and (T$::\tya$). For
example,
$\dt{\rtype^{\colora{\frac{1}{2}}},\btype^{\colorb{\frac{1}{4}}},
  \btype^{\colorb{\frac{1}{4}}} } =_{s}
\dt{\rtype^{\colora{\frac{1}{3}}},\rtype^{\colora{\frac{1}{6}}},\btype^{\colorb{\frac{1}{2}}}}$
because
${\colora{\frac{1}{2}}} = {\colora{\frac{1}{3}} } +
{\colora{\frac{1}{6}}}$ and
${\colorb{\frac{1}{4} + \frac{1}{4} }} =
{\colorb{\frac{1}{2}}}$. Formally, type equality is given by rules:
\begin{mathpar}%
\inference{}{\!\rtype =_{s} \rtype \!} \!\!\and\!\!
\inference{}{\!\btype =_{s} \btype \!} \!\!\and\!\!
\inference{\!\tys[1] =_{s} \tys[2] & \tya[1] =_{s} \tya[2]\! }
{\!\tys[1] -> \tya[1] =_{s} \tys[2] -> \tya[2]\!} \!\!\and\!\!
\inference{\!
  \forall \tys \in \supp(\tya[1]) \cup \supp(\tya[2]).\ \tya[1](\tys) = \tya[2](\tys) \\
  \rev[mt1]{\forall \tys[i][\ps[i]] \in \tya[1].\ \exists \tys[j][\ps[j]] \in \tya[2].\ \tys[i] =_{s} \tys[j]} &
  \rev[mt1]{\forall \tys[j][\ps[j]] \in \tya[2].\ \exists \tys[i][\ps[i]] \in \tya[1].\ \tys[i] =_{s} \tys[j]} \!}
{ \tya[1] =_{s} \tya[2] }
\end{mathpar}

\noindent where $\supp(\tya)$ represents the \emph{support}  of distribution
type $\tya$ defined by $\supp(\tya) = \{\tys ~|~ \tys[][\ps[]] \in \tya \land 
\ps>0\}$ and  $\tya(\tys)$ represents the \emph{probability} that $\tya$ assigns to a simple
type $\tys$, defined by $\d{\tys[i][\ps[i]]}[i \in \iSet](\tys) =
\sum_{i \in \iSet | \tys[i] =_{s} \tys} \ps[i]$.

\rev[fo1]{The last rule has two kinds of premises. The first one compares
probabilities, that is, both types must give the same total probability to
every class of equal simple types. The other two, which we call
\emph{coverage} premises, compare the entries themselves, that is, every entry
of one type must have an $=_{s}$-equal entry in the other, whatever
probability it carries. The first premise does not imply them, because the
support omits the entries of probability $0$.

If instead the rule had only its first premise, then it would equate
$\dt{\rtype^{1}}$ and $\dt{\rtype^{1}, \btype^{0}}$, and replacing a type
annotation by an equal one could turn a well-typed program into an ill-typed
one. The program
\[ \lett{x}{\asc{1}{\dt{\rtype^{1}}}}{\add{x}{1}} \]
type-checks, whereas the program
\[ \lett{x}{\asc{1}{\dt{\rtype^{1}, \btype^{0}}}}{\add{x}{1}} \]
does not, since rule (Tlet) type-checks the body once for each entry of the
type of the bound term, including $x : \btype$.}

Following the approach of~\citet{garciaAl:popl2016}, to ease the gradualization process we make all type relations and type functions explicit.
For the (Tapp) rule, we use partial functions $\dom$ and $\cod$ to extract
the domain and codomain of a function type, respectively. Also we make
explicit the fact that the type of the argument should be equal to the domain type of the function.
Rule (T$\oplus$) combines the distribution types $\tya[1]$ and
$\tya[2]$ of sub-expressions $m$ and $n$, by first scaling $\tya[1]$ by
$\ps$ and $\tya[2]$ by $1 - \ps$, and then adding the resulting 
scaled distribution types together. Scaling $\ps\cdot\tya$ is defined
pointwise, \ie by scaling all the probabilities in the distribution
type by $\ps$; the addition of two (sub) distribution types $\tya[1] +
\tya[2]$ is defined as the union of the two multi-sets, provided that
the sum of the resulting probabilities does not exceed $1$. For instance, consider program $1 \oplus_{\frac{1}{3}} \ttt$.
Expression $1$ is typed as $\dt{\rtype^{1}}$ and $\ttt$ as  
$\dt{\btype^{1}}$. After scaling both distribution types and adding them together, the resulting distribution type is
$\frac{1}{3}\cdot\dt{\rtype^{1}} + \frac{2}{3}\cdot\dt{\btype^{1}} = \dt{\rtype^{\frac{1}{3}}, \btype^{\frac{2}{3}}}$.
Rule (Tlet) propagates the type of $m$ to $n$ as follows.
If $m$ has a distribution type $\tya$, then for each type and probability $\tys[][\ps] \in \tya$, $n$ is type-checked under an extended environment where $x$ is typed as $\tys$. The resulting type of the let expression is computed by adding each distribution type of $n$ scaled by its corresponding $\ps$. 
\extends{Rule (T+) ensures that addition is performed between two numeric values.
Rule (Tif) requires \rev[fo2]{the types of the two branches of the conditional
expression to be equal up to $=_{s}$, and gives the conditional the type of
the first branch.}}
Finally, rule (V) is discussed in the next section as it assigns types to so-called distribution values, which are the result of evaluating terms. 

Note that, as expected, well-typed terms are assigned well-formed types, only.
\begin{lemma}[Type well-formedness]\label{lemma:welltyped-wellform-static2}\mechnote{\leanname{SPLC.wf\_val} and \leanname{SPLC.wf\_tm}.}
\rev[fo1]{For every well-formed environment $\Gamma$, every value $\v$ and
every term $\m$ of \slang{}, every simple type $\tys \in \Type$ and every
distribution type $\tya \in \DType$,}
  \begin{enumerate}
  \item If $\:\Gamma |-ss \v : \tys$, then $\justify{\tys}$.
  \item If\ $\:\Gamma |-ss \m : \tya$, then $\justify{\tya}$.
  \end{enumerate}
\end{lemma} 

\subsection{Dynamic Semantics}\label{sec:static-dynamics}

\begin{figure}[t]
  \begin{flushleft}
    \framebox{$ m \snreds{}  \sV $}
    \end{flushleft}
    \vspace{-2em}
    \begin{displaymath}
      \begin{array}{r@{\hspace{0.3em}}c@{\hspace{0.8em}}l@{\hspace{0.8em}}l}
         \sV & ::= &  \dt{ \pt{\v[i]}[\ps[i]] \mid i \in \iSet} & \text{(distribution values)}
      \end{array}
    \end{displaymath}\\
  \def \MathparLineskip {\lineskip=1.4ex}    
\begin{mathpar}
\mprset{andskip=1em plus 0.5fil minus 0.5em}%
\inference{ }
  { v \snreds{} \dt{ \pt{v} } } \and 
\inference{ m[v/x]  \snreds{}    \sV }
{(\lambda x: \tys. m)\; v \snreds{}    \sV } \and
\inference{{m} \snreds{} \sV[1] &  {\rev[mt1]{n}} \snreds{} \sV[2] }
{{m} \spsum {n} \snreds{} \ps \cdot \sV[1] + (1- \ps) \cdot \sV[2] }\and
\inference{ m \snreds{}   \dt{\pt{\v[i]}[\ps[i]] \mid i \in \iSet } \and
 \forall i \in \iSet.\ n[ \v[i] /x] \snreds{}   \sV[i]
}
{\lett{x}{m}{n} \snreds{}     \ssum[i \in \iSet] \ps[i] \cdot \sV[i] }
 \and
\inference{ }
{ v :: \tys \snreds{} \dt{ \pt{v} } } \and 
\inference{ m \snreds{}  \sV }
{ m :: \tya \snreds{}  \sV } \and 
\inference{r_3 = r_1 + r_2}
{ \add{r_1}{r_2} \snreds{} \dt{ \pt{r_3} } } 
\and 
\inference{\m \snreds{}  \sV  }
{ \ite{\true}{\m}{\n} \snreds{}  \sV} 
\and 
 \inference{\n \snreds{}  \sV  }
{ \ite{\false}{\m}{\n} \snreds{}  \sV
 } 
\end{mathpar}
\caption{Runtime semantics of \slang.}
\label{fig:static-distribution-reduction}
\end{figure}

We endow \slang with a big-step \emph{distribution-based} semantics that relates programs to probability distributions over final values~\citep{DBLP:journals/ita/LagoZ12}, following a call-by-value reduction strategy. Concretely, judgment $\m \snreds{} \sV$ denotes that expression $\m$ reduces to a \emph{distribution value} $\sV$, \ie a probability distribution over values. The reduction relation is formally defined in Figure~\ref{fig:static-distribution-reduction}.

 A value $\v$ reduces to a Dirac distribution, \ie a distribution that assigns probability $1$ to $\v$ (and 0 to any other value). A function application reduces by substituting the argument for the variable binder in the function body. 
 A probabilistic choice first reduces its pair of branches and then returns the weighted sum of the so-obtained distribution values. The scaling and addition operators for distribution values are defined analogously to those for type distributions (Figure~\ref{static-typing}).
  For example, program $(1 \oplus_{\frac{1}{2}} 2)
  \oplus_{\frac{2}{3}} \ttt$ reduces to distribution value 
  $\dt{ 1^{\frac{1}{3}}, 2^{\frac{1}{3}}, \ttt^{\frac{1}{3}} }$.
The reduction of a $\<let>\!$--expression $\lett{x}{\m}{\n}$ is more involved and proceeds as follows. First, subterm $\m$ is reduced to a distribution value $\ds{\pt{\v[i]}[\ps[i]] \mid i \in \iSet }$. Second, subterm $\n$ is reduced by substituting each $v_i$ (\ie each possible outcome of $m$) for $x$, resulting in distribution values $\sV[i]$.  The entire $\<let>\!$--expression then reduces to the weighted sum $\ssum[i \in \iSet] \ps[i] \cdot \sV[i]$.  %
 Finally, ascribed terms reduce by removing type ascription.

To develop the forthcoming metatheory, we need to assign types not only to \slang programs but also to distribution values. To this end, we follow  rule (V) from Figure~\ref{static-typing}: distribution value $\d{ \pt{\v[i]}[\ps[i]]}[i \in \iSet]$ is assigned distribution type $\d{ \tys[i][\ps[i]] }[i \in \iSet]$ provided each $\v[i]$ is assigned simple type $\tys[i]$ (under an empty environment).

\rev[fo1]{As expected, \slang is \emph{type safe}, that is, every well-typed
closed program reduces to a distribution value whose type is equal to the
type of the program.}

\begin{restatable}[{\rev[fo1]{Type safety for \slang}}]{theorem}{staticsafety}\label{theorem:static-safety}\mechnote{\leanname{SPLC.type\_safety}.}
  \rev[fo1]{For every closed term $\m$ of \slang{} and every distribution
  type $\tya \in \DType$, if $|-ss \m : \tya$, then $\m \snreds{} \sV$ for some
  distribution value $\sV$, and $|-ss \sV : \tyap$ for some distribution type
  $\tyap \in \DType$ with $\tyap =_{s} \tya$.}
\end{restatable}

\rev[fo1]{\noindent For instance, the program
$(1 \oplus_{\frac{1}{2}} 2) \oplus_{\frac{2}{3}} \ttt$ above has type
$\ds{\rtype^{\frac{1}{3}}, \rtype^{\frac{1}{3}}, \btype^{\frac{1}{3}}}$,
and the distribution value it reduces to assigns probability $\frac{2}{3}$
to reals and $\frac{1}{3}$ to booleans, as the type prescribes.}

\section{\glang: Gradual Source Language}\label{sec:source}
\rev[fo1]{In this section, we present \glang, our gradual source probabilistic
language. We start with its syntax and type system, followed by consistency
and precision between its types, together with the well-formedness of type
annotations. Finally, we establish the two refined criteria for \glang: the
static gradual guarantee, and that its type system conservatively extends that
of \slang. The dynamic semantics of \glang is defined through an elaboration
to a target language, presented in Section~\ref{sec:target}.}\footnote{In the remainder, we use  \textcolor{sourcecolor}{blue color} for source languages (\glang) and \textcolor{targetcolor}{red color} for target languages (\tlang).}

\subsection{Syntax} 
\label{sec:GPCL-syntax}
\begin{figure}[t]
\begin{displaymath}
\begin{array}{r@{\hspace{0.3em}}c@{\hspace{0.8em}}l@{\hspace{2em}}l}
  \multicolumn{4}{c}{\ssr\in \mathbb{R},\quad \ssb \in \mathbb{B}, \quad \sx \in \text{Var}, \quad \ps \in [0,1],\quad  \pty \in \GProb, \quad \sgtys \in \GType, \quad \sgtya \in \GDType}\\[1.5ex]
  \pty     & ::= & \ps ~|~ \Gbox{\?}  &  \text{(gradual probabilities)}\\
  \sgtys,\sgtysp        & ::=        & \rtype ~|~ \btype ~|~ \sgtys ->
                                       \sgtya ~|~ \Gbox{\?}  &
                                                               \text{(gradual simple types)} \\
  \sgtya,\sgtyb        & ::=        & \d{\sgtys[i][\pty[i]]}[i \in
                                      \iSet] & \text{(gradual distribution types)} \\[1.5ex]
  \sm,\srcn           
        & ::= & \sv ~|~ \sv \; \sw ~|~ \src{\lett{x}{\sm}{\srcn}} ~|~
                \src{\sm \: \pssum \: \srcn}                         & \text{(terms)}\\ 
        &     & \src{\asc{\sm}{\sgtya}} ~|~ \src{\asc{\sv}{\sgtys}} ~|~ \src{\ite{\sv}{\sm}{\srcn}} ~|~ \src{\add{\sv}{\sw}}  & \\
  \sv,\sw 
        & ::= & \sx ~|~ \ssr ~|~ \ssb ~|~ \src{\lambda x:\sgtys. \sm}
                                      & \text{(values)} \\
\end{array}
\end{displaymath}
\caption{Syntax of \glang.}
\label{fig:source-syntax}
\vspace{-1em}
\end{figure}

The syntax of \glang is presented in Figure~\ref{fig:source-syntax}. We introduce imprecision in the language by
extending probabilities and simple types with the unknown annotation $\?$.
The unknown probability $\?$ represents any probability in the interval
$[0,1]$, and similarly, the unknown simple type $\?$ represents any static
simple type. We do not need an (explicit) unknown distribution type as
it can already be encoded by the singleton distribution type $\ds{\?^\?}$ (of
unknown simple type, with unknown probability). Notationwise, we use
$\pty$ to range over gradual probabilities ($\GProb$), $\sgtys$, $\sgtysp$ to
range over simple gradual types ($\GType$), and $\sgtya, \sgtyb$ to range over
gradual distribution types ($\GDType$).

\subsubsection*{\rev[mt1]{Formula types}}
\rev[fo1]{The programmer writes gradual types ($\sgtys \in \GType$,
$\sgtya \in \GDType$) only as annotations, on $\lambda$-binders and
ascriptions, and the type system does not work with them directly. Probabilistic choice and \src{let} combine
distribution types by \emph{scaling} and \emph{adding} probabilities, and
scaling leaves the grammar $\ps \mid \?$ of gradual probabilities, since the
product $\nicefrac12 \cdot \?$ has no expression in it. Re-approximating the
product to $\?$ discards what the annotation recorded. Take the one-entry
annotation $\ds{\rtype^{\?}}$. Its probability is written as unknown, but the
probabilities of a distribution type sum to $1$, so it is $1$, and scaling by
$\nicefrac12$ gives $\ds{\rtype^{\nicefrac12}}$. Re-approximation gives
$\ds{\rtype^{\?}}$ instead, and loses the $\nicefrac12$.
Example~\ref{ex:sgg-refuted} shows that, with such re-approximating
operators, the static gradual guarantee fails.

We therefore equip the type system with \emph{formula types}
(Figure~\ref{fig:formula-types}), that is, distribution types whose
probabilities are \emph{symbolic} and which carry a \emph{constraint} $\phii$
over those symbols. The scaling above records the equation
$\cww = \nicefrac12 \cdot \cww'$ and keeps both probabilities.}

\begin{figure}[t]
\begin{displaymath}
\begin{array}{r@{\hspace{0.3em}}c@{\hspace{0.8em}}l@{\hspace{3em}}l}
  \multicolumn{4}{c}{\gtys,\gtysp \in \FGType,  \quad \gta,\gtb \in \FGDType, \quad \cww \in \TVar, \quad \phii \in \Formula}\\[1.5ex]
  \cww & ::= & \pr{\varr, \leftvar, \rightvar} & \text{(tagged variables)}\\
  \p & ::= &  \cww ~|~ r & \text{(symbolic probabilities)} \\
  \gtys,\gtysp        & ::=        & \rtype ~|~ \btype ~|~ \gtys -> \gtya ~|~ \?  & \text{(\fgt)} \\
  \gta,\gtb      & ::=        & \db{\phi}{\gtys[i][\gbox{\p[i]}]}[i \in \iSet] & \text{(\fgdt)} \\
  \phii       & ::=         & \varphi = \varphi ~|~ \varphi \leq \varphi ~|~ \rev[fo1]{\varphi < \varphi} ~|~ \phii \land \phii ~|~ \rev[mt1]{\phii \Rightarrow \phii} & \text{(formulas)} \\
    \varphi       & ::=         & \p   ~|~ \varphi + \varphi ~|~\varphi - \varphi ~|~ \varphi \cdot \varphi ~|~ \varphi/\varphi & \text{(expressions)} \\
\end{array}
\end{displaymath}
  \caption{Formula types.}
  \label{fig:formula-types}
\end{figure}

\rev[mt1]{Intuitively, a \ft is the same as an ordinary gradual
type, except that (1) unknown probabilities are replaced by variables,
and (2) distribution types are guarded by formulas (like in refinement types). Formally, a
\emph{symbolic probability} $\p$ is either a constant $r$ or a tagged
variable.
A tagged variable $\cww$ represents a symbolic variable
$\varr$ (to be interpreted over the $[0,1]$ interval), \rev[fo1]{tagged by a pair of natural numbers $\leftvar$ and $\rightvar$. In an evidence (Section~\ref{sec:evidence}), the tags are the indices of the entries of the left-hand and the right-hand type that the variable connects, and the runtime uses them to route values and to compare evidences (Figure~\ref{fig:tag-precision}); elsewhere, they do not affect the meaning of formulas.} %
For simplicity, we adopt the following notation conventions. First, we use $\projv{\cww}, \projl{\cww}$ and $\projr{\cww}$ to access the first, second and third component of $\cww$, respectively. Second, given $\cww = \pr{\varr, i, j}$, we use $\cw{i}{j}$ as a shorthand for $\varr$. Third, in order not to clutter formulas, we sometimes write $\cww$ for $\projv{\cww}$ (\eg, $\cww[1] + \cww[2] = 1$ for $\projv{\cww[1]} + \projv{\cww[2]} = 1$).  Finally, when clear from the context, we refer to tagged variables simply as variables. %
An \emph{expression} $\varphi$ represents either a symbolic probability
$\p$ or algebraic operations (addition, subtraction, multiplication
or division) between symbolic probabilities. A \emph{formula} $\phii$
is either a comparison \rev[fo1]{(an equality, or a weak or strict inequality)}
between two expressions, the conjunction of
two formulas, or an implication between two formulas. \rev[mt1]{Implications are
used in one place, the result of the $\<let>$ reduction rule of
Section~\ref{sec:target}, where the formula of a branch is guarded by the branch
having positive probability\rev[fo1]{; the strict inequality appears only in that
guard, and we write $\cww > 0$ for $0 < \cww$}.} \emph{\Fgt} are defined similarly to gradual
simple types (including the $\?$ type), except that the codomain of
function types are \fgdt.  \Fgdt are multi-sets of pairs of \fgt
and symbolic probabilities, closed under a formula~$\phii$.}

Let us introduce some handy notation for the rest of the presentation. First, given formula $\phii$, we use $\FV(\phii)$ to denote the set of (free) tagged variables occurring in $\phii$. Second, given a set of symbolic probabilities $\{\p[i]  \mid i \in \iSet \}$, we use $\TV{\{\p[i]  \mid i \in \iSet \}}$ to denote the subset of tagged variables, only (\ie the result of filtering out static probabilities). Lastly, given formula $\phii$ over tagged variables $\cww[1], \ldots, \cww[n]$, we use $\satisfiablee{\phii}$ to denote that $\phii$ is satisfiable, \ie as a shorthand for $\exists \cww[1], \ldots, \exists \cww[n]. \: \phii$.

\rev[fo1]{Every formula that the rest of the article attaches to a type is built
from these constructs. The lifting below and the type operators of
Figure~\ref{fig:source-type-system} add equations and conjunctions, and the
meet and reordering operators of Section~\ref{sec:evidence} add the marginal
equations of a coupling, which are finite sums. Finite quantifications over
index sets, as in $\bigwedge_{i \in \iSet}$, are finite conjunctions, and
$\cdot$ denotes the empty conjunction. Satisfiability $\satisfiablee{\phii}$ is
the only construct that binds variables, and it is used only as a side
condition, never inside a type.}

\rev[mt1]{User annotations are turned into formula types by a \emph{lifting}
$\liftT{\cdot}$. For example, the gradual distribution type
$\dt{\pt{\Int}[\frac{1}{3}], \pt{\Bool}[\?], \pt{\?}[\?]}$ is
represented by the \sfgdt $\dctx[\phii] \dt{\pt{\Int}[\cww[1]],
  \pt{\Bool}[\cww[2]], \pt{\?}[\cww[3]]}$, where
$
\phii ~=~ \cww[1] = \frac{1}{3} \land \cww[2] \in [0,1] \land \cww[3] \in [0,1]
\land  \cww[1] + \cww[2] + \cww[3] = 1$, and $\p \in [r_1,r_2]$ is syntactic sugar for $r_1 \leq \p \land \p
\leq r_2$. Formally, the lifting is captured by three mutually
recursive functions over gradual simple types, gradual
probabilities and gradual distribution types:}
\begin{align*}
    \liftT{\cdot}: \GType -> \FGType \\
    \liftT{\rtype} &~=~ \rtype \\
      \liftT{\btype} &~=~ \btype \\
    \liftT{\?} &~=~ \? \\
       \liftT{\sgtys -> \sgtya} &~=~ \liftT{\sgtys}
    -> \liftD{\sgtya}\\[1ex]
    \liftP{\cdot}[\cdot]:  \GProb \times \TVar -> \Formula \\
     \liftP{\ps}[\cww] &~=~ (\cww = \ps) \\
     \liftP{\?}[\cww] &~=~ (\cww \in
    [0,1] )\\[1ex]
    \liftD{\cdot}: \GDType -> \FGDType \\
     \liftD{\d{\sgtys[i][\pty[i]]}[i \in \iSet]} &~=~
  \dctx[ \Big(\bigwedge_{i \in \iSet} \liftP{\pty[i]} \land \sum_{i \in
                                                   \iSet} \cww[i] =
                                                   1\Big)
                                                   ][\d{\liftT{\sgtys[i]}^{\cww[i]}}[i
                                                   \in \iSet]] \\
    & \qquad \qquad \text{where $\cww[i] =
    \pr{\varr[i], i,i}$ and $\varr[i]$ is fresh}
  \end{align*}

\rev[mt1]{From now on $\gtys, \gtysp$ range over \fgt ($\FGType$) and $\gta, \gtb$
over \fgdt ($\FGDType$); the source (annotation) types $\sgtys, \sgtya$ appear
only as arguments~of~\mbox{$\liftT{\cdot}$}.}

\subsubsection*{Design driven by AGT}
To justify some of the design decisions behind \glang, we follow, in
parallel, the Abstracting Gradual Typing (AGT)
methodology~\citep{garciaAl:popl2016}. In short, the idea behind AGT is
that starting from a specification of the meaning of gradual types in
terms of sets of static types, we can systematically derive all
relevant notions of the gradual language, which by construction will 
enjoy a set of desired properties (to be discussed
later). Unfortunately, some of the so-obtained definitions turn out
not to be very amenable to implementation. To address this limitation,
we also derive alternative definitions, with a more
operational flavor\rev[fo1]{, which are equivalent to the AGT ones for
consistency and sound with respect to them for precision}.

As just hinted, we start by providing the meaning of \rev[mt1]{\emph{formula}
types via concretization functions that map a formula simple type (resp.\ a
formula distribution type) to a non-empty set of \emph{static} simple types
(resp.\ static distribution types). A formula distribution type carries a
constraint $\phi$ over its symbolic probabilities; its meaning ranges over the
\emph{solutions} of $\phi$. We write $\nu \models \phi$ to say that $\nu$ assigns
each tagged variable of $\FV(\phi)$ a value in $[0,1]$ so that $\phi$ holds, and
$\nu(\p)$ for the resulting value of a symbolic probability $\p$ (with
$\nu(r)=r$)}.
\begin{align*}
      \crtsymbol \colon \FGType \rightarrow  \mathcal{P}(\Type) \\
      \crt{\gtys -> \gta} &~=~
     \bigl\{ \tys -> \tya \mid \tys \in \crt{\gtys} \land \tya \in
        \crta{\gta} \bigr\} \\
        \crt{\?} &~=~
      \Type \\
       \crt{\rtype}
      &~=~ \{ \rtype  \}
      \\
      \crt{\btype} &~=~
      \{ \btype  \} \\[1ex]
      \crtasymbol \colon \FGDType \rightarrow  \mathcal{P}(\DType) \\
     \crta{ \db{\phi}{\gtys[i][\p[i]]}[i \in \iSet] }
    &~=~ \bigl\{ \d{ \tys[j][\ps[j]] }[j \in \jSet] \mid
       \exists\, \nu \models \phi.~ \exists f \colon \jSet \to \iSet.~ \\
    &\qquad\quad
       (\forall j \in \jSet.~ \tys[j] \in \crt{\gtys[f(j)]}) \;\land \\
    &\qquad\quad
       (\forall i \in \iSet.~ \textstyle\sum_{\{j \in \jSet \mid f(j) = i\}} \ps[j] = \nu(\p[i])) \;\land \\
    &\qquad\quad
       \rev[mt1]{(\forall i \in \iSet.~ \exists j \in \jSet.~ f(j) = i)}\, \bigr\}
    \end{align*}

\rev[fo1]{The concretization of the unknown simple type $\?$ is any static
simple type, and the concretization of a formula distribution type ranges over
the solutions of its formula. Each entry is \emph{split} into a finite block
of static entries, in the style of gradual
rows~\citep{garciaAl:popl2016}. The map $f$ assigns each static entry $j$ to
the entry $f(j)$ that it refines, and the probabilities of the static entries
assigned to an entry $i$ sum to the value $\nu(\p[i])$ that the solution gives
to its symbolic probability. Splitting is required because a single entry can
stand for several static entries. For instance,
$\ds{\rtype^{\nicefrac12}, \btype^{\nicefrac12}} \in \crta{\liftD{\ds{\?^\?}}}$,
since the only entry, of probability $1$, splits into a $\rtype$ entry and a
$\btype$ entry of probability $\nicefrac12$ each.}

\rev[fo1]{If instead a type and its concretizations shared the index set, then
this membership would not hold, and the equivalence between the two
consistencies (Lemma~\ref{agtconsistencyeq}) would fail. The inductive
consistency relates $\ds{\?^\?}$ and
$\ds{\rtype^{\nicefrac12}, \btype^{\nicefrac12}}$, whereas the two types would
have no concretization in common. This is the concretization of the
conference version of this article~\citep{yeAl:oopsla2023}, whose equivalence
of consistencies (Lemma~4.6 there) therefore does not hold.}

\rev[fo1]{Blocks are required to be non-empty ($f$ is surjective), so every
entry of a formula type contributes at least one static entry to each
concretization, possibly of probability zero. The coverage clauses of
consistency and precision (Sections~\ref{sec:src-consistency}
and~\ref{sec:src-precision}) rely on this requirement to agree with their AGT
counterparts.}

\subsection{Type System}
\label{sec:type-system}
\begin{figure}
\begin{small}
\begin{flushleft}
\framebox{$\Gamma |-t \sv: \gtys$, $\quad\Gamma |-d \sm: \gtya$}
\end{flushleft}
 \def \MathparLineskip {\lineskip=1.4ex}
\begin{mathpar}
\inference{}
{\Gamma |-t \ssr : \rtype} \and
\inference{}
{\Gamma |-t \ssb : \btype} \and
\inference{\Gamma(\sx) = \gtys}
{\Gamma |-t \sx : \gtys} \and
\inference{\Gamma |-t \sv : \gtys}
{\Gamma |-t \sv : \rev[mt1]{\dctx[(\cww = 1)][\ds{\gtys[][\cww]}]} } \and
\inference{\Gamma, \sx:\rev[mt1]{\liftT{\sgtys}} |-t \sm : \gtya & \jtf{}{\sgtys}}
{\Gamma |-t \src{\lambda \sx:\sgtys. \sm} : \rev[mt1]{\liftT{\sgtys}} -> \gtya} \and
\inference{\Gamma |-t \sv : \gtys &  \gtys \rel  \rev[mt1]{\liftT{\sgtysp}} &  \jtf{}{\sgtysp} }
{\Gamma |-t \src{\asc{\sv}{\sgtysp}} : \rev[mt1]{\dctx[(\cww=1)][\ds{{\liftT{\sgtysp}}^{\cww}}]}}  \and
\inference{
  \Gamma |-t \sv : \gtys  &
  \Gamma |-t \sw :  \gtysp &   \gtysp \rel  \cdom(\gtys)
}
{\Gamma |-d \sv \; \sw : \rev[mt1]{\ccod(\gtys)}} \and
\inference{
  \Gamma |-d \sm : \gtya  &
  \Gamma |-d \srcn : \gtyb & \\
    \rev[mt1]{\liftP{\pty}[\cww[1]] = \phi[][1]} & \rev[mt1]{\liftP{(1{-}\pty)}[\cww[2]] = \phi[][2]} &
  \rev[mt1]{\phi = \phi[][1] \land \phi[][2] \land (\cww[1] + \cww[2] = 1)}
}
{\Gamma |-d \src{{\sm}\, \pssum\,  {\srcn}} : \rev[mt1]{\sugarconj{\phi}{\cww[1] \cdot \gtya + \cww[2] \cdot \gtyb}}} \and
\inference{
  \Gamma |-d \sm : \rev[mt1]{\db{\phi}{\gtys[i][\p[i]]}[i \in \iSet]}  &
  \rev[mt1]{\forall i\in \iSet.~\Gamma, \sx : \gtys[i] |-d \srcn : \gtya[i]}
}
{\Gamma |-d \src{\lett{ \sx }{\sm}{\srcn}} : \rev[mt1]{\sugarconj{\phi}{\sum_{i \in \iSet} \p[i] \cdot \gtya[i]}}} \and
\inference{\Gamma |-d \sm : \gtya & \gtya \rel  \rev[mt1]{\liftD{\sgtyb}}  & \jtf{}{\sgtyb}
}
{\Gamma |-d \src{\asc{ \sm}{\sgtyb}} : \rev[mt1]{\liftD{\sgtyb}}} \\
  \inference{
    \Gamma |-t \sv : \gtys & \gtys \rel \rtype   &
    \Gamma |-t \sw : \gtysp & \gtysp \rel \rtype
  }
  {\Gamma |-d \src{\add{\sv}{\sw}} : \rev[mt1]{\dctx[(\cww=1)][\ds{\rtype^{\cww}}]} } \and
  \inference{
    \Gamma |-t \sv : \gtys & \gtys \rel \btype \\
    \Gamma |-d \sm : \gtya &
        \Gamma |-d \srcn : \gtyb & \rev[mt1]{\gtya \rel \gtyb}
  }
  {\Gamma |-d \src{\ite{\sv}{\sm}{\srcn}} : \rev[mt1]{\gtya \mathbin{\oplus_{\?}} \gtyb} }
\end{mathpar} %
\begin{tabular}{ll}
  $\cdom:\FGType \rightharpoonup \FGType$ & $\ccod:\FGType \rightharpoonup \FGDType$  \\
  $\cdom(\gtys -> \gtya) = \gtys $ & $\ccod(\gtys -> \gtya) = \gtya $ \\
  $\cdom(?) = \? $  &  $ \ccod(?) = \rev[mt1]{\liftD{\ds{\?^\?}}}$  \\
  $\cdom(\gtys)~\text{undef.~otherwise} $  &  $
                                              \ccod(\gtys)~\text{undef.~otherwise}$  \\[1ex]
\end{tabular}
\begin{talign*}
\\[-3ex]
\p \cdot \phty{\phi}{\d{\gtys[i][\p[i]]}[i \in \iSet]} &~=~
\phty{\phi \land \bigl(\bigwedge_{i \in \iSet} \cww[i] = \p \cdot \p[i]\bigr)}
{\d{\gtys[i][\cww[i]]}[i \in \iSet]} \qquad \\ & \text{where }\cww[i] = \pr{\varr[i],
                                                         \projl{\p[i]},
                                                         \projr{\p[i]}}\!, \text{ and }
                                                         \varr[i] \text{ is fresh}\\
\sugarconj{\phi}{\sum_{i \in \iSet}
  \phty{\phi[][i]}{\d{\gtys[j][\p[j]]}[j \in \jSet_i]}} &~=~
                                                          \phty{\phi
                                                          \land
                                                          \bigl(\bigwedge_{i
                                                          \in \iSet} \phi[][i]\bigr) \land \bigl(\sum_{i \in \iSet} \sum_{j \in \jSet_i} \p[ij]  = 1 \bigr)
   }{ \bigcup_{i \in \iSet} \d{\gtys[j][\p[ij]]}[j \in \jSet_i] }\\
   & \text{where } \p[ij] = \pr{\projv{\p[j]}, \sum_{k=0}^{i-1}|\jSet_k| + \projl{\p[j]}, \sum_{k=0}^{i-1}|\jSet_k| +  \projr{\p[j]}}
\end{talign*}
\end{small}
\caption{Type system of  \glang.}
\label{fig:source-type-system}
\vspace{-1em}
\end{figure}

Figure~\ref{fig:source-type-system} shows the type system of \glang,
which is obtained from the type system of the static language
(Figure~\ref{static-typing}) by replacing static elements with their
gradual counterpart. Let us briefly describe these liftings. The
lifting $\cdom \colon \FGType \rightharpoonup \FGType$ (resp.~$\ccod \colon \FGType \rightharpoonup \FGDType$) of type function $\dom$
is standard: $\cdom(\gtys -> \gtya) = \gtys$, $\cdom(\?) = \?$, and
$\cdom$ is undefined elsewhere. Function $\cod$ is defined analogously.

\rev[fo1]{The types of probabilistic choice and \src{let} are built with the
scaling $\p \cdot \gtya$ and the weighted union
$\sugarconj{\phi}{\sum_i \gtya[i]}$ defined in the figure. The type system
lifts each probability annotation $\pty$ to a formula (a concrete $\ps$
becomes the equation $\cww = \ps$, and $\?$ becomes the constraint
$\cww \in [0,1]$), and the two metafunctions \emph{add} a constraint to the
formula of the type instead of re-approximating the result into the grammar
$\ps \mid \?$ of gradual probabilities. Scaling adds
$\cww[i] = \p \cdot \p[i]$, and union adds $\sum_{ij} \p[ij] = 1$. Hence the
operators are monotone for precision, which the static gradual guarantee
(Theorem~\ref{theorem:sgg}) requires.}

\rev[fo1]{If instead the operators absorbed the unknown probability
(\eg $\nicefrac12 \cdot \? = \?$), then they would discard the constraint
that an annotation records, and the guarantee would fail.}

\begin{ExampleNoQED}
  \label{ex:sgg-refuted}
  \rev[mt1]{\rev[fo1]{Assume that scaling absorbs the unknown probability, and
  consider the terms}
  \[ \sm = \asc{\src{1}}{\ds{\rtype^{\?}}} \pssum[\frac{1}{2}]
           \asc{\src{\ttt}}{\ds{\btype^{\?}}},
     \qquad
     \srcn = \asc{\src{1}}{\ds{\rtype^{1}}} \pssum[\frac{1}{2}]
           \asc{\src{\ttt}}{\ds{\btype^{1}}}. \]
  Term precision relates $\sm \gprec \srcn$. It asks only for
  $\ds{\rtype^{\?}} \gprec \ds{\rtype^{1}}$ and
  $\ds{\btype^{\?}} \gprec \ds{\btype^{1}}$, and both hold, because the
  lifting of a one-entry annotation includes $\sum_i \cww[i] = 1$, which
  forces the unknown probability to $1$.

  The types of the two terms are not related, however. Scaling absorbs the
  unknown probability, so $\sm$ has type $\ds{\rtype^{\?}, \btype^{\?}}$,
  whereas the only type of $\srcn$ is
  $\ds{\rtype^{\frac{1}{2}}, \btype^{\frac{1}{2}}}$. The assignment
  $(0.3, 0.7)$ satisfies the formula of the former, and every coupling
  between the two types is diagonal, since $\rtype \not\gprec \btype$, so the
  probabilities of the latter would have to be $(0.3, 0.7)$. Hence no type of
  $\srcn$ is less precise than the type of $\sm$, \rev[fo1]{and the static
  gradual guarantee fails}.}
\end{ExampleNoQED}

\rev[fo1]{These absorbing operators are the ones of the conference
version~\citep{yeAl:oopsla2023}, whose static gradual guarantee (Theorem~4.12
there) therefore does not hold.}

\rev[fo1]{The conditional is typed at the \emph{convex hull}
$\gtya \mathbin{\oplus_\?} \gtyb$ of the types of its two branches, which must
be consistent. The hull is the instance of choice with unknown probability,
\[
  \gtya \mathbin{\oplus_\?} \gtyb ~\triangleq~
  \sugarconj{(\cww \in [0,1])}{\cww \cdot \gtya + (1 - \cww) \cdot \gtyb},
  \qquad \cww ~\text{fresh},
\]
and an ascription recovers the type of either branch.} \rev[mt1]{These
metafunctions are exactly the ones used by the target language
(Figure~\ref{fig:target-type-system})}.

\rev[fo1]{The \src{let} rule is \emph{exhaustive}. The body $\srcn$ must
type-check under every entry $\gtys[i]$ of the type of the bound term,
including the entries of probability zero, and the result is the weighted
union $\sugarconj{\phi}{\sum_{i \in \iSet} \p[i] \cdot \gtya[i]}$ of the types
of the body. The rule is compatible with the static gradual guarantee because
precision carries coverage clauses (Section~\ref{sec:src-precision}), so a
less precise annotation cannot introduce an entry under which the body fails
to type-check.}

The lifting of type equality, called type
\emph{consistency} and denoted by $\rel$ in \glang, plays a fundamental role in
gradual languages. It allows soundly handling the notion of
(im)precision, which is conveniently introduced via type
ascriptions. For example, program $\src{1} \pssum[\frac{1}{2}] \src{\ttt} ::
\ds{\?^\?} :: \ds{\rtype^{\frac{2}{3}}, \btype^{\frac{1}{3}}}$ is
(optimistically) accepted by the gradual type system of
\glang because $\ds{\rtype^{\frac{1}{2}}, \btype^{\frac{1}{2}}} \sim
\ds{\?^\?}$ and $\ds{\?^\?} \sim \ds{\rtype^{\frac{2}{3}},
  \btype^{\frac{1}{3}}}$. Following AGT, we define type consistency by the existential lifting of type equality~$=_{s}$:
\begin{definition}[Type consistency, by AGT]\label{def:cosistAGT}\mechnote{\leanname{GPLC.AgtConsTy} and \leanname{GPLC.AgtConsD}.}
  \rev[fo1]{The consistency relation $\grel$ between formula simple types
 ($\gtys, \gtysp \in \FGType$) and between formula distribution types
 ($\gta, \gtb \in \FGDType$) is defined as follows:}
 \begin{enumerate}
  \item $\gtys \grel \gtysp$ iff $\ \exists \tys[1] \in
    \crt{\gtys}, \tys[2] \in \crt{\gtysp}. \ \tys[1] =_{s} \tys[2]$,
  \item $\gta \grel \gtb$ iff $\ \exists \tya[1] \in
    \crta{\gta}, \tya[2] \in \crta{\gtb}. \ \tya[1] =_{s} \tya[2]$.
 \end{enumerate}
\end{definition}
\noindent In words, two formula types are consistent if there exist
static types in their concretizations that are equal. \rev[fo1]{Equivalently,
two types are consistent exactly when some static type is, up to $=_{s}$, more
precise than both of them; Section~\ref{sec:evidence} states the runtime form
of this characterization.} The
problem with this definition is that it is not practical, as it can
depend on sets of infinitely many types. For gradual
simple types, this can be partially addressed by stating that $\?$ is
consistent with every other gradual simple type, but for gradual
distribution types the problem is more challenging as probabilities
must also be taken into account.

\subsection{Consistency, Refined}
\label{sec:src-consistency}

We are thus interested in an \emph{inductive} definition of
consistency. To illustrate the main idea behind our alternative
characterization, consider the pair of gradual distribution types
\begin{itemize}
  \item $\sgtya \:=\: \dt{\pt{(\rtype -> \dt{\?^{1}}) }[\frac{1}{2}], \pt{(\? ->
           \dt{\rtype^{1}})}[\frac{1}{2}]}$, and
  \item $\sgtyb \:=\: \dt{\pt{(\? -> \dt{\btype^{1}}) }[\frac{1}{3}], \pt{(\rtype ->
    \dt{\?^{1}})}[\frac{2}{3}]}$
\end{itemize}
represented on the left and right hand
side of Figure~\ref{fig:splitting} (for concreteness, we assume that
the elements of $\sgtya$ and $\sgtyb$ are enumerated by index set
$\iSet = \{1,2\}$, thus, \eg, $\pt{(\rtype -> \dt{\?}) }[\frac{1}{2}]$
corresponds to the simple type of index $1$ in
$\sgtya$). Intuitively, $\sgtya$ and $\sgtyb$ will be consistent iff there exists a \emph{splitting} of the probabilities
$\tfrac{1}{2}, \tfrac{1}{2}$ from $\sgtya$ and $\tfrac{1}{3}, \tfrac{2}{3}$ from $\sgtyb$ that relates
the simple types in $\sgtya$ with the simple types in $\sgtyb$ as follows:
\begin{enumerate}
\item  Type $\rtype -> \dt{\?^{1}}$ in $\sgtya$ is consistent with both $\? -> \dt{\btype^{1}}$ and 
$\rtype -> \dt{\?^{1}}$ in $\sgtyb$. This means that $\frac{1}{2}$, the
probability of $\rtype -> \dt{\?^{1}}$ in $\sgtya$, must be split into
two, \ie $\frac{1}{2} = \var{1}{1} + \var{1}{2}$, where $\var{1}{1}$
(resp.~$\var{1}{2}$)
represents the probability of relating $\rtype -> \dt{\?^{1}}$, the first
simple type in $\sgtya$, with $\? -> \dt{\btype^{1}}$ (resp.~$\rtype ->
\dt{\?^{1}}$), the first (resp.~second) simple
type in $\sgtyb$.
\item Type  $\? -> \dt{\rtype^{1}}$ in $\sgtya$ is consistent only with
  $\rtype -> \dt{\?^{1}}$  in  $\sgtyb$. Therefore, the probability of 
$\? -> \dt{\rtype^{1}}$ in $\sgtya$ need not be split, leading to $\frac{1}{2} = \var{2}{2}$.
\item Similarly, type  $\? -> \dt{\btype^{1}}$ in $\sgtyb$ is consistent
  only with $\rtype -> \dt{\?^{1}}$ in  $\sgtya$, so $\frac{1}{3} = \var{1}{1}$.

\item Finally, type $\rtype -> \dt{\?^{1}}$ in $\sgtyb$ is consistent with
  $\rtype -> \dt{\?^{1}}$ and $\? -> \dt{\rtype^{1}}$ in $\sgtya$, resulting in  $\frac{2}{3} = \var{1}{2} + \var{2}{2}$.
\end{enumerate}

\begin{figure}[t]
\begin{small}
\begin{center}
\begin{tikzpicture}
  \node[draw, circle,  fill= white!20, inner sep=1pt] (p1) {1};
  \node[sloped, left = 3pt of p1] {$(\rtype -> \ds{\?^{1}}  )^{\frac{1}{2}}$};
  \node[sloped, above = 2pt of p1] {};
  
  \node[draw, circle, fill= white!20, inner sep=1pt, right=140pt of p1, yshift=18pt] (p2) {1};
  \node[sloped, right = 3pt of p2] {$(\?-> \ds{\btype^{1}} )^{\frac{1}{3}}$};
  \node[sloped, above = 2pt of p2] {};
  
  \node[draw, circle, fill= white!20,inner sep=1pt, below=8pt of p1, yshift=-18pt] (p3) {2};
  \node[sloped, left = 3pt of p3] {$(\?-> \ds{\rtype^{1}} )^{\frac{1}{2}}$};
  \node[sloped, above = 2pt of p3] {};
  
  \node[draw, circle, fill= white!20, inner sep=1pt, right=140pt of p3, yshift=18pt] (p4) {2};
  \node[sloped, right = 3pt of p4] {$(\rtype -> \ds{\?^{1}} )^{\frac{2}{3}}$};

  \node[draw, circle, fill, inner sep=1.0pt, right=50pt of p1, yshift=18pt] (p5) {};
  \node[sloped, above = 3pt of p5] {};

  \node[draw, circle, fill, inner sep=1.0pt, right=90pt of p1, yshift=18pt] (p8) {};
  \node[sloped, above = 3pt of p5] {};

  \node[draw, circle, fill, inner sep=1.0pt, right=50pt of p1, yshift=-8pt] (p6) {};
  \node[sloped, above = 3pt of p6] {};
  \node[draw, circle, fill, inner sep=1.0pt, right=90pt of p1, yshift=-8pt] (p9) {};
  \node[sloped, above = 3pt of p6] {};

  \node[draw, circle, fill, inner sep=1.0pt, right=50pt of p3] (p7) {};
  \node[sloped, above = 3pt of p7] {};
  \node[draw, circle, fill, inner sep=1.0pt, right=90pt of p3] (p10) {};
  \node[sloped, above = 3pt of p7] {};

  \draw[arrow] (p1) to (p5);
  \draw[arrow] (p2) to (p8);
  \draw[arrow] (p1) to (p6);
  \draw[arrow] (p4) to (p9);
  \draw[arrow] (p4) to (p10);
  \draw[arrow] (p3) to (p7);
  \draw[] (p5) edge  node[sloped, anchor=center, above] {$\var{1}{1} = \nicefrac{1}{3}$} (p8);
  \draw[] (p6) edge  node[sloped, anchor=center, above] {$\var{1}{2} = \nicefrac{1}{6}$} (p9);
  \draw[] (p7) edge  node[sloped, anchor=center, above] {$\var{2}{2} = \nicefrac{1}{2}$} (p10);
\end{tikzpicture} 
\end{center}
\end{small}\vspace{-1ex}
\caption{Probability splitting to justify consistency between gradual distribution types.}
\label{fig:splitting}
\end{figure}

Since the system of four equations so derived is feasible, witnessed
\eg by solution $\var{1}{1} = \frac{1}{3}$, $\var{1}{2} = \frac{1}{6}$
and $\var{2}{2} = \frac{1}{2}$, we can conclude that $\sgtya$ and
$\sgtyb$ are consistent. This thought process constitutes a lifting of the
consistency relation between simple types to distribution types
through \emph{couplings}~\citep{DBLP:journals/corr/abs-1103-4577}, a tool
that has already been exploited \eg in the context of probabilistic
bisimulation~\citep{DBLP:journals/njc/SegalaL95} and verification of
cryptographic properties~\citep{Barthe:2009}. 

\begin{definition}[Relation lifting]\label{def:coupling}
\mechnote{condition~\ref{cond:coupling}, \leanname{IsCoupling}; condition~\ref{cond:relation}, \leanname{Supp}; the lifting, \leanname{Lift}.}
Assume that  $\DA = \d{\pt{a_i}[\ps[i]]}[i \in \iSet]$ and $ \DB =
\d{\pt{b_j}[ q_j ] }[j \in \jSet]$ are multi-set representations of
discrete probability distributions over sets $A$ and $B$, respectively
(that is, $a_i \in A$ and $\ps[i] \in [0,1]$ for all $i \in \iSet$, $b_j \in B$
and $q_j \in [0,1]$ for all $j \in \jSet$, and $\sum_{i \in
  \iSet}p_i =   \sum_{j \in
  \jSet} q_j = 1$). Moreover, let $R \subseteq A \times B$ be a
relation between $A$ and $B$. We say that $\DA$ and $\DB$ are related
by the \emph{lifting} of $R$, written $\couplingLift{\DA}{\DB}{R}$,
iff there exist  $\CC = \d{\var{i}{j} \in [0,1]}[i \in \iSet \land j \in \jSet]$
such that for all $i \in \iSet$ and all $j \in \jSet$,
\begin{enumerate}
\item\label{cond:coupling} $\ps[i] = \sum_{j \in \jSet} \var{i}{j} \land \rev[fo1]{q_j} = \sum_{i
    \in \iSet} \var{i}{j}$, 
\item\label{cond:relation}  $ \var{i}{j} > 0 => a_i \mathbin{R} b_j$
\end{enumerate}
We write ${\cjudg{\CC}{\DA}{\,R\,}{\DB}}$ to denote
that $\CC$ is a witness of the relation $\couplingLift{\DA}{\DB}{R}$, \ie to denote the conjunction between conditions \ref{cond:coupling}  and \ref{cond:relation} above. Moreover, any $\CC$ satisfying (only) condition~\ref{cond:coupling} is called a
\emph{coupling} between $\DA$ and $\DB$.  
\end{definition}

\subsubsection*{Type equality via couplings}%
To make a uniform treatment of type equality ($=_{s}$) in 
\slang and type consistency ($\rel$) in \glang, we start by
redefining the type equality in \slang in terms of couplings, via
relation~$=$: 
\begin{mathpar}
\inference{}{\rtype = \rtype} \and
\inference{}{\btype = \btype} \and
\inference{\tys[1] = \tys[2] & \tya[1] = \tya[2] }
{\tys[1] -> \tya[1]= \tys[2] -> \tya[2]} \and
\rev[mt1]{
\inference{
  \couplingLift{\d{\tys[i][\ps[i]]}[i \in \iSet]}{\d{\tys[j][\ps[j]]}[j \in \jSet]}{=} \\
  \forall i \in \iSet.\ \exists j \in \jSet.\ \tys[i] = \tys[j] &
  \forall j \in \jSet.\ \exists i \in \iSet.\ \tys[i] = \tys[j]
}
{ \d{\tys[i][\ps[i]]}[i \in \iSet] = \d{\tys[j][\ps[j]]}[j \in \jSet] }}
\end{mathpar}
The last rule above says that two distribution types are equal if
there exists a coupling that justifies the lifting of equality on
simple types to distribution types (note that the $=$ symbol in the
rule premise refers to equality over simple types, while the $=$
symbol in the conclusion refers to equality over distribution
types). \rev[fo1]{Besides the coupling, the rule has the two coverage premises
of $=_{s}$ (Section~\ref{sec:static}). The coupling constrains only the
entries that carry probability, so by itself it would relate
$\dt{\rtype^{1}}$ and $\dt{\rtype^{1}, \btype^{0}}$.}
Using this rule we can, \eg, derive that
$\ds{\rtype^{\frac{1}{2}}, \btype^{\frac{1}{2}}} =
\ds{\rtype^{\frac{1}{3}}, \rtype^{\frac{1}{6}}, \btype^{\frac{1}{2}}}$
because the set of formulas  
$\tfrac{1}{2} = \var{1}{1} + \var{1}{2}$, $\tfrac{1}{2} = \var{2}{3}$,
$\tfrac{1}{3} = \var{1}{1}$, $\tfrac{1}{6} = \var{1}{2}$ and $\tfrac{1}{2} = \var{2}{3}$ is
satisfiable (by solution $\var{1}{1} = \frac{1}{3}$, $\var{1}{2} = \frac{1}{6}$, $\var{2}{3}
=\frac{1}{2}$).  

As expected, this alternative definition of equality is equivalent to
the original from Section~\ref{sec:static}.

\begin{restatable}[Alternative characterization of equality]{lemma}{couplingeq}\label{couplingeq}\mechnote{item~1, \leanname{SPLC.eqTy\_iff\_coupling}; item~2, \leanname{SPLC.eqD\_iff\_coupling}.}
  \rev[fo1]{For any} simple types $\tys[1], \tys[2] \in \Type $ and
  \rev[fo1]{any distribution} types $\tya[1], \tya[2] \in \DType$,
  \begin{enumerate}
    \item $\tys[1] =_{s} \tys[2] \;\; \text{iff} \; \;  \tys[1] = \tys[2]$.
    \item $\tya[1] =_{s} \tya[2] \; \; \text{iff} \; \;  \tya[1] = \tya[2]$.
  \end{enumerate}
\end{restatable}

Armed with this new definition of equality based on couplings we proceed to define consistency.

\subsubsection*{Type consistency, a straightforward approach}
A straightforward approach to define consistency in \glang consists in
(rule-wise) lifting the definition of type equality $=$ in \slang, and
extending it with rules stating that $\?$ is consistent with any gradual
simple type. An excerpt of the resulting set of rules would be:
\begin{mathpar}
\inference{}{\? \sim \sgtys} \and
\inference{}{\rtype \sim \rtype} \and
\inference{\sgtys[1] \sim \sgtys[2] & \sgtya[1] \sim \sgtya[2] }
{\sgtys[1] -> \sgtya[1]\sim \sgtys[2] -> \sgtya[2]} \and
\inference{{%
  \couplingLift{\sgtya[1]}{\sgtya[2]}{\sim}%
  }
}
{ \sgtya[1] \sim \sgtya[2] }
\end{mathpar}
According to this definition (in particular, by the last rule),
establishing the consistency, \eg, between gradual distribution types
$\d{\sgtys[i][\pty[i]]}[i \in \iSet]$ and
$\d{\sgtys[j][\pty[j]]}[j \in \jSet]$ requires exhibiting a coupling
between them. The problem here is that probabilities in either of the
distribution types can be only partially known, that is, $\pty[i]$
could be $\?$ for some $i \in \iSet$, rendering formula
$\sum_{j \in \jSet} \var{i}{j} = \?$ even ill-defined. 
A first approach to tackle this problem would be lifting these
formulas (from the static setting) to the gradual setting. Note,
however, that this lifting cannot be done for each formula
independently because the ``same'' $\?$  will probably occur in
multiple formulas.  We should then lift all related formulas at the
same time, but this still suffers from scoping problems because unknown
probabilities (represented by $\?$) must remain visible outside the
formulas lifting:  At runtime, we need to carry witness information
about consistency, where gradual probabilities ``flow'' across
reductions (see the $\<let>\!\!$ and probabilistic choice reduction rules in Figure~\ref{fig:static-distribution-reduction}).
To tackle this
problem, we introduce (fresh) symbolic variables representing unknown
probabilities, and analyze the existence of couplings for this
symbolic representation of distribution types.

To give the inductive definition of type consistency (and other forthcoming notions), we require a variant of the traditional notion of coupling. This variant differs from the traditional definition (see Definition~\ref{def:coupling}) in that it operates over \emph{symbolic} probability distributions, where probabilities are given by logical variables rather than concrete numbers, and these variables are subject to given constraints.

\begin{definition}[Coupling over symbolic distributions]\label{def:symcoupling}\mechnote{\leanname{SymLift}; the form of Figure~\ref{source-type-precision}, quantified universally over the solutions of $\psi_1$, is \leanname{SymLiftAll}.}
Assume that  $\DA = \d{\pt{a_i}[\ps[i]]}[i \in \iSet]$ and $ \DB =
\d{\pt{b_j}[ q_j] }[j \in \jSet]$ are (multi-set representations of)
symbolic discrete probability distributions over sets $A$ and $B$. Moreover, let $R \subseteq A \times B$ be a
relation between $A$ and $B$. Given $\CC = \d{\var{i}{j}}[i \in \iSet \land j \in \jSet]$, constraint $\psi_{1}$ over  $\ps[i]$, and constraint $\psi_{2}$ over $q_j$, we use  $\cjudgext{\CC}{\DA}{\,R\,}{\DB}{\psi_{1}}{\psi_{2}}{}$ to denote the conjunction of conditions \ref{cond:coupling}  and \ref{cond:relation} from Definition~\ref{def:coupling} (\ie that $\CC$ is a ``traditional'' coupling between $\DA$ and $\DB$), together with $\psi_{1} \land \psi_{2}$. We also use $\couplingLiftExt{\DA}{\DB}{R}{\psi_{1}}{\psi_{2}}{}$ to denote formula  $\exists \{ \ps[i] \mid i \in \iSet\} \cup \{ q_j \mid j \in \jSet\} \cup \{ \var{i}{j} \mid i \in \iSet \land j \in \jSet \}. \ \cjudgext{\CC}{\DA}{\,R\,}{\DB}{\psi_{1}}{\psi_{2}}{}$.
\end{definition}

Note that $\couplingLiftExt{\DA}{\DB}{R}{\psi_{1}}{\psi_{2}}{}$ requires the existence not only of a coupling $\CC$, but also of concretizations of (symbolic distributions) $\DA$ and $\DB$, respectively satisfying $\psi_1$ and $\psi_2$. Typically, $\psi_1$ and $\psi_2$ will require that probabilities sum up to $1$. Like in Definition~\ref{def:symcoupling}, in the rest of the presentation we allow ourselves some abuse of notation and write   
$\exists \{x_1, x_2, \ldots, x_n\}$ as a shorthand for $\exists x_1. \: \exists x_2. \: \ldots \exists x_n$, and similarly for $\forall \{x_1, x_2, \ldots, x_n\}$.

Armed with the above notion of relation lifting, we can define the lifting of any relation $R$ over gradual simple types to gradual distribution types as:
\[
  \couplingLift{\db{\phi[][1]}{\gtys[i][\p[i]]}[i \in \iSet]}{\db{\phi[][2]}{\gtys[j][\p[j]]}[j \in \jSet]}{R} \quad\text{iff}\quad \couplingLiftExt{\d{\gtys[i][\p[i]]}[i \in \iSet]} {\d{\gtys[j][\p[j]]}[j \in \jSet]}{R}{\phi[][1]}{\phi[][2]}{}
 \]
Now, we can readily provide an inductive characterization of consistency, by simply lifting the
definition of equality. \rev[fo1]{As for equality, the rule for distribution
types has, besides the coupling, two coverage clauses, which require every
entry of either type to have a consistent entry in the other, whatever its
probability.}

\begin{definition}[Type consistency, inductively]\label{def:consistency-source}\mechnote{\leanname{GPLC.ConsTy} and \leanname{GPLC.ConsD}.}
\rev[mt1]{The consistency relation $\sim$ between formula simple types
  ($\gtys, \gtysp \in \FGType$) and between formula distribution types
  ($\gtya, \gtyb \in \FGDType$) is defined \rev[fo1]{by the following rules}:}
\begin{mathpar}
\inference{}{\rtype \sim \rtype} \and
\inference{}{\btype \sim \btype} \and
\inference{}{\gtys \sim \?} \and
\inference{}{\? \sim \gtys} \\ \and
\inference{\rev[mt1]{\gtys[1] \sim \gtys[2]} & \gtya[1] \sim \gtya[2] }
{\gtys[1] -> \gtya[1]\sim \gtys[2] -> \gtya[2]} \and
\rev[mt1]{
\inference{
   \couplingLift{\db{\phi[][1]}{\gtys[i][\p[i]]}[i \in \iSet]}{\db{\phi[][2]}{\gtys[j][\p[j]]}[j \in \jSet]}{\sim} \\
   \forall i \in \iSet.\ \exists j \in \jSet.\ \gtys[i] \sim \gtys[j] &
   \forall j \in \jSet.\ \exists i \in \iSet.\ \gtys[i] \sim \gtys[j]
 }
{\db{\phi[][1]}{\gtys[i][\p[i]]}[i \in \iSet] \sim
 \db{\phi[][2]}{\gtys[j][\p[j]]}[j \in \jSet]}}
\and
\end{mathpar}
\end{definition}

As expected, \rev[fo1]{on well-formed types
(Definition~\ref{def:well-formed-source}),} the inductive definition of consistency (Definition~\ref{def:consistency-source}) coincides with
the one yielded by AGT (Definition~\ref{def:cosistAGT}):
\begin{restatable}[Equivalence of consistencies]{lemma}{agtconsistencyeq}\label{agtconsistencyeq}\mechnote{\leanname{GPLC.agt\_consistencyTy\_iff} and \leanname{GPLC.agt\_consistency\_iff}.}
  \rev[fo1]{For any well-formed formula simple types
  $\gtys, \gtysp \in \FGType$ and any well-formed formula distribution types
  $\gtya, \gtyb \in \FGDType$,}
\begin{enumerate}
  \item $\gtys \grel \gtysp \;\; \text{iff} \; \;  \gtys \rel \gtysp$.
  \item $\gtya \grel \gtyb \; \; \text{iff} \; \;  \gtya \rel \gtyb$.
\end{enumerate}
\end{restatable}

\subsubsection*{Type well-formedness}
Another relevant aspect of \glang type system is that, like \slang type system, programs are assigned well-formed types, only. \rev[fo1]{An annotation written by the programmer is well-formed when its probabilities can sum to $1$ (Appendix~\ref{app:source}). A formula distribution type is well-formed when its formula is satisfiable, every solution of the formula gives the type total probability $1$, the probabilities given by the solutions form a convex set, and all its tagged variables occur in the formula. An environment is well-formed when every type that it assigns is.}
\begin{definition}[{Type \rev[fo2]{and environment} well-formedness}]\label{def:well-formed-source}\mechnote{\leanname{GPLC.GoodTy}, \leanname{GPLC.GoodD} and \leanname{GPLC.CtxGood}.}
\rev[mt1]{The well-formedness of formula types \rev[fo2]{and of environments} (denoted by symbol $\vdash$) is defined \rev[fo1]{by the following rules}:}
\begin{mathpar}
  \inference{}{\jtf{}{\rtype}} \and
  \inference{}{\jtf{}{\btype}} \and
  \inference{}{\jtf{}{\?}} \and
  \inference{ \jtf{}{\gtys} & \jtf{}{\gtya}}
  {\jtf{}{\gtys -> \gtya}} \and
  \rev[fo2]{\inference{
   {\TV{\{\p[i]  \mid i \in \iSet \}} \subseteq \FV(\phi)}  &  \satisfiablee{\phi} &
   \forall \nu \models \phi.~ \textstyle\sum_{i \in \iSet} \nu(\p[i]) = 1 \\
   \{ (\nu(\p[i]))_{i \in \iSet} \mid \nu \models \phi \} \text{ is convex} &
   \forall i \in \iSet. \jtf{}{\gtys[i]}  \!
  }
  { \jtf{}{ \phty{\phi[][]}{ \d{\gtys[i][\p[i]]}}[i \in \iSet] } }} \and
  \rev[fo2]{\inference{}{\jtf{}{\cdot}}} \and
  \rev[fo2]{\inference{\jtf{}{\Gamma} & \jtf{}{\gtys}}{\jtf{}{\Gamma, x : \gtys}}}
  \end{mathpar}
\end{definition}

\rev[fo1]{The two conditions on solutions are invariants of the types that the
type system builds, since every constraint that it generates is linear in the
probability variables. The first condition makes each solution a probability
distribution, on which the dynamic semantics of Section~\ref{sec:target}
relies, and the second is needed for the \src{let} rule to be monotone for
precision.}

\rev[fo1]{The lifting of a well-formed annotation is a well-formed formula
type (Lemma~\ref{lemma:wellformed-lifting} in Appendix~\ref{app:source}).
Hence well-typed terms have well-formed types, provided that the environment
and the type annotations of the term are well-formed:}

\begin{lemma}[Type well-formedness]\label{lemma:welltyped-wellform-source2}\mechnote{\leanname{GPLC.good\_val} and \leanname{GPLC.good\_tm}.}
  \rev[fo1]{For any well-formed environment $\Gamma$, any value $\sv$ and
  any term $\sm$ of \glang{} whose type annotations are well-formed, any
  formula simple type $\gtys \in \FGType$ and any formula distribution type
  $\gtya \in \FGDType$,}%
\begin{enumerate}
  \item If $\: \Gamma |-t \sv : \gtys $, then $\justify{\gtys}$.
  \item If $\; \Gamma |-d  \sm : \gtya $, then $\justify{\gtya}$.
  \end{enumerate}
\end{lemma}

\subsection{Refined Criteria}
\label{sec:src-precision}
The refined criteria for gradual languages~\citep{siekAl:snapl2015}
establish a set of distinguishing properties for such a class of languages, where two such properties are related to the static semantics: the \emph{static gradual
  guarantee}, which guarantees that typing is monotone with respect to
imprecision, and the \emph{conservative extension of the static discipline}, which guarantees that every fully-statically-annotated well-typed term in the gradual language is also typeable in the static language (and vice versa).
To establish the first property, the static gradual guarantee for \glang, we first need to  define a notion of precision between types, and subsequently between terms.

\subsubsection*{Type precision}
\rev[fo1]{Precision is defined over formula types. AGT casts the definition of type precision in terms of}
set containment on the concretization, \ie $G_1 \gprec G_2$
\rev[mt1]{(meaning that the formula type $G_1$ is at least as precise as $G_2$)} if
and only if $\gamma(G_1) \subseteq \gamma(G_2)$.
Nevertheless, in the presence of distribution types,
the definition based on set containment is not satisfactory as it assumes a \emph{syntactic} equality between set elements. For instance, while $\ds{\rtype^1} \gprec \ds{\rtype^{\frac{1}{2}}, \rtype^{\frac{1}{2}}}$ is expected to hold since  the involved pair of types are equal (under our semantic view of equality), a naive definition of precision would reject this relation. Therefore, we adopt an alternative definition of precision by~\citep{lennonAl:toplas2022}, which can be successfully applied when equality is not syntactic.

\begin{definition}[Type precision]\label{def:tprecAGT}\mechnote{\leanname{GPLC.AgtPrecTy} and \leanname{GPLC.AgtPrecD}.}
  \rev[fo1]{The precision relation $\ggprec$ between formula simple types
  ($\gtys, \gtysp \in \FGType$) and between formula distribution types
  ($\gtya, \gtyb \in \FGDType$) is defined as follows:}
  \begin{enumerate}
  \item
  $\gtys \ggprec \gtysp$ if and only if
  $\, \forall \tys[1] \in \crt{\gtys}.\ \exists \tys[2] \in
  \crt{\gtysp}.\ \tys[1] = \tys[2]$.
  \item \rev[mt1]{$\gtya \ggprec \gtyb$ if and only if
   $\forall \tya[1] \in \crta{\gtya}. \ \exists \tya[2] \in
   \crta{\gtyb}. \ \tya[1] = \tya[2]$.}
  \end{enumerate}
\end{definition}

\rev[fo1]{Entries of probability zero matter for this definition. Since the
blocks of the concretization are non-empty, every static type in
$\crta{\phty{\phi}{\ds{\rtype^{\cww[1]}, \btype^{\cww[2]}}}}$ has a $\btype$ entry,
of probability $0$ when $\phi$ forces $\cww[2] = 0$, and consequently
$\ds{\rtype^{1}} \not\ggprec \phty{\phi}{\ds{\rtype^{\cww[1]}, \btype^{\cww[2]}}}$.
The inductive characterization below reproduces this through its two coverage
clauses, which require every entry of either type to have a
$\gprec$-related entry in the other, whatever its probability. These clauses
also make the exhaustive \src{let} rule compatible with the static gradual
guarantee.}

\rev[fo1]{If instead consistency and precision had no coverage clauses, then
a less precise annotation could introduce an entry of probability zero with no
counterpart, and the static gradual guarantee would fail.}

\begin{ExampleNoQED}
  \label{ex:coverage-sgg}
  \rev[fo1]{Consider the two programs of Section~\ref{sec:static},
  written in \glang:
  \[ \sm = \lett{x}{\asc{\src{1}}{\ds{\rtype^{1}}}}{\add{x}{1}},
     \qquad
     \srcn = \lett{x}{\asc{\src{1}}{\ds{\rtype^{1}, \btype^{0}}}}{\add{x}{1}}. \]
  They differ in one annotation, so term precision relates $\sm \gprec \srcn$
  exactly when $\ds{\rtype^{1}} \gprec \ds{\rtype^{1}, \btype^{0}}$ holds.
  Assume that consistency and precision have no coverage clauses.
  The two annotations are then related, because a coupling can put all the
  probability on the $\rtype$ pair and leave the $\btype$ entry at $0$, and for
  the same reason the ascription in $\srcn$ type-checks. The \src{let} rule of
  Figure~\ref{fig:source-type-system} types the body at every entry of
  $\ds{\rtype^{1}, \btype^{0}}$, including $x : \btype$, where $\add{x}{1}$ is
  ill typed. Hence $\sm$ is well typed, $\srcn$ is not, and
  $\sm \gprec \srcn$, so the static gradual guarantee fails.}
\end{ExampleNoQED}

\rev[fo1]{Precision without coverage clauses is the precision of
the conference version~\citep{yeAl:oopsla2023}, whose static gradual guarantee
(Theorem~4.12 there) fails for this reason as well. The runtime relations of
Section~\ref{sec:evidence} have no coverage clauses, for the reasons given
there.}

\begin{figure}[t]
\begin{mathpar}
\inference{}
{\rtype \gprec \rtype} \and
\inference{}
{\btype \gprec \btype} \and
\inference{}
{\gtys \gprec \? } \and
\inference{
  \gtys \gprec \gtysp &
  \gtya \gprec \gtyb
}
{\gtys -> \gtya \gprec \gtysp -> \gtyb } \and
{
\inference{
  \forall\: \FV(\phi[][1]).~  \phi[][1]  \implies \exists~\FV(\phi[][2]) \cup
  \{\cww[ij] ~|~ i \in \iSet \land j \in \jSet\}. %
  \\
  \cjudgext{\d{\cww[ij]}[i \in \iSet \land j \in \jSet]}{\d{\gtys[i][\p[i]]}[i \in \iSet]}{\, \gprec\,}{\d{\gtys[j][\p[j]]}[j \in \jSet]}{\phi[][1]}{\phi[][2]}{}%
  \\
  \rev[mt1]{\forall i \in \iSet.\ \exists j \in \jSet.\ \gtys[i] \gprec \gtys[j]} &
  \rev[mt1]{\forall j \in \jSet.\ \exists i \in \iSet.\ \gtys[i] \gprec \gtys[j]}
}
{\db{\phi[][1]}{\gtys[i][\p[i]]}[i \in \iSet] \gprec
  \db{\phi[][2]}{\gtys[j][\p[j]]}[j \in \jSet]}
}
\end{mathpar}
\caption{Type precision in \glang.}
\label{source-type-precision}
\vspace{-1em}
\end{figure}

\rev[mt1]{Like for consistency, the declarative definition of precision above,
despite being sound, is impractical. We thus present an alternative, inductive definition; only the case of formula distribution types deserves special
attention (Figure~\ref{source-type-precision}). Precision for \fgdt is slightly
different from consistency. Loosely speaking, \fgdt $\gtya[1]$ and
$\gtya[2]$ are related by precision iff every solution that makes the
probabilities of $\gtya[1]$ sum up to $1$ can be ``completed'' to form a coupling
between $\gtya[1]$ and $\gtya[2]$ that witnesses the liftings of precision.
Intuitively, the definition reproduces
the quantifier structure of Definition~\ref{def:tprecAGT}, with the coupling
playing the role of the splitting assignments of the concretization.}

\rev[fo1]{If instead Definition~\ref{def:tprecAGT} were stated over gradual
types, before lifting them, then it would not agree with the inductive rule.
The inductive rule derives $\ds{\rtype^{\?}} \gprec \ds{\rtype^{1}}$, because
the constraint of the lifting that the probabilities sum to $1$ forces the unknown probability to $1$,
whereas the declarative definition rejects it, since
$\ds{\rtype^{\nicefrac{3}{10}}}$ is a concretization of the left type with no
counterpart on the right. This is how the conference
version~\citep{yeAl:oopsla2023} states the declarative definition, and the
equivalence between the two relations that it claims (Lemma~4.11 there)
therefore does not hold.}

\begin{ExampleNoQED}
To illustrate how this new definition of type precision works,
consider the following examples:
\begin{itemize}
 \item
$\ds{\rtype^{\frac{1}{2}}, \?^{\frac{1}{2}}} \gprec \ds{\rtype^{\frac{1}{3}}, \rtype^{\frac{1}{6}}, \?^{\frac{1}{2}}}$ because 
$\fcw{1}{1}+\fcw{1}{2} + \fcw{1}{3} = \frac{1}{2} \land \fcw{2}{3} = \frac{1}{2} \land 
\fcw{1}{1} = \frac{1}{3} \land \fcw{1}{2} = \frac{1}{6} \land \fcw{1}{3}+\fcw{2}{3} = \frac{1}{2}$
is satisfiable by the solution set $\{ \fcw{1}{1} = \frac{1}{3}, \fcw{1}{2} = \frac{1}{6}, 
\fcw{1}{3} = 0 , 
\fcw{2}{3} = \frac{1}{2}  \}$.
\item
$\ds{\rtype^{\frac{1}{2}}, \?^{\frac{1}{2}}} \not\gprec \ds{\btype^{\frac{2}{3}}, \?^{\frac{1}{3}}}$  because 
$\fcw{1}{2} = \frac{1}{2} \land \fcw{2}{2} = \frac{1}{2} \land 
\fcw{1}{2}+ \fcw{2}{2} = \frac{1}{3} $
is not satisfiable.
\item
$\ds{\rtype^{\frac{1}{2}}, \?^{\frac{1}{2}}} \gprec \ds{\rtype^{\frac{1}{3}}, \?^{\frac{2}{3}}}$ because 
$\fcw{1}{1}+\fcw{1}{2} = \frac{1}{2} \land \fcw{2}{2} = \frac{1}{2} 
\land \fcw{1}{1} = \frac{1}{3} \land  
\fcw{1}{2}+\fcw{2}{2} = \frac{2}{3}$ is satisfiable by the 
solution set $\{ \fcw{1}{1} = \frac{1}{3}, \fcw{1}{2} = \frac{1}{6}, 
\fcw{2}{2} = \frac{1}{2} \}$.
\item 
$\ds{\rtype^1} \gprec \ds{\?^1}$ because $\fcw{1}{1} = 1$ is satisfiable by the solution set $\{ \fcw{1}{1} = 1 \}$.\hfill\ExEndSymbol
\end{itemize}%
\end{ExampleNoQED}

As already hinted, this inductive definition of precision \rev[mt1]{(Figure~\ref{source-type-precision})} is \rev[fo1]{sound with respect to} \rev[mt1]{the declarative} Definition~\ref{def:tprecAGT}:

\begin{restatable}[{\rev[mt1]{Soundness} of type precision}]{lemma}{precisioneq}\label{lemma:precisioneq}\mechnote{\leanname{GPLC.precTy\_to\_agtPrecTy} and \leanname{GPLC.precD\_to\_agtPrecD}.}
\rev[mt1]{For any formula simple types $\gtys, \gtysp \in \FGType$
  and any formula distribution types $\gtya, \gtyb \in \FGDType$,}
\begin{enumerate}
  \item \rev[fo1]{If $\gtys \gprec \gtysp$, then $\gtys \ggprec \gtysp$.}
  \item \rev[fo1]{If $\gtya \gprec \gtyb$, then $\gtya \ggprec \gtyb$.}
\end{enumerate}
\end{restatable}

\rev[fo1]{The type system and the theorems below use the inductive relation of
Figure~\ref{source-type-precision}, and the lemma guarantees that it never
relates two types that Definition~\ref{def:tprecAGT} does not relate. We do
not claim the converse. Consistency is existential on both sides, so a single
pair of equal concretizations yields it. Precision is universal on the left,
and the converse would require turning the fact that every concretization has
an equal counterpart into a single coupling that has positive weight only
between entries related by precision (see Appendix~\ref{app:agt-equivalences}).}

\subsubsection*{Term precision}
\begin{figure}[t]
  \def \MathparLineskip {\lineskip=1.4ex}

\begin{mathpar}
\inference{
}
{
   \sx \gprec \sx
} \and
\inference{
}
{
  \ssr \gprec \ssr
} \and
\inference{
}
{
  \ssb \gprec \ssb
} \and
\inference{
\sgtys \gprec \sgtysp &
\sm \gprec \srcn & 
}
{
 \src{(\lambda x:\sgtys. m)} \gprec \src{(\lambda x:\sgtysp. n)}
} \and
\inference{
\sv \gprec \src{\svp } &
\sgtys \gprec \sgtysp
}
{
\src{\asc{\sv}{\sgtys}} \gprec \src{\asc{ \svp }{\sgtysp}}
} \and
\inference{
\sm \gprec \srcn &
\sgtya \gprec \sgtyb
}
{
\src{\asc{ \sm}{\sgtya}} \gprec \src{\asc{ \srcn}{\sgtyb}}
} \and
\inference{}
{\ps \gprec \ps} \and
\inference{}
{\pty \gprec \? } \and
\inference{
\sm \gprec \src{m'} &
\srcn \gprec \src{n'} & \pty \gprec \src{\ptyp}
}
{
\src{\sm\: \pssum\, \srcn} \gprec \src{m' \:\pssum[\ptyp]\, n'}
} 
\and
\inference{
\sv \gprec \src{v'} &
\sw \gprec \src{w'} &
}
{
  \sv \; \sw \gprec \src{v'} \; \src{w'}
} \and
\inference{
\sm \gprec \src{m'} &
\srcn \gprec \src{n'} &
}
{
  \src{\lett{x}{m}{n}} \gprec \src{\lett{x}{m'}{n'}}
} \and
\inference{
\sv \gprec \rev[fo1]{\src{v'}} &
\sw \gprec \src{w'} &
}
{
  \src{\add{v}{w}} \gprec \src{\add{v'}{w'}}
} \and
\inference{
\sv \gprec \src{v'} &
\sm \gprec \src{m'} &
\srcn \gprec \src{n'} &
}
{
  \src{\ite{v}{m}{n}} \gprec \src{\ite{v'}{m'}{n'}}
}
\end{mathpar}
\caption{Term precision in \glang.}
\label{source-term-precision}
\vspace{-1em}
\end{figure}
Term precision is the natural lifting of type precision to the space of terms. Its definition is rather standard, by induction in the term structure, as presented in Figure~\ref{source-term-precision}. \rev[mt1]{The annotations occurring in terms are \emph{source} gradual types; an annotation precision $\sgtys \gprec \sgtysp$ in the rules abbreviates the precision $\liftT{\sgtys} \gprec \liftT{\sgtysp}$ of their liftings.}

\subsubsection*{Metatheory}
Armed with the definition of precision, we can now state the two fundamental properties that hold for the static semantics of \glang.
First, typeability is monotone \wrt imprecision:
\begin{restatable}[Static gradual guarantee for \glang]{theorem}{sgg}\label{theorem:sgg}\mechnote{\leanname{GPLC.static\_gradual\_guarantee}; open form, \leanname{GPLC.static\_gradual\_guarantee\_val} and \leanname{GPLC.static\_gradual\_guarantee\_tm}.}
\rev[fo1]{For every value $\sv$ and every term $\sm$ of \glang{}, every value
$\sw$ and every term $\srcn$ of \glang{} whose type annotations are
well-formed, every formula simple type $\gtys \in \FGType$ and every formula
distribution type $\gtya \in \FGDType$,}
\begin{enumerate}
  \item \rev[fo1]{If $|-t \sv : \gtys$ and $\sv \gprec \sw$, then
    $|-t \sw : \gtysp$ and $\gtys \gprec \gtysp$ for some formula simple type
    $\gtysp \in \FGType$.}
  \item \rev[fo1]{If $|-d \sm : \gtya$ and $\sm \gprec \srcn$, then
    $|-d \srcn : \gtyb$ and $\gtya \gprec \gtyb$ for some formula
    distribution type $\gtyb \in \FGDType$.}
\end{enumerate}
\end{restatable}

\rev[fo1]{The theorem assumes that the type annotations of the less precise
term are well-formed, because precision does not preserve well-formedness.
Recall from Example~\ref{ex:sgg-refuted} that the guarantee fails for operators
that re-approximate their result.}

Second, the static semantics of \slang and \glang are equivalent for
fully-statically-annotated terms. \rev[fo1]{Since \glang{} infers formula
types, the types that the two languages assign to such a term cannot be
required to coincide. Instead, the formula type \emph{realizes} the static
type, written $\realizes{\gtys}{\tys}$ and $\realizes{\gtya}{\tya}$, which
means that every solution of its formula distributes the static probabilities
over entries of $=_{s}$-equal simple types, and that the two types mention the
same simple types.}

\begin{definition}[{\rev[fo2]{Realization}}]\label{def:realization}\mechnote{\leanname{GPLC.RealizesTy} and \leanname{GPLC.RealizesD}.}
\rev[fo2]{Realization of a static type of \slang{} by a formula type is defined
by the following rules:}
\begin{mathpar}
  \rev[fo2]{\inference{}{\realizes{\rtype}{\rtype}}} \and
  \rev[fo2]{\inference{}{\realizes{\btype}{\btype}}} \and
  \rev[fo2]{\inference{\realizes{\gtys}{\tys} & \realizes{\gtya}{\tya}}
    {\realizes{\gtys -> \gtya}{\tys -> \tya}}} \and
  \rev[fo2]{\inference{
    \satisfiablee{\phi} &
    R \subseteq \iSet \times \jSet &
    \forall (i,j) \in R.~ \realizes{\gtys[i]}{\tys[j]} \\
    \forall i \in \iSet.~ \exists j.~ (i,j) \in R &
    \forall j \in \jSet.~ \exists i.~ (i,j) \in R \\
    \forall \nu \models \phi.~ \exists (w_{ij} \geq 0)_{i \in \iSet, j \in \jSet}.~
      (\forall i.~ \textstyle\sum_{j} w_{ij} = \nu(\p[i])) \land
      (\forall j.~ \textstyle\sum_{i} w_{ij} = \ps[j]) \\
    \land~ (\forall i, j.~ w_{ij} > 0 \Rightarrow
          \exists j'.~ (i,j') \in R \land \tys[j'] =_{s} \tys[j])
   }
   {\realizes{\db{\phi}{\gtys[i][\p[i]]}[i \in \iSet]}{\d{\tys[j][\ps[j]]}[j \in \jSet]}}}
\end{mathpar}
\end{definition}

\begin{revblock}[fo2]
\begin{example}\label{ex:realization}
The conditional $\src{\ite{\ssb}{1}{(1 \pssum[\frac{1}{2}] 1)}}$ has branches of
types $\ds{\rtype^{1}}$ and $\ds{\rtype^{\frac{1}{2}}, \rtype^{\frac{1}{2}}}$,
which are $=_{s}$-equal, and \glang{} types it at the convex hull
\[
  \gtya = \dctx[\phi] \ds{\rtype^{\cww[1]}, \rtype^{\cww[2]}, \rtype^{\cww[3]}}
  \qquad \phi = \bigl(\cww[1] = t \land \cww[2] = \cww[3] = \tfrac{1-t}{2} \land t \in [0,1]\bigr).
\]
Then $\realizes{\gtya}{\ds{\rtype^{1}}}$, since every solution of $\phi$
distributes probability $1$ over entries of type $\rtype$. In contrast, the
hull
$\dctx[\cww[1] = t \land \cww[2] = 1-t \land t \in [0,1]] \ds{\rtype^{\cww[1]},
\btype^{\cww[2]}}$ of $\src{\ite{\ssb}{1}{\ttt}}$ realizes no static type, since
different solutions give different probabilities to $\rtype$; accordingly,
\slang{} rejects that conditional.
\end{example}
\end{revblock}

\rev[fo1]{With this relation, a term of \slang{} is well typed if and only if
it is well typed in \glang{}, at a formula type that realizes its static
type:}
\begin{restatable}[Conservative extension of the static semantics]{theorem}{staticeq}\label{theorem:staticeq}\mechnote{\leanname{GPLC.static\_conservative\_extension\_forward}, \leanname{GPLC.static\_conservative\_extension\_backward} and \leanname{GPLC.static\_conservative\_extension}; values, \leanname{GPLC.conservative\_forward\_val} and \leanname{GPLC.conservative\_backward\_val}.}~
  \rev[fo1]{For every value $\v$ and every term $\m$ of \slang{}, every simple
  type $\tys \in \Type$, every formula simple type $\gtys \in \FGType$, every
  distribution type $\tya \in \DType$ and every formula distribution type
  $\gtya \in \FGDType$,}
  \begin{enumerate}
    \item \rev[fo1]{If $|-ss \v : \tys$, then $|-t {\color{sourcecolor} \v} :
    \gtysp$ and $\realizes{\gtysp}{\tys}$ for some formula simple type
    $\gtysp \in \FGType$; and if $|-t {\color{sourcecolor} \v} : \gtys$, then
    $|-ss \v : \tysp$ and $\realizes{\gtys}{\tysp}$ for some simple type
    $\tysp \in \Type$.}
    \item \rev[fo1]{If $|-ss \m : \tya$, then $|-d {\color{sourcecolor} \m} :
    \gtyb$ and $\realizes{\gtyb}{\tya}$ for
    some formula distribution type $\gtyb \in \FGDType$; and if
    $|-d {\color{sourcecolor} \m} : \gtya$, then $|-ss \m : \tyap$ and
    $\realizes{\gtya}{\tyap}$ for some distribution type $\tyap \in \DType$.}
  \end{enumerate}
\end{restatable}

\rev[fo1]{Realization also explains why rule (Tif) of \slang{} accepts
branches of $=_{s}$-equal types. The conditional of \glang{} is typed at the
convex hull of its branches, which realizes the type of either branch when
the two are $=_{s}$-equal, and a static rule requiring identical branch types
would reject such conditionals.}

\rev[fo1]{If instead item 2 were stated as an equivalence with the same
$\tya$ on both sides, then it would fail. For instance, \glang{} types the
constant $1$ at $\ds{\rtype^{1}}$, which realizes
$\ds{\rtype^{\frac{1}{2}}, \rtype^{\frac{1}{2}}}$, but
$|-ss 1 : \ds{\rtype^{\frac{1}{2}}, \rtype^{\frac{1}{2}}}$ does not hold.}

\subsection{Dynamic Semantics}
Traditionally, when designing gradual languages, the runtime semantics are not defined directly over the gradual source language. The program is translated or elaborated into a \emph{cast calculus} program, inserting casts at the boundaries between static and dynamic typing, ensuring at runtime that no static assumptions are violated. If a static assumption is violated, then a runtime error is raised. This cast calculus is usually called the gradual target language.
The dynamic semantics of \glang is no exception: taking inspiration from AGT, we elaborate \glang into an evidence-based gradual target language, where evidence plays the role of casts that justify consistency judgments. The gradual target language for \glang, dubbed \tlang, is presented next.

\section{\tlang: Gradual Target Language}\label{sec:target}
In this section, we introduce \tlang, an evidence-based target language for \glang.
We start by presenting the static semantics, followed by the
  dynamic semantics, which relates programs to probability
  distributions over values. Finally, we establish type safety and two
  refined criteria for \tlang (dynamic counterparts of the refined criteria
  already established for \glang): the gradual guarantee, and that the language is a conservative extension of \slang, its static counterpart.

\subsection{Static Semantics}\label{sec:target-static}
\begin{figure}[t]
\begin{displaymath}
  \begin{array}{r@{\hspace{0.3em}}c@{\hspace{0.8em}}l@{\hspace{2em}}l}
    \multicolumn{4}{c}{\ttr\in  \mathbb{R},\quad \ttb \in  \mathbb{B},\quad \tx \in \text{Var},\quad \rev[fo2]{\pty \in \GProb,}\quad \gtys \in \GFType,\quad \gtya \in \GFDType}\\[1.5ex]
    \tm,\tn           
          & ::= & \tv ~|~ \tv \; \tw ~|~ \trg{\lett{x}{\tm}{\rev[mt1]{\overline{\tn}}}} ~|~ \rev[fo2]{\trg{ \tm \tpsum \tn }} ~|~ \trg{ \asc{\evd \tm}{\gtya} }  & \text{(terms)}\\
          &  & \trg{\asc{\ev \tv}{\gtys}} ~|~ \trg{ \ite{\tv}{\tm}{\tn} } ~| ~ \trg{ \add{\tv}{\tw} } ~|~ \errort{\gtya}           & \\
    \tv,\tw 
        & ::= & \tx ~|~ \trg{ \asc{\ev u}{\gtys} } ~|~ \errort{\gtys} & \text{(values)}\\
    \tu & := & \ttr ~|~ \ttb ~|~ \trg{(\lambda x:\gtys. \tm)} & \text{(raw values)}\\[1.5ex]
    \V & ::= &  \dt{\tv[i][\p[i]]~|~ i \in \iSet} & \text{(distribution values)}
  \end{array}
  \end{displaymath}
  \caption{Syntax of \tlang (excerpt).} %
  \label{fig:target-syntax}
 \end{figure}
The static semantics of \tlang differs from \glang in five 
key aspects: (1) we use \fgdt from the beginning, (2) consistency judgments 
are augmented with concrete type information (called \emph{evidence}) that 
justifies judgment validity, (3) explicit ascriptions are incorporated along type derivations to push all consistency judgments to the ascription type rules, (4) ascriptions carry their underlying evidence to justify consistency transitivity at runtime, and (5) to simplify the reduction rules and proofs, all values are ascribed.

\subsubsection*{Syntax}
Figure~\ref{fig:target-syntax} presents the syntax of \tlang. Types are the formula types $\FGType$ and $\FGDType$ from \glang (see Figure~\ref{fig:formula-types}). Terms are now annotated with \fgt and \fgdt from the previous section. 
\rev[fo2]{The probabilistic choice operator $\trg{\tm \tpsum \tn}$ keeps the
gradual probability $\pty$ of \glang, either a concrete probability or $\?$.
The typing and reduction rules turn it into fresh probability variables
$\cww[1]$, $\cww[2]$ for the two branches, closed by the lifting of $\pty$, as
the typing rule of \glang does.}
\rev[mt1]{The body of a $\<let>$ is a \emph{family} of terms $\overline{\tn}$, namely one body $\tn[i]$ per entry of the distribution type of the bound term $\tm$. This generalizes the usual single-body $\<let>$: because the bound variable $\tx$ may receive a different (simple) type in each entry, its body may elaborate differently in each case, and forcing a single body would be unnecessarily restrictive.} \rev[fo1]{The length of the family is the number of entries of that type, so typing remains syntax-directed.}
Ascriptions $\trg{\ev \asc{\tv}{\gtys}}$ and $\trg{\evd \asc{\tm}{\gtya}}$ are 
augmented with evidences, where $\ev$ is an evidence for 
a \sfgt consistency judgment, and $\evd$ for a \sfgdt consistency judgment (both kinds of evidence, to be defined in Section~\ref{sec:evidence}). A raw value is either a real number $\ttr$, a constant $\ttb$ or a lambda 
abstraction $\trg{\lambda \tx:\gtys.\tm}$.
As previously mentioned, all values in \glang become ascribed values in \tlang. 
Therefore, a value $\tv$ is either a variable $\tx$, an ascribed raw value $\tu$, or a tagged error
$\errort{\gtys}$.
Note that in contrast to classical gradual approaches, in \tlang error is also a term, and can be either a redex ($\errort{\gtya}$) or a value ($\errort{\gtys}$). 
The main reason for this is to simplify the metatheory when 
accounting for probabilistic branches that may fail during runtime. Errors also carry type information related to the expected type of the expression in order to establish type safety, and can be removed in a real implementation.
Finally, a \emph{distribution value} $\V$ stands for a distribution over values $\tv$.

\begin{figure}[t]
\begin{flushleft}
\framebox{$\rev[fo1]{\Gamma |-t \tu: \gtys,} \quad \Gamma |-t \tv: \gtys, \quad \Gamma |-d \tm: \gtya,
\quad  \Gamma |-d \phty{\phii}{\V}: \gtya $ }
\end{flushleft}
  \def \MathparLineskip {\lineskip=1.4ex}
\begin{mathpar}
\rev[fo1]{\inference[(Gr)]{}
{\Gamma |-t \ttr : \rtype}} \and
\rev[fo1]{\inference[(Gb)]{}
{\Gamma |-t \ttb : \btype}} \and
\rev[fo1]{\inference[(Gx)]{\Gamma(\tx) = \gtys}
{\Gamma |-t \tx : \gtys}} \and
\inference[(Gerr$_{\gtys}$)]{
  \jtf{}{\gtys}
}
{\Gamma |-t \errort{\gtys} : \gtys } \and
\inference[(Gerr$_{\gtya}$)]{
  \jtf{}{\gtya}
}
{\Gamma |-d \errort{\gtya} : \gtya} \and
\inference[(Gv)]{\Gamma |-t \tv : \gtys}
{\Gamma |-d \tv : \ds{\gtys[][1]} } \and
\inference[(G$\lambda$)]{\Gamma, \tx :\gtys |-d \tm : \gtya  & \jtf{}{\gtys}}
{\Gamma |-t \trg{ \lambda \tx:\gtys. \tm} : \trg{\gtys -> \gtya} } \and
\rev[fo1]{\inference[(G$\tu\mathord{::}\gtys$)]{\Gamma |-t \tu : \gtys &  \rev[fo2]{\vjtf{\ev}{\gtys \erel \gtysp}} & \rev[fo2]{\jtf{}{\ev}} & \jtf{}{\gtysp}}
{\Gamma |-t \trg{ \asc{\ev \tu}{\gtysp}} : \gtysp}}  \and
\inference[(G$\mathord{::}\gtys$)]{\Gamma |-t \tv : \gtys &  \rev[fo2]{\vjtf{\ev}{\gtys \erel \gtysp}} & \rev[fo2]{\jtf{}{\ev}} & \jtf{}{\gtysp}}
{\Gamma |-d \trg{ \asc{\ev \tv}{\gtysp}} : \ds{\gtysp[][1]}  }  \and
\inference[(Gapp)]{
  \Gamma |-t \tv : \gtys -> \gtya  &  
  \Gamma |-t \tw : \gtys & 
}
{\Gamma |-d \tv\;\tw : \gtya} \and
\inference[(Glet)]{
  \Gamma |-d \tm : \phty{\phi}{\d{\gtys[i][\p[i]]}[i \in \iSet]}   &
  \rev[mt1]{\forall i\in \iSet.~\Gamma, \tx : \gtys[i] |-d \tn[i] : \gtya[i]}
}
{\Gamma |-d \trg{ \lett{\tx}{\tm}{\rev[mt1]{\overline{\tn[i]}^{\,i \in \iSet}}} } : \rev[mt1]{\sugarconj{\phi}{\sum_{i \in \iSet} \p[i] \cdot \gtya[i]}}} \and
\inference[(G$\mathord{::}\gtya$)]{\Gamma |-d \tm : \gtya & \rev[fo2]{\vjtf{\evd}{\gtya \erel \gtyb}} & \rev[fo2]{\jtf{}{\evd}} & \jtf{}{\gtyb}
}
{\Gamma |-d \trg{ \asc{\evd \tm}{\gtyb} } : \gtyb} \and
  \inference[(G$+$)]{
    \Gamma |-t \tv :  \rtype   &  
    \Gamma |-t \tw :  \rtype 
  }
  {\Gamma |-d \trg{ \add{\tv}{\tw} } : \ds{\rtype^1} } \and
  \inference[(Gif)]{
    \Gamma |-t \tv : \btype \\
    \Gamma |-d \tm : \gtya &
    \Gamma |-d \tn : \gtyb
  }
  {\Gamma |-d \trg{ \ite{\tv}{\tm}{\tn} } : \rev[mt1]{\gtya \mathbin{\oplus_{\?}} \gtyb} } \and
\rev[fo2]{\inference[(G$\oplus$)]{
  \Gamma |-d \tm : \gtya  &  
  \Gamma |-d \tn : \gtyb \\
  \liftP{\pty}[\cww[1]] = \phi[][1] & \liftP{(1{-}\pty)}[\cww[2]] = \phi[][2] &
  \phi = \phi[][1] \land \phi[][2] \land (\cww[1] + \cww[2] = 1)
}
{\Gamma |-d \trg{ {\tm\:} \tpsum {\: \tn} } : \sugarconj{\phi}{\cww[1] \cdot \gtya + \cww[2] \cdot \gtyb}}} 
\and
\inference[(GV)]{
  \forall i \in \iSet. |-t \tv[i] : \gtys[i] 
  & \rev[fo2]{\jtf{}{\phty{\phi}{ \d{ \gtys[i][\p[i]] }[i \in \iSet] }}}
}
{\Gamma |-d \phty{\phi}{\d{ \pt{\tv[i]}[\p[i]]}[i \in \iSet]}  : \phty{\phi}{ \d{ \gtys[i][\p[i]] }[i \in \iSet] }} 
\end{mathpar}
\caption{Type system of \tlang.}
\label{fig:target-type-system}
\end{figure}

\subsubsection*{Type System}
The type system of \tlang is presented in Figure~\ref{fig:target-type-system}. 
Compared to \glang, the only rules that use consistency are the ascription rules \rev[fo1]{(G$\tu\mathord{::}\gtys$),} (G$\mathord{::}\gtys$) and (G$\mathord{::}\gtya$), making all top-level constructors match in the remaining type rules.
\rev[fo1]{The ascription rules also require the evidence to be well-formed
and \emph{valid} for the judgment (Definition~\ref{def:valid-evidence}), that
is, the coupling that justifies the judgment must be the one that the tags of
the evidence describe.
Rule (G$\oplus$) is the rule of \glang. It lifts $\pty$ into fresh variables
$\cww[1]$ and $\cww[2]$, which the formula $\phi$ fixes when $\pty$ is a
concrete probability, and only requires them to lie in $[0,1]$ and to sum to $1$
when $\pty$ is $\?$.}
Note that formula $\phi$ is also pushed as part of the constraints of the
resulting type---the weighted sum between the branch types.
Similarly, rule (Glet) scales each $\gtya[i]$ with variable $\p[i]$, so $\phi$ is
pushed to the resulting distribution type to close the type.
\rev[fo1]{As in \glang, the rule types one body per entry of the type of the
bound term, and the result is the weighted union of the types of the bodies.}
Both of these rules use the $\p \cdot \gtya$ and $\sugarconj{\phi}{\sum_{i} \gtya[i]}$ metafunctions, \rev[fo1]{defined in Figure~\ref{fig:source-type-system}}. The latter combines formulas and shifts variable indices accordingly.
\rev[fo1]{Also as in \glang, rule (Gif) types the conditional at the convex
hull $\gtya \mathbin{\oplus_\?} \gtyb$ of the types of its branches.}
Consistency judgments are now justified by some evidence, written
$\jtf{\ev}{\gtys \erel \gtysp}$ (for simple types) and 
$\jtf{\evd}{\gtya \erel \gtyb}$ (for distribution types).  
Intuitively, evidences $\ev$ and $\evd$ correspond to the most precise type 
information that support the respective consistency judgment; we
elaborate on this in Section~\ref{sec:evidence}.

\rev[mt1]{The notion of consistency of formula types is defined exactly as in \glang\ (Definition~\ref{def:consistency-source}), and the definition of well-formedness \rev[fo2]{of types and environments} is identical to Definition~\ref{def:well-formed-source}; both are omitted for brevity.}
Like in \slang and \glang, all well-typed terms type-check to well-formed formula types:
\begin{lemma}[{\rev[fo1]{Type well-formedness}}]\label{lemma:welltyped-wellform-target2}\mechnote{\leanname{TPLC.wf\_val} and \leanname{TPLC.wf\_tm}.}
\rev[fo1]{For any well-formed environment $\Gamma$, any value $\tv$ and any
term $\tm$ of \tlang{}, any formula simple type $\gtys \in \GFType$ and any
formula distribution type $\gtya \in \GFDType$,}
\begin{enumerate}
    \item If $\ \Gamma |-t \tv : \gtys$, then $\justify{\gtys}$.
    \item If $\ \Gamma |-d  \tm : \gtya$, then $\justify{\gtya}$.
\end{enumerate}
  \end{lemma}

\subsection{Evidence}\label{sec:evidence}

Following AGT, evidences are  encoded as pairs of (gradual) types of the form
$\langle G_1, G_2 \rangle$. Intuitively, each type of the pair
corresponds to a type in the gradual judgment that the evidence
shall justify, \eg in $\langle G_1, G_2 \rangle \vdash G'_1
\sim G'_2$, $G_1$ corresponds to $G'_1$ and $G_2$ to
$G'_2$. Furthermore, each type in the evidence is at least as precise
as its corresponding type in the judgment, \ie $G_1 \gprec G_1'$ and $G_2 \gprec G_2'$. 
When dealing with consistency (the gradual counterpart of equality), both types in evidence coincide, and therefore evidence is represented by single types, namely
\[
  \egtys~::=~~\gtys \quad \text{(simple evidences)}%
    \qquad\qquad
    \gtyd~::=~~\gtya \quad \text{(distribution evidences)}
\]
A simple evidence $\ev$ (resp.~distribution evidence $\evd$) is just a \sfgt (resp.~\sfgdt) that justifies a consistency judgment between two \fgt (resp.~two \fgdt).
\rev[fo1]{The definition below uses a consistency and a precision that differ
from the ones of Section~\ref{sec:source}. Recall that those relations have
three premises, namely a coupling and two coverage clauses, which require
every entry of one type to have a related entry in the other. They apply to
annotations, which the programmer writes. Evidences are instead computed at
runtime, and the two relations that apply to them, written $\erel$ and
$\eprec$, keep the coupling and drop the two coverage clauses.
Figure~\ref{fig:runtime-relations} gives their rules for distribution types.
The remaining rules are those of $\rel$
(Definition~\ref{def:consistency-source}) and of $\gprec$
(Figure~\ref{source-type-precision}), with $\erel$ and $\eprec$ in the
recursive positions.

On simple types, $\erel$ and $\eprec$ coincide with $\rel$ and $\gprec$. On
distribution types, they allow an entry with no related entry in the other
type, provided that the coupling gives it probability zero. For instance,
$\ds{\rtype^{1}}$ and $\ds{\rtype^{1}, \btype^{0}}$ are related by $\erel$ but
not by $\rel$. Example~\ref{ex:coverage-mono} shows why the runtime relations
must drop the coverage clauses.}

\begin{figure}[t]
\rev[mt1]{
\begin{mathpar}
\inference{
   \couplingLift{\db{\phi[][1]}{\gtys[i][\p[i]]}[i \in \iSet]}{\db{\phi[][2]}{\gtys[j][\p[j]]}[j \in \jSet]}{\erel}
}
{\db{\phi[][1]}{\gtys[i][\p[i]]}[i \in \iSet] \erel
  \db{\phi[][2]}{\gtys[j][\p[j]]}[j \in \jSet]} \and
\inference{
   \rev[fo1]{\forall\: \FV(\phi[][1]).~  \phi[][1]  \implies \exists~\FV(\phi[][2]) \cup
  \{\cww[ij] ~|~ i \in \iSet \land j \in \jSet\}.} \\
  \rev[fo1]{\cjudgext{\d{\cww[ij]}[i \in \iSet \land j \in \jSet]}{\d{\gtys[i][\p[i]]}[i \in \iSet]}{\, \eprec\,}{\d{\gtys[j][\p[j]]}[j \in \jSet]}{\phi[][1]}{\phi[][2]}{}}
}
{\db{\phi[][1]}{\gtys[i][\p[i]]}[i \in \iSet] \eprec
  \db{\phi[][2]}{\gtys[j][\p[j]]}[j \in \jSet]}
\end{mathpar}}
\caption{\rev[mt1]{Runtime consistency and precision of formula distribution types.}}
\label{fig:runtime-relations}
\end{figure}

Formally, an evidence justifies a consistency judgment iff the evidence is at least as precise as each type:
\begin{definition}[Evidence]\label{def:evidence}\mechnote{\leanname{TPLC.EEvTy} and \leanname{TPLC.EEvD}.}
\rev[fo1]{The judgment $\jtf{\ev}{\gtys \erel \gtysp}$, for formula simple
  types $\ev, \gtys, \gtysp \in \GFType$, and the judgment
  $\jtf{\evd}{\gtya \erel \gtyb}$, for formula distribution types
  $\evd, \gtya, \gtyb \in \GFDType$, are defined as follows:}
  \begin{enumerate}
    \item $\jtf{\ev}{\gtys \erel \gtysp}$ iff $\ev \eprec \gtys$ and $\ev \eprec \gtysp$.
    \item $\jtf{\evd}{\gtya \erel \gtyb}$ iff $\evd \eprec \gtya$ and $\evd \eprec \gtyb$.
  \end{enumerate}
\end{definition}
For instance, $\jtf{\rtype -> \btype}{\rtype -> \? \rev[mt1]{\erel} \? -> \btype}$. Now we can justify the role of tags in tagged variables. Consider judgment $\jtf{\dctx[\phi][\dt{\gtys[\kk][\cww[\kk]]}]}{\dctx[\phi[][1]][\dt{\gtys[i][\p[i]]}] \erel \dctx[\phi[][2]][\dt{\gtysp[j][\p[j]]}]}$. Tagged variables connect evidence with their underlying types, \ie type $\gtys[\kk]$ justifies that $\gtys[i]$ is consistent with $\gtysp[j]$, because $\projl{\cww[k]} = i$ and $\projr{\cww[k]}=j$, where $\projv{\cww[\kk]}$ is the \emph{weight} of the connection between $\gtys[i]$ and $\gtysp[j]$. Note that a pair of simple types can be connected through multiple evidences.

\rev[fo1]{Definition~\ref{def:evidence} accepts any coupling between the
evidence and each of the two types. The evidences that occur in a term must
satisfy a stronger condition, namely that the coupling be the one that their
tags describe, at every depth, because the reduction rules of
Section~\ref{sec:dynamics} route values along the tags. We call such
evidences \emph{valid}.}

\begin{definition}[{\rev[fo2]{Valid evidence}}]\label{def:valid-evidence}\mechnote{\leanname{TPLC.HVTagTy} and \leanname{TPLC.TagD.HValidFor}.}
\rev[fo2]{Left validity $\vjtfl{\ev}{\gtys}$ of an evidence for a type is
defined by the following rules; right validity $\vjtfr{\ev}{\gtys}$ is defined
by the same rules with $\projr{\cdot}$ in place of $\projl{\cdot}$.}
\begin{mathpar}
  \rev[fo2]{\inference{}{\vjtfl{\rtype}{\rtype}}} \and
  \rev[fo2]{\inference{}{\vjtfl{\btype}{\btype}}} \and
  \rev[fo2]{\inference{}{\vjtfl{\ev}{\?}}} \and
  \rev[fo2]{\inference{\vjtfl{\ev}{\gtys} & \vjtfl{\evd}{\gtya}}
    {\vjtfl{\ev -> \evd}{\gtys -> \gtya}}} \and
  \rev[fo2]{\inference{
     \forall \kk \in \kSett.~ \projl{\cww[\kk]} \in \iSet \land
       \gtys[\kk] \eprec \gtysp[\projl{\cww[\kk]}] \land
       \vjtfl{\gtys[\kk]}{\gtysp[\projl{\cww[\kk]}]} \\
     \forall \nu \models \phi.~ \exists \nu' \models \phi[][1].~
       \textstyle\bigwedge_{i \in \iSet} \nu'(\p[i]) = \sum_{\kk \mid \projl{\cww[\kk]} = i} \nu(\cww[\kk])
    }
    {\vjtfl{\dctx[\phi][\d{\gtys[\kk][\cww[\kk]]}[\kk \in \kSett]]}{\db{\phi[][1]}{\gtysp[i][\p[i]]}[i \in \iSet]}}}
\end{mathpar}
\rev[fo2]{An evidence is valid for a judgment when it is left valid for its
left side and right valid for its right side: $\vjtf{\ev}{\gtys \erel
\gtysp}$ iff $\vjtfl{\ev}{\gtys}$ and $\vjtfr{\ev}{\gtysp}$, and likewise
$\vjtf{\evd}{\gtya \erel \gtyb}$.}
\end{definition}

\rev[fo1]{A valid evidence is an evidence in the sense of
Definition~\ref{def:evidence}. The converse does not hold, since
Definition~\ref{def:evidence} admits couplings that split the probability of an
entry of the evidence between two entries of the type.}

\begin{restatable}[{\rev[fo2]{Valid evidence is evidence}}]{lemma}{validevidence}\label{lemma:valid-evidence}\mechnote{\leanname{TPLC.valid\_evidence\_ty} and \leanname{TPLC.valid\_evidence\_d}.}
\rev[fo2]{For any well-formed evidences $\ev$ and $\evd$\rev[fo1]{, any formula
simple types $\gtys, \gtysp \in \GFType$ and any formula distribution types
$\gtya, \gtyb \in \GFDType$}, if
$\vjtf{\ev}{\gtys \erel \gtysp}$\rev[fo1]{,} then $\jtf{\ev}{\gtys \erel \gtysp}$, and if
$\vjtf{\evd}{\gtya \erel \gtyb}$\rev[fo1]{,} then $\jtf{\evd}{\gtya \erel \gtyb}$.}
\end{restatable}

As usual in gradual languages, consistency is not transitive, \eg $\ds{\rtype^{\frac{2}{3}}, \btype^{\frac{1}{3}}} \sim \ds{\?^{\?}}$ and $\ds{\?^{\?}} \sim \ds{\rtype^{\frac{1}{3}}, \btype^{\frac{2}{3}}}$, but $\ds{\rtype^{\frac{2}{3}}, \btype^{\frac{1}{3}}} \not\sim \ds{\rtype^{\frac{1}{3}}, \btype^{\frac{2}{3}}}$. Therefore, during runtime, evidence is combined to try to justify transitivity. If the combination succeeds, the resulting (and possibly more precise) evidence justifies the resulting judgment from transitivity; otherwise a runtime error is raised.
The combination of evidence is formalized using the \emph{consistent transitivity} operator, which coincides with the meet (least upper bound) operator \wrt the (im)precision order, \ie
\[
  \ev[1] \trans{} \ev[2] ~=~ \ev[1] \meet \ev[2]%
  \qquad\quad\text{and}\qquad\quad
  \evd[1] \trans{} \evd[2] ~=~ \evd[1] \meet \evd[2]
\]

\subsubsection*{Meet Operator}
The meet operator is partial.\mechnote{simple types, \leanname{TPLC.meetTy}; distribution types, \leanname{TPLC.meetD}, an instance of the construction \leanname{TPLC.witness}.} For simple types, it is defined by the following clauses:
\begin{mathpar}
\rtype \meet \rtype ~=~ \rtype \and
\btype \meet \btype ~=~ \btype \and
\? \meet \gtys ~=~ \gtys \and
\gtys \meet \? ~=~ \gtys \and
\gtys[1] -> \gtya[1] \meet \gtys[2] -> \gtya[2] ~=~
\gtys[1] \meet \gtys[2] -> \gtya[1] \meet \gtya[2] \and
\end{mathpar}
\begin{revblock}[fo1]
For distribution types, the definition follows the same approach as used for
defining consistency and precision, in terms of couplings that witness the
lifting; but where consistency asks for the \emph{existence} of a witness
coupling, the meet \emph{records} all of them, as a formula over fresh
variables. Formally,
$$\gtya[1] \meet \gtya[2] ~=~ \witness{\gtya[1]}{\gtya[2]}{\meet}$$
where $\witnessname \colon (\GFType \times \GFType \rightharpoonup \GFType) \times \GFDType \times \GFDType \rightharpoonup \GFDType$ is defined, for
$\gtya[1] = \dctx[\phi[][1]][\d{\gtys[i][\cww[i]]}[i \in \iSet]]$ and
$\gtya[2] = \dctx[\phi[][2]][\d{\gtysp[j][\cww[j]]}[j \in \jSet]]$, by
\begin{align*}
\witness{\gtya[1]}{\gtya[2]}{f} ~&=~
\dctx[\phii] \d{\pt{f(\gtys[i], \gtysp[j])}[\cww[ij]]}[(i,j) \in \kSet]
\qquad \text{provided } \satisfiablee{\phii}\\[0.5ex]
\text{where}\qquad
\kSet &= \{ (i,j) \in \iSet \times \jSet \mid (\gtys[i], \gtysp[j]) \in \dom(f) \} \\
\cww[ij] &= \pr{\varr[ij], \projl{\cww[i]}, \projr{\cww[j]}}, \quad \varr[ij] \text{ fresh} \\
\phii &= \phi[][1] \land \phi[][2]
  \land \bigwedge_{(i,j) \in \kSet} \cww[ij] \in [0,1] \\
  &\qquad \land \bigwedge_{i \in \iSet} \cww[i] = \textstyle\sum_{j \mid (i,j) \in \kSet} \cww[ij]
  \land \bigwedge_{j \in \jSet} \cww[j] = \textstyle\sum_{i \mid (i,j) \in \kSet} \cww[ij]
\end{align*}
The formula $\phii$ is the coupling condition of
Definition~\ref{def:symcoupling} between the two operands, where the pairs
$(i,j)$ outside $\dom(f)$, whose weight the relation clause of the coupling
would force to zero, are left out of the index set $\kSet$. Hence $\phii$
mentions no relation and binds no variable, and it is a formula of
Figure~\ref{fig:formula-types}. The tags of each fresh variable $\cww[ij]$
connect the entry $(i,j)$ to the entries $i$ and $j$ of the operands. The side
condition $\satisfiablee{\phii}$ (Section~\ref{sec:GPCL-syntax}) is not part
of the stored formula.
\end{revblock}

\rev[fo1]{The meet $\gtya[1] \meet \gtya[2]$ is therefore defined if and only
if $\gtya[1] \erel \gtya[2]$ (Figure~\ref{fig:runtime-relations}). In
particular, an entry with no compatible entry in the other operand does not
make the meet undefined; the formula forces its weights to zero.}
\begin{example}\label{ex:meet}
Let 
\begin{align*}
  \gtya[1] &= \phty{(\cww[1]= \tfrac{1}{2} \land \cww[2]= \tfrac{1}{2})}{\ds{(\rtype -> \?)^{\cww[1]}, (\? -> \rtype)^{\cww[2]} }}\\
  \gtya[2] &= \phty{\rev[fo1]{(\cwwp[1]= \tfrac{1}{3} \land \cwwp[2]= \tfrac{2}{3})}}{\ds{(\rtype -> \?)^{\cwwp[1]},  (\? -> \rtype)^{\cwwp[2]} } }
\end{align*}
Then $$\gtya[1] \meet \gtya[2] = \phty{\phi}{\ds{(\rtype -> \?)^{\fcw{1}{1}},  (\rtype -> \rtype)^{\fcw{1}{2}}, (\rtype -> \rtype)^{\fcw{2}{1}}, (\? -> \rtype)^{\fcw{2}{2}} }}$$
where
\begin{multline*}
  \phi = (\fcw{1}{1} + \fcw{1}{2} = \cww[1] \land  \fcw{2}{1} + \fcw{2}{2} = \cww[2] \land \fcw{1}{1} + \fcw{2}{1} = \cwwp[1] \land
  \fcw{1}{2} + \fcw{2}{2} = \cwwp[2]) \, \land\\ (\cww[1]= \tfrac{1}{2} \land \cww[2]= \tfrac{1}{2}) \land (\cwwp[1]= \tfrac{1}{3} \land \cwwp[2]= \tfrac{2}{3})
\end{multline*}
since $\phi$ is satisfiable by the solution set $\{ \fcw{1}{2} = \frac{1}{2} - \fcw{1}{1}, \fcw{2}{1} = \frac{1}{3}- \fcw{1}{1}, \fcw{2}{2} =  \fcw{1}{1} + \frac{1}{6} \}$. 
\end{example}

\rev[fo1]{If instead the relations on runtime objects carried the coverage
clauses of the relations on annotations, then the monotonicity of evidence
composition (Lemma~\ref{lemma:mono}) and the dynamic gradual guarantee
(Theorem~\ref{theorem:dgg}) would fail. The next example shows it on two
\glang{} programs that differ in one binder and one annotation.}

\begin{example}\label{ex:coverage-mono}
  \rev[fo1]{Let
  \begin{align*}
    \gtys[A] &= \btype -> \ds{\rtype^{1}} &
    \gtys[A][\prime] &= \? -> \ds{\rtype^{1}} \\
    \gtys[B] &= \rtype -> \ds{\btype^{1}} &
    \gtys[B][\prime] &= \rtype -> \ds{\?^{1}}
  \end{align*}
  so that $\gtys[A] \gprec \gtys[A][\prime]$ and $\gtys[B] \gprec \gtys[B][\prime]$, and
  consider the functions
  \begin{align*}
    \sv &= \src{\lambda x{:}\rtype.\, ((\lambda b{:}\btype.\, 1) \pssum[\?]
           (\lambda y{:}\rtype.\, true))}, &
    \sw &= \src{\lambda x{:}\rtype.\, ((\lambda b{:}\?.\, 1) \pssum[\?]
           (\lambda y{:}\rtype.\, true))},
  \end{align*}
  of types $\rtype -> \ds{\gtys[A][\?], \gtys[B][\?]}$ and
  $\rtype -> \ds{\gtys[A][\prime\?], \gtys[B][\?]}$, and the programs
  \begin{align*}
    \sm   &= \src{\lett{x}{\asc{\sv}{\gtys[1]}}{\asc{x}{\gtys[2]}}}, &
    \srcn &= \src{\lett{x}{\asc{\sw}{\gtys[1][\prime]}}{\asc{x}{\gtys[2]}}},
  \end{align*}
  where $\gtys[1] = \gtys[2] = \rtype -> \ds{\gtys[A][\?], \gtys[B][\?]}$
  and $\gtys[1][\prime] = \rtype -> \ds{\gtys[A][\prime\?], \gtys[B][\prime\?]}$. Both
  programs type-check, and $\sm \gprec \srcn$.

  The elaboration gives each ascription the meet of the type of the ascribed
  term and the annotation as evidence. Let $\ev[1]$ and $\ev[2]$ be the two
  evidences of $\sm$, and $\ev[1][\prime]$ and $\ev[2][\prime]$ those of $\srcn$. The
  codomains of $\ev[1]$ and $\ev[2]$ have the two entries $\gtys[A]$ and
  $\gtys[B]$, and the codomain of $\ev[2][\prime]$ has the same two entries. The
  codomain of $\ev[1][\prime]$ has a third entry,
  \[ \gtys[A][\prime] \meet \gtys[B][\prime] =
     (\? -> \ds{\rtype^{1}}) \meet (\rtype -> \ds{\?^{1}}) =
     \rtype -> \ds{\rtype^{1}}. \]
  This entry is related to no entry of the other evidences. It is neither
  less precise than nor consistent with $\gtys[A]$, whose domain is $\btype$,
  and likewise for $\gtys[B]$, whose codomain is $\ds{\btype^{1}}$.

  With the coverage clauses, $\ev[1] \eprec \ev[1][\prime]$ would therefore not
  hold, although Lemma~\ref{lemma:mono} requires it. Moreover, the
  composition $\ev[1][\prime] \trans{} \ev[2][\prime]$ that the body of $\srcn$ performs
  would be undefined, so $\srcn$ would raise an error where $\sm$ returns a
  value, and Theorem~\ref{theorem:dgg} would fail. Without the coverage
  clauses, the precision holds and the composition is defined, since the
  couplings give the third entry probability zero. The relations on annotations keep the coverage clauses, which the
  static gradual guarantee requires (Example~\ref{ex:coverage-sgg}).}
\end{example}

  Importantly, the meet between a pair of types is at least as precise as either of them (and therefore, a valid  evidence for their consistency). \rev[fo1]{It is moreover the greatest such type, and it is defined whenever the two types have a well-formed common lower bound.}
  \begin{restatable}[Reductivity of the meet operator]{lemma}{meetreductivity}\label{meet-more-precise2}\mechnote{item~1, \leanname{TPLC.eprec\_meetTy}; item~2, \leanname{TPLC.eprec\_meetD}; item~3, \leanname{TPLC.eprec\_meetTy\_glb}, \leanname{TPLC.eprec\_meetD\_glb}, \leanname{TPLC.consTy\_of\_common\_lb} and \leanname{TPLC.consD\_of\_common\_lb}.}
    \rev[fo1]{For any} formula simple types $\gtys[1], \gtys[2], \gtys[3] \in \GFType$ and \rev[fo1]{any} formula distribution types $ \gtya[1],  \gtya[2],  \gtya[3] \in \GFDType$,
    \begin{enumerate}
      \item \rev[fo1]{If $\jtf{}{\gtys[1]}$, $\jtf{}{\gtys[2]}$ and $\gtys[3] = \gtys[1] \meet \gtys[2]$, then $\gtys[3] \eprec \gtys[1]$ and $\gtys[3] \eprec \gtys[2]$.}
     \item \rev[fo1]{If $\jtf{}{\gtya[1]}$, $\jtf{}{\gtya[2]}$ and $\gtya[3] = \gtya[1] \meet \gtya[2]$, then $\gtya[3] \eprec \gtya[1]$ and $\gtya[3] \eprec \gtya[2]$.}
    \item \rev[fo1]{If $\jtf{}{\gtys[3]}$, $\gtys[3] \eprec \gtys[1]$
      and $\gtys[3] \eprec \gtys[2]$, then $\gtys[1] \meet \gtys[2]$ is defined and
      $\gtys[3] \eprec \gtys[1] \meet \gtys[2]$; likewise, if $\jtf{}{\gtya[3]}$,
      $\gtya[3] \eprec \gtya[1]$ and $\gtya[3] \eprec \gtya[2]$, then
      $\gtya[1] \meet \gtya[2]$ is defined and $\gtya[3] \eprec \gtya[1] \meet \gtya[2]$.}
  \end{enumerate}
  \end{restatable}
  \rev[fo1]{The lemma is stated for the runtime precision $\eprec$. If
  instead it were stated for the precision $\gprec$ on annotations, which
  carries the two coverage clauses, then it would fail;
  Appendix~\ref{app:meet} gives the witnesses.
  The lemma also yields a characterization of runtime consistency by
  precision. For well-formed $\gtya[1]$ and $\gtya[2]$,
  \[
    \gtya[1] \erel \gtya[2]
    \quad\text{iff}\quad
    \exists \gtya[3].\ \jtf{}{\gtya[3]} \land \gtya[3] \eprec \gtya[1] \land \gtya[3] \eprec \gtya[2],
  \]
  and likewise for simple types. A common lower bound makes the meet defined
  by part~3, and when the two types are consistent, their meet is a common
  lower bound by parts 1 and~2. With the precision on annotations, the characterization fails as well: the two types of the reductivity witness in Appendix~\ref{app:meet} are consistent but have no common lower bound with coverage.}

\subsubsection*{Evidence construction via consistent transitivity}

Armed with the definition of evidence and the meet operator, we now present a set of sufficient conditions (in the form of a proof system) for constructing the evidence of consistency judgments: 
\begin{mathpar}
\inference[(wdB)]{
  \ev \in \{\rtype,\btype, \?\} &
  \ev \rev[mt1]{\eprec} \gtys & \ev \rev[mt1]{\eprec} \gtysp
  }{\ev |- \gtys \rev[mt1]{\erel} \gtysp} \and
\inference[(wd$->$)]{
  \ev |- \gtys \rev[mt1]{\erel} \gtysp &
  \evd |- \gtya \rev[mt1]{\erel} \gtyb
}{\ev -> \evd |- \gtys -> \gtya \rev[mt1]{\erel} \gtysp -> \gtyb} \and
{ 
\inference[(wd$\evd$)]{
  \exists\ \FV(\phi[][]) \cup  \FV(\phi[][1]) \cup  \FV(\phi[][2]).\ %
    \phi[][] \land \phi[][1] \land \phi[][2] \land  \phi[][L] \land \phi[][R] \land  \phi[][\sim]\\
   \phi[][L] = \forall i \in \iSet.\ \sum_{\kk|\projl{\cww[k]} = i}
  \cww[\kk] 
  = \p[i] &
   \phi[][R] = \forall j \in \jSet, \sum_{\kk|\projr{\cww[k]} = j}
  \cww[\kk] 
  = \p[j] \\
  \phi[][\sim] = \forall \kk \in \kSett. \cww[\kk] > 0  =>  \gtys[\kk] |- \gtys[\projl{\cww[k]}] \rev[mt1]{\erel} \gtys[\projr{\cww[k]}]
}
{\dctx[\phi[][]][\d{\gtys[\kk][\cww[\kk]]}[\kk \in \kSett]] |- \db{\phi[][1]}{\gtys[i][\p[i]]}[i \in \iSet] \rev[mt1]{\erel} \db{\phi[][2]}{\gtys[j][\p[j]]}[j \in \jSet]}
}
\and
\end{mathpar}
The conditions require that (1) the sum of all the weights connected to a simple type be equal to the probability of that type, and (2) each evidence in the distribution 
evidence of weight greater than zero be well-defined (and thus more precise than the pair of types involved in the consistency judgment).

Appealing to these sufficient conditions, we show that the consistent transitivity operator indeed allows an incremental construction of evidence:

\begin{restatable}[Evidence from consistent transitivity]{lemma}{transinvariant}\label{lemma:transinvariant}\mechnote{\leanname{TPLC.transTy\_invariant} and \leanname{TPLC.transD\_invariant}.}
\rev[fo1]{For any well-formed formula simple types $\ev[1], \ev[2] \in \GFType$,
any formula simple types $\gtys[1], \gtys[2], \gtysp \in \GFType$, any
well-formed formula distribution types $\evd[1], \evd[2] \in \GFDType$ and any
formula distribution types $\gtya[1], \gtya[2], \gtyb \in \GFDType$,}
  \begin{enumerate}
    \item \rev[fo1]{If $\jtf{\ev[1]}{\gtys[1] \erel \gtysp}$,
      $\jtf{\ev[2]}{\gtysp \erel \gtys[2]}$ and $\ev[1] \trans{} \ev[2]$ is
      defined, then $\jtf{\ev[1] \trans{} \ev[2]}{\gtys[1] \erel \gtys[2]}$.}
    \item \rev[fo1]{If $\jtf{\evd[1]}{\gtya[1] \erel \gtyb}$,
      $\jtf{\evd[2]}{\gtyb \erel \gtya[2]}$ and $\evd[1] \trans{} \evd[2]$ is
      defined, then $\jtf{\evd[1] \trans{} \evd[2]}{\gtya[1] \erel \gtya[2]}$.}
  \end{enumerate}
\end{restatable}

\rev[fo1]{Consistent transitivity also preserves the validity of evidences
(Lemma~\ref{lemma:trans-validity} in Appendix~\ref{app:target-safety}).}

\subsection{Dynamic Semantics}\label{sec:dynamics}
We now present the dynamic semantics for \tlang, which relates programs to probability distributions over final values, through a big-step reduction relation (like in \slang). 
\begin{figure}[t]
    \begin{flushleft}
      \framebox{$ \tm \nreds{}{\j} \rctx[\phi] \V$}
      \end{flushleft}
      \vspace{-2em}
      \begin{displaymath}
    \begin{array}{r@{\hspace{0.3em}}c@{\hspace{0.8em}}l@{\hspace{0.8em}}l}
       \V & ::= &  \dt{\tv[i][\p[i]]~|~ i \in \iSet} & \text{(distribution values)}
    \end{array}
  \end{displaymath}
  \begin{mathpar}
  \inference[\trg{(\mathit{Dapp})}]{ \rrctx[\phii[]] \trg{ \asc{\cdom(\ev[])\tv}{\gtysp} } \nreds{}{1} \rctx[\cdot] \dt{\pt{\tw}} &
  \rrctx[\phii[']]  \ssub{\trg{(\asc{\ccod(\ev[])\tm[]}{\gtya})}}{\tw}{\tx}  \nreds{}{\j}  \rctx[\phi] \V }
  {\rrctx[\phii[]] \trg{(\asc{\ev[] (\lambda \tx: \gtysp. \tm)}{\gtys -> \gtya}) \; \tv} \nreds{}{\j+1} \rctx[\phi] \V } \and
  \rev[fo2]{{\setpremisesspace{1em}\inference[\trg{(\mathit{D}\oplus)}]{\rrctx[\phi[][]] {\tm} \nreds{}{\j[1]} \rctx[\phi[][1]] \V[1] &
  \rrctx[\phi[][]]  {\tn} \nreds{}{\j[2]} \rctx[\phi[][2]] \V[2] \\
  \cww[1], \cww[2] \; \text{fresh} &
  \phi[][\oplus] = \liftP{\pty}[\cww[1]] \land \liftP{(1{-}\pty)}[\cww[2]] \land (\cww[1] + \cww[2] = 1) &
  \phip = \phi[][1] \land \phi[][2] \land \phi[][\oplus] }
  {\rrctx[\phi[][]]  \trg{ {\tm} \tpsum {\tn} } \nreds{}{\j[1]+\j[2]+1} \rctx[\phip] \cww[1] \cdot \V[1] + \cww[2] \cdot \V[2] }}}\and
  \inference[\trg{(Dv)}]{}
  {\rrctx[\phii[]] \tv \nreds{}{1} \rctx[\cdot] \dt{\pt{\tv}} }
  \and 
  \inference[\trg{(\mathit{Derr})}]{
    \gtya = \dctx[\phi] \d{\gtys[i][\p[i]]}[i \in \iSet]
  }{
    \rrctx[\phi[]] \errort{\gtya} \nreds{}{1}
    \rctx[\phi] \d{\errort{\gtys[i]}^{\p[i]}}[i  \in \iSet] 
  }
  \and
  \rev[mt1]{  \inference[\trg{(\mathit{Dlet})}]{
   \rrctx[\phii[]] \tm \nreds{}{\j[1]} \rrctx[\phii[']] \V &
   \V = \d{\tv[i][\p[i]]}[i \in \iSet] &
   \jtf{}{\phty{\phii[']}{\V} : \gtyap} &
   \jtf{}{\tm : \gtya} \\
   \gtya = \phty{\phi}{\d{\gtysp[j][\p[j]]}[j \in \jSet]} &
   \initReorder{\gtyap}{\gtya} = \phty{\pphi[2]}{\d{\ev[k][\cww[k]]}[k \in \kSet]} \\
   \forall k \in \kSet.~ (i = \projl{\cww[k]} \land j = \projr{\cww[k]}) \implies
   {\begin{block}
     \trg{(\asc{\ev[k] \tv[i]}{\gtysp[j]})} \nreds{}{1} \rctx[\cdot] \dt{\pt{\tw[k]}} \quad\land\\
     \rrctx[\phii[']] \ssub{\tn[j]}{\tw[k]}{\tx} \nreds{}{\j[2]} \rctx[\phi[][k]] \V[k]
   \end{block}}
  }
  { \rrctx[\phii[]] \trg{\lett{\tx}{\tm}{\overline{\tn[j]}}} \nreds{}{\j[1]+\j[2]+1} \rctx[\pphi[2] \land (\bigwedge_{k \in \kSet} \rev[mt1]{\cww[k] > 0 \Rightarrow} \phi[][k])] \ssum[k \in \kSet] \cww[k] \cdot \V[k]  }}
  \and
  \inference[\trg{(\mathit{D+})}]{ \ev[1] \trans{} \ev[2] = \ev[3] & \trg{r_3} = \trg{r_1} + \trg{r_2}   }
  { \rrctx[\phii] \trg{\add{\ev[1] \asc{r_1}{\rtype} }{ \asc{\ev[2] r_2}{\rtype} }} \nreds{}{1}  \rctx[\cdot] \dt{\pt{\trg{\asc{\ev[3] r_3}{\rtype}}}}
  }
  \end{mathpar}
    \begin{align*}
    \omit\rlap{$\mathit{sub} :\Term \times \Value \times \Var \rightharpoonup \Term$}\\[-0.5ex]
    \ssub{\tm[]}{\asc{\ev \tu }{\gtys}}{\tx}  & =  \tm[][\asc{\ev \tu}{\gtys} / \tx] \\
    \ssub{\tm[]}{\errort{\gtys}}{\tx} &= \errort{\gtya} \quad
     \text{where } \tx : \gtys |- \tm : \gtya 
    \end{align*}
  \caption{Distribution semantics of \tlang (part 1).}
  \label{fig:target-reduction-dis}
\end{figure}

\begin{figure}
\begin{flushleft}
\framebox{$ \tm \nreds{}{\j} \rctx[\phi] \V$}
\end{flushleft}
\begin{mathpar}
  \inference[\trg{(\mathit{Dit})}]{
    \tm \nreds{}{\j} \rctx[\phi]  \V 
 }{ \rrctx[\phii] \trg{\ite{ \asc{\ev \ttt}{\btype} }{\tm}{\tn}} 
     \nreds{}{\j+1}  \rctx[\phi]  \V
  } \and 
 \inference[\trg{(\mathit{Dif})}]{
    \tn \nreds{}{\j} \rctx[\phi]  \V 
 }{ \rrctx[\phii] \trg{\ite{ \asc{\ev \fff}{\btype} }{\tm}{\tn}} 
     \nreds{}{\j+1}  \rctx[\phi]  \V
  } \and 
  
  \inference[\trg{(\mathit{Dmon})}]{ \tm \nreds{}{k} \V}
  { \tm \nreds{}{k+1} \V}
  \and
  \inference[\trg{(\mathit{D}\mathord{::}\gtys)}]{}
  { \rrctx[\phii[]] \trg{\asc{\ev[2](\asc{\ev[1] \tu}{\gtys})}{\gtysp}} \nreds{}{1}  \rctx[\cdot]
    \begin{cases}
      \dt{\pt{\trg{(\asc{\ev[3] \tu }{\gtysp})}}} &  \text{If}~ \ev[1] \trans{} \ev[2] = \ev[3]  \\
      \dt{\pt{\errort{\rev[mt1]{\gtysp}}}} & \text{otherwise}
    \end{cases}
  }
  \and
  \rev[mt1]{
  \inference[\trg{(\mathit{Derr}\mathord{::}\gtys)}]{}
  { \rrctx[\phii[]] \trg{\asc{\ev \errort{\gtys}}{\gtysp}} \nreds{}{1}
    \rctx[\cdot] \dt{\pt{\errort{\gtysp}}} }
  }
  \and
  \rev[fo2]{\inference[\trg{(\mathit{Derr}{+})}]{}
  { \rrctx[\phii[]] \trg{\add{\errort{\gtys}}{\tw}} \nreds{}{1}
    \rctx[\cdot] \dt{\pt{\errort{\rtype}}} }}
  \and
  \rev[fo2]{\inference[\trg{({+}\mathit{Derr})}]{}
  { \rrctx[\phii[]] \trg{\add{\tv}{\errort{\gtys}}} \nreds{}{1}
    \rctx[\cdot] \dt{\pt{\errort{\rtype}}} }}
  \and
  \rev[fo2]{\inference[\trg{(\mathit{Derr\,app})}]{
    \gtya = \dctx[\phi] \d{\gtys[i][\p[i]]}[i \in \iSet]}
  { \rrctx[\phii[]] \trg{\errort{\gtysp -> \gtya}\; \tw} \nreds{}{1}
    \rctx[\phi] \d{\errort{\gtys[i]}^{\p[i]}}[i \in \iSet] }}
  \and
  \rev[fo2]{\inference[\trg{(\mathit{Derr\,if})}]{
    |-d \tm : \gtya[1] & |-d \tn : \gtya[2] &
    \gtya[1] \mathbin{\oplus_{\?}} \gtya[2] = \dctx[\phi] \d{\gtys[i][\p[i]]}[i \in \iSet]}
  { \rrctx[\phii[]] \trg{\ite{\errort{\gtys}}{\tm}{\tn}} \nreds{}{1}
    \rctx[\phi] \d{\errort{\gtys[i]}^{\p[i]}}[i \in \iSet] }}
  \and 
    \rev[mt1]{\setpremisesend{0em}%
    \inference[\trg{(\mathit{D}\mathord{::}\gtya)}]{
      \rrctx[\phi[][1]] \tm \nreds{}{\kp} \rctx[\phi[][1]] \d{\tv[i][\p[i]]}[i \in \iSet]
      &  |-d \phty{\phi[][1]}{\d{\tv[i][\p[i]]}[i \in \iSet] }: \gtyap \\
      \rev[fo2]{|-d \tm : \gtya} &
      \rev[fo1]{\vjtf{\evd}{\gtya \erel \gtyb}}
       &
       \gtyb = \phty{\pphi[3]}{ \d{\gtysp[j][\p[j]]}[j \in \jSet]}
      }
      { \rrctx[\phi[][1]] \trg{ (\asc{\evd \tm}{ \gtyb }) } \nreds{}{\kp+1} %
      \begin{cases}
        \rctx[\phi[][2] ] \ssum[\kk \in \kSet] 
        \cww[k] \cdot \V[\kk]  
        & \begin{block}
          \text{If}~ (\initReorder{\gtyap}{\gtya}) \trans{} \evd = \phty{\pphi[2]}{\d{\ev[k][\cww[k]]}[k \in \kSet]}, \\ \text{where}~  \forall k \in \kSet,
           (i = \projl{\cww[k]} \land\\  j = \projr{\cww[k]}) \implies  
           \trg{ (\asc{\ev[\kk] \tv[i]}{\gtysp[j]}) } \nreds{}{1} \rctx[\cdot] \V[\kk] 
        \end{block}\\
           \rctx[\pphi[3]] \d{\errort{\gtysp[j]}^{\p[j]}}[j \in \jSet]
           & \text{otherwise}
      \end{cases} 
      } 
    }
\end{mathpar}
\caption{Distribution semantics of \tlang (part 2).}
  \label{fig:target-reduction-dis-2}
\end{figure}  
The distribution semantics is presented in Figures~\ref{fig:target-reduction-dis} and~\ref{fig:target-reduction-dis-2}. Reduction judgment $\rrctx[\phi[]] \tm \nreds{}{\j} \rctx[\phi] \V$ denotes that term $\tm$ reduces to \emph{distribution configuration} $\rctx[\phi] \V$ within $k$ steps, where $\phi$ closes the symbolic probabilities occurring in $\V$. The number of steps of reduction judgments are only required   
to establish the metatheory, and can be removed in a real implementation. Several rules are defined similarly to \slang, but  accounting for the fact that values are always ascribed. 
Rule $\trg{(\mathit{Dapp})}$ first ascribes argument $\tv$ to $\gtysp$, appealing to 
transitivity with the domain of $\ev$ as evidence. After its reduction, the obtained value $\tw$ is 
substituted for $x$ in the body of the function using the auxiliary function $\mathit{sub}$. 
If during $\tv$ reduction transitivity does not hold (and $\tw$ is thus $\errort{\gtysp}$), $\mathit{sub}$ yields term $\errort{\gtya}$, $\gtya$ being the expected distribution type of the application.
\rev[mt1]{Rule $\trg{(\mathit{Dlet})}$ first reduces $\tm$ to a distribution value $\V$ of type $\gtyap$. Since the body of the $\<let>$ is a \emph{family} $\overline{\tn[j]}$ of terms---one $\tn[j]$ per entry $\gtysp[j]$ of the (static) distribution type $\gtya$ of $\tm$---the rule must route each runtime outcome $\tv[i]$ of $\V$ to the appropriate body. To do so it follows the same strategy as rule $\trg{(\mathit{D}\mathord{::}\gtya)}$: it reorders $\gtyap$ against $\gtya$, and the resulting evidence $\initReorder{\gtyap}{\gtya} = \phty{\pphi[2]}{\d{\ev[k][\cww[k]]}[k \in \kSet]}$ connects, with weight $\cww[k]$, the outcome $\tv[\projl{\cww[k]}]$ to the type $\gtysp[\projr{\cww[k]}]$ of body $\tn[\projr{\cww[k]}]$. For each such connection, $\tv[i]$ is ascribed to $\gtysp[j]$ (yielding $\tw[k]$), the body $\tn[j]$ is then evaluated with $\tx$ substituted by $\tw[k]$ (using function $\mathit{sub}$ to handle the case where $\tw[k]$ is an error), and the so obtained distribution configurations are combined---values via their weighted sum (using the weights $\cww[k]$) and formulas via their conjunction---to form the final outcome of the $\<let>\!$--expression.}
\rev[fo1]{The rule reduces every connection $k$, whatever probability a solution of
$\pphi[2]$ gives to it, and the resulting formula includes the formula
$\phi[][k]$ of a connection only under the guard $\cww[k] > 0$, so a
connection of probability zero imposes no constraint.}
Rule $\trg{(\mathit{D}\oplus)}$ reduces the pair of branches and combines their results like the $\trg{(\mathit{Dlet})}$ rule, the major difference being that \rev[fo2]{the lifting $\phi[][\oplus]$ of the probability $\pty$} is also included in the resulting distribution configuration.
\rev[fo1]{The variables $\cww[1]$ and $\cww[2]$ are fresh at each application
of the rule, so two reductions of the same choice, for instance when a
function is applied twice, constrain independent variables.}
Rule $\trg{(\mathit{Dv})}$ lifts values to (Dirac) distribution values.
\extends{Rule $\trg{(\mathit{D+})}$ combines the evidences of the two numbers and adds them to produce the final result with probability one.
Rule $\trg{(\mathit{Dit})}$ (resp.~$\trg{(\mathit{Dif})}$) applies when the guard of a conditional is true (resp.~false), reducing $\tm$ (resp.~$\tn$) as the final result.}
Rule $\trg{(\mathit{Derr})}$ reduces an error over distribution type $\gtya$ to a distribution of errors over simple types $\gtys[i]$. 
Rule $\trg{(\mathit{Dmon})}$ establishes the monotonicity of the reduction relation with respect to the step index. 
Rule $\trg{(\mathit{D}\mathord{::}\gtys)}$ analyzes whether type of $\tu$ is consistent with $\gtysp$, combining the respective evidences through the consistent transitivity operator. The rule yields either a Dirac distribution of a newly ascribed value (if consistent transitivity succeeds), or (else) an error \rev[mt1]{tagged with the type $\gtysp$ of the failed ascription}.

\rev[fo1]{The remaining rules propagate errors. Rule
$\trg{(\mathit{Derr}\mathord{::}\gtys)}$ reduces the ascription of an error
to the error at the ascribed type; rules $\trg{(\mathit{Dlet})}$ and
$\trg{(\mathit{D}\mathord{::}\gtya)}$ rely on it, since an outcome that they
ascribe may itself be an error. Rules $\trg{(\mathit{Derr}{+})}$,
$\trg{({+}\mathit{Derr})}$, $\trg{(\mathit{Derr\,app})}$ and
$\trg{(\mathit{Derr\,if})}$ reduce an addition, an application and a
conditional with an error as operand, function or guard to the error at the
type of the term, distributed over its entries as in
$\trg{(\mathit{Derr})}$.}

Rule $\trg{(\mathit{D}\mathord{::}\gtya)}$ is the most challenging. Intuitively, it ``pushes'' simple evidences within $\gtyd$ into the outcomes of $\tm$. However, note that pushing every simple evidence within $\gtyd$ into every possible outcome of $\tm$ is not what we want. For instance,  
  given the (informal) program $\ds{\rtype^{\frac{1}{2}}, \btype^{\frac{1}{2}}} (1 \rev[fo2]{\tpsum[\frac{1}{2}]} \ttt) \trg{::} \allowbreak \ds{\rtype^{\frac{1}{2}}, \btype^{\frac{1}{2}}}$, it would be futile to push evidence $\rtype$ into $\ttt$, or $\btype$ into $1$.
Here, we have two problems to address. First, to determine what evidences must be pushed into what values. Second, to determine the probability of each such combination. We address both problems simultaneously, taking advantage of the consistent transitivity operator, and a subsidiary relation over formula types we introduce next.

Observe that rule $\trg{(\mathit{D}\mathord{::}\gtya)}$ proceeds by first reducing $\tm$  to a value distribution of type $\gtyap$. However, this $\gtyap$ can be (syntactically) different from $\gtya$, the actual type of $\tm$. What we require here is that $\gtyap$ be a \emph{reordering} of $\gtya$. We thus introduce the reordering relation $\reoder{}{}$ over formula types, which (for distribution types) is nothing more than the coupling lifting of the syntactic equality.

\begin{definition}[Reordering]\label{def:reorder}\mechnote{\leanname{TPLC.EReordTy} and \leanname{TPLC.EReordD}.}
    The reordering relation $\reoder{}{}$ over formula simple and distribution types is defined by the following \rev[fo1]{rules}:
\begin{mathpar}
\mprset{andskip=1em plus 0.5fil minus 0.5em}%
 \inference{}{\reoder{\rtype}{\rtype}} \!\!\and\!\!
  \inference{}{\reoder{\btype}{\btype}} \!\!\and\!\!
  \inference{}{\reoder{\?}{\?}} \!\!\and\!\!
  \inference{\reoder{\gtys[2]}{\gtys[1]} & \reoder{\gtya[1]}{\gtya[2]} }
  {\reoder{\gtys[1] -> \gtya[1]}{\gtys[2] -> \gtya[2]}} \!\!\and\!\!
  \inference{
      \rev[mt1]{\couplingLift{\gtya}{\gtyb}{=}}
    }
    {\reoder{\gtya}{\gtyb}}%
\end{mathpar}

\end{definition}
\rev[fo1]{In the distribution clause, the relation carried by the coupling is
the syntactic equality of the entries, not $\reoderSym$ itself, so a positive
weight $\cww[ij]$ forces $\gtys[i]$ and $\gtysp[j]$ to be the same simple
type. Hence the routing of a $\<let>$ sends each runtime outcome to an entry
that has exactly its type, and such an entry exists because every entry of a
runtime distribution type is an entry of the corresponding static type
(Theorem~\ref{lm:typesafety}). Like the runtime relations, the reordering
relation carries no coverage clause.}

We can construct an initial evidence for reordering judgments similarly to the meet operator:

\begin{definition}[Reordering initial evidence]\mechnote{\leanname{TPLC.reorderTy} and \leanname{TPLC.reorderD}.}
The partial operator $\initReorder{}{}$ over formula simple and distribution types is defined as follows\rev[fo2]{, where $\mathit{id}_{=}$ is the partial identity on simple types, defined only on pairs of equal types, $\mathit{id}_{=}(\gtys, \gtys) = \gtys$}:
  \begin{mathpar}
  \initReorder{\rtype}{\rtype} ~=~ \rtype \and
  \initReorder{\btype}{\btype} ~=~ \btype \and
  \initReorder{\?}{\?} ~=~ \? \and
  (\initReorder{\gtys[1] -> \gtya[1]}{\gtys[2] -> \gtya[2]}) ~=~
  (\initReorder{\gtys[1]}{\gtys[2]}) -> (\initReorder{\gtya[1]}{\gtya[2]}) \and
  \initReorder{\gtya[1]}{\gtya[2]} ~=~ \witness{\gtya[1]}{\gtya[2]}{\rev[fo2]{\mathit{id}_{=}}}
  \end{mathpar}
  \end{definition}
  \rev[fo1]{As in Definition~\ref{def:reorder}, the distribution clause has
  one entry for each pair of equal entries of the operands. As with the meet,
  $\initReorder{\gtya[1]}{\gtya[2]}$ is defined exactly when its formula is
  satisfiable, which for well-formed operands coincides with
  $\reoder{\gtya[1]}{\gtya[2]}$:}

    \begin{lemma}[{\rev[fo1]{Initial evidence for reordering}}]\label{reorderevidence2}\mechnote{item~1, \leanname{TPLC.reord\_reorderTy\_isSome} and \leanname{TPLC.eprec\_reorderTy}; item~2, \leanname{TPLC.reorderD\_sat\_iff\_ereordD}, \leanname{TPLC.eprec\_reorderD} and \leanname{TPLC.goodD\_reorderD}.}
\rev[fo1]{For any well-formed formula simple types $\gtys[1], \gtys[2] \in \GFType$
and any well-formed formula distribution types
$\gtya[1], \gtya[2] \in \GFDType$,}
      \begin{enumerate}
        \item \rev[fo1]{If $\reoder{\gtys[1]}{\gtys[2]}$, then
        $\initReorder{\gtys[1]}{\gtys[2]}$ is defined and
        $\initReorder{\gtys[1]}{\gtys[2]} |- \reoder{\gtys[1]}{\gtys[2]}$.}
        \item \rev[fo1]{If $\reoder{\gtya[1]}{\gtya[2]}$, then
        $\initReorder{\gtya[1]}{\gtya[2]}$ is defined and
        $\initReorder{\gtya[1]}{\gtya[2]} |- \reoder{\gtya[1]}{\gtya[2]}$.}
      \end{enumerate}
    \end{lemma}
\noindent Here, $\egtys |- \reoder{\gtys[1]}{\gtys[2]}$ means that evidence $\egtys$ justifies reordering $\reoder{\gtys[1]}{\gtys[2]}$ and holds if $\egtys  \eprec  \initReorder{ \gtys[1]}{ \gtys[2]}$. The notion of evidence for the reordering between distribution types is defined analogously.

\rev[fo1]{An evidence for a reordering composes with an evidence for a
consistency through $\trans{}$, and the result justifies the consistency
between the outer two types.}

Returning to rule $\trg{(\mathit{D}\mathord{::}\gtya)}$, observe that evidence  $(\initReorder{\gtyap}{\gtya}) \trans{} \evd$  addresses both of the mentioned problems: every simple evidence $\ev[k][]$ in $(\initReorder{\gtyap}{\gtya}) \trans{} \evd$ connects a simple type from $\gtyap$ with a simple type from $\gtyb$, via its corresponding weight $\cww[k]$. %
To form the final distribution value, evidence $\ev[k]$ is pushed to the ascription of value $\tv[\projl{\cww[k]}]$ with type $\gtysp[\projr{\cww[k]}]$. After reducing these terms, the obtained distribution values are combined (through a weighted sum) to yield the final distribution value. %

\begin{example}
  To illustrate this process, consider an expression 
  $\evd \tm \trg{::} \, \dt{?^{\frac{2}{3}}, \rtype^{\frac{1}{3}} }$ such that $|-d  \tm : \gtya$, where
  
  \[
  \begin{array}{c}
  \evd  \:=\: \phty{\bigl(\cw[1]{1}{1}=\tfrac{1}{6} \land \cw[2]{1}{2}=\tfrac{1}{3} \land \cw[3]{2}{1} = \tfrac{1}{2}\bigr)}{\dt{\rtype^{\cww[1]} , \rtype^{\cww[2]}, \btype^{\cww[3]}}}\\
  \rev[fo1]{\vjtf{\evd}{\gtya \erel \gtyb}} %
  \qquad \qquad
  \gtya \:=\: \dt{\rtype^{\frac{1}{2}}, \btype^{\frac{1}{2}}} 
  \qquad \qquad
  \gtyb \:=\: \ds{?^{\frac{2}{3}}, \rtype^{\frac{1}{3}} }     
  \end{array}
  \]
  For simplicity, we consider only concrete probabilities, thus omitting formulas, and also omitting some trivial ascriptions. 

  Now, say $\tm \nreds{} {k} \rctx[\cdot] \V$ and $|-d \V : \gtyap$, where
  \[
  \begin{array}{c}
  \V \:=\: \dt{((\btype) \trg{\ttt} :: \btype)^{\frac{1}{2}}, ((\rtype) \trg{1} :: \rtype)^{\frac{1}{2}}}
  \qquad \qquad 
  \gtyap \:=\: \dt{\btype^{\frac{1}{2}}, \rtype^{\frac{1}{2}}}  
  \end{array}
  \]
  Letting  $\evdp = (\initReorder{\gtyap}{\gtya})$, we have
  $$\evdp = \phty{ \bigl(\cwp[1]{1}{2} = \tfrac{1}{2} \land \cwp[2]{2}{1} = \tfrac{1}{2}\bigr) }{ \rev[fo1]{\dt{\btype^{\cwwp[1]}, \rtype^{\cwwp[2]} }} }$$
  Types and evidence are illustrated in Figure~\ref{fig:example-red}. Notice that $\jtf{\evdp \trans{} \evd}{\gtyap \erel \gtyb}$, where
  $$\evdp \trans{} \evd = \phty{\bigl(\cwpp[1]{1}{1}=\tfrac{1}{2} \land \cwpp[2]{2}{1}=\tfrac{1}{6} \land \cwpp[3]{2}{2} = \tfrac{1}{3}\bigr)}{\ds{\btype^{\cwwpp[1]} , \rtype^{\cwwpp[2]}, \rtype^{\cwwpp[3]}}}$$
  Finally, the whole expression $\evd \tm \trg{::} \, \dt{?^{\frac{2}{3}}, \rtype^{\frac{1}{3}} }$ reduces to
  $$\dt{((\btype) \trg{\ttt} :: \?)^{\frac{1}{2}}, 
  \rev[fo1]{((\rtype) \trg{1} :: \?)^{\frac{1}{6}}}, 
  ((\rtype) \trg{1} :: \rtype)^{\frac{1}{3}} }$$
in $k+1$ steps.
\end{example}

\begin{figure}[t]
\begin{small}
\begin{center}
\begin{tikzpicture}[colgroup/.style={draw=gray, rounded corners, thick, inner sep=5pt}]
  \node[draw, circle,  fill= white!20, inner sep=1pt] (p1) {1};
  \node[sloped, above = 2pt of p1] {$\btype^{\frac{1}{2}}$};

  \node[draw, circle,  fill= white!20, inner sep=1pt, right=100pt of p1] (p21) {1};
  \node[sloped, above = 2pt of p21] {$\rtype^{\frac{1}{2}}$};

  \node[draw, circle,  fill= white!20, inner sep=1pt, right=200pt of p1] (p31) {1};
  \node[sloped, above = 2pt of p31] {$?^{\frac{2}{3}}$};

  \node[draw, circle,  fill= white!20, inner sep=1pt, below=46pt of p1] (p2) {2};
  \node[sloped, above = 2pt of p2] {$\rtype^{\frac{1}{2}}$};

  \node[draw, circle,  fill= white!20, inner sep=1pt, right=100pt of p2] (p22) {2};
  \node[sloped, above = 2pt of p22] {$\btype^{\frac{1}{2}}$};

  \node[draw, circle,  fill= white!20, inner sep=1pt, right=200pt of p2] (p32) {2};
  \node[sloped, above = 2pt of p32] {$\rtype^{\frac{1}{3}}$};

  \draw[] (p1) edge  node[sloped, anchor=center, above, xshift=-25, font=\footnotesize] {$\rev[fo1]{\cwp[1]{1}{2}=\frac{1}{2}}$} (p22);
  \draw[] (p2) edge  node[sloped, anchor=center, below, xshift=-25, font=\footnotesize] {$\rev[fo1]{\cwp[2]{2}{1}=\frac{1}{2}}$} (p21);

  \draw[] (p21) edge  node[sloped, anchor=center, above, font=\footnotesize] {$\rev[fo1]{\cw[1]{1}{1}=\frac{1}{6}}$} (p31);
  \draw[] (p21) edge  node[sloped, anchor=center, above, xshift=-20, font=\footnotesize] {$\rev[fo1]{\cw[2]{1}{2}=\frac{1}{3}}$} (p32);
  \draw[] (p22) edge  node[sloped, anchor=center, below, xshift=-25, font=\footnotesize] {$\rev[fo1]{\cw[3]{2}{1}=\frac{1}{2}}$} (p31);

  \node[inner sep=0pt, minimum height=0pt, above = 15pt of p1] (top1) {};
  \node[colgroup, fit=(top1)(p2), inner xsep=8pt] (col1) {};
  \node[inner sep=0pt, minimum height=0pt, above = 15pt of p21] (top21) {};
  \node[colgroup, fit=(top21)(p22), inner xsep=8pt] (col2) {};
  \node[inner sep=0pt, minimum height=0pt, above = 15pt of p31] (top31) {};
  \node[colgroup, fit=(top31)(p32), inner xsep=8pt] (col3) {};

  \node[below=2pt of col1] (mu) {$\gtyap$};
  \node[below=2pt of col2] (mup) {$\gtya$};
  \node[below=2pt of col3] (mux) {$\gtyb$};
  \node at ($(mu)!0.5!(mup)$) {$\evdp =(\initReorder{\gtyap}{\gtya})$};
  \node at ($(mup)!0.5!(mux)$) {$\evd$};

\end{tikzpicture}  
\end{center}
\end{small}\vspace{-1ex}
\caption{Illustration of types and evidence for reduction example.}
\label{fig:example-red}
\end{figure}

\subsection{Elaboration}
\label{sec:elaboration}

\begin{figure}[t]

\begin{flushleft}
  \minibox[frame]{$\elaborate{\Gamma |-t \sv}{\tv}{\sgtys}$\\ $\elaborate{\Gamma |-d \sm}{\tm}{\sgtya}$}
\end{flushleft}
  \def \MathparLineskip {\lineskip=1.8ex}
  \begin{mathpar}
    \inference[\src{(E\ssr)}]{ \ev = \rtype \meet \rtype }
    {\ela{\Gamma |-t \ssr }{\asc{\ev \ttr }{\rtype}}{\rtype}} \and
    \inference[\src{(E\ssb)}]{\ev = \btype \meet \btype}
    {\ela{\Gamma |-t \ssb }{\asc{\ev \ttb }{\btype}}{\btype}} \and
    \inference[\src{(E\sx)}]{\Gamma(\sx) = \rev[mt1]{\gtys}}
    {\ela{\Gamma |-t \sx }{ \tx }{\rev[mt1]{\gtys}}} \and
    \inference[\src{(Eapp)}]{
      \ela{\Gamma |-t \sv}{\tv}{\rev[mt1]{\gtys}}  &
      \ela{\Gamma |-t \sw}{\tw}{\rev[mt1]{\gtysp}} &  \rev[mt1]{\gtysp  \rel \cdom(\gtys)} \\
      \rev[mt1]{\ev[1] = \gtysp \meet \cdom(\gtys)}  & \rev[mt1]{\ev[2] = \gtys \meet (\cdom(\gtys)->\ccod(\gtys))}
    }
    {\ela{\Gamma |-d \sv \; \sw}{
      \trg{\letnl{ \tx }{\asc{\ev[1] \tw }{ \rev[mt1]{\cdom(\gtys)} } }{ \lett{\ty}{\asc{\ev[2] \tv}{ \rev[mt1]{\cdom(\gtys) -> \ccod(\gtys)} }}{\ty \;\tx} } } }{\rev[mt1]{\ccod(\gtys)}}
    } \and
    {\setpremisesspace{1em}\inference[\src{(E\oplus)}]{
      \ela{\Gamma |-d \sm}{\tm}{\rev[mt1]{\gtya}} &
      \ela{\Gamma |-d \srcn}{\tn}{\rev[mt1]{\gtyb}} \\
      \cww[1], \cww[2] \; \text{fresh} &
      \liftP{\pty}[\cww[1]] = \phi[][1] & \liftP{(1-\pty)}[\cww[2]] = \phi[][2]
      & \phi = \phi[][1] \land \phi[][2] \land (\cww[1] + \cww[2] = 1)
    }
    {\ela{\Gamma |-d \src{ { \sm} \pssum { \srcn} } }{ \rev[fo2]{\trg{ {\tm} \tpsum {\tn} }} }{
      \rev[mt1]{\sugarconj{\phi}{\cww[1] \cdot \gtya + \cww[2] \cdot \gtyb}}}}} \and
    \inference[\src{(Elet)}]{
      \ela{\Gamma |-d \sm}{\tm}{\rev[mt1]{\db{\phi}{\gtys[i][\p[i]]}[i \in \iSet]}}  &
      \rev[mt1]{\forall i\in \iSet.~\ela{\Gamma, \sx : \gtys[i] |-d \srcn}{\tn[i]}{\gtya[i]}}
    }
    {\elaborate{\Gamma |-d \src{\lett{\sx}{\sm}{\srcn}} }{
      \trg{ \lett{\tx}{\tm}{\rev[mt1]{\overline{\tn[i]}^{\,i \in \iSet}}} } }{
        \rev[mt1]{\sugarconj{\phi}{\sum_{i \in \iSet} \p[i] \cdot \gtya[i]}}}} \and
    \inference[\src{(E\lambda)}]{\ela{\Gamma, \sx :\rev[mt1]{\liftT{\sgtys}} |-t \sm}{\tm}{\rev[mt1]{\gtya}}
    & \rev[mt1]{\ev = \liftT{\sgtys} -> \gtya}
    & \jtf{}{\sgtys} }
    {\ela{\Gamma |-t \src{\lambda \sx : \sgtys. \sm} }{  \trg{ \asc{\ev \lambda \tx: \liftT{\sgtys}. \tm}{ \rev[mt1]{\liftT{\sgtys} -> \gtya}} } }{\rev[mt1]{\liftT{\sgtys} -> \gtya}}} \and
    \inference[\src{(E\mathord{::}\sgtya)}]{\elaborate{\Gamma |-d \sm }{\rev[mt1]{\tm} }{\rev[mt1]{\gtya}} &  \rev[mt1]{\gtya \rel  \liftD{\sgtyb}} & \rev[mt1]{\evd = \gtya \meet \liftD{\sgtyb}}  & \jtf{}{\sgtyb}
    }
    {\elaborate{\Gamma |-d \src{\asc{ \sm}{\sgtyb} } }{ \trg{ \asc{\evd \tm}{ \liftD{\sgtyb} } } }{\rev[mt1]{\liftD{\sgtyb}}} } \and
    \inference[\src{(E\mathord{::}\sgtys)}]{\ela{\Gamma |-t \sv}{\tv}{\rev[mt1]{\gtys}} &  \rev[mt1]{\gtys \rel \liftT{\sgtysp}} & \rev[mt1]{\ev = \gtys \meet \liftT{\sgtysp}}  & \jtf{}{\sgtysp} }
    {\ela{\Gamma |-t \src{\asc{\sv}{\sgtysp}}}{\asc{\ev \tv}{ \liftT{\sgtysp} }}{\rev[mt1]{\dctx[(\cww=1)][\ds{{\liftT{\sgtysp}}^{\cww}}]}}}  \and
    \inference[\src{(E\sv)}]{\ela{\Gamma |-t \sv}{\tv}{\rev[mt1]{\gtys}} }
    {\ela{\Gamma |-t \sv}{\tv}{\rev[mt1]{\dctx[(\cww = 1)][\ds{\gtys[][\cww]}]}} } \and
    {\setpremisesend{0em}\inference[\src{(E+)}]{
      \ela{\Gamma |-t \sv}{\tv}{\rev[mt1]{\gtys}}   & \rev[mt1]{\gtys \rel \rtype}   &
      \ela{\Gamma |-t \sw}{\tw}{\rev[mt1]{\gtysp}} & \rev[mt1]{\gtysp \rel \rtype} \\
      \rev[mt1]{\ev[1] = \gtys \meet \rtype} & \rev[mt1]{\ev[2] = \gtysp \meet \rtype}
    }
    {\ela{\Gamma |-d \src{\add{\sv}{\sw}}}{\trg{  \lett{\tx}{ \asc{\ev[1] \tv }{\rtype} }{ \lett{\ty}{ \asc{\ev[2] \tw }{\rtype} }{\add{\tx}{\ty}}  } }}{\rev[mt1]{\dctx[(\cww=1)][\ds{\rtype^{\cww}}]}}  }} \and
    \inference[\src{(Eif)}]{
      \ela{\Gamma |-t \sv}{\tv}{\rev[mt1]{\gtys}} & \rev[mt1]{\gtys \rel \btype} &
      \rev[mt1]{\gtya \rel \gtyb} \\
      \ela{\Gamma |-d \sm }{\tm}{\rev[mt1]{\gtya}} &
      \ela{\Gamma |-d \srcn }{\tn}{\rev[mt1]{\gtyb}} & \rev[mt1]{\ev= \gtys \meet \btype}
    }
    { \ela{\Gamma |-d \src{\ite{\sv}{\sm}{\srcn}}}{
      \trg{\lett{\tx}{ \asc{\ev \tv }{\btype}}{ \ite{\tx}{\tm}{\tn}} }
    }{ \rev[mt1]{\gtya \mathbin{\oplus_{\?}} \gtyb}}  } \and
  
  \end{mathpar}

\caption{Elaboration from \glang to \tlang.}
\label{fig:elaboration}
\end{figure}

As previously mentioned, the runtime semantics of \glang\  is given via translation to the \tlang target language.
Figure~\ref{fig:elaboration} presents the type-driven elaboration rules from \glang to \tlang.
Judgment $\elaborate{\Gamma |-t \sv}{\tv}{\sgtys}$ (resp.\ $\elaborate{\Gamma |-d \sm}{\tm}{\sgtya}$), denotes the elaboration of value $\tv$ (resp.\ term $\tm$) from value $\sv$ (resp.\ term $\sm$), where $\sv$ (resp.  $\sm$) is typed $\sgtys$ (resp. $\sgtya$) under environment $\Gamma$. For simplicity, we write $\elaborate{\sv}{\tv}{\sgtys}$ (resp.\ $\elaborate{\sm}{\tm}{\sgtya}$) as a shorthand for $\elaborate{\cdot |-t \sv}{\tv}{\sgtys}$ (resp.\ $\elaborate{\cdot |-t \sm}{\tm}{\sgtya}$).
\extends{Rules $\src{(Ex)}$ and $\src{(Ev)}$ elaborate variables and values to themselves, without any modification.}
Rule $\src{(E\lambda)}$ and the ones for elaborating other values, elaborate by inserting ascriptions to their types.
The initial evidence between two gradual types is computed using the meet of the lifted types.
Rule $\src{(Eapp)}$ inserts ascriptions in both the function and argument to make top-level constructor match using \rev[fo1]{monadic normal form}.
\extends{Rules $\src{(E\mathord{::}\sgtya)}$ and $\src{(E\mathord{::}\sgtys)}$ generate initial evidence to justify the consistency judgment between the annotated types and the inner types.}
\rev[fo1]{Rule $\src{(Elet)}$ elaborates the body $\srcn$ once per entry $i \in \iSet$ of the distribution type of $\sm$, under the binding $\sx : \gtys[i]$, and collects the resulting bodies $\tn[i]$ in the family of the target $\<let>\!$.}
Rule $\src{(E\oplus)}$ is designed to carefully deal with the probability annotations. First, it generates two fresh variables: $\cww[1]$ for annotation $\pty$, and $\cww[2]$ for the complement $1 - \pty$.\footnote{For $\cww[1]$ and $\cww[2]$ we do not care about the indexes because they do not flow into evidences. For this reason we can arbitrarily choose $0$ and $0$ as default values.} Second, these two fresh variables are used to generate two formulas $\phi[][1]$ and $\phi[][2]$ by lifting $\pty$ and $1 - \pty$. Third, the annotation formula $\phi$ is computed by combining $\phi[][1]$ and $\phi[][2]$, with the extra requirement that the sum of the probability variables must be one (in case $\pty = \?$).
\rev[fo1]{Neither $\src{(Elet)}$ nor $\src{(E\oplus)}$ inserts an ascription, since the target term has exactly the type of the source term.}
\extends{Rule $\src{(E+)}$ inserts ascriptions of type $\rtype$ along with corresponding evidences that justify the relationship between the value types and $\rtype$ for both operands. Similarly, rule $\src{(Eif)}$ inserts \rev[mt1]{an ascription and evidence only} for the condition value\rev[fo1]{; the branches need no ascription, since rule (Gif) types the conditional at the convex hull of their types}.}

\begin{example}
To illustrate the elaboration procedure, consider the program 
$$\src{\ttt}\pssum[\frac{1}{2}]\src{1}$$ 
The source values $\src{\ttt}$ and $\src{1}$ are elaborated as
\begin{align*}
  \elalign{\cdot |-d \src{\ttt}}{\trg{\ev[1]\asc{\ttt}{\btype}}}{\ds{\btype^{1}}}\\
  \elalign{\cdot |-d \src{1}}{\trg{\ev[2]\asc{1}{\rtype}}}{\ds{\rtype^{1}}}
\end{align*}
where  $\ev[1] = \btype$ and $\ev[2] = \rtype$, respectively.
\rev[mt1]{We know that
\begin{align*}
  \phi[][1] &= \liftP{\frac{1}{2}}[\cww[1]] = (\cww[1] = \frac{1}{2})\\
  \phi[][2] &= \liftP{1-\frac{1}{2}}[\cww[1]] = (\cww[2] = \frac{1}{2})
\end{align*}
for $\cww[1]$ and $\cww[2]$ fresh, and then
$\phi = (\cww[1] = \frac{1}{2} \land \cww[2] = \frac{1}{2} \land
\cww[1] + \cww[2] = 1)$. The whole program thus elaborates to
$\trg{(\ev[1]\asc{\ttt}{\btype})} \rev[fo2]{\tpsum[\frac{1}{2}]} \trg{(\ev[2]\asc{1}{\rtype})}$,
of type $\sugarconj{\phi}{\cww[1] \cdot \ds{\btype^{1}} + \cww[2] \cdot
\ds{\rtype^{1}}}$---the exact type the source term receives, so no outer
ascription (and no distribution evidence) is inserted.}
\end{example}

\rev[fo1]{Each elaboration rule has the premises of the corresponding typing
rule of \glang{}, including its consistency premises, and every evidence is
the meet of the two types of the judgment that it justifies. Those
consistency premises are exactly the conditions under which the meets are
defined. Hence every well-typed term of \glang{} has an elaboration, and the
elaboration has the same formula type:}

\begin{restatable}[Elaboration preserves types]{lemma}{elaprvty}\label{theorem:ep}\mechnote{\leanname{TPLC.elaboration\_preserves\_types}.}
\rev[fo1]{For every well-formed environment $\Gamma$, every value $\sv$ and
every term $\sm$ of \glang{}, every formula simple type $\gtys \in \GFType$ and
every formula distribution type $\gtya \in \GFDType$,}
 \begin{enumerate}
 \item \rev[fo1]{If $\Gamma |-t \sv : \gtys$, then
   $\ela{\Gamma |-t \sv}{\tv}{\gtys}$ and $\Gamma |-t \tv : \gtys$ for some
   value $\tv$ of \tlang{}.}
 \item \rev[fo1]{If $\Gamma |-d \sm : \gtya$, then
   $\ela{\Gamma |-d \sm}{\tm}{\gtya}$ and $\Gamma |-d \tm : \gtya$ for some
   term $\tm$ of \tlang{}.}
 \end{enumerate}
\end{restatable}

\subsection{Type Safety and Gradual Guarantee}
\label{sec:properties}
We can now establish several properties about \glang, based on the elaboration to \tlang. To this end, we start highlighting that even though our static language (\slang) is terminating, when introducing unknown types, one can encode statically well-typed programs that diverge, rendering our gradual languages (\glang and \tlang) non-terminating. The emblematic program illustrating this phenomenon is the omega term $\Omega = (\lambda x:?. x\;x ) (\lambda x: ? . x\;x)$.

In the remainder of this section, for simplicity, given $\sm$ we write $\converge[\sm][\phty{\pphi}{\dset}]$ if 
$\ela{ |-d \sm}{\tm}{\gtya}$ (for some $\gtya$ and $\tm$) and there exists $k$ such that $ \rrctx[\cdot] \tm \nreds{}{k} \rctx[\phi] \V$ (for some $\phi$ and $\V\,$), and $\diverge[\sm]$ if 
$\ela{ |-d \sm}{\tm}{\gtya}$ (for some $\gtya$ and $\tm$) and there exists no such $k$. In the former case we say that $\sm$ (similarly, $\tm$) \emph{terminates}, reducing to a distribution value, while in the latter case we say that  $\sm$ (similarly, $\tm$) \emph{diverges} (also denoted by $\rrctx[\cdot] \tm \Uparrow $). \rev[fo1]{This
terminology corresponds to certain termination
(Section~\ref{sec:overview-design}): a program that terminates only almost
surely has no reduction and counts as divergent, so the results of this
section claim nothing about it.}

\subsubsection*{Type Safety}
Following~\citep{lennonAl:toplas2022}, we model errors ($\errort{\gtys}$ and $\errort{\gtya}$) as expressions, thereby simplifying the statement of type safety, as we do not need to reason about error separately. Type safety for \glang then states that if a term $\sm$ is well-typed, then it either reduces to a distribution value of an equivalent type, or diverges:

\begin{restatable}[Type safety for \glang]{theorem}{typesafetyx}\label{lm:typesafety}\mechnote{\leanname{GPLC.type\_safety}.}
\rev[fo1]{For every term $\sm$ of \glang{} and every formula distribution type
$\gtyb \in \GFDType$, if $\empty |-d \sm : \gtyb$, then either}
  \begin{enumerate}
      \item \rev[fo1]{$\converge[\sm][\rctx[\phi] \V ]$, $|- \rctx[\phi] \V : \gtya$,
      $\reoder{\gtya}{\gtyb}$ and every entry of $\gtya$ is an entry of
      $\gtyb$, for some formula $\phi$, some distribution value $\V$ and
      some formula distribution type $\gtya \in \GFDType$}, or
      \qquad \qquad \item $\diverge[\sm]$.
    \end{enumerate}
\end{restatable}

\begin{figure}[t]
    \def \MathparLineskip {\lineskip=1.4ex}  
  \begin{mathpar}
    { 
    \inference[\trg{(\mathord{\gprec}\V)}]{
  \forall\: \FV(\phi[][1]).~  \phi[][1]  \implies \exists~\FV(\phi[][2]) \cup
  \{\cww[ij] ~|~ i \in \iSet \land j \in \jSet\}. %
  \\
  \cjudgext{\d{\cww[ij]}[i \in \iSet \land j \in \jSet]}{\d{\tv[i][\p[i]]}[i \in \iSet]}{\, \gprec\,}{\d{\tv[j][\p[j]]}[j \in \jSet]}{\phi[][1]}{\phi[][2]}{}%
}
      {
       \phty{\phi[][1]}{\dt{\tv[i][\p[i]] ~|~ i \in \iSet}}   \gprec 
       \phty{\phi[][2]}{\dt{\tv[j][\p[j]] ~|~ j \in \jSet}}
      } 
    }
     \\ \and 
    \inference{}
    {
      \tx \gprec \tx
    } \and
    \inference{
    }
    {
      \ttr \gprec \ttr
    } \and
    \inference{
    }
    {
      \ttb \gprec \ttb
    } \and
    \inference{
      \rev[mt1]{\Gamma' |-t \tv : \gtysp} &  \gtys \gprec \gtysp
    }
    {
      \rev[mt1]{\Gamma \gprec \Gamma' |-} \errort{\gtys} \gprec \rev[mt1]{\tv}
    } \and
    \rev[mt1]{
    \inference{
      \gtya \gprec \gtyb
    }
    {
      \errort{\gtya} \gprec \errort{\gtyb}
    }} \and
    \inference{
      \gtys \gprec \gtysp &
      \rev[mt1]{\Gamma, \tx {:} \gtys \gprec \Gamma', \tx {:} \gtysp |- \tm \gprec \tmp} &
    }
    {
      \rev[mt1]{\Gamma \gprec \Gamma' |-}
      \trg{(\lambda \tx:\gtys. \tm)} \gprec \trg{(\lambda \tx:\gtysp. \tmp)}
    } \and
    \rev[fo1]{
    \inference{
    \tagprecr{\gtys}{\gtysp}{\ev}{\evp} \\
    \Gamma |-t \tu : \gtys[1] & \Gamma' |-t \tup : \gtysp[1] &
    \tagprecl{\gtys[1]}{\gtysp[1]}{\ev}{\evp} \\
    \Gamma \gprec \Gamma' |- \tu \gprec \tup &
    \gtys \gprec \gtysp
    }
    {
     \Gamma \gprec \Gamma' |-
     \trg{ \asc{\ev \tu}{\gtys} } \gprec \trg{\asc{\evp \tup}{\gtysp}}
    }} \and
    \inference{
    \rev[mt1]{\tagprecr{\gtys}{\gtysp}{\ev}{\evp}} \\
    \rev[mt1]{\Gamma |-t \tv : \gtys[1]} & \rev[mt1]{\Gamma' |-t \tvp : \gtysp[1]} &
    \rev[mt1]{\tagprecl{\gtys[1]}{\gtysp[1]}{\ev}{\evp}} \\
    \rev[mt1]{\Gamma \gprec \Gamma' |-} \tv \gprec \tvp &
    \gtys \gprec \gtysp
    }
    {
     \rev[mt1]{\Gamma \gprec \Gamma' |-}
     \trg{ \asc{\ev \tv}{\gtys} } \gprec \trg{\asc{\evp \tvp}{\gtysp}
    }} \and
    \inference{
    \rev[mt1]{\tagprecr{\gtya}{\gtyb}{\evd}{\evdp}} \\
    \rev[mt1]{\Gamma |-d \tm : \gtya[1]} & \rev[mt1]{\Gamma' |-d \tn : \gtyb[1]} &
    \rev[mt1]{\tagprecl{\gtya[1]}{\gtyb[1]}{\evd}{\evdp}} \\
    \rev[mt1]{\Gamma \gprec \Gamma' |-} \tm \gprec \tn &
    \gtya \gprec \gtyb
    }{
      \rev[mt1]{\Gamma \gprec \Gamma' |-}
      \trg{\asc{\evd \tm}{\gtya}} \gprec \trg{\asc{\evdp \tn}{\gtyb}}
    } \and
    
    \inference[\trg{(\mathord{\gprec}\oplus)}]{
       \tm  \gprec \tmp  &
      \tn \gprec \tnp
     & %
     \rev[fo2]{\pty \gprec \ptyp}
     }
    {
      \rev[fo2]{\trg{\tm \tpsum \tn}  \gprec 
      \trg{\tmp \tpsum[\ptyp] \tnp}} 
    } \and
    \inference{
      \tv \gprec \tvp &
      \tw \gprec \twp &
    }
    {
      \tv \; \tw \gprec \tvp \; \twp 
    } \and
    \rev[mt1]{
    \inference{
      \Gamma \gprec \Gamma' |- \tm \gprec \tmp \\
      \Gamma |-d \tm : \phty{\phi[][1]}{\d{\gtys[i][\p[i]]}[i \in \iSet]} &
      \Gamma' |-d \tmp : \phty{\phi[][2]}{\d{\gtysp[j][\p[j]]}[j \in \jSet]} \\
      \forall i \in \iSet.~\forall j \in \jSet.~ \gtys[i] \gprec \gtysp[j]
        \implies
        \Gamma, \tx {:} \gtys[i] \gprec \Gamma', \tx {:} \gtysp[j] |- \tn[i] \gprec \tnp[j] \\
      \forall\: \FV(\phi[][1]).~ \phi[][1] \implies \exists~\FV(\phi[][2]) \cup
        \{\cww[ij] ~|~ i \in \iSet \land j \in \jSet\}. \\
      \cjudgext{\d{\cww[ij]}[i \in \iSet \land j \in \jSet]}{\d{\gtys[i][\p[i]]}[i \in \iSet]}{\, \gprec\,}{\d{\gtysp[j][\p[j]]}[j \in \jSet]}{\phi[][1]}{\phi[][2]}{}
    }
    {
      \Gamma \gprec \Gamma' |-
      \trg{\lett{\tx}{\tm}{\overline{\tn[i]}^{\,i \in \iSet}}}  \gprec
      \trg{\lett{\tx}{\tmp}{\overline{\tnp[j]}^{\,j \in \jSet}}} }
    }
    \and
    \inference{
    \tv \gprec \rev[fo1]{\tvp} &
    \tw \gprec \twp &
    }
    {
      \trg{\add{\tv}{\tw}} \gprec \trg{\add{\tvp}{\twp}}
    } \and
    \inference{
    \tv \gprec \tvp &
    \tm \gprec \tmp &
    \tn \gprec \tnp &
    }
    {
    \trg{\ite{\tv}{\tm}{\tn}} \gprec \trg{\ite{\tvp}{\tmp}{\tnp}}
   } 
    \end{mathpar}
  \caption{Term precision of \tlang.}
  \label{fig:term-precision}
  \end{figure}

  \subsubsection*{Dynamic Gradual Guarantee} To establish the dynamic gradual guarantee (DGG) for \glang, we start by establishing the DGG for \tlang. This requires defining the notion of type and term precision for \tlang. Type precision is defined in the same way as for \glang (see Figure~\ref{source-type-precision}). Term precision is the natural lifting of type precision to terms, and is defined in Figure~\ref{fig:term-precision}.

\rev[fo1]{Term precision is a judgment indexed by a pair of environments
related by pointwise type precision,
$\Gamma \gprec \Gamma' |- \tm \gprec \tmp$. The rules for ascriptions, for
$\<let>$ and for error values consult the types of subterms, which are unique
because typing in \tlang{} is deterministic. The rules for $\lambda$ and
$\<let>$ extend the environments; the other rules keep them fixed, and the
figure omits them there.}

\begin{figure}[t]
\rev[mt1]{
\begin{mathpar}
\pi \in \{\leftvar, \rightvar\} \\
\inference[($\pi$B)]{
  \rev[fo2]{\{\gtys, \gtysp, \ev, \evp\} \cap \{\rtype,\btype,\?\} \neq \emptyset} &
  \ev \eprec \evp
}{\gtys \gprec \gtysp \vdash_{\pi} \ev \eprec \evp} \and
\inference[($\pi{->}$)]{
  \gtys[1] \gprec \gtysp[1] \vdash_{\pi} \ev[1] \eprec \evp[1] &
  \gtya[1] \gprec \gtyb[1] \vdash_{\pi} \evd[1] \eprec \evdp[1]
}{
  \gtys[1] -> \gtya[1] \gprec \gtysp[1] -> \gtyb[1] \vdash_{\pi}
  \ev[1] -> \evd[1] \eprec \evp[1] -> \evdp[1]} \and
\inference[($\pi\evd$)]{
  \forall \kk \in \kSett.\ \projpi{\cww[\kk]} \in \iSet \ad
  \forall \kk \in \kkSett.\ \projpi{\cwwp[\kk]} \in \jSet \\
  \forall\, \FV(\phi).\ \phi \implies \exists\, \FV(\phip) \cup
    \{\cw{\kk}{\kk'} \mid \kk \in \kSett \land \kk' \in \kkSett\}.\
    \phip \land {} \\
  \forall \kk \in \kSett.\ \textstyle\sum_{\kk' \in \kkSett} \cw{\kk}{\kk'} = \cww[\kk]
  \ad
  \forall \kk' \in \kkSett.\ \textstyle\sum_{\kk \in \kSett} \cw{\kk}{\kk'} = \cwwp[\kk'] \\
  \forall \kk, \kk'.\ \cw{\kk}{\kk'} > 0 \implies
    \gtys[\projpi{\cww[\kk]}] \gprec \gtysp[\projpi{\cwwp[\kk']}]
    \vdash_{\pi} \ev[\kk] \eprec \evp[\kk']
    \ad
    \gtys[\projpi{\cww[\kk]}] \gprec \gtysp[\projpi{\cwwp[\kk']}]
}{
  \begin{array}{@{}l@{}}
  \db{\phi[][1]}{\gtys[i][\p[i]]}[i \in \iSet] \gprec
  \db{\phi[][2]}{\gtysp[j][\p[j]]}[j \in \jSet] \\
  \quad \vdash_{\pi}
  \dctx[\phi][\d{\ev[\kk][\cww[\kk]]}[\kk \in \kSett]] \eprec
  \dctx[\phip][\d{\evp[\kk][\cwwp[\kk]]}[\kk \in \kkSett]]
  \end{array}}
\end{mathpar}}
\caption{\rev[mt1]{Tag-aware precision of evidences.}}
\label{fig:tag-precision}
\end{figure}

\rev[fo1]{The ascription rules compare evidences with the two judgments of
Figure~\ref{fig:tag-precision}, $\vdash_{\leftvar}$ and $\vdash_{\rightvar}$,
whose rules are stated once for a pointer $\pi$ that stands for either
$\leftvar$ or $\rightvar$. Recall from Section~\ref{sec:evidence} that
the tagged variables of an evidence for $\gtys[1] \erel \gtys[2]$ point,
through $\projl{\cdot}$ and $\projr{\cdot}$, at the entries of $\gtys[1]$ and
of $\gtys[2]$ that each entry of the evidence connects. Each judgment compares
two evidences entry by entry along one of the two pointers, together with the
entries that the pointers name; when one of the types or of the evidences is a base type or
$\?$, rule ($\pi$B) compares the evidences by the runtime precision $\eprec$
alone. The ascription rules use $\vdash_{\rightvar}$ against the ascribed types
and $\vdash_{\leftvar}$ against the types of the ascribed terms. The latter is
needed because rule $\trg{(\mathit{Dapp})}$ coerces the argument with the
domain of the function's evidence, which exchanges the two sides.}

  Rule $\trg{(\mathord{\gprec}\V)\,}$ relates two value distribution by lifting the precision relation on values to distributions via couplings, similarly to type precision. \rev[fo1]{Following~\citep{newAl:popl2019,lennonAl:toplas2022}, an error value $\errort{\gtys}$ is more precise than any value of a less precise type. An error term $\errort{\gtya}$, in contrast, is more precise only than an error term $\errort{\gtyb}$ of a less precise type, because rule $\trg{(\mathit{Derr})}$ reduces it to a distribution of errors that keeps every solution of the formula of $\gtya$, whereas a term that runs keeps only the probabilities of its own outcomes.}

Rule $\trg{(\mathord{\gprec}\oplus)}$ relates two probabilistic choices if the corresponding subterms are in precision relation, and most importantly, the formula in the more precise probabilistic choice entails the formula in the less precise probabilistic choice. The notion of entailment over \Formula is rather standard, and thus omitted. Finally, note that the pair of probabilistic choices share the same variable names after alpha renaming. %
To illustrate this rule, assume that we want to show that a probabilistic choice with a static probability $\frac{1}{2}$ is more precise than the one that we obtain replacing the static probability with $\?$. The rule application would then generate, as premise, the entailment $(\cww[1] = \frac{1}{2} \land \cww[2] = \frac{1}{2} \land \cww[1] + \cww[2] = 1) => (\cww[1] \in [0,1] \land \cww[2] \in [0,1] \land \cww[1] + \cww[2] = 1)$, with $\cww[1], \cww[2]$ universally quantified.

\rev[fo1]{The rule for $\<let>$ relates the bodies by the types of their
branches, not by their positions. For every pair of entries
$\gtys[i] \gprec \gtysp[j]$ of the types of the two bound terms, it asks that the
bodies $\tn[i]$ and $\tnp[j]$ be related in the environments extended with
those entries. Positions do not suffice, because rule $\trg{(\mathit{Dlet})}$
may route an outcome to any branch of its type. For instance, for a bound term
of type $\ds{\rtype^{\nicefrac12}, \rtype^{\nicefrac12}}$, either body may
receive an outcome of type $\rtype$. The last premise is the coupling clause
of type precision between the types of the two bound terms.}

The remaining rules are standard.

\rev[mt1]{Having defined the notion of precision, the central step of the DGG is that evidence combination is monotone with respect to imprecision.} \rev[fo1]{If two evidences compose, then so do two less precise ones, and the two results are related by precision. The lemma is stated for the runtime precision $\eprec$; recall from Example~\ref{ex:coverage-mono} that it fails for a precision that carries the coverage clauses.}

\begin{restatable}[Monotonicity of evidence combination]{lemma}{mono}\label{lemma:mono}\mechnote{\leanname{TPLC.meetTy\_mono} and \leanname{TPLC.meetD\_mono}.}
\rev[fo1]{For any} \rev[fo2]{well-formed} formula simple types $\ev[1], \ev[2], \ev[3], \ev[4] \in \GFType$ and \rev[fo1]{any} \rev[fo2]{well-formed} formula distribution types $\evd[1], \evd[2], \evd[3], \evd[4] \in \GFDType$,
  \begin{enumerate}
    \item \rev[fo1]{If $\ev[1] \eprec \ev[2]$, $\ev[3] \eprec \ev[4]$ and
      $\ev[1] \trans{} \ev[3]$ is defined, then so is $\ev[2] \trans{} \ev[4]$,
      and $\ev[1] \trans{} \ev[3] \eprec \ev[2] \trans{} \ev[4]$.}
    \item \rev[fo1]{If $\evd[1] \eprec \evd[2]$, $\evd[3] \eprec \evd[4]$ and
      $\evd[1] \trans{} \evd[3]$ is defined, then so is
      $\evd[2] \trans{} \evd[4]$, and
      $\evd[1] \trans{} \evd[3] \eprec \evd[2] \trans{} \evd[4]$.}
  \end{enumerate}
\end{restatable}

\rev[fo1]{\tlang{} satisfies a weaker form of the classical dynamic gradual
guarantee. When both terms reduce and the more precise one raises no error,
imprecision preserves the result up to precision. We first discuss why each of
these hypotheses is needed.}

\rev[fo2]{\paragraph*{Reduction of the less precise term.}
\rev[fo1]{If instead the guarantee asked, as the classical one does, that the
less precise term terminate whenever the more precise one does, then it would
fail in \tlang{}.} Consider the (informal) programs
\begin{center}
$\trg{\lett{\tx}{(1 \tpsum[\frac{1}{2}] \ttt) :: \ds{\rtype^{\frac{1}{2}}, \btype^{\frac{1}{2}}}}{\{0, \mathit{b}\}}}$
\\[0.3ex] $\gprec$ \\[0.3ex]
$\trg{\lett{\tx}{(1 \tpsum[\frac{1}{2}] \ttt) :: \ds{\?^{\frac{1}{2}}, \?^{\frac{1}{2}}}}{\{0, \mathit{b}\}}}$
\end{center}
where $\trg{\mathit{b} = \lett{y}{\tx :: \ds{\rtype^{\?}, \btype^{\?}}}{\{\Omega, 0\}}}$
binds $\tx$ at $\btype$ above and at $\?$ below.
\rev[fo1]{Above, rule $\trg{(\mathit{Dlet})}$ sends $1$ to the first body and
$\ttt$ to $\trg{\mathit{b}}$, where $\ttt$ reaches only the $\btype$ entry of
the inner $\<let>\!$, so the term reduces to $0$ and $\Omega$ never runs.
Below, both entries are $\?$, so the rule also reduces the connection that
sends $1$ to $\trg{\mathit{b}}$. In that connection the inner $\<let>$ routes
$1$ to its $\rtype$ entry and runs $\Omega$, so the term has no reduction.}}

\rev[fo2]{\paragraph*{Errors in the more precise reduction.}
\rev[fo1]{If instead the reduction of the more precise term were allowed to
raise an error, then the guarantee would fail as well.} Let $\gtya = \phty{{(\cww[1]+\cww[2]=1)}}{\ds{\rtype^{\cww[1]},\btype^{\cww[2]}}}$,
the type that rule $\trg{(\mathit{Gif})}$ gives to
\begin{center}
$\trg{\mathit{c} = \ite{\asc{\ev \ttt}{\btype}}{\asc{\ev 3}{\rtype}}{\asc{\ev \ttt}{\btype}}}$,
\end{center}
and consider
\begin{center}
$\trg{\errort{\rtype -> \gtya}~\asc{\ev 3}{\rtype}}$
\qquad $\gprec$ \qquad
$\trg{(\asc{\ev \lambda \tx{:}\rtype.\,\mathit{c}}{\rtype -> \gtya})~\asc{\ev 3}{\rtype}}$
\end{center}
\rev[fo1]{The term on the left reduces by rule $\trg{(\mathit{Derr\,app})}$ to
$\rctx[\cww[1]+\cww[2]=1] \ds{\errort{\rtype}^{\cww[1]},\errort{\btype}^{\cww[2]}}$,
which keeps the whole formula of $\gtya$, whereas the term on the right
reduces to $\rctx[\cdot] \dt{\pt{\asc{\ev 3}{\rtype}}}$. Rule
$\trg{(\mathord{\gprec}\V)}$ quantifies over every solution of the left
formula, and the solution $\cww[1] = 0$, $\cww[2] = 1$ gives probability $1$ to
$\errort{\btype}$, which is not more precise than a value of type $\rtype$.}}

\rev[fo2]{\paragraph*{The converse direction.}
\rev[fo1]{If instead the guarantee included the converse, that the more
precise term terminates whenever the less precise one does, then it would fail
for the same reason.} With $\gtya$ and $\mathit{c}$ as above,
consider the (informal) programs
\begin{center}
$\trg{\lett{\tx}{(\errort{\rtype -> \gtya}~3 \tpsum[\frac{1}{2}] \ttt) :: \ds{\btype^{1}}}{\{\ite{\tx}{\Omega}{0}\}}}$
\\[0.3ex] $\gprec$ \\[0.3ex]
$\trg{\lett{\tx}{((\lambda z{:}\rtype.\,\mathit{c})~3 \tpsum[\frac{1}{2}] \ttt) :: \ds{\btype^{1}}}{\{\ite{\tx}{\Omega}{0}\}}}$
\end{center}
\rev[fo1]{Below, the application yields $3$, so the bound term gives probability
$\frac{1}{2}$ to $\rtype$, its ascription to $\ds{\btype^{1}}$ fails, and the
term reduces to an error. Above, the error keeps the whole formula of
$\gtya$, so the ascription holds under the solution $\cww[1] = 0$ and lets
$\ttt$ through. The $\<let>$ then runs $\trg{\ite{\ttt}{\Omega}{0}}$, so the
more precise term has no reduction.}}

\rev[fo1]{The guarantee for \tlang{} therefore takes the following form.}

\begin{restatable}[Dynamic gradual guarantee for \tlang]{theorem}{targetdgg}\label{theorem:dgg}\mechnote{\leanname{TPLC.dynamic\_gradual\_guarantee}.}
  \rev[fo1]{For any terms $\tm, \tn$ of \tlang{}, any formula distribution
  types $\gtya, \gtyb \in \GFDType$, any indices $k_1, k_2$, any formulas
  $\phi[][1], \phi[][2]$ and any distribution values $\dset[1], \dset[2]$, if
  $\tm \gprec \tn$, $|-d \tm : \gtya$, $|-d \tn : \gtyb$,
  $\tm \nreds{}{k_1}\!\rctx[\phi[][1]] \dset[1]$ by a reduction that raises
  no error, and $\tn \nreds{}{k_2}\! \rctx[\phi[][2]] \dset[2]$, then
  $\phty{\phi[][1]}{\dset[1]} \gprec \phty{\phi[][2]}{\dset[2]}$.}
\end{restatable}

\rev[fo2]{By ``raises no error'' we mean that no step of the derivation takes
the second case of $\trg{(\mathit{D}\mathord{::}\gtys)}$ or of
$\trg{(\mathit{D}\mathord{::}\gtya)}$, none is $\trg{(\mathit{Derr})}$,
$\trg{(\mathit{Derr}\mathord{::}\gtys)}$ or one of the four error-propagation
rules, and none applies $\trg{(\mathit{Dv})}$ to an error value. This is a
condition on the derivation rather than on the result.}

The DGG for \glang is given by first elaborating the source terms to \tlang and then reducing the \tlang terms.
\rev[fo1]{Thus, making the annotations of a \glang{} program less precise
preserves its result up to precision, for a suitable elaboration of the less
precise program and any elaboration of the more precise one that reduces
without raising an error.}

\begin{restatable}[Dynamic gradual guarantee for \glang]{theorem}{sourcegg}\label{theorem:gg}\mechnote{\leanname{GPLC.dynamic\_gradual\_guarantee}.}
\rev[fo1]{For any terms $\sm, \srcn$ of \glang{} and any formula distribution
types $\gtya, \gtyb \in \GFDType$, if $\sm \gprec \srcn$, $|-d \sm : \gtya$
and $|-d \srcn : \gtyb$, then}
  \rev[fo1]{there is an elaboration
  $\ela{|-d \srcn}{\tn}{\gtyb}$ such that, for any elaboration
  $\ela{|-d \sm}{\tm}{\gtya}$, any indices $k_1, k_2$, any formulas
  $\phi[][1], \phi[][2]$ and any distribution values $\dset[1], \dset[2]$, if
  $\tm \nreds{}{k_1}\! \rctx[\phi[][1]] \dset[1]$ by a reduction that raises no
  error and $\tn \nreds{}{k_2}\! \rctx[\phi[][2]] \dset[2]$, then
  $\phty{\phi[][1]}{\dset[1]} \gprec \phty{\phi[][2]}{\dset[2]}$}.
\end{restatable}

\begin{figure}[t]
  \rev[mt1]{
  \[
  \begin{array}{@{}l@{}}
  \erase{\cdot} : \trg{\Value} \rightharpoonup \Value
  \\[0.4ex]
  \quad
  \begin{array}{r@{~=~}l@{\qquad}r@{~=~}l@{\qquad}r@{~=~}l}
    \erase{\tx} & \static{x} &
    \erase{\asc{\ev \ttr}{\gtys}} & \static{r} &
    \erase{\asc{\ev \ttb}{\gtys}} & \static{b}
  \end{array}
  \\[0.6ex]
  \quad
  \begin{array}{r@{~=~}l}
    \erase{\asc{\ev (\lambda x{:}\gtysp.\ \tm)}{\gtys}} &
      \static{\lambda x{:}\ph.\ \erase{\tm}}
  \end{array}
  \\[1.6ex]
  \erase{\cdot} : \trg{\Term} \rightharpoonup \Term
  \\[0.4ex]
  \quad
  \begin{array}{r@{~=~}l@{\qquad}r@{~=~}l}
    \erase{\tv\;\tw} & \erase{\tv}\;\erase{\tw} &
    \erase{\add{\tv}{\tw}} & \add{\erase{\tv}}{\erase{\tw}}
    \\
    \rev[fo2]{\erase{\tm \tpsum[\ps] \tn}} & \erase{\tm} \spsum \erase{\tn} &
    \erase{\ite{\tv}{\tm}{\tn}} & \ite{\erase{\tv}}{\erase{\tm}}{\erase{\tn}}
    \\
    \erase{\asc{\ev \tv}{\gtys}} & \asc{\erase{\tv}}{\ph} &
    \erase{\asc{\evd \tm}{\gtya}} & \asc{\erase{\tm}}{\ph}
  \end{array}
  \\[0.6ex]
  \quad
  \begin{array}{r@{~=~}l}
    \erase{\trg{\lett{x}{\tm}{\overline{\tn[j]}^{\,j \in \jSet}}}} &
      \lett{x}{\erase{\tm}}{\static{n}}
      \quad \text{where } \jSet \neq \emptyset \text{ and }
      \forall j \in \jSet.\ \erase{\tn[j]} = \static{n}
  \end{array}
  \\[1.6ex]
  \quad
  \begin{array}{r@{~=~}l@{\quad}r@{~=~}l}
    \erase{\V}_{\cww}(\static{v}) &
      \textstyle\sum \{\, \cww[j] ~\mid~ \tv[j] \in \V \ad
        \erase{\tv[j]} = \static{v} \,\} &
    \sV(\static{v}) &
      \textstyle\sum \{\, \ps[i] ~\mid~ \pt{\static{v_i}}[\ps[i]] \in \sV \ad
        \static{v_i} = \static{v} \,\}
  \end{array}
  \end{array}
  \]}
  \caption{\rev[mt1]{Erasure from \tlang to \slang.}}
  \label{fig:erasure}
\end{figure}

\subsubsection*{Conservative extension of the dynamic semantics}
Section~\ref{sec:source} shows that a fully-annotated program
\rev[fo1]{is well typed in \slang{} exactly when it is well typed in \glang{},
at a formula type that realizes its static type}. \rev[mt1]{This section shows that it also computes the
same distribution. \rev[fo1]{We state this through an \emph{erasure}, which
sends a \tlang{} value back to the \slang{} value it came from, and we compare
the probability that each \slang{} value receives. We first discuss why we do
not use the usual instrument for such a statement, a logical relation between
the two languages.}}

\rev[fo1]{\paragraph*{A logical relation between outcomes.}
If instead the statement used a logical relation, which pairs each outcome of
one distribution with an outcome of the other, then its fundamental property
would fail, already on a program of \slang{}. Write
\begin{align*}
  \sv &= \src{\lambda y{:}\rtype.\ (2 \pssum[1] \true)}
    &&: \tys[A] = \functype{\rtype}{\ds{\rtype^{1}, \btype^{0}}}, \\
  \sw &= \src{\lambda y{:}\rtype.\ 1}
    &&: \tys[B] = \functype{\rtype}{\ds{\rtype^{1}}},
\end{align*}
and consider
$\sm = \src{\lett{x}{\asc{(\sv \pssum[\frac{1}{2}] \sw)}{\tya}}{x\ 3}}$ with
$\tya = \ds{\tys[A][\frac{1}{2}], \tys[B][\frac{1}{2}]}$, a term of \slang{}
ascribed to its own type. In \tlang{}, the evidence of the ascription pairs
$\tys[A]$ with $\tys[B]$, since the entry $\btype^{0}$ of $\tys[A]$ has
probability zero and the runtime consistency $\erel$ asks no partner for it. The $\<let>$
therefore also runs its body with $\sv$ coerced to $\tys[B]$. The relation
would need an \slang{} value related to $\sv{::}\tys[B]$ at $\tys[B]$, and
there is none, since $\sv$ has type $\tys[A]$, which is not $\ceq$ to
$\tys[B]$, and $\sw$ computes $1$ where $\sv{::}\tys[B]$ computes $2$.}

\rev[mt1]{\paragraph*{Erasure.}
\rev[fo1]{The erasure compares probabilities instead of pairing outcomes.}
Figure~\ref{fig:erasure} defines the \emph{erasure} $\erase{\cdot}$, which sends
a \tlang{} value to the \slang{} value it came from by discarding the
ascription, the evidence and the type annotations, and a \tlang{} term to the
\slang{} term left underneath its types and evidences. \rev[fo1]{The erasure
is a partial function. An error value has no erasure, and a $\<let>$ erases
only if it has at least one body and all its bodies erase to the same term. A
value that occurs as a term erases by the clauses for values. The figure also
defines the total probability $\erase{\V}_{\cww}(\static{v})$ that a
distribution value $\V$ of \tlang{} gives, under a solution $\cww$ of its
formula, to the values that erase to $\static{v}$, and the total probability
$\sV(\static{v})$ that a distribution value $\sV$ of \slang{} gives to
$\static{v}$.}}

\rev[fo1]{A value and
any coercion of it erase to the same value, so $\sv{::}\tys[B]$ above erases to
$\sv$ and needs no partner. Moreover, the statement compares the total
probability that each \slang{} value receives, to which an outcome of
probability $0$ adds nothing.}

\rev[mt1]{\paragraph*{Discarded annotations.}
The lifting $\liftT{\cdot}$ sends different static types to the same formula
type, so no erasure can tell which annotation a $\lambda$ originally carried;
Figure~\ref{fig:erasure} replaces every annotation by the placeholder $\ph$.
The erasure also keeps the $\<let>$s that the \rev[fo1]{monadic}
elaboration of $\sv\;\sw$, $\add{\sv}{\sw}$ and $\ite{\sv}{\sm}{\sn}$
introduces. \rev[fo1]{To compare an \slang{} outcome with an erased one, we
therefore} write $\anf{\sv}$ for the \slang{} value $\sv$ with its annotations
discarded and those three constructs \rev[fo1]{let-normalized}
(Appendix~\ref{app:conservative} gives the definition). Neither rewriting
changes the distribution a program computes, and neither is visible on a value
of base type, \rev[fo1]{so} $\anf{\cdot}$ appears in the first item of the
theorem below and not in the second.}

We can now establish the conservative extension of the dynamic semantics
of  \tlang with respect to \slang for fully-annotated terms.

\begin{restatable}[Dynamic conservative extension of \tlang \wrt \slang]{theorem}{dyneq}\label{tm:dynamic-cons-ext}\mechnote{\leanname{TPLC.dynamic\_conservative\_extension} and \leanname{TPLC.dynamic\_conservative\_extension\_obs}.}
\rev[fo1]{For any closed term $\sm$ of \slang{}, any distribution type
$\tya \in \DType$, any term $\tmp$ of \tlang{}, any formula distribution type
$\gtya \in \GFDType$, any distribution value $\sV[1]$ of \slang{} and any
distribution value $\V[2]$ of \tlang{}, if $|-ss \sm : \tya$,
$\sm \leadsto \tmp : \gtya$ is the elaboration of the embedding of $\sm$ into
\glang{}, $\sm \snreds{*} \sV[1]$ and $\tmp \jreds[*][] \V[2]$, then for every
solution $\cww$ of the formula of $\V[2]$,}
\begin{enumerate}
     \item \rev[mt1]{$\erase{\V[2]}_{\cww}(\static{v}) =
     \sum \{\, \ps[i] ~\mid~ \pt{\static{v_i}}[\ps[i]] \in \sV[1] \ad
     \anf{\static{v_i}} = \static{v} \,\}$ for every \slang{} value
     $\static{v}$; and}
     \item \rev[mt1]{$\erase{\V[2]}_{\cww}(\static{v}) = \sV[1](\static{v})$ for
     every \rev[fo1]{\slang{} value} $\static{v}$ of base type.}
  \end{enumerate}
\end{restatable}

\rev[fo1]{The first item states that the two languages assign the same
probability to every value, reading the \slang{} outcomes with their
annotations erased. The second states that, on real numbers and booleans,
where there is no annotation to erase, the two distributions are equal.}

\rev[fo1]{\paragraph*{Termination.}
Theorem~\ref{tm:dynamic-cons-ext} assumes that both runs exist, and it does
not carry termination from one language to the other. Doing so calls for a
normalization result for \tlang{} on the terms that the elaboration produces,
which we leave as future work. The obstruction is the exhaustive $\<let>\!$, whose rule
$\trg{(\mathit{Dlet})}$ reduces a body at every connection of its routing,
connections of probability zero included, whereas a run of the source program
reduces the body once per outcome.}

\section{Discussion about \rev[fo1]{Semantics}}\label{sec:discussion}

\rev[fo1]{The runtime semantics of \tlang is a distribution semantics: it relates programs to their full distribution of final values. In this section, we discuss this design choice by comparing it with a sampling semantics, which (non-deterministically) relates programs to individual execution paths and their associated probabilities.}

Although \rev[fo1]{a sampling semantics} may seem more appealing, we argue that it is not expressive enough to properly capture the checks that shall be done at runtime (about the optimistic assumptions made during type checking). In particular, the presence of ascriptions to distribution types requires considering all possible traces simultaneously, a task mostly incompatible with a sampling semantics. \rev[fo1]{Section~\ref{sec:discussion-example} illustrates this point on an example, and Section~\ref{sec:discussion-sampling} formalizes a sampling variant of the semantics of \tlang and identifies the rule responsible for its unsoundness. Section~\ref{sec:discussion-limits} discusses the limitations of our distribution semantics, namely the cost of its checks and its restriction to programs that terminate certainly.}

\subsection{Motivating Example}\label{sec:discussion-example}

\rev[fo1]{Consider again the program of Section~\ref{sec:overview-design}:}
\begin{lstlisting}
let x: ? = 1 (*$\psum[\frac{1}{3}]$*) (2 (*$\psum[\frac{1}{3}]$*) true) in
let y: {Int(*${}^{\frac{1}{3}}$*), Bool(*${}^{\frac{2}{3}}$*)} = x in 
y
\end{lstlisting}
which can be alternatively cast through a pair of nested ascriptions:
$$\bigl((1 \psum[\frac{1}{3}] (2 \psum[\frac{1}{3}] \ttt)) ::  \dt{?^1}\bigr) :: \dt{\rtype^{\frac{1}{3}}, \btype^{\frac{2}{3}}}$$
Under the distribution semantics, the term \lstinline[mathescape]|x| evaluates to a distribution of type 
\lstinline[mathescape]|{Int${}^{\frac{1}{3}}$, Int${}^{\frac{2}{9}}$, Bool${}^{\frac{4}{9}}$}|. This is inconsistent with the declared type of \lstinline[mathescape]|y|, which expects a distribution of type  
\lstinline[mathescape]|{Int${}^{\frac{1}{3}}$, Bool${}^{\frac{2}{3}}$}|. Therefore, the program fails at runtime due to a type violation.

In contrast, a sampling semantics models evaluation as the reduction through  individual execution paths. This is represented by (non-deterministic) judgments of the form $\tm \nreds{}{} \rctx[\phi] \tv[][\p]$, which denote that term $\tm$ reduces to value $\tv$ with probability $\p$, under constraint context~$\phi$. To illustrate how this works in our example, let us consider the different execution paths (or traces) of $\tm = 1 \psum[\frac{1}{3}] (2 \psum[\frac{1}{3}] \ttt)$.

\begin{itemize}
  \item First, suppose we take left branch (of the outermost probabilistic choice). Then,
  \[
  \tm \nreds{}{} \rctx[\cdot] 1^{\frac{1}{3}}
  \]
  Since the result is an integer with probability $\frac{1}{3}$, and the expected type is $\dt{\Int^{\frac{1}{3}}, \Bool^{\frac{2}{3}}}$, the reduction succeeds.

  \item Now assume that we take the right branch, and then the left branch (of the innermost probabilistic choice):
  \[
  \tm \nreds{}{} \rctx[\cdot] 2^{\frac{2}{9}}
  \]
  Again, since $\frac{2}{9} \leq \frac{1}{3}$, this value is within the expected distribution bounds, and the reduction succeeds.

  \item Finally, if we take the right branch of both probabilistic choices, we get:
  \[
  \tm \nreds{}{} \rctx[\cdot] \ttt^{\frac{4}{9}}
  \]
  Similarly, since $\frac{4}{9} \leq \frac{2}{3}$ the term also reduces successfully.
\end{itemize}

Collecting the three traces yields the value distribution $\dt{1^{\frac{1}{3}}, 2^{\frac{2}{9}}, \ttt^{\frac{4}{9}}}$ with inferred type $\dt{\Int^{\frac{5}{9}}, \Bool^{\frac{4}{9}}}$, which is inconsistent with the ascribed type $\dt{\Int^{\frac{1}{3}}, \Bool^{\frac{2}{3}}}$ declared for \lstinline|y|. To detect this inconsistency, a sampling semantics would need to inspect all traces simultaneously. However, this contradicts the very notion of sampling, which evaluates only one path at a time. 

In fact, the only inconsistencies that could be detected at runtime under a sampling semantics would be those induced by single---individual---executions, like in the program below:
\[
(1 \psum[\frac{2}{3}] \ttt :: \dt{?^1}) :: \dt{\Int^{\frac{1}{3}}, \Bool^{\frac{2}{3}}}
\]
Here, taking the left branch yields an integer with probability $\frac{2}{3}$, which readily violates the expected distribution type.

In summary, type ascriptions to distribution types create global constraints over all execution paths, while sampling semantics only refer to individual execution paths. Therefore, the runtime enforcement of distribution types is,  in general, not possible under sampling semantics. Only the static type checking can handle such enforcement, optimistically.

\subsection{Limitations of a Sampling Semantics}\label{sec:discussion-sampling}

To better understand the limitations of sampling semantics, we now formalize a variant of the distribution semantics that evaluates a single execution trace at a time. The goal is to illustrate the tension between type annotations and runtime behavior, and to identify the specific rule responsible for the loss of soundness.

The idea behind sampling semantics is to follow the structure of the distribution semantics, but to non-deterministically select one value from a distribution instead of computing the full outcome set. Step numbers are omitted because we do not provide a metatheoretical development for this semantics. For readability, we introduce the following syntactic sugar:
\[
\rctx[\phi] \trg{\tv^{\p[1]\p[2]}} \triangleq \rctx[(\phi \land \cww = \p[1]\p[2])] \trg{\tv^{\cww}}
\]
where $\cww$ is a fresh symbolic variable (for simplicity we are ignoring indexes $\leftvar$ and $\rightvar$). We now present the rules of this sampling semantics.

When $\tm$ is already a value $\tv$, the term reduces to itself with probability $1$:
\begin{displaymath}
  \inference[\trg{(Sv)}]{}
  {\rrctx[\phii[]] \tv \nreds{}{} \rctx[\cdot] \trg{\tv[][1]}}
\end{displaymath}

Function application reduces in two steps:
\begin{displaymath}
  \inference[\trg{(\mathit{Sapp})}]{ \rrctx[\phii[]] \trg{ \asc{\cdom(\ev[])\tv}{\gtysp} } \nreds{}{} \rctx[\cdot] \tw[][1] &
  \rrctx[\phii[']]  \ssub{\trg{(\asc{\ccod(\ev[])\tm[]}{\gtya})}}{\tw}{\tx}  \nreds{}{}  \rctx[\phi] \tv^{\p} }
  {\rrctx[\phii[]] \trg{(\asc{\ev[] (\lambda \tx: \gtysp. \tm)}{\gtys -> \gtya}) \; \tv} \nreds{}{} \rctx[\phi] \tv^{\p} }
\end{displaymath}

After the argument is cast to the expected input type of the function and reduces to a value $\tw$ with probability $1$, the body of the function—previously ascribed to the expected return type—is substituted and evaluated as in the distribution semantics. If the substitution results in an error, that is, when $\tw = \errort{\gtys}$ and the substitution $\ssub{\tm[]}{\errort{\gtys}}{\tx}$ yields $\errort{\gtya}$, the evaluation proceeds using the following rule:
\begin{displaymath}
  \inference[\trg{(\mathit{Serr})}]{
    \gtya = \dctx[\phi] \d{\gtys[i][\p[i]]}[i \in \iSet]
  }{
    \rrctx[\phi[]] \errort{\gtya} \nreds{}{}
    \rctx[\phi] \errort{\gtys[i]}^{\p[i]}
  }\
\end{displaymath}
This rule selects one of the possible simple types non-deterministically, since it is not known which choices the term $\tm$ would have taken (the error propagates eagerly and prevents the evaluation of alternative branches).

One of the most interesting rules in the sampling semantics is the choice operator. In contrast to the distribution semantics, which uses a single rule, the sampling semantics introduces two distinct rules depending on which branch is selected:
\begin{mathpar}
	\rev[fo2]{\inference[\trg{(\mathit{S}\oplus l)}]{\rrctx[\phi[][]] {\tm} \nreds{}{} \rctx[\phi[][]] \tv^{\p} }
  {\rrctx[\phi[][]]  \trg{ {\tm} \tpsum {\tn} } \nreds{}{} \rctx[\phi \land \phi[][\oplus]] \trg{\tv^{\cww[1]\p}}}} \and
  \rev[fo2]{\inference[\trg{(\mathit{S}\oplus r)}]{\rrctx[\phi[][]] {\tn} \nreds{}{} \rctx[\phi[][]] \tv^{\p}  }
  {\rrctx[\phi[][]]  \trg{ {\tm} \tpsum {\tn} } \nreds{}{} \rctx[\phi \land \phi[][\oplus]] \trg{\tv^{\cww[2]\p}}}}
\end{mathpar}
Rule $\trg{(S\oplus l)}$ selects the left branch, while rule $\trg{(S\oplus r)}$ selects the right. The entire expression reduces to the value produced by the chosen subterm. The resulting probability is computed by multiplying the probability of selecting the branch (either \rev[fo2]{$\cww[1]$ or $\cww[2]$, fresh variables closed by the lifting $\phi[][\oplus]$ of $\pty$ as in rule $\trg{(\mathit{D}\oplus)}$}) with the probability of reducing to the final value within that branch. 

The rule for \texttt{let} expressions is defined as follows:
\begin{displaymath}
   \inference[\trg{(\mathit{Slet})}]{\rrctx[\phii[]] \tm \nreds{}{} \rrctx[\phii[']] \tv[1][\p[1]] &
   \rrctx[\phii[']] \ssub{\tn[]}{\tv[1]}{\tx} \nreds{}{} \rctx[\phi[][]] \tv[2][\p[2]] 
  }
  { \rrctx[\phii[]] \trg{\lett{\tx}{\tm}{\tn}} \nreds{}{} \rctx[\phi] \tv[2][\p[1]\p[2]]  }
\end{displaymath}
This rule evaluates the bound expression $\tm$ to a value $\tv[1]$ with probability $\p[1]$. Then, it substitutes $\tv[1]$ for variable $\tx$ in the body $\tn$, and evaluates the resulting term to value $\tv[2]$ with probability $\p[2]$. As a result, the entire \texttt{let} expression reduces to $\tv[2]$ with probability $\p[1]\p[2]$.

The most complex rule in the sampling semantics, and the one responsible for the lack of soundness, is the reduction of distribution types. It is defined as follows:
\begin{displaymath}
	\inference[\trg{(\mathit{D}\mathord{::}\gtya)}]{
      \rrctx[\phi[][1]] \tm \nreds{}{} \rctx[\phi[][1]] \tv[1][\p[1]]
      &  \cdot |- \tv[1] : \gtys &
      \cww \; \text{fresh} &
      \gtyap = \phty{(\pphi[1] \land \cww = 1-\p[1])}{\dt{\gtys[][\p[1]], ?^{\cww}}}
      \\
      \evd |- \gtya \rel \gtyb
       &
       \gtyb = \phty{\pphi[3]}{ \d{\gtysp[j][\p[j]]}[j \in \jSet]}
      }
      { \rrctx[\phi[][1]] \trg{ (\asc{\evd \tm}{ \gtyb }) } \nreds{}{} %
      \begin{cases}
        \rctx[\phi[][2] ]
       \cdot \tv[2][\trg{\cww[k]}\p[1]]  
        & \begin{block}
          \text{If}~ (\initReorder{\gtyap}{\gtya}) \trans{} \evd = \phty{\pphi[2]}{\d{\ev[k][\cww[k]]}[k \in \kSet]}, \\ \text{where}~  \exists k \in \kSet,
           (\projl{\cww[k]}  = 1 \land\\  j = \projr{\cww[k]}) \implies  
           \trg{ (\asc{\ev[\kk] \tv[1]}{\gtysp[j]}) } \nreds{}{} \rctx[\cdot] \tv[2][1] 
        \end{block}\\
           \rctx[\cdot] \errort{\gtyb}^{\trg{1}} & \text{otherwise}
      \end{cases} 
      } 
\end{displaymath}

The rule starts by reducing the term $\tm$ to a value $\tv[1]$ with probability $\p[1]$. Let $\gtys$ be the type of $\tv[1]$. Since only a single execution trace is observed, we cannot access the types of the values that would have resulted from other branches. Therefore, we optimistically approximate them with an unknown type. The combined type $\gtyap$ consists of $\gtys$ with probability $\p[1]$ and a wildcard type with probability $1 - \p[1]$.

The rest of the rule proceeds similarly to the distribution semantics, with one key difference. Instead of quantifying over all possible simple evidences ($k \in \kSet$), the sampling semantics selects a single $k$ non-deterministically. The simple evidence associated with $\tv[1]$ is then combined with the corresponding $\ev[k]$. 

This overly-optimistic non-deterministic behavior is the main source of unsoundness in the sampling semantics. When distribution types are present, correctness depends on the global distribution over all execution traces. However, sampling semantics only observe one trace at a time, and thus cannot guarantee that type constraints are satisfied in aggregate.

\rev[fo1]{\subsection{Limitations of the Distribution Semantics}\label{sec:discussion-limits}}

\rev[fo1]{Although the distribution semantics avoids the unsoundness of a sampling semantics, it has two limitations, which we now discuss: the cost of the checks and the class of programs to which it assigns a distribution.}

\rev[fo1]{First, since a single execution cannot perform the checks, an implementation must account for the probabilities of all the outcomes of a term before it checks an ascription to a distribution type, at a cost that grows with the number of probabilistic choices made in sequence. \glang is a core calculus, and its distribution semantics specifies which checks the ascriptions perform, not how an implementation (necessarily) performs them. Exact inference engines (see Section~\ref{sec:relwork}) compute the distributions of large programs by restricting them so that each distribution has a finite symbolic representation. An efficient representation of distributions and of evidence for \tlang is left for future work.}

\rev[fo1]{Second, the distribution semantics assigns a distribution only to programs that terminate certainly, and treats a program that terminates only almost surely as divergent. The semantics itself could accommodate almost-sure termination as \citet{DBLP:journals/ita/LagoZ12} do, by taking the distribution of a program as the supremum of its finite approximants. The checks, however, could not. At any finite stage, an approximant determines the probability of each outcome only up to the probability of the executions not yet completed. At a finite stage, a check can therefore refute an ascription whose probabilities are all known, but it can confirm an ascription only in the limit. Extending the checks to programs that terminate almost surely is left as future work.}

\section{Related Work}\label{sec:relwork}

\paragraph*{Gradual typing.}
As previously mentioned, gradual typing has been applied to many type disciplines and language constructs. To the best of our knowledge, gradual typing has not been applied to probabilistic languages, nor to non-deterministic languages. 

\citet{lehmannTanter:popl2017} presented gradual refinement types, which allow a 
smooth transition---and interoperability---between simple types and logically-refined types.
In this work, we use statically-typed refinement types to implement cast/evidence, but we do not support gradual refinement types at the source level.
\citet{phippsEtAl:oopsla2021} present \textsf{TypeWhich}, an approach for automatic type migration, which tries to infer additional or improved type annotations in gradually typed languages. Similarly to this work, \textsf{TypeWhich} also generates constraints (formulas) during type checking, and relies on an SMT solver to find solutions to their objectives.
\extends{\citet{Hattori23} propose a hybrid gradual type system for detecting shape mismatches in deep learning programs, combining best-effort inference with runtime assertion insertion. They also use formulas solved by SMT solvers, but applied to shape constraints in machine learning rather than probabilistic correctness. \citet{migeed24} extend this direction with a gradually typed tensor calculus that supports reasoning about static migration, shape constraints, and control-flow elimination. Their framework enables more precise shape analysis across dynamic inputs. Separately, \citet{Ye24merge} integrate gradual typing with a merge operator to model dynamic object-oriented features. Interestingly, their use of merge mirrors distribution semantics in spirit, but without probabilities—capturing multiple behaviors in a type-directed, deterministic way.
}

\rev[fo1]{Type inference for gradually typed languages was studied by \citet{siekVachharajani:dls2008} and \citet{garciaCimini:popl2015}, whose algorithm for an implicitly typed gradual language with let-polymorphism infers static types but may leave some type variables undetermined. \citet{miyazakiAl:popl2019} observe that different instantiations of those variables can make the translated program behave differently, and propose \emph{dynamic type inference}, in which the undetermined type variables are kept in the target blame calculus and instantiated during reduction, when a value tagged at the unknown type is cast to a type variable, with the resulting substitution applied to the whole program. They prove dynamic type inference sound and complete with respect to evaluating the fully instantiated program, and establish the gradual guarantees.}

\rev[fo1]{Like theirs, our target language defers the resolution of type information to run time, but what remains unresolved are not type variables but the unknown probabilities of distribution types. Formula types carry a constraint over symbolic probabilities, and at run time the combination of evidence (the meet) accumulates these constraints and checks that they remain satisfiable, so our runtime solves constraints over probabilities rather than unifying types.}

There exist many flavors to define the runtime semantics of gradual languages. The classical 
approach is via a 
translation to a cast calculus
\citep{wadlerFindler:esop2009,siekAl:pldi2015,siekAl:popl10,garcia2013calculating,herman07,hermanAl:hosc10,Siek2009ExploringTD};  
\citet{garciaAl:popl2016} defined the runtime semantics directly in the source language, by mimicking the proof normalization steps done in type safety; and recently, \rev[fo1]{Ye et al.~\citep{Ye2021TypeDirectedOS,yeOliveira:jfp2024}} also presented
direct dynamic semantics by using type-directed operational semantics (TDOS)~\citep{Huang2020ATO}. 
In this work, we follow the classical approach---defining a source and target language (cast calculus)---, where casts are implemented by using evidences from the AGT methodology. 

There has been active work on designing gradual languages that allows the combination/collection of types.
\citet{castagnaLanvin:icfp2017,castagnaAl:popl2019} proposed a theory for gradual set-theoretic types, supporting union, intersection and the unknown type. In parallel, \citet{toroTanter:sas2017}
explored tagged and untagged union types, and
\citet{jaferyDunfield:popl2017} sums types.
Besides many fundamental differences, this work could be seen as a generalization of gradual union types with gradual weights.

\paragraph*{Probabilistic $\lambda$-calculus.}
We can trace the origin of probabilistic $\lambda$-calculus to the work of \citet{SahebDjahromi1978ProbabilisticL}, who \rev[fo1]{presents a typed, higher-order calculus---an extension of PCF with probabilistic choice---with an operational semantics given by a random walk on terms, that is, a Markov chain, and a denotational semantics over probabilistic domains. The latter extends an earlier study of cpos of probability measures by the same author~\citep{SahebDjahromi1980CPOsOM} and, as the paper acknowledges, Plotkin's work on probabilistic domains. \citet{Jones1989APP} later generalized Saheb-Djahromi's construction to the general probabilistic powerdomain, with continuous valuations on arbitrary directed-complete partial orders.}

\citet{DBLP:journals/ita/LagoZ12} survey a variety of operational semantics for a $\lambda$-calculus with a probabilistic choice operator including  small/big-step, inductive/coinductive and call-by-value/name variants. They all fall under the category of distribution-based semantics, relating programs to probability distributions of values. On the contrary, sampling-based semantics interpret probabilistic programs as deterministic programs that are parametrized by the sequence of random choices made during the execution~\citep{DBLP:conf/icfp/BorgstromLGS16,DBLP:conf/esop/MakOPW21}. 

\citet{Ramsey2002StochasticLC} develop a denotational semantics for a stochastic $\lambda$-calculus exploiting the monadic structure of probability distributions\rev[fo1]{, a structure that \citet{Erwig:FPP:06} turn into a library that computes the distribution of a program exactly. Exact inference engines such as \textsf{PSI}~\citep{gehrAl:cav2016}, \textsf{Dice}~\citep{holtzenAl:oopsla2020} and \textsf{SPPL}~\citep{saadAl:pldi2021} compute such distributions for large programs by restricting them, for instance to discrete distributions or to bounded loops, so that the distribution of a program has a finite symbolic representation}. More recently, \citet{DANOS2011966} and \citet{10.1145/3158147} study a denotational semantics for higher-order programs in terms of coherence spaces. Finally, \citet{DBLP:journals/pacmpl/ScibiorKVSYCOMH18} use quasi-Borel spaces~\citep{DBLP:conf/lics/HeunenKSY17} as a semantic model for a functional language that combines higher-order programming with continuous distributions.

Different type systems have been developed for probabilistic $\lambda$-calculi, aimed at establishing different program invariants. \citet{Lago2017ProbabilisticTB} use sized types to reason about almost-sure termination of higher-order programs, while  \citet{DBLP:conf/csl/HeijltjesM25} propose a simpler approach, based on simple types, to reason about  termination probabilities. \citet{10.5555/3470152.3470193} \rev[fo1]{extend the distribution types of \citet{Lago2017ProbabilisticTB}, on which Section~\ref{sec:static} builds, with refinements for the analysis of expected runtime; their subtyping relation between distribution types also rests on probabilistic couplings.} \citet{Reed:2010} (and many subsequent extensions) present a type system for reasoning about program sensitivity, used for establishing differential privacy properties of programs.  

\section{Conclusion}\label{sec:conclusion}

In this work, we provide a first step into the theoretical foundations of gradual probabilistic programming. We develop \glang, to the best of our knowledge, the first gradual probabilistic language. The language enables an increased flexibility and expressivity, allowing some form of probabilistic specifications and also of program refinement via unknown probabilities in probabilistic choices.  %
The development of \glang is justified using the AGT methodology. The dynamic semantics of \glang is given via translation 
to an evidence-based calculus, called \tlang, which features a distribution-based dynamic semantics. %
The development of \glang and \tlang heavily relies on the notion of probabilistic coupling, as required for defining several relations and functions, such as type consistency, precision and consistent transitivity.
As for the metatheory, \glang satisfies type safety as well as the refined criteria for gradual languages.

\rev[fo1]{Two directions for future work arise from the distribution semantics, namely an efficient representation of distributions and of evidence, and runtime checks for programs that terminate only almost surely. Another direction is to enrich the type system with subtyping, polymorphism and recursive algebraic data types.}

\paragraph*{Conflicts of Interest.} 
The authors report no conflict of interest.

\clearpage
\appendix
\etocdepthtag.toc{appendix}

\section*{\rev[fo1]{About this appendix}}
\label{app:intro}

\rev[fo1]{This appendix states the auxiliary definitions and lemmas that the
proofs of Sections~\ref{sec:static}--\ref{sec:target} use, and proves the
results of those sections. The metatheory is mechanized in \textsf{Lean}~4. The mechanization is publicly available~\citep{gplc-lean}, and a footnote on each numbered definition, lemma and theorem names its counterpart in the mechanization.

To convey the proof arguments and their main insights while keeping the
appendix at a reasonable length, we give proofs at three levels of detail,
depending on the role each result plays in the article. Each theorem (main result) of the
article body has a subsection of its own, which restates the theorem and
describes the structure of its proof, namely the induction it rests on, the
lemmas it combines and where the difficulty lies. The theorem and the lemmas
that carry its main inductions are then proved in full, one case per rule,
with the routine cases grouped. Each lemma of the article body is proved in
outline, giving its main induction and its non-routine cases. Finally, for the
auxiliary lemmas of this appendix we state how the proof proceeds, and leave
the details to the mechanization.}

\etocsettagdepth{body}{none}
\etocsettagdepth{appendix}{subsection}
\etocsettocstyle{\section*{\rev[fo1]{Contents of the appendix}}}{}
\tableofcontents

\section{The Static Language \slang }
\label{app:static}
\rev[fo1]{This section collects the definitions and proofs for the results
of Section~\ref{sec:static} on \slang: well-formedness and type safety
(Theorem~\ref{theorem:static-safety}). The complete definitions are given in
the article body: syntax and type system in Figures~\ref{fig:static-syntax}
and~\ref{static-typing}, and the distribution semantics in
Figure~\ref{fig:static-distribution-reduction}.}

\subsection{Type System}

\begin{definition}[Well-formedness of types]\label{def:well-formed-static}\mechnote{\leanname{SPLC.WfTy} and \leanname{SPLC.WfDTy}.}
  \rev[fo1]{The well-formedness of simple and distribution types of \slang{}
  (denoted by symbol $\vdash$) is defined by the following rules:}
  \begin{mathpar}
    \inference{}{\jtf{}{\rtype}} \and
    \inference{}{\jtf{}{\btype}} \and
    \inference{ \jtf{}{\tys} & \jtf{}{\tya}}
    {\jtf{}{\tys -> \tya}} \and
    \inference{
     \sum_{i \in \iSet} \ps[i] = 1 &
     \forall i \in \iSet, \jtf{}{\tys[i]}
    }
    { \jtf{}{\d{\tys[i][\ps[i]]}[i \in \iSet]} }
    \end{mathpar}
\end{definition}

\begin{definition}[Well-formedness of contexts]\label{def:well-formed-static-ct}\mechnote{\leanname{SPLC.CtxWf}.}
  \rev[fo1]{The well-formedness of contexts of \slang{} is defined by the
  following rules:}
  \begin{mathpar}
    \inference{}{\jtf{}{ \cdot }} \and
    \inference{ \rev[fo2]{\jtf{}{\Gamma}} & \jtf{}{\tys} }{\jtf{ }{ \Gamma  , x : \tys}} \and
    \end{mathpar}
\end{definition}

\begin{lemma}[Well-formedness (equality)]\label{lemma:wellformed-equal-static}\mechnote{\leanname{SPLC.wf\_eq\_ty} and \leanname{SPLC.wf\_eq\_d}.}~
  \rev[fo1]{For any well-formed simple type $\tys[1] \in \Type$, any simple
  type $\tys[2] \in \Type$, any well-formed distribution type
  $\tya[1] \in \DType$ and any distribution type $\tya[2] \in \DType$,
  \begin{enumerate}
  \item if $\tys[1] =_{s} \tys[2]$, then $\justify{\tys[2]}$;
  \item if $\tya[1] =_{s} \tya[2]$, then the probabilities of $\tya[2]$ sum to
    $1$ and every entry of $\tya[2]$ of \emph{positive probability} is
    well-formed.
  \end{enumerate}}
\end{lemma}
\begin{proof}
  \rev[fo1]{By induction on the derivation of $\tys[1] =_{s} \tys[2]$ for
  item~1, the arrow case using item~2 on the codomains. For item~2, the
  coupling that witnesses $\tya[1] =_{s} \tya[2]$ has the probabilities of
  $\tya[1]$ as row sums and those of $\tya[2]$ as column sums, so the two
  types have the same total probability; and every entry of $\tya[2]$ is $=_{s}$
  some entry of $\tya[1]$, which is well-formed by hypothesis, so item~1
  transfers well-formedness to it. The mechanization proves item~2 for every
  entry of $\tya[2]$; the typing rules only rely on the ones of positive probability.}
\end{proof}

\rev[fo1]{The next lemma is Lemma~\ref{lemma:welltyped-wellform-static2} of
the article body, for environments that are well-formed in the sense of
Definition~\ref{def:well-formed-static-ct}.}

\begin{lemma}[Well-formed types]\label{lemma:welltyped-wellform-static}\mechnote{\leanname{SPLC.wf\_val} and \leanname{SPLC.wf\_tm}.}~
  \rev[fo1]{For any well-formed environment $\Gamma$, any value $\v$ and any
  term $\m$ of \slang{}, any simple type $\tys \in \Type$ and any
  distribution type $\tya \in \DType$,
  \begin{enumerate}
  \item if $\Gamma |-ss \v : \tys$, then $\justify{\tys}$;
  \item if $\Gamma |-ss \m : \tya$, then $\justify{\tya}$.
  \end{enumerate}}
\end{lemma}
\begin{proof}
  \rev[fo1]{By induction on the typing derivation, the two items together.
  A variable has the type that $\Gamma$ gives it, well-formed by hypothesis.
  Rule (T$\lambda$) and the two ascription rules carry the well-formedness
  of their annotation as a premise; the function type of (T$\lambda$) is
  well-formed because its body is typed under $\Gamma$ extended with that
  annotation. Rule (Tapp) returns the codomain of a well-formed function
  type, (Tv), (T$+$) and (T$::\tys$) a Dirac type, (Tif) the type of its
  first branch, and (T$\oplus$) scales and adds two well-formed types, with
  $\ps + (1-\ps) = 1$. The case with content is (Tlet). The type
  $\d{\tys[i][\ps[i]]}[i \in \iSet]$ of the bound term is well-formed by the
  induction hypothesis, so each $\tys[i]$ is, the environment
  $\Gamma, x{:}\tys[i]$ is well-formed, and the induction hypothesis applies
  to each body type $\tya[i]$; the entries of
  $\sum_{i \in \iSet} \ps[i] \cdot \tya[i]$ are entries of the $\tya[i]$, and
  its total probability is $\sum_{i \in \iSet} \ps[i] \cdot 1 = 1$.}
\end{proof}

\begin{lemma}[{\rev[fo1]{Equality and the type operators}}]\label{lemma:static-congruence}\mechnote{\leanname{SPLC.eqD\_singleton}, \leanname{SPLC.eqD\_scale}, \leanname{SPLC.eqD\_add} and \leanname{SPLC.eqD\_letRes}.}~
  \rev[fo1]{For any simple types $\tys, \tysp \in \Type$, any distribution
  types $\tya[1], \tya[1][\prime], \tya[2], \tya[2][\prime] \in \DType$, any
  probability $\ps \in [0,1]$, any distribution types
  $\d{\tys[i][\ps[i]]}[i \in \iSet], \d{\tysp[j][\ps[j]']}[j \in J] \in \DType$,
  any nonnegative matrix $w$ indexed by $\iSet \times J$, and any families
  $(\tya[i])_{i \in \iSet}$ and $(\tya[j][\prime])_{j \in J}$ of distribution
  types,
  \begin{enumerate}
  \item if $\tys =_{s} \tysp$, then $\ds{\tys[][1]} =_{s} \ds{\tysp[][1]}$;
  \item if $\tya[1] =_{s} \tya[1][\prime]$ and $\tya[2] =_{s} \tya[2][\prime]$, then
    $\ps \cdot \tya[1] + (1-\ps) \cdot \tya[2] =_{s}
     \ps \cdot \tya[1][\prime] + (1-\ps) \cdot \tya[2][\prime]$;
  \item if $\d{\tys[i][\ps[i]]}[i \in \iSet] =_{s} \d{\tysp[j][\ps[j]']}[j \in J]$
    by the coupling $w$, and (a) $\tya[i] =_{s} \tya[j][\prime]$ for all
    $i \in \iSet$ and $j \in J$ with $w_{ij} > 0$, (b) for every
    $i \in \iSet$ there is a $j \in J$ with $\tya[i] =_{s} \tya[j][\prime]$,
    and (c) for every $j \in J$ there is an $i \in \iSet$ with
    $\tya[i] =_{s} \tya[j][\prime]$, then
    $\sum_{i \in \iSet} \ps[i] \cdot \tya[i] =_{s}
     \sum_{j \in J} \ps[j]' \cdot \tya[j][\prime]$.
  \end{enumerate}}
\end{lemma}
\begin{proof}
  \rev[fo1]{Each item builds the coupling of the conclusion from the
  couplings of the hypotheses. A Dirac couples with a Dirac; scaling
  multiplies a coupling by the scalar, and the coupling of a sum is the
  block-diagonal matrix of the two couplings; for the weighted sum, the
  coupling between the entries of $\tya[i]$ and those of $\tya[j][\prime]$,
  scaled by $w_{ij}$, is placed in block $(i,j)$, and the coverage clauses
  of the hypotheses supply those of the conclusion.}
\end{proof}

\begin{lemma}[{\rev[fo1]{Determinism modulo equality}}]\label{lemma:static-det}\mechnote{\leanname{SPLC.det\_eq\_val} and \leanname{SPLC.det\_eq\_tm}.}~
  \rev[fo1]{For any well-formed environments $\Gamma$ and $\Gamma'$, any
  value $\v$ and any term $\m$ of \slang{}, any simple types
  $\tys[1], \tys[2] \in \Type$ and any distribution types
  $\tya[1], \tya[2] \in \DType$, if $\Gamma$ and $\Gamma'$ bind the same
  variables at types pointwise $=_{s}$, then
  \begin{enumerate}
  \item if $\Gamma |-ss \v : \tys[1]$ and $\Gamma' |-ss \v : \tys[2]$, then
    $\tys[1] =_{s} \tys[2]$;
  \item if $\Gamma |-ss \m : \tya[1]$ and $\Gamma' |-ss \m : \tya[2]$, then
    $\tya[1] =_{s} \tya[2]$.
  \end{enumerate}}
\end{lemma}
\begin{proof}
  \rev[fo1]{By induction on the first typing derivation; the second ends in
  the same rule, since the rules are directed by the syntax. A variable
  receives $=_{s}$ types from the two environments, and every other rule is
  a congruence for $=_{s}$ (Lemma~\ref{lemma:static-congruence}). In (Tlet)
  the types of the two bound terms are $=_{s}$ by the induction hypothesis, so their
  coupling pairs entries that are $=_{s}$; for each such pair the induction
  hypothesis under the two extended environments gives $=_{s}$ body types,
  and item~3 of Lemma~\ref{lemma:static-congruence} concludes.}
\end{proof}

\subsection{Type Safety}\label{app:static-safety}

\staticsafety*

\begin{revblock}[fo1]
\proofheading{Structure of the proof.}
The theorem combines two lemmas. Preservation
(Lemma~\ref{lemma:static-preservation}) states that the result of a
reduction is typed $=_{s}$ the type of the program, and termination
(Lemma~\ref{lemma:static-termination}) that every well-typed closed program
reduces. Distribution values are typed by rule $(V)$, whose result displays
the probabilities of the value itself, so the reading of the theorem in
terms of probabilities follows from the coupling characterization of $=_{s}$
(Lemma~\ref{couplingeq}).

Preservation is proved by induction on the reduction derivation. Its two
cases with content, application and $\<let>$, substitute into a body a value
whose type is only $=_{s}$ the type of the binder. This is why substitution
is proved modulo $=_{s}$ (Lemma~\ref{lemma:static-subst}), and why the type
system has to be deterministic modulo $=_{s}$
(Lemma~\ref{lemma:static-det}). In a $\<let>$, two outcomes of the bound term
whose types fall in the same $=_{s}$-class must give $=_{s}$ body types for
the weighted sum of the results to be $=_{s}$ the weighted sum of the rule.
Termination is proved by Tait's method~\citep{pierce:tapl}, with a reducibility predicate
indexed by the static type and invariant under $=_{s}$; the fundamental
lemma is by induction on the typing derivation, and its $\<let>$ case
repeats the assembly of preservation. Throughout, $=_{s}$ is an equivalence
(Lemma~\ref{srtrans}) and the type operators are congruences for it
(Lemma~\ref{lemma:static-congruence}). The complete proof is given at the
end of this subsection, after the lemmas it uses.
\end{revblock}

\begin{lemma}[{\rev[fo1]{Equality is an equivalence}}]\label{srtrans}\mechnote{\leanname{SPLC.EqTy.refl}/\leanname{symm}/\leanname{trans} and \leanname{SPLC.EqD.refl}/\leanname{symm}/\leanname{trans}.}
  \rev[mt1]{$=_{s}$ is an equivalence relation on static simple and
  distribution types.}
\end{lemma}
\begin{proof}
  \rev[fo1]{By Lemma~\ref{couplingeq}, two distribution types are $=_{s}$
  exactly when they assign the same total probability to every $=_{s}$-class of
  simple types and every entry of either type has an $=_{s}$ entry in the
  other. Both conditions are reflexive and symmetric, and they are transitive, since
  the probabilities of the classes are equated through the intermediate type, and the
  coverage witnesses compose. The simple-type case follows by induction on
  the type. The direct composition of two couplings would require
  normalizing by the probabilities of the intermediate type; the characterization avoids it.}
\end{proof}

\begin{lemma}[{\rev[fo1]{Substitution modulo equality}}]\label{lemma:static-subst}\mechnote{\leanname{SPLC.subst\_val} and \leanname{SPLC.subst\_tm}.}~
  \rev[fo1]{For any environments $\Gamma$, $\Delta$, $\Delta'$, any variable
  $x$, any values $v$, $w$ and any term $\m$ of \slang{}, any simple types
  $\tys, \tysp, \tys[1] \in \Type$ and any distribution type
  $\tya \in \DType$, if the environment $\Gamma, x{:}\tys, \Delta$ is
  well-formed, $|-ss w : \tysp$, $\tysp =_{s} \tys$, and $\Delta'$ binds the
  variables of $\Delta$ at types pointwise $=_{s}$ those of $\Delta$, then
  \begin{enumerate}
  \item if $\Gamma, x{:}\tys, \Delta |-ss v : \tys[1]$, then
    $\Gamma, \Delta' |-ss v[w/x] : \tys[2]$ for some simple type
    $\tys[2] \in \Type$ with $\tys[2] =_{s} \tys[1]$;
  \item if $\Gamma, x{:}\tys, \Delta |-ss \m : \tya$, then
    $\Gamma, \Delta' |-ss \m[w/x] : \tyap$ for some distribution type
    $\tyap \in \DType$ with $\tyap =_{s} \tya$.
  \end{enumerate}}
\end{lemma}
\begin{proof}
  \rev[fo1]{By induction on the typing derivation, the two items together.
  The variable $x$ receives the type $\tysp$ of $w$ and a variable of
  $\Delta$ the type that $\Delta'$ gives it, both $=_{s}$ the original. Every
  other rule is a congruence for $=_{s}$. A premise that compares two types
  by $=_{s}$ still holds after replacing either side by an $=_{s}$ type,
  since $=_{s}$ is an equivalence that respects $\dom$ and $\cod$, and the
  result of the rule is $=_{s}$ by Lemma~\ref{lemma:static-congruence}. In
  (Tlet) the substituted bound term has a type only $=_{s}$ the original
  $\d{\tys[i][\ps[i]]}[i \in \iSet]$, so each body is re-bound at an entry
  $=_{s}$ one of the $\tys[i]$; this is the perturbation of $\Delta$ that the
  statement allows, the induction hypothesis under the perturbed environment
  types each new body $=_{s}$ the old one, Lemma~\ref{lemma:static-det}
  identifies the body types of entries in one $=_{s}$-class, and item~3 of
  Lemma~\ref{lemma:static-congruence} assembles the weighted sums. The
  well-formedness of the environment enters there. The entries of the
  type of the bound term are well-formed by
  Lemma~\ref{lemma:welltyped-wellform-static}.}
\end{proof}

\begin{lemma}[{\rev[fo1]{Preservation}}]\label{lemma:static-preservation}\mechnote{\leanname{SPLC.preservation}.}~
  \rev[fo1]{For any term $\m$ of \slang{}, any distribution type
  $\tya \in \DType$ and any distribution value $\sV$, if $|-ss \m : \tya$ and
  $\m \snreds{} \sV$, then $|-ss \sV : \tyap$ for some distribution type
  $\tyap \in \DType$ with $\tyap =_{s} \tya$.}
\end{lemma}
\begin{proof}
  \begin{revblock}[fo1]
  By induction on the derivation of $\m \snreds{} \sV$. In each case the
  typing derivation of $\m$ ends in the rule for the top constructor of
  $\m$, and the type that rule $(V)$ gives to $\sV$ is compared with
  $\tya$.

  \begin{case}[values, ascriptions, addition and conditional]
    A value $v$ reduces to $\dt{\pt{v}}$; rule (Tv) gives $|-ss v : \tys$
    and $\tya = \ds{\tys[][1]}$, and rule $(V)$ types $\dt{\pt{v}}$ at
    $\ds{\tys[][1]}$. An ascription $\asc{v}{\tys}$ reduces to
    $\dt{\pt{v}}$ as well; rule (T$::\tys$) gives $|-ss v : \tysp$ with
    $\tysp =_{s} \tys$ and $\tya = \ds{\tys[][1]}$, so $(V)$ types the result
    at $\ds{\tysp[][1]} =_{s} \tya$ by item~1 of
    Lemma~\ref{lemma:static-congruence}. An ascription $\asc{\m[0]}{\tya}$
    reduces as $\m[0]$ does; rule (T$::\tya$) gives $|-ss \m[0] : \tya[0]$
    with $\tya[0] =_{s} \tya$, the induction hypothesis types $\sV$ at
    $\tyap =_{s} \tya[0]$, and transitivity concludes. An addition of two
    reals reduces to the Dirac of a real, typed at $\ds{\rtype^1}$, the type
    (T$+$) gives. A conditional reduces as the branch selected by its guard;
    rule (Tif) gives $|-ss \m[1] : \tya[1]$, $|-ss \m[2] : \tya[2]$,
    $\tya[1] =_{s} \tya[2]$ and $\tya = \tya[1]$, and the induction hypothesis
    on the branch taken types $\sV$ at $\tyap =_{s} \tya[1]$, directly or
    through $\tya[2] =_{s} \tya[1]$.
  \end{case}

  \begin{case}[probabilistic choice]
    $\m = \m[1] \spsum \m[2]$ with $\m[1] \snreds{} \sV[1]$,
    $\m[2] \snreds{} \sV[2]$ and
    $\sV = \ps \cdot \sV[1] + (1-\ps) \cdot \sV[2]$. Rule (T$\oplus$) gives
    $|-ss \m[i] : \tya[i]$ and $\tya = \ps \cdot \tya[1] + (1-\ps) \cdot \tya[2]$.
    By the induction hypothesis $|-ss \sV[i] : \tya[i][\prime]$ with
    $\tya[i][\prime] =_{s} \tya[i]$; rule $(V)$ types $\sV$ at
    $\ps \cdot \tya[1][\prime] + (1-\ps) \cdot \tya[2][\prime]$, which is $=_{s} \tya$ by
    item~2 of Lemma~\ref{lemma:static-congruence}.
  \end{case}

  \begin{case}[application]
    $\m = (\lambda x{:}\tys[0]. \m[0])\; v$ with
    $\m[0][v/x] \snreds{} \sV$. Rule (Tapp) gives
    $|-ss \lambda x{:}\tys[0]. \m[0] : \tys[1]$, $|-ss v : \tys[2]$,
    $\dom(\tys[1]) =_{s} \tys[2]$ and $\tya = \cod(\tys[1])$; the typing
    derivation of the function ends in (T$\lambda$), whose premises give
    $x{:}\tys[0] |-ss \m[0] : \tya[0]$, $\justify{\tys[0]}$ and
    $\tys[1] = \tys[0] -> \tya[0]$. The environment $x{:}\tys[0]$ is
    well-formed and $\tys[2] =_{s} \tys[0]$, so
    Lemma~\ref{lemma:static-subst} gives $|-ss \m[0][v/x] : \tya[0][\prime]$ with
    $\tya[0][\prime] =_{s} \tya[0]$; the induction hypothesis types $\sV$ at
    $\tyap =_{s} \tya[0][\prime]$, and transitivity gives $\tyap =_{s} \tya$.
  \end{case}

  \begin{case}[let]
    $\m = \lett{x}{\m[0]}{\n}$ with
    $\m[0] \snreds{} \dt{\pt{\v[j]}[\ps[j]'] \mid j \in J}$,
    $\n[\v[j]/x] \snreds{} \sV[j]$ for every $j \in J$, and
    $\sV = \ssum[j \in J] \ps[j]' \cdot \sV[j]$. Rule (Tlet) gives
    $|-ss \m[0] : \d{\tys[i][\ps[i]]}[i \in \iSet]$,
    $x{:}\tys[i] |-ss \n : \tya[i]$ for every $i \in \iSet$, and
    $\tya = \sum_{i \in \iSet} \ps[i] \cdot \tya[i]$. By the induction
    hypothesis on the bound term, rule $(V)$ types its result at
    $\d{\tysp[j][\ps[j]']}[j \in J]$, with $|-ss \v[j] : \tysp[j]$, and this
    type is $=_{s} \d{\tys[i][\ps[i]]}[i \in \iSet]$. There is a coupling $w$
    between the two with $\tys[i] =_{s} \tysp[j]$ whenever $w_{ij} > 0$,
    and by coverage every $\tysp[j]$ is $=_{s}$ some $\tys[f(j)]$. The
    entries $\tys[i]$ are well-formed by
    Lemma~\ref{lemma:welltyped-wellform-static}, so
    Lemma~\ref{lemma:static-subst} gives $|-ss \n[\v[j]/x] : \tya[j][\prime]$ with
    $\tya[j][\prime] =_{s} \tya[f(j)]$, and the induction hypothesis on each body
    types $\sV[j]$ at $\tya[j][{\prime\prime}] =_{s} \tya[j][\prime]$. Rule $(V)$ types $\sV$ at
    $\sum_{j \in J} \ps[j]' \cdot \tya[j][{\prime\prime}]$. It remains to compare this type
    with $\tya$ through item~3 of Lemma~\ref{lemma:static-congruence},
    with the coupling $w$. Whenever $\tys[i] =_{s} \tysp[j]$, also
    $\tys[i] =_{s} \tys[f(j)]$, so Lemma~\ref{lemma:static-det} gives
    $\tya[i] =_{s} \tya[f(j)]$, both being types of $\n$ under the
    well-formed environments $x{:}\tys[i]$ and $x{:}\tys[f(j)]$, and
    therefore $\tya[i] =_{s} \tya[j][{\prime\prime}]$; the two coverage
    clauses of the item hold for the same reason, since by coverage every
    $\tys[i]$ is $=_{s}$ some $\tysp[j]$ and every $\tysp[j]$ is $=_{s}$ some
    $\tys[i]$. \qedhere
  \end{case}
  \end{revblock}
\end{proof}

\begin{lemma}[{\rev[fo1]{Termination}}]\label{lemma:static-termination}\mechnote{\leanname{SPLC.termination}, from \leanname{SPLC.fundamental\_val} and \leanname{SPLC.fundamental\_tm} over the predicate \leanname{SPLC.RedV}.}~
  \rev[fo1]{For any term $\m$ of \slang{} and any distribution type
  $\tya \in \DType$, if $|-ss \m : \tya$, then $\m \snreds{} \sV$ for some
  distribution value $\sV$.}
\end{lemma}
\begin{proof}
  \begin{revblock}[fo1]
  By Tait's method~\citep{pierce:tapl}. We define reducibility of closed values at simple types
  and of closed terms at distribution types, prove it invariant under
  $=_{s}$, and prove the fundamental lemma, which states that a well-typed term closed by
  reducible values is reducible at its type. Termination is its instance
  for the empty environment.

  \paragraph*{Reducibility.}
  A closed value is reducible at $\rtype$ and at $\btype$. A closed value
  $v$ is reducible at $\tys -> \tya$ when, for every closed value $w$ with
  $|-ss w : \tysp$, $\tysp =_{s} \tys$ and $w$ reducible at $\tys$, the
  application $v\;w$ reduces to a distribution value $\sV$ with
  $|-ss \sV : \tyap$ for some $\tyap =_{s} \tya$ and, for every outcome $u$
  of $\sV$ and every entry $\tys[0]$ of $\tya$, if the type of $u$ is
  $=_{s} \tys[0]$, then $u$ is reducible at $\tys[0]$. The definition
  recurses on the entries of $\tya$ and on $\tys$, so it is well-founded. A closed term $\m$ is reducible at $\tya$ when
  $\m \snreds{} \sV$ for some $\sV$ with $|-ss \sV : \tyap$,
  $\tyap =_{s} \tya$, and each outcome of $\sV$ reducible at its own type.
  A \emph{reducible environment} for $\Gamma$ is a closing substitution
  $\rho$ that maps each $x{:}\tys$ of $\Gamma$ to a closed value $\rho(x)$
  with $|-ss \rho(x) : \tysp$, $\tysp =_{s} \tys$, and $\rho(x)$ reducible at
  $\tys$.

  Reducibility is invariant under $=_{s}$, that is, if $\tys =_{s} \tysp$, then a closed
  value is reducible at $\tys$ if and only if it is reducible at $\tysp$.
  The proof is by induction on the derivation of $\tys =_{s} \tysp$; the
  arrow case uses the hypothesis on the domains in both directions, and on
  the codomains the coverage of the coupling, which pairs every entry of
  one codomain with an $=_{s}$ entry of the other. As a consequence, if
  $|-ss \sV : \tyap$ with $\tyap =_{s} \tya$ and every outcome of $\sV$ is
  reducible at every entry of $\tya$ that is $=_{s}$ its type, then every
  outcome is reducible at its own type, since by coverage its type is $=_{s}$ some
  entry of $\tya$. Finally, if $\Gamma |-ss \v : \tys$ and $\rho$ is a
  reducible environment for $\Gamma$, then $|-ss \rho(\v) : \tysp$ for some
  $\tysp =_{s} \tys$, by Lemma~\ref{lemma:static-subst} applied once per
  variable of $\Gamma$.

  \paragraph*{Fundamental lemma.}
  For any well-formed $\Gamma$ and any reducible environment $\rho$ for
  $\Gamma$, if $\Gamma |-ss \v : \tys$, then $\rho(\v)$ is reducible at
  $\tys$, and if $\Gamma |-ss \m : \tya$, then $\rho(\m)$ is reducible at
  $\tya$. By induction on the typing derivation, the two statements
  together.

  \begin{case}[variables, constants, values and ascriptions]
    $\rho(x)$ is reducible at $\Gamma(x)$ by definition of $\rho$, and a
    real or boolean constant is reducible at its type. For a value $v$
    used as a term, $\rho(v)$ reduces to $\dt{\pt{\rho(v)}}$, typed by
    $(V)$ at $\ds{\tysp[][1]} =_{s} \ds{\tys[][1]}$ where $\tysp =_{s} \tys$
    is the type of $\rho(v)$, and its only outcome is reducible at $\tysp$
    by the induction hypothesis and invariance. The ascription
    $\asc{v}{\tys}$ is handled in the same way, with the premise
    $\tysp =_{s} \tys$ of (T$::\tys$) composed by transitivity; the
    ascription $\asc{\m}{\tya}$ reduces as $\m$ does, and the induction
    hypothesis gives a type $=_{s}$ the type of $\m$, which is $=_{s} \tya$
    by the premise of (T$::\tya$).
  \end{case}

  \begin{case}[addition and conditional]
    By the induction hypothesis on the operands and the typing of closed
    terms, $|-ss \rho(v) : \tysp$ with $\tysp =_{s} \tys[1] =_{s} \rtype$,
    and only $\rtype$ is $=_{s} \rtype$, so $\rho(v)$ is a closed value of
    type $\rtype$, hence a real literal; likewise for $w$. The sum reduces
    to the Dirac of a real, typed at $\ds{\rtype^1}$ with a reducible
    outcome. In a conditional the guard is, in the same way, a boolean
    literal, so the conditional reduces as the branch it selects, which is
    reducible at $\tya[1]$ or at $\tya[2]$ by the induction hypothesis, and
    $\tya[2] =_{s} \tya[1]$.
  \end{case}

  \begin{case}[probabilistic choice]
    By the induction hypothesis $\rho(\m[i]) \snreds{} \sV[i]$ with
    $|-ss \sV[i] : \tya[i][\prime]$, $\tya[i][\prime] =_{s} \tya[i]$ and reducible outcomes.
    Then $\rho(\m[1] \spsum \m[2]) \snreds{}
    \ps \cdot \sV[1] + (1-\ps) \cdot \sV[2]$, which $(V)$ types at
    $\ps \cdot \tya[1][\prime] + (1-\ps) \cdot \tya[2][\prime] =_{s}
     \ps \cdot \tya[1] + (1-\ps) \cdot \tya[2]$ by item~2 of
    Lemma~\ref{lemma:static-congruence}; its outcomes are those of
    $\sV[1]$ and $\sV[2]$.
  \end{case}

  \begin{case}[abstraction]
    $\v = \lambda x{:}\tys[0]. \m[0]$, and rule (T$\lambda$) gives
    $\Gamma, x{:}\tys[0] |-ss \m[0] : \tya[0]$ and $\justify{\tys[0]}$. Let
    $w$ be a closed value with $|-ss w : \tysp$, $\tysp =_{s} \tys[0]$ and
    $w$ reducible at $\tys[0]$. Then $\rho, x \mapsto w$ is a reducible
    environment for $\Gamma, x{:}\tys[0]$, which is well-formed, so by the
    induction hypothesis $\rho(\m[0])[w/x] \snreds{} \sV$ with
    $|-ss \sV : \tyap$, $\tyap =_{s} \tya[0]$ and reducible outcomes. The
    application rule then gives
    $(\lambda x{:}\tys[0]. \rho(\m[0]))\;w \snreds{} \sV$, and each outcome
    is reducible at every entry of $\tya[0]$ that its type is $=_{s}$ to, by
    invariance. This is the arrow clause of reducibility at
    $\tys[0] -> \tya[0]$.
  \end{case}

  \begin{case}[application]
    Rule (Tapp) gives $\Gamma |-ss v : \tys[1]$, $\Gamma |-ss w : \tys[2]$
    and $\dom(\tys[1]) =_{s} \tys[2]$, with
    $\tya = \cod(\tys[1])$. By the induction hypothesis $\rho(v)$ is
    reducible at $\tys[1]$ and $\rho(w)$ at $\tys[2]$, hence at
    $\dom(\tys[1])$ by invariance; and $|-ss \rho(w) : \tysp$ with
    $\tysp =_{s} \tys[2] =_{s} \dom(\tys[1])$. The arrow clause of the
    reducibility of $\rho(v)$ gives $\rho(v)\;\rho(w) \snreds{} \sV$ with
    $|-ss \sV : \tyap$, $\tyap =_{s} \cod(\tys[1])$, and every outcome
    reducible at every entry of $\cod(\tys[1])$ its type is $=_{s}$ to,
    hence at its own type.
  \end{case}

  \begin{case}[let]
    Rule (Tlet) gives $\Gamma |-ss \m[0] : \d{\tys[i][\ps[i]]}[i \in \iSet]$
    and $\Gamma, x{:}\tys[i] |-ss \n : \tya[i]$ for every $i \in \iSet$, with
    $\tya = \sum_{i \in \iSet} \ps[i] \cdot \tya[i]$. By the induction
    hypothesis on the bound term,
    $\rho(\m[0]) \snreds{} \dt{\pt{\v[j]}[\ps[j]'] \mid j \in J}$, typed by
    $(V)$ at $\d{\tysp[j][\ps[j]']}[j \in J] =_{s} \d{\tys[i][\ps[i]]}[i \in \iSet]$,
    each $\v[j]$ reducible at $\tysp[j]$. By coverage each $\tysp[j]$ is
    $=_{s}$ some $\tys[f(j)]$, so $\rho, x \mapsto \v[j]$ is a reducible
    environment for $\Gamma, x{:}\tys[f(j)]$, well-formed by
    Lemma~\ref{lemma:welltyped-wellform-static}, and the induction
    hypothesis on the body under $x{:}\tys[f(j)]$ gives
    $\rho(\n)[\v[j]/x] \snreds{} \sV[j]$ with $|-ss \sV[j] : \tya[j][{\prime\prime}]$,
    $\tya[j][{\prime\prime}] =_{s} \tya[f(j)]$ and reducible outcomes. The $\<let>$ rule
    gives $\rho(\lett{x}{\m[0]}{\n}) \snreds{} \ssum[j \in J] \ps[j]' \cdot \sV[j]$,
    typed by $(V)$ at $\sum_{j \in J} \ps[j]' \cdot \tya[j][{\prime\prime}]$, which is
    $=_{s} \tya$ exactly as in the $\<let>$ case of
    Lemma~\ref{lemma:static-preservation}, by
    Lemma~\ref{lemma:static-det} and item~3 of
    Lemma~\ref{lemma:static-congruence}; the outcomes of the weighted sum
    are those of the $\sV[j]$.
  \end{case}

  For a closed program the environment is empty, so the fundamental lemma
  applies with the empty substitution, and reducibility at $\tya$ includes
  the reduction $\m \snreds{} \sV$.
  \end{revblock}
\end{proof}

\begin{theoremproof}{Theorem~\ref{theorem:static-safety} (Type safety for \slang)}
  \rev[fo1]{
  Lemma~\ref{lemma:static-preservation} types the result of the reduction
  that Lemma~\ref{lemma:static-termination} gives. The reading in
  terms of probabilities given in Section~\ref{sec:static-dynamics} follows by unfolding $=_{s}$ in
  $\tyap =_{s} \tya$. The type that $(V)$ assigns
  to $\sV$ carries the probabilities of $\sV$, so by Lemma~\ref{couplingeq} the
  probability that $\sV$ gives to the values of each $=_{s}$-class of simple
  types is the probability that $\tya$ declares for that class.}
\end{theoremproof}

\section{The Source Language \glang }
\label{app:source}
\rev[fo1]{This section collects the proofs for the results of
Section~\ref{sec:source} on \glang: well-formedness and well-typedness, the
equivalences with the AGT definitions (Lemmas~\ref{couplingeq},
\ref{agtconsistencyeq} and~\ref{lemma:precisioneq}), the static gradual
guarantee (Theorem~\ref{theorem:sgg}), and the static conservative extension
of \glang{} \wrt \slang{} (Theorem~\ref{theorem:staticeq}). The complete
definitions are given in the article body: the type system in
Figure~\ref{fig:source-type-system}, and type and term precision in
Figures~\ref{source-type-precision} and~\ref{source-term-precision}.}

\subsection{Type System}

\rev[mt1]{Recall (Section~\ref{sec:source}) that the type system of \glang{}
manipulates \emph{formula} types throughout, and that source gradual types
occur only as user annotations, entering the system through the lifting
$\liftT{\cdot}$/$\liftD{\cdot}$ at the $\lambda$ and ascription rules.
Well-formedness of formula types \rev[fo2]{and of environments} is
Definition~\ref{def:well-formed-source}; for \emph{annotations} it is the following
definition, whose distribution rule requires plausibility of the annotated
probabilities.}

\begin{definition}[{\rev[fo1]{Well-formedness of source annotations}}]\label{def:well-formed-source2}\mechnote{\leanname{GPLC.WfTy}, \leanname{GPLC.WfDTy} and \leanname{GPLC.WfAnnT}.}
  \rev[mt1]{Well-formedness of source gradual types ($\sgtys \in \GType$,
  $\sgtya \in \GDType$) is defined \rev[fo1]{by the following rules}:}
  \begin{mathpar}
    \inference{}{\jtf{}{\rtype}} \and
    \inference{}{\jtf{}{\btype}} \and
    \inference{}{\jtf{}{\?}} \and
    \inference{ \jtf{}{\sgtys} & \jtf{}{\sgtya}}
    {\jtf{}{\sgtys -> \sgtya}} \and
    \inference{
      \exists \overline{\ps[i] \in \crp{\pty[i]}}. \ \textstyle\sum_{i \in \iSet} \ps[i] = 1 &
      \forall i \in \iSet. \ \jtf{}{\sgtys[i]}
    }
    { \jtf{}{ \d{\sgtys[i][\pty[i]]}[i \in \iSet] } } \and
  \end{mathpar}
\end{definition}

\begin{lemma}[Lifting well-formedness]\label{lemma:wellformed-lifting}\mechnote{\leanname{GPLC.goodTy\_liftF} and \leanname{GPLC.goodD\_liftFD}.}~
  \rev[fo1]{For any well-formed gradual simple type $\sgtys \in \GType$ and
  any well-formed gradual distribution type $\sgtya \in \GDType$,
  $\jtf{}{\liftT{\sgtys}}$ and $\jtf{}{\liftD{\sgtya}}$.}
\end{lemma}
\begin{proof}
  \rev[fo1]{By induction on the annotation. For a distribution annotation
  $\d{\sgtys[i][\pty[i]]}[i \in \iSet]$, the lift conjoins, for fresh
  variables $\cww[i]$, the liftings $\liftP{\pty[i]}[\cww[i]]$ of the
  probabilities with the condition $\sum_i \cww[i] = 1$. The plausibility
  premise of Definition~\ref{def:well-formed-source2} supplies admissible
  probabilities $\ps[i] \in \crp{\pty[i]}$ whose sum is $1$, so the lifted formula
  is satisfiable. Its conjuncts are linear constraints, so its solutions
  are closed under convex combinations, and the equation
  $\sum_i \cww[i] = 1$ gives every solution total probability $1$. The induction
  hypothesis gives the well-formedness of every lifted entry, and the
  variables $\cww[i]$ occur in the formula by construction. Hence the lifted
  distribution type is well-formed.}
\end{proof}

\rev[fo1]{The next lemma is Lemma~\ref{lemma:welltyped-wellform-source2} of
the article body.}

\begin{lemma}[Well-formed types]\label{lemma:welltyped-wellform-source}\mechnote{\leanname{GPLC.good\_val} and \leanname{GPLC.good\_tm}.}~
  \rev[fo1]{For any well-formed environment $\Gamma$, any value $\sv$ and
  any term $\sm$ of \glang{} whose type annotations are well-formed, any
  formula simple type $\gtys \in \FGType$ and any formula distribution type
  $\gtya \in \FGDType$,
  \begin{enumerate}
  \item if $\Gamma |-t \sv : \gtys$, then $\justify{\gtys}$;
  \item if $\Gamma |-d \sm : \gtya$, then $\justify{\gtya}$.
  \end{enumerate}}
\end{lemma}
\begin{proof}
  \rev[fo1]{By induction on the typing derivation, the two items together.
  A variable reads the well-formed environment; the $\lambda$ rule and the
  ascription rules introduce annotations through the lift, well-formed by
  Lemma~\ref{lemma:wellformed-lifting}; the application rule returns
  $\ccod(\gtys)$, well-formed when $\gtys$ is, since
  $\ccod(\?) = \liftD{\ds{\?^{\?}}}$ is well-formed; and the rules for
  values and for addition return a singleton. The three remaining rules
  build their result from well-formed operand types. A solution of the type
  of a choice, of a conditional or of a \src{let} is a convex combination of
  solutions of the operand types, with weights given by a solution of the
  probability, of the hull parameter or of the formula of the bound term, so its
  total probability is $1$, the set of such solutions is convex, and a solution exists
  because the operands have one.}
\end{proof}

\rev[fo1]{Two formula types that differ only in the names of the fresh
variables that the typing rules introduce are identified throughout. The
type system is directed by the syntax of terms and assigns a unique type.}

\begin{lemma}[{\rev[fo1]{Determinism of typing}}]\label{lemma:source-det}\mechnote{\leanname{GPLC.det\_val} and \leanname{GPLC.det\_tm}.}~
  \rev[fo1]{For any environment $\Gamma$, any value $\sv$ and any term $\sm$
  of \glang{}, any formula simple types $\gtys[1], \gtys[2] \in \FGType$ and
  any formula distribution types $\gtya[1], \gtya[2] \in \FGDType$,
  \begin{enumerate}
  \item if $\Gamma |-t \sv : \gtys[1]$ and $\Gamma |-t \sv : \gtys[2]$, then
    $\gtys[1] = \gtys[2]$;
  \item if $\Gamma |-d \sm : \gtya[1]$ and $\Gamma |-d \sm : \gtya[2]$, then
    $\gtya[1] = \gtya[2]$.
  \end{enumerate}}
\end{lemma}
\begin{proof}
  \rev[fo1]{By induction on the first derivation. The second ends in the
  same rule, since every rule is determined by the top constructor of the
  term; the induction hypothesis identifies the types of the premises, and
  the result type is a function of them, $\cdom$ and $\ccod$ being partial
  functions.}
\end{proof}

\subsection{Equivalences with AGT Definition}\label{app:agt-equivalences}
\couplingeq*
\begin{proof}
  \rev[fo1]{Item~1 holds because the two definitions coincide on simple
  types. For item~2, the coverage clauses of the two definitions are the
  same, so only the premise that compares probabilities has to be related to the existence of a
  coupling. Write $\tya[1] = \d{\tys[i][\ps[i]]}[i \in \iSet]$ and
  $\tya[2] = \d{\tysp[j][q_j]}[j \in \jSet]$, and let $\ps[C]$ be the common
  probability of an $=_{s}$-class $C$ on the two sides. The matrix with
  $w_{ij} = \ps[i] \, q_j / \ps[C]$ when $\tys[i] \in C$ and $\tysp[j] \in C$
  for a class $C$ of positive probability, and $w_{ij} = 0$ otherwise, has row sums
  $\ps[i]$ and column sums $q_j$, since the probabilities of the entries of a class sum to
  $\ps[C]$ on each side; the division is possible only for classes of
  positive probability, which
  is why the coverage clauses are carried separately. Conversely, given a
  coupling, the probability of a class on the left is the sum of its rows, which
  equals the sum of the columns of the same class, because a coupling
  supported on $=_{s}$ has positive weight only between entries of the same
  class.}
\end{proof}

\agtconsistencyeq*
\begin{proof}
  \rev[fo1]{By induction on the types, the two items together. Recall that
  a static distribution type $\tya \in \crta{\gtya}$ comes with a solution
  $\{\ps[i]\}$ of the closing formula of $\gtya$ and a splitting of its
  entries into non-empty blocks, the block of entry $i$ having probability
  $\ps[i]$. For simple types, $\rtype$, $\btype$ and the cases with $\?$ agree
  directly, as $\crt{\gtys}$ is non-empty by well-formedness, and arrows
  decompose componentwise.

  For distribution types, from the declarative to the inductive consistency,
  equal concretizations $\tya[1] \in \crta{\gtya}$ and $\tya[2] \in \crta{\gtyb}$
  are related by a coupling $\{c_{kl}\}$ between their static entries
  (Lemma~\ref{couplingeq}); summing it over the blocks of $i$ and of $j$
  gives a coupling $\{w_{ij}\}$ between the two solutions, and every positive
  weight $w_{ij}$ contains a positive $c_{kl}$ with equal static entries, so a common
  static type concretizes $\gtys[i]$ and $\gtysp[j]$ and the induction
  hypothesis gives $\gtys[i] \rel \gtysp[j]$. The coverage clauses come from
  those of $=_{s}$ through the non-emptiness of the blocks. Conversely, let
  $\{w_{ij}\}$ be a coupling with $\gtys[i] \rel \gtysp[j]$ whenever its
  weight $w_{ij}$ is positive. We build two equal concretizations with one
  static entry for each pair of indices $(i,j)$, of probability $w_{ij}$ on
  both sides. If $w_{ij} > 0$, the entry $(i,j)$ receives a common
  concretization of $\gtys[i]$ and $\gtysp[j]$, which exists by the induction
  hypothesis; if $w_{ij} = 0$, it receives, on the left, the common
  concretization that the coverage clauses provide for $\gtys[i]$ and, on the
  right, the one they provide for $\gtysp[j]$, which keeps the coverage
  clauses of $=_{s}$ satisfied. Grouping the entries by $i$ and by $j$ gives
  the two splittings, and the diagonal coupling, which relates the entry
  $(i,j)$ of one side to the entry $(i,j)$ of the other with weight $w_{ij}$,
  witnesses the equality.}
\end{proof}

\precisioneq*
\begin{proof}
  \rev[fo1]{By induction on the types, item~1 as in
  Lemma~\ref{agtconsistencyeq} with containment
  $\crt{\gtys} \subseteq \crt{\gtysp}$ in place of a common element. For
  item~2, let $\tya[1] \in \crta{\gtya}$, with solution $\{\ps[i]\}$ and
  blocks, and write $b(k)$ for the entry of $\gtya$ whose block contains the
  static entry $k$. The coupling premise of $\gtya \gprec \gtyb$ at
  $\{\ps[i]\}$ gives a solution $\{q_j\}$ of $\gtyb$ and a coupling
  $\{w_{ij}\}$ with $\gtys[i] \gprec \gtysp[j]$ whenever $w_{ij} > 0$;
  composing the splitting of $\tya[1]$ with it gives a coupling $\{g_{kj}\}$
  between the probabilities of the static entries and $\{q_j\}$. Let $\ell$ and $r$ witness the
  two coverage clauses of $\gtya \gprec \gtyb$, that is, $\gtys[i] \gprec \gtysp[\ell(i)]$
  for every $i$ and $\gtys[r(j)] \gprec \gtysp[j]$ for every $j$. We build
  $\tya[2] \in \crta{\gtyb}$ with one entry per pair $(k,j)$, of probability
  $g_{kj}$, whose type is a concretization of $\gtysp[j]$ chosen as follows.
  If $g_{kj} > 0$ or $j = \ell(b(k))$, then $\gtys[b(k)] \gprec \gtysp[j]$, and
  item~1 gives a concretization of $\gtysp[j]$ that is $=_{s} \tys[k]$.
  Otherwise, item~1 applied to $\gtys[r(j)] \gprec \gtysp[j]$ gives a
  concretization of $\gtysp[j]$ that is $=_{s}$ the type of some static entry
  in the block of $r(j)$. The block of $j$ has probability $\sum_k g_{kj} = q_j$ and
  is non-empty, so $\tya[2] \in \crta{\gtyb}$. Finally
  $\tya[1] =_{s} \tya[2]$. The coupling sends the static entry $k$ to the
  entries $(k,j)$ of $\tya[2]$ with weight $g_{kj}$, and the two types are
  $=_{s}$ whenever $g_{kj} > 0$; the entry $k$ is covered by the entry
  $(k, \ell(b(k)))$ of $\tya[2]$, and every entry of $\tya[2]$ is covered by
  the static entry its type was copied from.}
\end{proof}

\rev[fo1]{We do not claim the converse of Lemma~\ref{lemma:precisioneq},
and nothing in the article uses it. Reversing the argument means producing,
from the family of couplings that the concretizations of $\gtya$ induce one
by one, a single coupling that has positive weight only between entries
related by precision; that
is a transportation-feasibility problem~\citep{bertsimasTsitsiklis:1997} over the two index sets, in the
manner of the completeness proof of the semantic precision
of~\citet{lennonAl:toplas2022}, and a single $=_{s}$-equal pair of
concretizations witnesses the consistency of two entries, not their
precision.}

\subsection{Static Gradual Guarantee}\label{app:sgg}

\sgg*

\begin{revblock}[fo1]
\proofheading{Structure of the proof.}
The theorem is the case of empty environments of the following open
statement, which the induction proves.
\begin{quote}
For any environments $\Gamma \gprec \Gamma'$ with $\jtf{}{\Gamma}$ and
$\jtf{}{\Gamma'}$, if $\Gamma |-t \sv : \gtys$, $\sv \gprec \sw$ and the
annotations of $\sw$ are well-formed, then $\Gamma' |-t \sw : \gtysp$ for
some $\gtysp$ with $\gtys \gprec \gtysp$; and likewise for terms and
distribution types.
\end{quote}
The proof is by induction on the typing derivation of the more precise term.
Term precision is a congruence, so the less precise term has the same top
constructor and its typing derivation is built with the rule for that
constructor. Each case has the same three steps: the induction hypothesis
types the subterms of the less precise term at types related to those of the
more precise one; the consistency premises of the rule transport along
precision (Lemma~\ref{conp}), and its well-formedness premises are the
hypothesis on the annotations of the less precise term; and the result type
is related to the more precise one because every type operator of
Figure~\ref{fig:source-type-system} is monotone for precision
(Lemma~\ref{lemma:prec-operators}).

The content is in the monotonicity of the operators, whose proofs construct
couplings, and in the \src{let} case. There the rule types the body once per
entry of the type of the bound term, on both sides. The coverage clause of precision
supplies, for every entry of the type of the less precise bound term, a related entry
of the more precise one at which the induction hypothesis on the body can be
applied; the determinism of typing (Lemma~\ref{lemma:source-det})
identifies the body types obtained from different related entries; and the
coupling of the two weighted unions combines, in each column, the solutions
obtained from several more precise entries into a convex mixture, which is a
solution because the body types of the less precise term are well-formed.
The complete proof is given at the end of this subsection, after the lemmas
it uses.
\end{revblock}

\rev[mt1]{The complete precision rules are given in
Figures~\ref{source-type-precision} and~\ref{source-term-precision} in the
article body. Context precision $\Gamma \gprec \Gamma'$ is the pointwise lifting
of type precision. Both contexts bind the same variables, and
$\Gamma(\sx) \gprec \Gamma'(\sx)$ for each of them.}

\begin{lemma}[Consistency precision]\label{conp}\mechnote{item~1, \leanname{GPLC.cons\_prec\_ty}; item~2, \leanname{GPLC.cons\_prec\_d}; for the runtime relations, \leanname{GPLC.econs\_eprec\_ty} and \leanname{GPLC.econs\_eprec\_d}.}
  \rev[fo1]{For any formula simple types
  $\gtys[1], \gtys[2], \gtys[3], \gtys[4] \in \FGType$ and any formula
  distribution types $\gtya[1], \gtya[2], \gtya[3], \gtya[4] \in \FGDType$,
  \begin{enumerate}
  \item if $\gtys[1] \rel \gtys[3]$, $\gtys[1] \gprec \gtys[2]$ and
    $\gtys[3] \gprec \gtys[4]$, then $\gtys[2] \rel \gtys[4]$;
  \item if $\gtya[1] \rel \gtya[3]$, $\gtya[1] \gprec \gtya[2]$ and
    $\gtya[3] \gprec \gtya[4]$, then $\gtya[2] \rel \gtya[4]$.
  \end{enumerate}}
\end{lemma}
\begin{proof}
  \rev[fo1]{By induction on the types, the two items together. For item~1,
  $\?$ on either side of the conclusion is consistent with everything, base
  types are preserved by precision, and arrows decompose componentwise. For
  item~2, let $W$ be the consistency coupling between solutions $p$ of
  $\gtya[1]$ and $q$ of $\gtya[3]$, and $u$ and $v$ the precision couplings
  that $\gtya[1] \gprec \gtya[2]$ and $\gtya[3] \gprec \gtya[4]$ give at $p$
  and at $q$. The composition (Lemma~\ref{lemma:coupling-composition}) of the
  transpose of $u$, of $W$ and of $v$ couples the two transported
  solutions, and each of its positive weights comes from a positive weight
  of $W$, whose entries are consistent by item~1 after the two precisions. The coverage
  clauses chain in the same way. An entry of $\gtya[2]$ is covered by an
  entry of $\gtya[1]$, which is covered by an entry of $\gtya[3]$, which is
  covered by an entry of $\gtya[4]$, and the three are consistent by item~1;
  symmetrically for an entry of $\gtya[4]$. The same argument, without the
  coverage clauses, proves the transport for the runtime relations of
  Figure~\ref{fig:runtime-relations}. If $\gtya[1] \erel \gtya[3]$,
  $\gtya[1] \eprec \gtya[2]$ and $\gtya[3] \eprec \gtya[4]$, then
  $\gtya[2] \erel \gtya[4]$, and likewise for simple types.}
\end{proof}

\begin{lemma}[{\rev[fo1]{Monotonicity of the type operators}}]\label{lemma:prec-operators}\mechnote{item~1, \leanname{GPLC.prec\_point}; item~2, \leanname{GPLC.domcod\_mono}; item~3, \leanname{GPLC.prec\_choose}, \leanname{GPLC.prec\_choose\_unk} and \leanname{GPLC.prec\_chooseU}; item~4, \leanname{GPLC.prec\_let}.}~
  \rev[fo1]{For any formula simple types $\gtys, \gtysp \in \FGType$, any
  gradual probabilities $\pty, \ptyp \in \GProb$, any formula distribution
  types $\gtya[1], \gtya[2], \gtyb[1], \gtyb[2] \in \FGDType$, any formula
  distribution types $\db{\phi}{\gtys[i][\p[i]]}[i \in \iSet]$ and
  $\db{\phi[\prime]}{\gtysp[j][q_j]}[j \in \jSet]$ in $\FGDType$, and any
  families $(\gtya[i])_{i \in \iSet}$ and $(\gtyb[j])_{j \in \jSet}$ of
  formula distribution types,
  \begin{enumerate}
  \item if $\gtys \gprec \gtysp$, then
    $\dctx[(\cww=1)][\ds{\gtys[][\cww]}] \gprec \dctx[(\cww=1)][\ds{\gtysp[][\cww]}]$;
  \item if $\gtys \gprec \gtysp$, $\jtf{}{\gtys}$ and $\cdom(\gtys)$ is
    defined, then $\cdom(\gtysp)$ and $\ccod(\gtysp)$ are defined,
    $\cdom(\gtys) \gprec \cdom(\gtysp)$ and $\ccod(\gtys) \gprec \ccod(\gtysp)$;
  \item if $\pty \gprec \ptyp$, $\gtya[1] \gprec \gtyb[1]$ and
    $\gtya[2] \gprec \gtyb[2]$, then the type that the choice rule of
    Figure~\ref{fig:source-type-system} assigns to a choice with probability
    $\pty$ between terms of types $\gtya[1]$ and $\gtya[2]$ is $\gprec$ the
    type it assigns with probability $\ptyp$ between terms of types
    $\gtyb[1]$ and $\gtyb[2]$, and
    $\gtya[1] \mathbin{\oplus_{\?}} \gtya[2] \gprec \gtyb[1] \mathbin{\oplus_{\?}} \gtyb[2]$;
  \item if $\db{\phi}{\gtys[i][\p[i]]}[i \in \iSet] \gprec
           \db{\phi[\prime]}{\gtysp[j][q_j]}[j \in \jSet]$, every $\gtyb[j]$ is
    well-formed, and $\gtya[i] \gprec \gtyb[j]$ for all $i \in \iSet$ and
    $j \in \jSet$ with $\gtys[i] \gprec \gtysp[j]$, then
    $\sugarconj{\phi}{\sum_{i \in \iSet} \p[i] \cdot \gtya[i]} \gprec
     \sugarconj{\phi[\prime]}{\sum_{j \in \jSet} q_j \cdot \gtyb[j]}$.
  \end{enumerate}}
\end{lemma}
\begin{proof}
  \rev[fo1]{Each item builds the coupling and the coverage of the
  conclusion from those of the hypotheses. For item~1 both types have the
  single solution $\cww = 1$ and the coupling is the single weight $1$. For
  item~2, $\gtys$ is an arrow or $\?$; if $\gtysp$ is an arrow too, the
  precision of the components is the claim; if $\gtysp = \?$, then
  $\cdom(\?) = \?$ is above everything, and
  $\ccod(\gtys) \gprec \ccod(\?) = \liftD{\ds{\?^{\?}}}$ because
  $\ccod(\gtys)$ is well-formed. Every solution has total probability $1$, which
  the single entry of $\liftD{\ds{\?^{\?}}}$ receives, and that entry, $\?$, is
  above every entry. For item~3, a solution of the more precise type is
  $(a \cdot p, (1-a) \cdot q)$ with $a \in \crp{\pty} \subseteq \crp{\ptyp}$
  and $p$, $q$ solutions of the branch types; the precisions of the
  branches transport $p$ and $q$ to solutions $p'$, $q'$ with couplings
  $u$, $v$, and the block-diagonal matrix with blocks $a \cdot u$ and
  $(1-a) \cdot v$ couples $(a \cdot p, (1-a) \cdot q)$ with
  $(a \cdot p', (1-a) \cdot q')$, a solution of the less precise type; the
  coverage clauses hold block by block. For item~4, a solution of the more
  precise union is $(\p[i] \cdot b^{i})_{i}$ for a solution $p$ of $\phi$ and
  solutions $b^{i}$ of the $\gtya[i]$. The precision of the types of the bound terms
  gives a solution $q$ of $\phi[\prime]$ and a coupling $u$ of $p$ and $q$ with
  $\gtys[i] \gprec \gtysp[j]$ whenever $u_{ij} > 0$, hence
  $\gtya[i] \gprec \gtyb[j]$, which transports $b^{i}$ to a solution $c^{ij}$
  of $\gtyb[j]$ with a coupling $v^{ij}$. For each $j$ with $q_j > 0$, the
  convex mixture $c_j = \sum_{i} (u_{ij} / q_j) \cdot c^{ij}$ is a solution
  of $\gtyb[j]$ because $\gtyb[j]$ is well-formed, and for $q_j = 0$ any
  solution serves; $(q_j \cdot c_j)_j$ is then a solution of the less precise
  union, coupled with the more precise one by the matrix with block
  $u_{ij} \cdot v^{ij}$ at $(i,j)$, which has positive weight only between
  $\gprec$-related entries. The coverage clauses follow from those of the
  types of the bound terms and of the pairs $\gtya[i] \gprec \gtyb[j]$.}
\end{proof}

\begin{theoremproof}{Theorem~\ref{theorem:sgg} (Static gradual guarantee for \glang)}
  \begin{revblock}[fo1]
  We prove the open statement of the structure paragraph above, by
  induction on the typing derivation of the more precise term, the
  statements for values and for terms together. Let $\Gamma \gprec \Gamma'$,
  both well-formed, and let the annotations of the less precise term be
  well-formed.

  \begin{case}[constants, variables, values and addition]
    A constant is related only to itself and keeps its type. For a variable,
    $\Gamma(\sx) = \gtys$ and $\Gamma \gprec \Gamma'$ give $\Gamma'(\sx) = \gtysp$
    with $\gtys \gprec \gtysp$. For a value $\sv \gprec \sw$ used as a term,
    the induction hypothesis gives $\Gamma' |-t \sw : \gtysp$ with
    $\gtys \gprec \gtysp$, and the two singletons are related by item~1 of
    Lemma~\ref{lemma:prec-operators}. For an addition, the induction
    hypotheses relate the types of the operands, Lemma~\ref{conp} transports
    the premises $\gtys \rel \rtype$ and $\gtysp \rel \rtype$ (with
    $\rtype \gprec \rtype$), and both results are the singleton of $\rtype$.
  \end{case}

  \begin{case}[abstraction]
    $\sv = \lambda \sx{:}\sgtys. \sm \gprec \lambda \sx{:}\sgtysp. \srcn$
    with $\liftT{\sgtys} \gprec \liftT{\sgtysp}$ and $\sm \gprec \srcn$. The
    typing derivation of $\sv$ ends in the $\lambda$ rule, whose premises
    give $\Gamma, \sx{:}\liftT{\sgtys} |-d \sm : \gtya$ and $\jtf{}{\sgtys}$.
    The annotation $\sgtysp$ is well-formed by hypothesis, so the two extended
    environments are well-formed (Lemma~\ref{lemma:wellformed-lifting}) and
    related; the induction hypothesis gives
    $\Gamma', \sx{:}\liftT{\sgtysp} |-d \srcn : \gtyb$ with
    $\gtya \gprec \gtyb$, and the $\lambda$ rule types the less precise
    abstraction at $\liftT{\sgtysp} -> \gtyb$, related to
    $\liftT{\sgtys} -> \gtya$ by the arrow rule of precision.
  \end{case}

  \begin{case}[ascriptions]
    $\asc{\sm}{\sgtyb[1]} \gprec \asc{\srcn}{\sgtyb[2]}$ with
    $\sm \gprec \srcn$ and $\liftD{\sgtyb[1]} \gprec \liftD{\sgtyb[2]}$. The
    typing derivation of the more precise term ends in the ascription rule,
    whose premises give $\Gamma |-d \sm : \gtya$,
    $\gtya \rel \liftD{\sgtyb[1]}$ and $\jtf{}{\sgtyb[1]}$. The induction
    hypothesis gives $\Gamma' |-d \srcn : \gtyb$ with $\gtya \gprec \gtyb$;
    Lemma~\ref{conp} gives $\gtyb \rel \liftD{\sgtyb[2]}$, and $\sgtyb[2]$ is
    well-formed by hypothesis, so the rule types the less precise ascription
    at $\liftD{\sgtyb[2]}$, related to $\liftD{\sgtyb[1]}$. A value
    ascription is the same, with the singletons of the lifted annotations
    related by item~1 of Lemma~\ref{lemma:prec-operators}.
  \end{case}

  \begin{case}[application]
    $\sv\;\sw \gprec \svp\;\swp$. The typing derivation ends in the
    application rule, whose premises give $\Gamma |-t \sv : \gtys$,
    $\Gamma |-t \sw : \gtysp$ and $\gtysp \rel \cdom(\gtys)$, with result
    $\ccod(\gtys)$. The induction hypotheses give $\Gamma' |-t \svp : \gtys[1]$
    with $\gtys \gprec \gtys[1]$ and $\Gamma' |-t \swp : \gtysp[1]$ with
    $\gtysp \gprec \gtysp[1]$. The type $\gtys$ is well-formed
    (Lemma~\ref{lemma:welltyped-wellform-source}, whose hypotheses hold
    because $\jtf{}{\Gamma}$ and the typing rules check the annotations of
    $\sv$), so item~2 of Lemma~\ref{lemma:prec-operators} gives
    $\cdom(\gtys) \gprec \cdom(\gtys[1])$ and
    $\ccod(\gtys) \gprec \ccod(\gtys[1])$; Lemma~\ref{conp} gives
    $\gtysp[1] \rel \cdom(\gtys[1])$, and the application rule types
    $\svp\;\swp$ at $\ccod(\gtys[1])$.
  \end{case}

  \begin{case}[probabilistic choice and conditional]
    $\sm[1] \pssum \sm[2] \gprec \srcn[1] \pssum[\ptyp] \srcn[2]$ with
    $\pty \gprec \ptyp$ and $\sm[i] \gprec \srcn[i]$. The induction
    hypotheses give $\gtya[i] \gprec \gtyb[i]$ for the types of the
    branches, the choice rule applies in $\Gamma'$, and item~3 of
    Lemma~\ref{lemma:prec-operators} relates the two results. For a
    conditional, the premises of the conditional rule give
    $\Gamma |-t \sv : \gtys$, $\gtys \rel \btype$,
    $\Gamma |-d \sm[i] : \gtya[i]$ and $\gtya[1] \rel \gtya[2]$, with the hull
    $\gtya[1] \mathbin{\oplus_{\?}} \gtya[2]$ as result; the induction
    hypotheses relate the types of the guard and of the branches,
    Lemma~\ref{conp} gives $\gtysp \rel \btype$ and $\gtyb[1] \rel \gtyb[2]$,
    the rule applies, and item~3 relates the two hulls.
  \end{case}

  \begin{case}[let]
    $\lett{\sx}{\sm}{\srcn} \gprec \lett{\sx}{\sm'}{\srcn'}$ with
    $\sm \gprec \sm'$ and $\srcn \gprec \srcn'$. The typing derivation ends
    in the \src{let} rule, whose premises give
    $\Gamma |-d \sm : \gtya = \db{\phi}{\gtys[i][\p[i]]}[i \in \iSet]$ and
    $\Gamma, \sx{:}\gtys[i] |-d \srcn : \gtya[i]$ for every $i \in \iSet$,
    with result $\sugarconj{\phi}{\sum_{i \in \iSet} \p[i] \cdot \gtya[i]}$.
    The induction hypothesis on the bound term gives
    $\Gamma' |-d \sm' : \gtyb = \db{\phi[\prime]}{\gtysp[j][q_j]}[j \in \jSet]$
    with $\gtya \gprec \gtyb$. Both $\gtya$ and $\gtyb$ are well-formed
    (Lemma~\ref{lemma:welltyped-wellform-source}), so the environments
    $\Gamma, \sx{:}\gtys[i]$ and $\Gamma', \sx{:}\gtysp[j]$ are. For each
    $j \in \jSet$, the coverage clause of $\gtya \gprec \gtyb$ gives $i_j$
    with $\gtys[i_j] \gprec \gtysp[j]$, so the environments extended at
    $i_j$ and at $j$ are related, and the induction hypothesis on the body
    at $i_j$ gives $\Gamma', \sx{:}\gtysp[j] |-d \srcn' : \gtyb[j]$. The
    \src{let} rule then types the less precise term at
    $\sugarconj{\phi[\prime]}{\sum_{j \in \jSet} q_j \cdot \gtyb[j]}$, and each
    $\gtyb[j]$ is well-formed. The two results are related by item~4 of
    Lemma~\ref{lemma:prec-operators}, provided $\gtya[i] \gprec \gtyb[j]$
    for every pair with $\gtys[i] \gprec \gtysp[j]$. For such a pair, the
    induction hypothesis on the body at $i$, under the related environments
    $\Gamma, \sx{:}\gtys[i]$ and $\Gamma', \sx{:}\gtysp[j]$, gives a type of
    $\srcn'$ related to $\gtya[i]$, and that type is $\gtyb[j]$ by
    Lemma~\ref{lemma:source-det}. \qedhere
  \end{case}
  \end{revblock}
\end{theoremproof}

\subsection{Static Conservative Extension}\label{app:conservative-static}

\staticeq*

\begin{revblock}[fo1]
\proofheading{Structure of the proof.}
The theorem relates the typing of a term of \slang{} in the two systems
through realization (Definition~\ref{def:realization}). Its two directions
are proved by two separate inductions, on the typing derivation of \slang{}
for the forward direction and on the typing derivation of \glang{} for the
backward one. Both prove the open statement below, of which the theorem is
the case of empty environments; an environment $\Phi$ of \glang{}
\emph{realizes} an environment $\Gamma$ of \slang{} when both bind the same
variables and $\realizes{\Phi(x)}{\Gamma(x)}$ for each of them.
\begin{quote}
For any environment $\Gamma$ of \slang{} with $\jtf{}{\Gamma}$ and any
environment $\Phi$ of \glang{} that realizes $\Gamma$, if
$\Gamma |-ss \v : \tys$, then $\Phi |-t \v : \gtys$ and
$\realizes{\gtys}{\tys}$ for some $\gtys$; if $\Phi |-t \v : \gtys$, then
$\Gamma |-ss \v : \tys$ and $\realizes{\gtys}{\tys}$ for some $\tys$; and
likewise for terms and distribution types.
\end{quote}

Each case translates the premises of a rule of one system into those of the
corresponding rule of the other and shows that the two result types are
related by realization. Three kinds of lemmas carry the translation. The
consistency premises of \glang{} and the $=_{s}$ premises of \slang{}
correspond through Lemma~\ref{staticconsistency}. Two formula types that
realize static types are consistent exactly when the static types are
$=_{s}$. The annotations of $\lambda$ and of the ascriptions enter \glang{}
through the lifting, which realizes the annotation
(Lemma~\ref{lemma:lift-realizes}). And the type operators of the two systems
preserve realization (Lemma~\ref{lemma:realizes-operators}), which rests on
the composition of couplings (Lemma~\ref{lemma:coupling-composition}), known
as the gluing lemma~\citep{villani:2009}, and on
the closure of realization under $=_{s}$ (Lemma~\ref{lemma:realizes-eq}).

The cases with content are the conditional, whose convex hull must realize
the type of the first branch, and the \src{let}, whose exhaustive rule types
the body once per entry of the type of the bound term of its own system. The
coverage clauses of realization supply a body typing for every entry of the
other system, and the determinism of the two type systems
(Lemmas~\ref{lemma:static-det} and~\ref{lemma:source-det}) identifies the
body types that the induction hypothesis returns for different entries. On
static types the well-formedness of \slang{}
(Definition~\ref{def:well-formed-static}) and that of source annotations
(Definition~\ref{def:well-formed-source2}) coincide, since a concrete
probability has itself as its only concretization; we do not distinguish
them below. The complete proof is given at the end of this subsection,
after the lemmas it uses.
\end{revblock}

\begin{lemma}[{\rev[fo1]{Composition of couplings}}]\label{lemma:coupling-composition}\mechnote{\leanname{CouplingLemma.coupling\_glue}.}~
  \rev[fo1]{For any nonnegative vectors $(p_i)_{i \in \iSet}$,
  $(q_j)_{j \in \jSet}$ and $(r_k)_{k \in K}$, any nonnegative matrix $u$
  with row sums $p$ and column sums $q$, and any nonnegative matrix $v$ with
  row sums $q$ and column sums $r$, there is a nonnegative matrix $w$ with
  row sums $p$ and column sums $r$ such that, for any $i \in \iSet$ and
  $k \in K$, $w_{ik} > 0$ implies $u_{ij} > 0$ and $v_{jk} > 0$ for some
  $j \in \jSet$.}
\end{lemma}
\begin{proof}
  \rev[fo1]{Take $w_{ik} = \sum_{j : q_j > 0} u_{ij} \, v_{jk} / q_j$. Its
  row sums are $\sum_{j : q_j > 0} u_{ij} = p_i$, since $u_{ij} = 0$ whenever
  $q_j = 0$, and its column sums are $r_k$ in the same way; a positive
  $w_{ik}$ has a positive summand.}
\end{proof}

\begin{lemma}[{\rev[fo1]{Lifting realizes}}]\label{lemma:lift-realizes}\mechnote{\leanname{GPLC.realizesTy\_lift} and \leanname{GPLC.realizesD\_lift}.}~
  \rev[fo1]{For any well-formed simple type $\tys \in \Type$ and any
  well-formed distribution type $\tya \in \DType$,
  $\realizes{\liftT{\tys}}{\tys}$ and $\realizes{\liftD{\tya}}{\tya}$.}
\end{lemma}
\begin{proof}
  \rev[fo1]{By induction on the type. For a distribution type, the relation
  $R$ pairs each entry with itself, the only solution of the lifted formula
  gives each entry its concrete probability, which is why the probabilities
  must sum to $1$, and the coupling is the diagonal.}
\end{proof}

\begin{lemma}[{\rev[fo1]{Realization and static equality}}]\label{lemma:realizes-eq}\mechnote{\leanname{GPLC.realizesTy\_eq} and \leanname{GPLC.realizesD\_eq}.}~
  \rev[fo1]{For any formula simple type $\gtys \in \FGType$, any formula
  distribution type $\gtya \in \FGDType$, any simple types
  $\tys, \tysp \in \Type$ and any distribution types $\tya, \tyap \in \DType$,
  if $\realizes{\gtys}{\tys}$ and $\tys =_{s} \tysp$, then
  $\realizes{\gtys}{\tysp}$, and if $\realizes{\gtya}{\tya}$ and
  $\tya =_{s} \tyap$, then $\realizes{\gtya}{\tyap}$.}
\end{lemma}
\begin{proof}
  \rev[fo1]{By induction on the type. For a distribution type realized
  through a relation $R$, the new relation pairs $i$ with $k$ when
  $(i,j) \in R$ and $\tys[j] =_{s} \tysp[k]$ for some $j$; its coverage
  clauses follow from those of $R$ and of $\tya =_{s} \tyap$. The coupling
  for a solution is the composition
  (Lemma~\ref{lemma:coupling-composition}) of the coupling of the
  realization with the coupling of $\tya =_{s} \tyap$
  (Lemma~\ref{couplingeq}); a positive weight of it comes through some $j$,
  whose partner in $R$ is the partner that the support clause requires.}
\end{proof}

\begin{lemma}[{\rev[fo1]{Realization and the type operators}}]\label{lemma:realizes-operators}\mechnote{item~1, \leanname{GPLC.realizes\_point}; item~2, \leanname{GPLC.realizes\_choice}; item~3, \leanname{GPLC.realizes\_hull}; item~4, \leanname{GPLC.realizes\_letRes}.}~
  \rev[fo1]{For any formula simple type $\gtys \in \FGType$, any simple type
  $\tys \in \Type$, any formula distribution types
  $\gtya[1], \gtya[2] \in \FGDType$, any distribution types
  $\tya[1], \tya[2] \in \DType$, any probability $\ps \in [0,1]$, any formula
  distribution type $\db{\phi}{\gtys[i][\p[i]]}[i \in \iSet] \in \FGDType$,
  any distribution type $\d{\tys[j][\ps[j]]}[j \in \jSet] \in \DType$, and
  any families $(\gtya[i])_{i \in \iSet}$ of formula distribution types and
  $(\tya[j])_{j \in \jSet}$ of distribution types,
  \begin{enumerate}
  \item if $\realizes{\gtys}{\tys}$, then
    $\realizes{\dctx[(\cww=1)][\ds{\gtys[][\cww]}]}{\ds{\tys[][1]}}$;
  \item if $\realizes{\gtya[1]}{\tya[1]}$ and $\realizes{\gtya[2]}{\tya[2]}$,
    then
    $\realizes{\sugarconj{\phi[][\oplus]}{\cww[1] \cdot \gtya[1] + \cww[2] \cdot \gtya[2]}}
              {\ps \cdot \tya[1] + (1-\ps) \cdot \tya[2]}$,
    where $\phi[][\oplus]$ is the closing formula of the choice rule of
    Figure~\ref{fig:source-type-system} for the probability $\ps$;
  \item if $\realizes{\gtya[1]}{\tya[1]}$, $\realizes{\gtya[2]}{\tya[2]}$
    and $\tya[1] =_{s} \tya[2]$, then
    $\realizes{\gtya[1] \mathbin{\oplus_{\?}} \gtya[2]}{\tya[1]}$;
  \item if $\realizes{\db{\phi}{\gtys[i][\p[i]]}[i \in \iSet]}
                     {\d{\tys[j][\ps[j]]}[j \in \jSet]}$,
    $\realizes{\gtya[i]}{\tya[j]}$ for all $i \in \iSet$ and $j \in \jSet$
    with $\realizes{\gtys[i]}{\tys[j]}$, and $\tya[j] =_{s} \tya[j']$ for
    all $j, j' \in \jSet$ with $\tys[j] =_{s} \tys[j']$, then
    $\realizes{\sugarconj{\phi}{\sum_{i \in \iSet} \p[i] \cdot \gtya[i]}}
              {\sum_{j \in \jSet} \ps[j] \cdot \tya[j]}$.
  \end{enumerate}}
\end{lemma}
\begin{proof}
  \rev[fo1]{Each item builds the relation and the couplings of the
  conclusion from those of the hypotheses. For item~1 the relation pairs the
  two entries, the only solution is $\cww = 1$ and the coupling is the single
  weight $1$. For item~2 a solution of the result is
  $(\ps \cdot \nu_1, (1-\ps) \cdot \nu_2)$ for solutions $\nu_1, \nu_2$ of
  the operands; the relation is the union of the two relations, block by
  block, and the coupling is $\ps$ times the first coupling on the first
  block, $(1-\ps)$ times the second on the second, and zero elsewhere. For
  item~3 a solution of the hull is $(t \cdot \nu_1, (1-t) \cdot \nu_2)$ with
  $t \in [0,1]$; the relation pairs the entries of $\gtya[1]$ as before and
  an entry $i$ of $\gtya[2]$ with every entry $j$ of $\tya[1]$ such that
  $(i,k) \in R_2$ and $\tys[k] =_{s} \tys[j]$ for some entry $k$ of
  $\tya[2]$; the coupling is $t$ times the coupling of $\gtya[1]$ on the
  first block and, on the second, $(1-t)$ times the composition
  (Lemma~\ref{lemma:coupling-composition}) of the coupling of $\gtya[2]$
  with the coupling of $\tya[2] =_{s} \tya[1]$ (Lemma~\ref{couplingeq}),
  so that its column sums are $t \cdot \ps[j] + (1-t) \cdot \ps[j] = \ps[j]$.
  For item~4 a solution of the weighted union is $(\nu(\p[i]) \cdot b^{i})_{i}$
  for a solution $\nu$ of $\phi$ and solutions $b^{i}$ of the $\gtya[i]$; the
  relation pairs an entry of $\gtya[i]$ with an entry of $\tya[j]$ when the
  first realizes the second, with the coverage clauses inherited from those
  of the hypotheses. Let $w$ be the coupling of the realization of the bound term at
  $\nu$. A positive weight $w_{ij}$ comes with a partner $j'$ of $i$, so
  $\realizes{\gtya[i]}{\tya[j']}$ gives a coupling of $b^{i}$ onto the probabilities
  of $\tya[j']$, which composes with the coupling of $\tya[j'] =_{s} \tya[j]$
  into a coupling $v^{ij}$ of $b^{i}$ onto the probabilities of $\tya[j]$; the
  coupling of the weighted union places $w_{ij} \cdot v^{ij}$ in block
  $(i,j)$, and its row sums are $\sum_j w_{ij} \, b^{i} = \nu(\p[i]) \, b^{i}$
  while its column sums are $\sum_i w_{ij} = \ps[j]$ times the probabilities of
  $\tya[j]$.}
\end{proof}

\begin{lemma}[Static equality and consistency]\label{staticconsistency}\mechnote{\leanname{GPLC.eq\_of\_cons\_ty}/\leanname{\_d} and \leanname{GPLC.cons\_of\_eq\_ty}/\leanname{\_d}; the particular case, with \leanname{GPLC.realizesTy\_lift}/\leanname{GPLC.realizesD\_lift}.}
  \rev[fo1]{For any simple types $\tys[1], \tys[2] \in \Type$, any
  distribution types $\tya[1], \tya[2] \in \DType$, any formula simple types
  $\gtys[1], \gtys[2] \in \FGType$ and any formula distribution types
  $\gtya[1], \gtya[2] \in \FGDType$,
  \begin{enumerate}
  \item if $\realizes{\gtys[1]}{\tys[1]}$ and $\realizes{\gtys[2]}{\tys[2]}$,
    then $\tys[1] =_{s} \tys[2]$ iff $\gtys[1] \rel \gtys[2]$;
  \item if $\realizes{\gtya[1]}{\tya[1]}$ and $\realizes{\gtya[2]}{\tya[2]}$,
    then $\tya[1] =_{s} \tya[2]$ iff $\gtya[1] \rel \gtya[2]$;
  \item if $\tys[1]$, $\tys[2]$, $\tya[1]$ and $\tya[2]$ are well-formed,
    then $\tys[1] =_{s} \tys[2]$ iff $\liftT{\tys[1]} \rel \liftT{\tys[2]}$,
    and $\tya[1] =_{s} \tya[2]$ iff $\liftD{\tya[1]} \rel \liftD{\tya[2]}$.
  \end{enumerate}}
\end{lemma}
\begin{proof}
  \rev[fo1]{By induction on the types, the two items together; the
  particular case follows by Lemma~\ref{lemma:lift-realizes}. No formula
  type that realizes a static type contains $\?$, so on simple types
  consistency reduces to the equality of base types and to the arrow case,
  where both relations decompose componentwise. For distribution types,
  from left to right, the coupling of $\gtya[1] \rel \gtya[2]$ at a pair of
  solutions is the composition (Lemma~\ref{lemma:coupling-composition}) of
  the coupling of the first realization, the coupling of
  $\tya[1] =_{s} \tya[2]$ and the inverse of the coupling of the second
  realization; its positive weights come through $=_{s}$ static entries,
  whose formula entries are consistent by the induction hypothesis, and its
  coverage clauses cross through those of the realizations and of $=_{s}$.
  From right to left the composition runs from the probabilities of $\tya[1]$
  through the two realizations and the consistency coupling to the
  probabilities of $\tya[2]$, and it has positive weight only between static entries realized by
  consistent formula types, hence $=_{s}$ by the induction hypothesis.}
\end{proof}

\begin{lemma}[Equality defined]\label{eq-defined}\mechnote{\leanname{GPLC.cons\_of\_eq\_d} and \leanname{TPLC.meetD\_sat\_iff\_econsD}.}
  \rev[fo1]{For any well-formed distribution types
  $\tya[1], \tya[2] \in \DType$, if $\tya[1] =_{s} \tya[2]$, then
  $\liftD{\tya[1]} \meet \liftD{\tya[2]}$ is defined.}
\end{lemma}
\begin{proof}
  \rev[mt1]{By Lemma~\ref{staticconsistency}, the lifts are consistent; by
  Lemma~\ref{consistency-defined} (definedness \emph{is} consistency),
  their meet is defined.}
\end{proof}

\begin{theoremproof}{Theorem~\ref{theorem:staticeq} (Conservative extension of the static semantics)}
  \begin{revblock}[fo1]
  It suffices to prove the open statement of the structure paragraph above.

  \paragraph*{Forward direction.}
  By induction on the typing derivation of \slang{}, the statements for
  values and for terms together. Let $\Phi$ realize $\Gamma$, with
  $\jtf{}{\Gamma}$.

  \begin{case}[constants, variables, values and addition]
    A constant has the same type in both systems, and that type realizes
    itself. For a variable, $\Gamma(x) = \tys$, so $\Phi(x) = \gtys$ with
    $\realizes{\gtys}{\tys}$. For a value $v$ used as a term, rule (Tv)
    gives $\Gamma |-ss v : \tys$ and $\tya = \ds{\tys[][1]}$; the induction
    hypothesis gives $\Phi |-t v : \gtys$ with $\realizes{\gtys}{\tys}$, and
    the corresponding rule of \glang{} types $v$ at the singleton of
    $\gtys$, which realizes $\ds{\tys[][1]}$ by item~1 of
    Lemma~\ref{lemma:realizes-operators}. For an addition, rule (T$+$) gives
    $\Gamma |-ss v : \tys[1]$ with $\tys[1] =_{s} \rtype$, and likewise for
    $w$; the induction hypothesis gives realizing types $\gtys[1]$ and
    $\gtys[2]$, Lemma~\ref{staticconsistency} with
    $\realizes{\rtype}{\rtype}$ gives $\gtys[1] \rel \rtype$ and
    $\gtys[2] \rel \rtype$, and the addition rule of \glang{} types the sum
    at the singleton of $\rtype$, which realizes $\ds{\rtype^1}$.
  \end{case}

  \begin{case}[abstraction]
    Rule (T$\lambda$) gives $\Gamma, x{:}\tys |-ss \m : \tya$ and
    $\justify{\tys}$. By Lemma~\ref{lemma:lift-realizes},
    $\realizes{\liftT{\tys}}{\tys}$, so $\Phi, x{:}\liftT{\tys}$ realizes
    the well-formed environment $\Gamma, x{:}\tys$; the induction
    hypothesis gives $\Phi, x{:}\liftT{\tys} |-d \m : \gtya$ with
    $\realizes{\gtya}{\tya}$, and the $\lambda$ rule of \glang{}, whose
    well-formedness premise is that of (T$\lambda$), gives
    $\Phi |-t \lambda x{:}\tys. \m : \liftT{\tys} -> \gtya$, which realizes
    $\tys -> \tya$.
  \end{case}

  \begin{case}[ascriptions]
    Rule (T$::\tya$) gives $\Gamma |-ss \m : \tyap$, $\tyap =_{s} \tya$ and
    $\justify{\tya}$. The induction hypothesis gives $\Phi |-d \m : \gtyap$
    with $\realizes{\gtyap}{\tyap}$; by Lemma~\ref{lemma:lift-realizes},
    $\realizes{\liftD{\tya}}{\tya}$, so Lemma~\ref{staticconsistency} gives
    $\gtyap \rel \liftD{\tya}$, and the ascription rule of \glang{} types
    $\asc{\m}{\tya}$ at $\liftD{\tya}$, which realizes $\tya$. A value
    ascription is handled in the same way, its result being the singleton of
    $\liftT{\tys}$, which realizes $\ds{\tys[][1]}$ by item~1 of
    Lemma~\ref{lemma:realizes-operators}.
  \end{case}

  \begin{case}[application]
    Rule (Tapp) gives $\Gamma |-ss v : \tys[1]$, $\Gamma |-ss w : \tys[2]$
    and $\dom(\tys[1]) =_{s} \tys[2]$, with $\tya = \cod(\tys[1])$, so
    $\tys[1] = \tys -> \tya$. The induction hypothesis gives
    $\Phi |-t v : \gtys[v]$ with $\realizes{\gtys[v]}{\tys -> \tya}$; by the
    rules of realization, $\gtys[v] = \gtys -> \gtya$ with
    $\realizes{\gtys}{\tys}$ and $\realizes{\gtya}{\tya}$. It also gives
    $\Phi |-t w : \gtys[w]$ with $\realizes{\gtys[w]}{\tys[2]}$, so
    Lemma~\ref{staticconsistency} gives $\gtys[w] \rel \gtys = \cdom(\gtys[v])$,
    and the application rule of \glang{} types $v\;w$ at
    $\ccod(\gtys[v]) = \gtya$.
  \end{case}

  \begin{case}[probabilistic choice]
    Rule (T$\oplus$) gives $\Gamma |-ss \m[i] : \tya[i]$ for $i = 1, 2$ and
    $\tya = \ps \cdot \tya[1] + (1-\ps) \cdot \tya[2]$. The induction
    hypothesis gives $\gtya[1]$ and $\gtya[2]$ realizing $\tya[1]$ and
    $\tya[2]$; the choice rule of \glang{} with the concrete probability
    $\ps$ types the choice at the type of item~2 of
    Lemma~\ref{lemma:realizes-operators}, which realizes $\tya$.
  \end{case}

  \begin{case}[conditional]
    Rule (Tif) gives $\Gamma |-ss v : \tys$, $\tys =_{s} \btype$,
    $\Gamma |-ss \m[i] : \tya[i]$ for $i = 1, 2$, $\tya[1] =_{s} \tya[2]$ and
    $\tya = \tya[1]$. The induction hypothesis gives $\gtys$, $\gtya[1]$
    and $\gtya[2]$ realizing $\tys$, $\tya[1]$ and $\tya[2]$;
    Lemma~\ref{staticconsistency} gives $\gtys \rel \btype$ and
    $\gtya[1] \rel \gtya[2]$, so the conditional rule of \glang{} types the
    conditional at $\gtya[1] \mathbin{\oplus_{\?}} \gtya[2]$, which realizes
    $\tya[1]$ by item~3 of Lemma~\ref{lemma:realizes-operators}.
  \end{case}

  \begin{case}[let]
    Rule (Tlet) gives $\Gamma |-ss \m : \d{\tys[j][\ps[j]]}[j \in \jSet]$,
    $\Gamma, x{:}\tys[j] |-ss \n : \tya[j]$ for every $j \in \jSet$, and
    $\tya = \sum_{j \in \jSet} \ps[j] \cdot \tya[j]$. The induction
    hypothesis on the bound term gives $\Phi |-d \m : \gtya[0]$ with
    $\gtya[0] = \db{\phi}{\gtys[i][\p[i]]}[i \in \iSet]$ realizing the
    type of the bound term through a relation $R$. Each $\tys[j]$ is well-formed by
    Lemma~\ref{lemma:welltyped-wellform-static}, so each
    $\Gamma, x{:}\tys[j]$ is. For each $i \in \iSet$, coverage gives $j_i$
    with $(i, j_i) \in R$, so $\Phi, x{:}\gtys[i]$ realizes
    $\Gamma, x{:}\tys[j_i]$, and the induction hypothesis on the body at
    $j_i$ gives $\Phi, x{:}\gtys[i] |-d \n : \gtya[i]$ with
    $\realizes{\gtya[i]}{\tya[j_i]}$; the \src{let} rule of \glang{} then
    types the term at $\sugarconj{\phi}{\sum_{i \in \iSet} \p[i] \cdot \gtya[i]}$.
    This type realizes $\tya$ by item~4 of
    Lemma~\ref{lemma:realizes-operators}, whose two hypotheses hold as
    follows. For a pair $(i,j)$ with $\realizes{\gtys[i]}{\tys[j]}$, the
    induction hypothesis on the body at $j$, under $\Phi, x{:}\gtys[i]$,
    gives a type of $\n$ that realizes $\tya[j]$, and this type is
    $\gtya[i]$ by Lemma~\ref{lemma:source-det}. And if $\tys[j] =_{s} \tys[j']$,
    then $\tya[j] =_{s} \tya[j']$ by Lemma~\ref{lemma:static-det}, both
    being types of $\n$ under the well-formed environments
    $\Gamma, x{:}\tys[j]$ and $\Gamma, x{:}\tys[j']$.
  \end{case}

  \paragraph*{Backward direction.}
  By induction on the typing derivation of \glang{}, the statements for
  values and for terms together. Since $\m$ is a term of \slang{}, its
  annotations are static types and its probabilities are concrete, so the
  choice rule with an unknown probability does not occur; and no formula
  type that realizes a static type contains $\?$.

  \begin{case}[constants, variables, values and addition]
    As in the forward direction, in the other sense. A constant realizes
    itself; a variable reads $\Phi(x) = \gtys$ and $\Gamma(x) = \tys$ with
    $\realizes{\gtys}{\tys}$; for a value used as a term the induction
    hypothesis gives $\tys$, rule (Tv) gives $\ds{\tys[][1]}$, realized by
    the singleton (item~1 of Lemma~\ref{lemma:realizes-operators}); for an
    addition the induction hypothesis gives $\tys[1]$ and $\tys[2]$
    realized by the types of the operands, which are consistent with
    $\rtype$, so $\tys[1] =_{s} \rtype$ and $\tys[2] =_{s} \rtype$ by
    Lemma~\ref{staticconsistency}, rule (T$+$) applies, and the singletons
    of $\rtype$ are related by item~1.
  \end{case}

  \begin{case}[abstraction]
    The derivation ends in the $\lambda$ rule, whose premises give
    $\Phi, x{:}\liftT{\tys} |-d \m : \gtya$ and $\justify{\tys}$. By
    Lemma~\ref{lemma:lift-realizes}, $\Phi, x{:}\liftT{\tys}$ realizes the
    well-formed environment $\Gamma, x{:}\tys$; the induction hypothesis
    gives $\Gamma, x{:}\tys |-ss \m : \tya$ with $\realizes{\gtya}{\tya}$,
    and rule (T$\lambda$) gives $\tys -> \tya$, realized by
    $\liftT{\tys} -> \gtya$.
  \end{case}

  \begin{case}[ascriptions]
    The derivation ends in the ascription rule, whose premises give
    $\Phi |-d \m : \gtyap$, $\gtyap \rel \liftD{\tya}$ and $\justify{\tya}$.
    The induction hypothesis gives $\Gamma |-ss \m : \tyap$ with
    $\realizes{\gtyap}{\tyap}$; since $\realizes{\liftD{\tya}}{\tya}$
    (Lemma~\ref{lemma:lift-realizes}), Lemma~\ref{staticconsistency} gives
    $\tyap =_{s} \tya$, and rule (T$::\tya$) types $\asc{\m}{\tya}$ at
    $\tya$, realized by $\liftD{\tya}$. A value ascription is handled in the
    same way, with the singletons related by item~1 of
    Lemma~\ref{lemma:realizes-operators}.
  \end{case}

  \begin{case}[application]
    The derivation ends in the application rule, whose premises give
    $\Phi |-t v : \gtys$, $\Phi |-t w : \gtysp$ and
    $\gtysp \rel \cdom(\gtys)$, with result $\ccod(\gtys)$. The induction
    hypothesis gives $\Gamma |-ss v : \tys[v]$ with
    $\realizes{\gtys}{\tys[v]}$; as $\gtys \neq \?$ and $\cdom(\gtys)$ is
    defined, $\gtys = \gtys[s] -> \gtya[d]$ and therefore
    $\tys[v] = \tys[s] -> \tya[d]$ with $\realizes{\gtys[s]}{\tys[s]}$ and
    $\realizes{\gtya[d]}{\tya[d]}$. It also gives $\Gamma |-ss w : \tys[w]$
    with $\realizes{\gtysp}{\tys[w]}$, so Lemma~\ref{staticconsistency}
    gives $\tys[w] =_{s} \tys[s] = \dom(\tys[v])$, and rule (Tapp) types
    $v\;w$ at $\cod(\tys[v]) = \tya[d]$, realized by
    $\gtya[d] = \ccod(\gtys)$.
  \end{case}

  \begin{case}[probabilistic choice and conditional]
    For a choice, the derivation ends in the choice rule with the concrete
    probability $\ps$; the induction hypothesis gives $\tya[1]$ and
    $\tya[2]$ realized by the types of the branches, rule (T$\oplus$)
    applies, and item~2 of Lemma~\ref{lemma:realizes-operators} relates the
    results. For a conditional, the premises of the conditional rule give
    $\Phi |-t v : \gtys$, $\gtys \rel \btype$, $\Phi |-d \m[i] : \gtya[i]$
    and $\gtya[1] \rel \gtya[2]$, with the hull
    $\gtya[1] \mathbin{\oplus_{\?}} \gtya[2]$ as result. The induction
    hypothesis gives $\tys$, $\tya[1]$ and $\tya[2]$ realized by $\gtys$,
    $\gtya[1]$ and $\gtya[2]$; Lemma~\ref{staticconsistency} gives
    $\tys =_{s} \btype$ and $\tya[1] =_{s} \tya[2]$, so rule (Tif) types the
    conditional at $\tya[1]$, which the hull realizes by item~3.
  \end{case}

  \begin{case}[let]
    The derivation ends in the \src{let} rule, whose premises give
    $\Phi |-d \m : \gtya[0]$ with
    $\gtya[0] = \db{\phi}{\gtys[i][\p[i]]}[i \in \iSet]$ and
    $\Phi, x{:}\gtys[i] |-d \n : \gtya[i]$ for every $i \in \iSet$, with
    result $\sugarconj{\phi}{\sum_{i \in \iSet} \p[i] \cdot \gtya[i]}$. The
    induction hypothesis on the bound term gives
    $\Gamma |-ss \m : \d{\tys[j][\ps[j]]}[j \in \jSet]$, realized by
    $\gtya[0]$ through a relation $R$; each $\tys[j]$ is well-formed by
    Lemma~\ref{lemma:welltyped-wellform-static}. For each $j \in \jSet$,
    coverage gives $i_j$ with $(i_j, j) \in R$, so $\Phi, x{:}\gtys[i_j]$
    realizes $\Gamma, x{:}\tys[j]$, and the induction hypothesis on the
    body at $i_j$ gives $\Gamma, x{:}\tys[j] |-ss \n : \tya[j]$ with
    $\realizes{\gtya[i_j]}{\tya[j]}$; rule (Tlet) then types the term at
    $\sum_{j \in \jSet} \ps[j] \cdot \tya[j]$. The result of \glang{}
    realizes it by item~4 of Lemma~\ref{lemma:realizes-operators}. For a
    pair $(i,j)$ with $\realizes{\gtys[i]}{\tys[j]}$, the induction
    hypothesis on the body at $i$, under $\Phi, x{:}\gtys[i]$ realizing
    $\Gamma, x{:}\tys[j]$, gives $\Gamma, x{:}\tys[j] |-ss \n : \tyap$ with
    $\realizes{\gtya[i]}{\tyap}$; by Lemma~\ref{lemma:static-det},
    $\tyap =_{s} \tya[j]$, and $\realizes{\gtya[i]}{\tya[j]}$ follows by
    Lemma~\ref{lemma:realizes-eq}. And if $\tys[j] =_{s} \tys[j']$, then
    $\tya[j] =_{s} \tya[j']$ by Lemma~\ref{lemma:static-det}. \qedhere
  \end{case}
  \end{revblock}
\end{theoremproof}

\section{The Target Language \tlang }
\label{app:target}
\rev[fo1]{This section collects the proofs for the results of
Section~\ref{sec:target} on \tlang{} and the elaboration: well-typedness
implies well-formedness, the properties of the meet operator, elaboration
preserves types (Lemma~\ref{theorem:ep}), type safety
(Theorem~\ref{lm:typesafety}), the dynamic gradual guarantee for \tlang{}
(Theorem~\ref{theorem:dgg}), the monotonicity of the canonical elaboration,
from which the dynamic gradual guarantee for \glang{}
(Theorem~\ref{theorem:gg}) follows, and the dynamic conservative extension
of \tlang{} \wrt \slang{} (Theorem~\ref{tm:dynamic-cons-ext}) via the erasure
of Figure~\ref{fig:erasure}. The complete definitions are given in the
article body: the type system in Figure~\ref{fig:target-type-system}, the
distribution semantics in Figures~\ref{fig:target-reduction-dis}
and~\ref{fig:target-reduction-dis-2}, the elaboration in
Figure~\ref{fig:elaboration}, and term precision in
Figure~\ref{fig:term-precision}.
In Appendix~\ref{app:conservative}, we use blue color $(\sm)$ for static
terms, because static source terms coincide with static terms.}

\subsection{Type System}

\rev[mt1]{The type system of \tlang is presented in Figure~\ref{fig:target-type-system} in the article body.}

\rev[fo1]{The next lemma is Lemma~\ref{lemma:welltyped-wellform-target2} of
the article body, for environments that are well-formed in the sense of
Definition~\ref{def:well-formed-source}.}

\begin{lemma}[{\rev[fo1]{Well-formed types}}]\label{lemma:welltyped-wellform-target}\mechnote{\leanname{TPLC.wf\_raw}, \leanname{TPLC.wf\_val} and \leanname{TPLC.wf\_tm}.}~
  \rev[fo1]{For any well-formed environment $\Gamma$, any value $\tv$ and any
  term $\tm$ of \tlang{}, any formula simple type $\gtys \in \FGType$ and
  any formula distribution type $\gtya \in \FGDType$,
  \begin{enumerate}
  \item if $\Gamma |-t \tv : \gtys$, then $\justify{\gtys}$;
  \item if $\Gamma |-d \tm : \gtya$, then $\justify{\gtya}$.
  \end{enumerate}}
\end{lemma}
\begin{proof}
  \rev[fo1]{By induction on the typing derivation, the two items together
  with the corresponding statement for raw values. A variable reads the
  well-formed environment; (G$\lambda$), the three ascription rules,
  (Gerr$_{\gtys}$), (Gerr$_{\gtya}$) and (GV) carry the well-formedness of
  their result, or of the annotation it is built from, as a premise, the
  function type of (G$\lambda$) being well-formed because its body is typed
  under the environment extended with the well-formed annotation; (Gapp)
  reads the codomain of a well-formed arrow; (Gv), (G$+$) and
  (G$\mathord{::}\gtys$) return a Dirac type. For (G$\oplus$), (Gif) and
  (Glet), a solution of the result type is a convex combination of
  solutions of the operand types, with weights given by a solution of the
  lifted probability, by the hull parameter $t \in [0,1]$ or by a solution of
  the formula of the bound term; the operand types are well-formed by the induction
  hypothesis, for the bodies of (Glet) under the environments extended with
  the well-formed entries of the type of the bound term, so the result has a
  solution, every solution has total probability $1$, the set of solutions
  is convex, and the entries are well-formed.}
\end{proof}

\begin{lemma}[{\rev[fo1]{Weakening and single-entry lets}}]\label{lemma:target-weakening}\mechnote{\leanname{TPLC.hasTy\_wk0\_val}, \leanname{TPLC.hasTy\_wk\_tm} and \leanname{TPLC.hasTy\_letin1}.}~
  \rev[fo1]{For any environment $\Gamma$, any variable $\tx \in \text{Var}$
  that does not occur in $\Gamma$, any value $\tv$ and any terms $\tm, \tn$
  of \tlang{}, any formula simple types $\gtys, \gtysp \in \FGType$ and any
  formula distribution type $\gtya \in \FGDType$,
  \begin{enumerate}
  \item if $\Gamma |-t \tv : \gtysp$, then $\Gamma, \tx{:}\gtys |-t \tv : \gtysp$,
    and if $\Gamma |-d \tm : \gtya$, then $\Gamma, \tx{:}\gtys |-d \tm : \gtya$;
  \item if $\Gamma |-d \tm : \ds{\gtys[][1]}$ and
    $\Gamma, \tx{:}\gtys |-d \tn : \gtya$, then
    $\Gamma |-d \lett{\tx}{\tm}{\tn} : \gtya$.
  \end{enumerate}}
\end{lemma}
\begin{proof}
  \rev[fo1]{Item~1 is by induction on the typing derivation; no rule reads
  the environment except at a variable. Item~2 is rule (Glet) with a single
  entry. The weighted union $\sugarconj{(\cww = 1)}{\cww \cdot \gtya}$ has
  the entries and the solutions of $\gtya$.}
\end{proof}

\subsection{The Meet Operator}\label{app:meet}

\rev[mt1]{Recall that the meet $\meet$ over formula types is \emph{partial},
and that under the satisfiability reading its definedness is exactly the
satisfiability of its witness constraint. The following \rev[fo1]{three}
lemmas are the central lemmas of the evidence calculus: reductivity (with
the greatest-lower-bound property), the identity between definedness and
consistency\rev[fo1]{, and the validity of the meet as an evidence}.}

\meetreductivity*
\begin{proof}
  \rev[fo1]{By induction on the types, the statements for simple and for
  distribution types together. For items~1 and~2, on simple types the meet
  of two equal base types is that type, $\?$ is absorbed, since
  $\gtys \meet \? = \gtys \eprec \?$, and the arrow case is componentwise. On
  distribution types, $\gtya[1] \meet \gtya[2]$ has one entry for each pair
  $(i,j)$ of entries of the operands whose meet is defined, and the entry
  $(i,j)$ carries that meet; the formula of the meet asserts the two marginal
  clauses, which tie the weights $\cww[ij]$ of these entries to the
  probabilities of $\gtya[1]$ along the rows and to those of $\gtya[2]$ along
  the columns. A solution of that formula is therefore a coupling of a
  solution of $\gtya[1]$ with a solution of $\gtya[2]$, and every entry of
  positive weight carries the meet of two entries of the operands, which is
  $\eprec$ each of them by the induction hypothesis; read along the rows, the
  coupling witnesses
  $\gtya[1] \meet \gtya[2] \eprec \gtya[1]$, and read along the columns,
  $\gtya[1] \meet \gtya[2] \eprec \gtya[2]$.

  For item~3, on simple types the induction is on the common lower bound
  $\gtys[3]$. If either operand is $\?$, the meet is the other operand and
  the remaining hypothesis concludes; on equal base types the meet is that
  type; the arrow case is componentwise. On distribution types, let $\rho$
  be a solution of $\gtya[3]$, and let $u$ and $v$ be the couplings that
  $\gtya[3] \eprec \gtya[1]$ and $\gtya[3] \eprec \gtya[2]$ give at $\rho$,
  onto solutions $p$ of $\gtya[1]$ and $q$ of $\gtya[2]$. The assignment
  $W(\kk[1],\kk[2]) = \sum_{\kk} u(\kk,\kk[1]) \cdot v(\kk,\kk[2]) /
  \rho(\kk)$, where an entry $\kk$ with $\rho(\kk) = 0$ contributes $0$, has
  marginals $p$ and $q$, and each positive weight $W(\kk[1], \kk[2])$ comes
  from some entry $\kk$ of $\gtya[3]$ with $u(\kk,\kk[1]) > 0$ and
  $v(\kk,\kk[2]) > 0$, whose type $\gtys[\kk]$ is $\eprec$ both
  $\gtys[\kk[1]]$ and $\gtys[\kk[2]]$. By the simple-type statement the
  meet $\gtys[\kk[1]] \meet \gtys[\kk[2]]$ is defined, so
  $\gtya[1] \meet \gtya[2]$ has an entry $(\kk[1], \kk[2])$, and $W$ is a
  solution of its formula, so the meet is defined. Routing the probability
  that $\rho$ gives to $\kk$ to the entries $(\kk[1], \kk[2])$ through $u$
  and $v$ gives a coupling of $\rho$ with $W$ with
  $\gtys[\kk] \eprec \gtys[\kk[1]] \meet \gtys[\kk[2]]$ whenever its weight
  at $(\kk, (\kk[1], \kk[2]))$ is positive, by the induction hypothesis,
  which witnesses $\gtya[3] \eprec \gtya[1] \meet \gtya[2]$.}
\end{proof}

\rev[fo1]{The lemma is stated with the runtime precision $\eprec$ of
Figure~\ref{fig:runtime-relations}, the relation that compares evidences.
If instead it were stated with the precision $\gprec$ on annotations, which
carries the two coverage clauses, then both halves would fail. For reductivity, the meet of
$\ds{\rtype^{1}, \btype^{0}}$ with $\ds{\rtype^{1}}$ has no entry for the
inconsistent pair $(\btype, \rtype)$, leaving the entry $\btype$ of the first
operand with nothing below it.
For the greatest lower bound, write $\gtys[a] = \rtype -> \ds{\?^{1}}$ and
$\gtys[b] = \? -> \ds{\rtype^{1}}$. The two are consistent, and their meet
refines both positions at once,
$\gtys[a] \meet \gtys[b] = \rtype -> \ds{\rtype^{1}}$, so neither of them is
$\gprec$-below it ($\gtys[a]$ fails in the codomain, $\gtys[b]$ in the
domain). Now let $\evd[0] = \ds{\gtys[a][\nicefrac{1}{2}],
\gtys[b][\nicefrac{1}{2}]}$, let $\gtya[a] = \evd[0]$ and let $\gtya[b]$ be
$\evd[0]$ with its two entries exchanged. Then $\evd[0] \gprec \gtya[a]$ by
reflexivity and $\evd[0] \gprec \gtya[b]$ by the crossing coupling, both with
coverage; but the entry $(\gtys[a], \gtys[b])$ of $\gtya[a] \meet \gtya[b]$
is live and no entry of $\evd[0]$ is $\gprec$-below it, so the coverage
clause of $\evd[0] \gprec \gtya[a] \meet \gtya[b]$ cannot be met. The
greatest-lower-bound half also depends on the satisfiability reading of
definedness. A common lower bound that is a coupling need not cover the
entries of the meet, so under a reading that requires coverage it would not
make the meet defined.}

\begin{lemma}[{Consistency \rev[fo1]{is definedness}}]\label{consistency-defined}\mechnote{\leanname{TPLC.meetD\_sat\_iff\_econsD}, \leanname{TPLC.cons\_meetTy\_isSome} and \leanname{TPLC.cons\_of\_meetTy\_good}.}
  \rev[fo1]{For any well-formed formula simple types
  $\gtys[1], \gtys[2] \in \FGType$ and any formula distribution types
  $\gtya[1], \gtya[2] \in \FGDType$, $\gtys[1] \erel \gtys[2]$ \emph{iff}
  $\gtys[1] \meet \gtys[2]$ is defined, and $\gtya[1] \erel \gtya[2]$
  \emph{iff} $\gtya[1] \meet \gtya[2]$ is defined.}
\end{lemma}
\begin{proof}
  \rev[fo1]{The relation is the runtime consistency $\erel$ of
  Figure~\ref{fig:runtime-relations}, without coverage clauses. On simple
  types, the proof is by induction on the types. Base types and $\?$ are immediate, and
  the arrow case is componentwise. On distribution types, the formula of
  $\gtya[1] \meet \gtya[2]$ asserts the two marginal clauses over the pairs
  of entries that have a defined meet, so a solution of it is a coupling of a
  solution of $\gtya[1]$ with a solution of $\gtya[2]$ that has positive
  weight only between consistent entries, by the simple-type statement, and
  such a coupling is a witness of $\gtya[1] \erel \gtya[2]$; conversely, a
  witness of the consistency gives weight only to consistent pairs, for which
  the meet has an entry, and satisfies the two marginal clauses. A pair of entries that no
  solution can give positive probability, such as $\btype^{0}$ in
  $\ds{\rtype^{1}, \btype^{0}}$ against $\ds{\rtype^{1}}$, has no entry in
  the meet and does not prevent its definedness; the marginal clause gives
  its entry probability $0$.}
\end{proof}

\begin{lemma}[{\rev[fo1]{Validity of the meet}}]\label{lemma:meet-validity}\mechnote{\leanname{TPLC.hvtag\_tagMeetTy}, \leanname{TPLC.goodTy\_meetTy}, \leanname{TPLC.hvalid\_tagMeetD} and \leanname{TPLC.goodD\_meetD}.}~
  \rev[fo1]{For any well-formed formula simple types
  $\gtys[1], \gtys[2] \in \FGType$ and any well-formed formula distribution
  types $\gtya[1], \gtya[2] \in \FGDType$,
  \begin{enumerate}
  \item if $\gtys[1] \erel \gtys[2]$, then the meet $\gtys[1] \meet \gtys[2]$
    is defined, well-formed, and valid for the judgment it justifies,
    $\vjtf{\gtys[1] \meet \gtys[2]}{\gtys[1] \erel \gtys[2]}$;
  \item if $\gtya[1] \erel \gtya[2]$, then the meet $\gtya[1] \meet \gtya[2]$
    is defined, well-formed and
    $\vjtf{\gtya[1] \meet \gtya[2]}{\gtya[1] \erel \gtya[2]}$.
  \end{enumerate}}
\end{lemma}
\begin{proof}
  \rev[fo1]{Definedness is Lemma~\ref{consistency-defined}. Validity is by
  induction on the consistency derivation, the two levels together. On
  equal base types the meet is that type. When one operand is $\?$ the meet
  is the other operand, read with identity pointers; it is valid for $\?$ by
  the rule for $\?$, and valid for itself by induction on the type, each
  entry of a distribution type pointing to itself. The arrow case is
  componentwise. At the distribution level, each entry of the meet is the
  meet of two consistent entries, which it is $\eprec$ by item~1 of
  Lemma~\ref{meet-more-precise2} and valid for by the simple-type statement,
  and its two pointers name those entries; the push-forward clause of
  validity towards each operand is the marginal clause of the formula of
  the meet on that side. For well-formedness, the entries of the meet are
  well-formed by the simple-type statement, its formula is satisfiable by
  definedness, and since the marginal clauses are linear and transfer the
  probabilities of the operands, its solutions are closed under convex combinations
  and have total probability $1$.}
\end{proof}

\subsection{Elaboration}\label{app:elaboration}

\elaprvty*
\begin{proof}
  \begin{revblock}[fo1]
  Each item has two halves, proved separately: every well-typed term has an
  elaboration, and every elaboration of a term under a well-formed
  environment is well typed at the elaboration type. The second half holds
  for every elaboration, not only for the one that the first half
  constructs, and the proofs of Theorems~\ref{lm:typesafety}
  and~\ref{theorem:gg} use it in this form. Both halves are by induction,
  the statements for values and for terms together; each rule of
  Figure~\ref{fig:elaboration} mirrors the typing rule of
  Figure~\ref{fig:source-type-system} for the same constructor premise by
  premise, and adds the equations that fix its evidences.

  \emph{Existence}, by induction on the typing derivation of the source
  term. In every case the induction hypotheses supply the elaborations of
  the subterms, namely under $\Gamma, \sx{:}\liftT{\sgtys}$ for the body of a
  $\lambda$, an environment well-formed by
  Lemma~\ref{lemma:wellformed-lifting}, and under $\Gamma, \sx{:}\gtys[i]$
  for the body of a \src{let} at each entry $i$ of the type of the bound term, an
  environment well-formed because that type is
  (Lemma~\ref{lemma:welltyped-wellform-source}). What remains is the
  definedness of the meets that the rule computes, each of a pair of
  well-formed types (Lemmas~\ref{lemma:welltyped-wellform-source}
  and~\ref{lemma:wellformed-lifting}). For the ascriptions, the addition,
  the conditional and the first coercion of $\src{(Eapp)}$, the pair is
  consistent by a premise of the typing rule, consistency implies the
  runtime consistency $\erel$, and Lemma~\ref{consistency-defined} applies.
  The second coercion of $\src{(Eapp)}$ takes the meet of the type $\gtys$
  of the function with $\cdom(\gtys) -> \ccod(\gtys)$, and the two are
  consistent for every well-formed $\gtys$. If $\gtys$ is an arrow, then they are
  equal, and $\?$ is consistent with every type.

  \emph{Preservation}, by induction on the elaboration derivation. The
  well-formedness premises of the rules of \tlang{} hold by
  $\jtf{}{\Gamma}$, by Lemma~\ref{lemma:welltyped-wellform-target} for the
  types of the elaborated subterms and by
  Lemma~\ref{lemma:wellformed-lifting} for the lifted annotations. A
  variable, a value used as a term, a choice and a \src{let} are typed by
  the rule of \tlang{} for the same constructor, which assigns the same type
  as the source rule, from the induction hypotheses; for $\src{(Elet)}$ the
  family has one body per entry of the type of the bound term, each typed under the
  corresponding extension by its induction hypothesis. A literal elaborates
  to its ascription with the evidence $\rtype$ or $\btype$, typed by
  (G$\tu\mathord{::}\gtys$). For $\src{(E\lambda)}$, the induction
  hypothesis types the body under $\Gamma, \sx{:}\liftT{\sgtys}$, rule
  (G$\lambda$) types the abstraction at $\liftT{\sgtys} -> \gtya$, and
  (G$\tu\mathord{::}\gtys$) wraps it with the evidence
  $\liftT{\sgtys} -> \gtya$ read with identity pointers, which is
  well-formed and valid for the consistency of that type with itself, by
  induction on the type.

  The value ascription $\src{(E\mathord{::}\sgtys)}$ is the model of the
  cases that insert a meet. The induction hypothesis gives
  $\Gamma |-t \tv : \gtys$. The type $\gtys$ is well-formed
  (Lemma~\ref{lemma:welltyped-wellform-target}), $\liftT{\sgtysp}$ is
  well-formed (Lemma~\ref{lemma:wellformed-lifting}), and the premise
  $\gtys \rel \liftT{\sgtysp}$ gives $\gtys \erel \liftT{\sgtysp}$, so by
  Lemma~\ref{lemma:meet-validity} the evidence
  $\ev = \gtys \meet \liftT{\sgtysp}$ is well-formed and
  $\vjtf{\ev}{\gtys \erel \liftT{\sgtysp}}$. Rule (G$\mathord{::}\gtys$)
  types $\asc{\ev \tv}{\liftT{\sgtysp}}$ at $\ds{\liftT{\sgtysp}^{1}}$, the
  source type. The term ascription $\src{(E\mathord{::}\sgtya)}$ is the
  same with (G$\mathord{::}\gtya$), at the distribution level.

  The three rules that bind a coerced operand with a fresh \src{let} are
  typed from the inside out. For $\src{(Eapp)}$, write $\gtys[1] =
  \cdom(\gtys)$ and $\gtya = \ccod(\gtys)$, both well-formed since $\gtys$
  is. The evidences $\ev[1] = \gtysp \meet \gtys[1]$ and
  $\ev[2] = \gtys \meet (\gtys[1] -> \gtya)$ are well-formed and valid by
  Lemma~\ref{lemma:meet-validity}, the second because
  $\gtys \erel \gtys[1] -> \gtya$, as in the existence half. Rule (Gapp)
  types $\ty\;\tx$ at $\gtya$ under $\Gamma, \tx{:}\gtys[1], \ty{:}\gtys[1] -> \gtya$.
  Rule (G$\mathord{::}\gtys$) types $\asc{\ev[2] \tv}{\gtys[1] -> \gtya}$ at
  $\ds{(\gtys[1] -> \gtya)^{1}}$ under $\Gamma, \tx{:}\gtys[1]$, where $\tv$
  is typed by weakening (item~1 of Lemma~\ref{lemma:target-weakening}), and
  item~2 of the same lemma types the inner \src{let} at $\gtya$. Rule
  (G$\mathord{::}\gtys$) types $\asc{\ev[1] \tw}{\gtys[1]}$ at
  $\ds{\gtys[1][1]}$ under $\Gamma$, and item~2 again types the outer
  \src{let} at $\gtya = \ccod(\gtys)$, the source type. For $\src{(E+)}$
  the two coercions are to $\rtype$, the body $\add{\tx}{\ty}$ is typed by
  (G$+$), and the result is $\ds{\rtype^{1}}$. For $\src{(Eif)}$ only the
  guard is coerced, to $\btype$; the branches, typed at $\gtya$ and $\gtyb$
  by the induction hypotheses and weakened under $\tx$, are typed by (Gif)
  at the hull $\gtya \mathbin{\oplus_{\?}} \gtyb$, which is the source type,
  and item~2 of Lemma~\ref{lemma:target-weakening} types the \src{let} at
  it.
  \end{revblock}
\end{proof}

\subsection{Type Safety}\label{app:target-safety}

\typesafetyx*

\begin{revblock}[fo1]
\proofheading{Structure of the proof.}
By Lemma~\ref{theorem:ep} the term $\sm$ has an elaboration $\tm$ with
$\cdot |-d \tm : \gtyb$, at the same formula type. Either $\tm$ has no
reduction, and $\sm$ diverges, or it reduces, and the preservation lemma
(Lemma~\ref{lemma:tplc-preservation}) types its result. The reduction
relation is big-step, so there is no progress lemma. The second alternative
of the theorem is the complement of the first.

The preservation lemma is proved by induction on the reduction derivation,
in four separate inductions, each for one clause of its conclusion: the
values of the result are well typed at the types they display; the runtime
type of the result is well-formed and refines the static type, assuming the
formula of the result is satisfiable; that formula is satisfiable; and every
entry of the runtime type is an entry of the static type. The first clause
is proved without the satisfiability hypothesis because the routed rules
$\trg{(\mathit{Dlet})}$ and $\trg{(\mathit{D}\mathord{::}\gtya)}$ coerce
every value of the result of the term they reduce first, including the values that no solution
of the formula gives positive probability, and typing those coercions needs the type
of each value. The runtime type of a configuration $\rctx[\phi]\V$ with
$\V = \d{\tv[i][\p[i]]}[i \in \iSet]$ is the type
$\phty{\phi}{\d{\gtys[i][\p[i]]}[i \in \iSet]}$ that rule (GV) assigns to it,
where $\gtys[i]$ is the type of the closed value $\tv[i]$, unique by
Lemma~\ref{lemma:target-det}. We say that a formula distribution type
$\gtyap$ \emph{refines} $\gtya$ when, for every solution $\nu$ of the formula
of $\gtyap$, there are a solution $\nu'$ of the formula of $\gtya$ and a
coupling between the probabilities of $\nu$ and those of $\nu'$ that has
positive weight only between equal entries. A well-formed type that refines $\gtya$ reorders into it
(Definition~\ref{def:reorder}), since its formula has a solution; the proof
of the dynamic gradual guarantee (Appendix~\ref{app:target-gg}) uses the
refinement itself.

The cases with content are the routed rules. In $\trg{(\mathit{Dlet})}$
the coercion at each entry is typed through the validity of the reordering
evidence (Lemma~\ref{lemma:routing-validity}) and the substituted body
through Lemma~\ref{subtyp}; the satisfiability of the result formula uses
the refinement of the bound term's runtime type, and the refinement of the
result is assembled from the refinements of the branches
(Lemma~\ref{lemma:refinement-operators}), each column being a convex
mixture of branch solutions. In $\trg{(\mathit{D}\mathord{::}\gtya)}$ the
definedness of the routing evidence is the condition under which the rule
takes its first case, and the result refines the target type because every
entry of the routing evidence is sent, by its right pointer, to one entry of
the target, whose formula the validity of the ascription evidence satisfies.
Every coercion of a value, in $\trg{(\mathit{D}\mathord{::}\gtys)}$,
$\trg{(\mathit{Dapp})}$ and in the routed rules, is typed by
Lemma~\ref{lemma:value-coercion}, which rests on the fact that consistent
transitivity preserves the validity of evidences
(Lemma~\ref{lemma:trans-validity}). The complete proof is given at the end
of this subsection, after the lemmas it uses.
\end{revblock}

\rev[mt1]{The complete dynamic semantics of \tlang is given in
Figures~\ref{fig:target-reduction-dis} and~\ref{fig:target-reduction-dis-2}
in the article body.}

\begin{lemma}[{\rev[fo1]{Determinism of typing}}]\label{lemma:target-det}\mechnote{\leanname{TPLC.det\_raw}, \leanname{TPLC.det\_val} and \leanname{TPLC.det\_tm}.}~
  \rev[fo1]{For any environment $\Gamma$, any value $\tv$ and any term $\tm$
  of \tlang{}, any formula simple types $\gtys[1], \gtys[2] \in \FGType$ and
  any formula distribution types $\gtya[1], \gtya[2] \in \FGDType$, if
  $\Gamma |-t \tv : \gtys[1]$ and $\Gamma |-t \tv : \gtys[2]$, then
  $\gtys[1] = \gtys[2]$, and if $\Gamma |-d \tm : \gtya[1]$ and
  $\Gamma |-d \tm : \gtya[2]$, then $\gtya[1] = \gtya[2]$.}
\end{lemma}
\begin{proof}
  \rev[fo1]{By induction on the first derivation, together with the
  statement for raw values. Every rule is determined by the top constructor
  of the term and its annotations, the induction hypothesis identifies the
  types of the premises, and the result type is a function of them and of
  the annotations.}
\end{proof}

\validevidence*
\begin{proof}
  \rev[fo1]{By induction on the evidence. In the distribution case, given a
  solution $\nu$ of the formula of the evidence, validity provides the
  solution $\nu'$ of the side; the tags provide the coupling, which sends the
  whole probability $\nu(\cww[\kk])$ of each entry $\kk$ to the entry
  $\projl{\cww[\kk]}$, with marginals $\nu$ and $\nu'$ (nonnegative because
  the evidence is well-formed); and every entry $\kk$ is precision-related to
  the entry it is sent to.}
\end{proof}

\transinvariant*
\begin{proof}
  \rev[fo1]{By induction on the evidences, the two items together.
  Consistent transitivity is the meet, and the meet is $\eprec$ each of its
  operands (Lemma~\ref{meet-more-precise2}); composing with the hypotheses
  $\ev[1] \eprec \gtys[1]$ and $\ev[2] \eprec \gtys[2]$ by the transitivity
  of $\eprec$, which composes couplings as in
  Lemma~\ref{lemma:coupling-composition}, gives
  $\ev[1] \trans{} \ev[2] \eprec \gtys[1]$ and
  $\ev[1] \trans{} \ev[2] \eprec \gtys[2]$, that is,
  $\jtf{\ev[1] \trans{} \ev[2]}{\gtys[1] \erel \gtys[2]}$. For distribution
  evidences the argument is the same, the reductivity of the distribution
  meet being item~2 of Lemma~\ref{meet-more-precise2}.}
\end{proof}

\begin{lemma}[{\rev[fo1]{Consistent transitivity preserves validity}}]\label{lemma:trans-validity}\mechnote{\leanname{TPLC.hemeetTy\_valid}, \leanname{TPLC.hemeetD\_valid}, \leanname{TPLC.hetransTy\_invariant} and \leanname{TPLC.hvalidFor\_emeetD}.}~
  \rev[fo1]{For any well-formed evidences $\ev[1], \ev[2]$ and any formula
  simple types $\gtys[1], \gtys[2], \gtysp \in \FGType$,
  \begin{enumerate}
  \item if $\vjtfl{\ev[1]}{\gtys[1]}$, $\vjtfr{\ev[2]}{\gtys[2]}$ and
    $\ev[1] \trans{} \ev[2]$ is defined, then
    $\vjtfl{\ev[1] \trans{} \ev[2]}{\gtys[1]}$ and
    $\vjtfr{\ev[1] \trans{} \ev[2]}{\gtys[2]}$;
  \item if $\vjtf{\ev[1]}{\gtys[1] \erel \gtysp}$,
    $\vjtf{\ev[2]}{\gtysp \erel \gtys[2]}$ and $\ev[1] \trans{} \ev[2]$ is
    defined, then $\vjtf{\ev[1] \trans{} \ev[2]}{\gtys[1] \erel \gtys[2]}$.
  \end{enumerate}
  Likewise for distribution evidences.}
\end{lemma}
\begin{proof}
  \rev[fo1]{By induction on the evidences. On simple types the two
  validities are structural and the meet is componentwise. For distribution
  evidences, the entry $(\kk[1], \kk[2])$ of the meet carries
  $\ev[\kk[1]] \trans{} \ev[\kk[2]]$ and inherits the left pointer of
  $\kk[1]$ and the right pointer of $\kk[2]$, so the entry of $\gtya[1]$ it
  names is the one $\kk[1]$ names, and it is left valid for that entry by the
  induction hypothesis; the push-forward of a solution of the meet along the
  left pointers factors through its row sums, which are a solution of
  $\ev[1]$ by the marginal clause of the meet, and from there through the
  push-forward clause of the validity of $\ev[1]$. The right side is
  symmetric.}
\end{proof}

\begin{lemma}[Substitution preserves typing]\label{subtyp}\mechnote{\leanname{TPLC.hasTy\_subst\_tm}, \leanname{TPLC.hasTy\_subst0\_tm} and \leanname{TPLC.subErr\_typed}.}~
  \rev[fo1]{For any well-formed environment $\Gamma$, any variable
  $\tx \in \text{Var}$, any term $\tm$ and any closed value $\tv$ of
  \tlang{}, any formula simple
  type $\gtys \in \FGType$ and any formula distribution type
  $\gtya \in \FGDType$, if $\Gamma, \tx{:}\gtys |- \tm : \gtya$ and
  $\cdot |- \tv : \gtys$, then
  \begin{enumerate}
  \item $\Gamma |- \tm{}[\tv/\tx] : \gtya$;
  \item $\Gamma |- \ssub{\tm}{\tv}{\tx} : \gtya$.
  \end{enumerate}}
\end{lemma}
\begin{proof}
  \rev[fo1]{Item~1 is by induction on the typing derivation of $\tm$,
  together with the statements for raw values and values. The variable
  $\tx$ is replaced by $\tv$, typed at $\gtys$ in the ambient environment by
  weakening (Lemma~\ref{lemma:target-weakening}); another variable is
  unchanged; the binders of $\lambda$ and of each body of a \src{let}
  family push the substitution under the binder; every other rule is a
  congruence. For item~2, if $\tv$ is an error, then
  $\ssub{\tm}{\tv}{\tx} = \errort{\gtya}$, typed at $\gtya$ by
  (Gerr$_{\gtya}$) since $\gtya$ is well-formed
  (Lemma~\ref{lemma:welltyped-wellform-target}); otherwise $\mathit{sub}$ is
  the substitution of item~1.}
\end{proof}

\begin{lemma}[{\rev[fo1]{One-step coercion of a value}}]\label{lemma:value-coercion}\mechnote{\leanname{TPLC.ascV\_total}, \leanname{TPLC.red\_ascV\_inv}, \leanname{TPLC.cell\_coercion\_typed} and \leanname{TPLC.tyEntry\_of\_red\_ascV}.}~
  \rev[fo1]{For any closed value $\tv$ of \tlang{}, any well-formed evidence
  $\ev$, any formula simple type $\gtys \in \FGType$ and any well-formed
  formula simple type $\gtysp \in \FGType$, if $\cdot |-t \tv : \gtys$ and
  $\vjtf{\ev}{\gtys \erel \gtysp}$, then the term $\asc{\ev \tv}{\gtysp}$
  reduces in one step to $\dt{\pt{\tw}}$ for some value $\tw$, and every
  value $\tw$ with such a reduction satisfies $\cdot |-t \tw : \gtysp$ and
  is $\asc{\ev[1] \trans{} \ev\;\tu}{\gtysp}$ when
  $\tv = \asc{\ev[1] \tu}{\gtys}$ for some evidence $\ev[1]$ and raw value
  $\tu$ with $\ev[1] \trans{} \ev$ defined, and $\errort{\gtysp}$
  otherwise.}
\end{lemma}
\begin{proof}
  \rev[fo1]{The rules that apply to a value ascription are
  $\trg{(\mathit{D}\mathord{::}\gtys)}$, when $\tv$ is an ascription, and
  $\trg{(\mathit{Derr}\mathord{::}\gtys)}$, when it is an error, and each
  applies unconditionally. An error result is typed at $\gtysp$ by
  (Gerr$_{\gtys}$). In the first case of
  $\trg{(\mathit{D}\mathord{::}\gtys)}$, the typing derivation of $\tv$ ends
  in (G$\tu\mathord{::}\gtys$), whose premises give $\cdot |-t \tu : \gtys[u]$,
  $\vjtf{\ev[1]}{\gtys[u] \erel \gtys}$ and $\jtf{}{\ev[1]}$; by
  Lemma~\ref{lemma:trans-validity},
  $\vjtf{\ev[1] \trans{} \ev}{\gtys[u] \erel \gtysp}$, and the composed
  evidence is well-formed because its formula is satisfiable (the meet is
  defined) over well-formed operands, so (G$\tu\mathord{::}\gtys$) types
  the result at $\gtysp$.}
\end{proof}

\begin{proof}[Proof of Lemma~\ref{reorderevidence2} (initial evidence for reordering)]
  \rev[fo1]{By induction on the types, the simple case being structural.
  For distribution types, $\initReorder{\gtya[1]}{\gtya[2]}$ is the
  construction $\witnessname$ for the partial identity: it has one entry for
  each pair of equal entries of the operands, carrying that entry, and its
  formula is the coupling condition over those pairs. The argument of
  Lemma~\ref{consistency-defined}, with equality in place of consistency,
  gives that the operator is defined exactly when
  $\reoder{\gtya[1]}{\gtya[2]}$, and the two marginal clauses make it
  $\eprec$ each operand, as for the meet in item~2 of
  Lemma~\ref{meet-more-precise2}. The operator is well-formed, its formula
  being satisfiable with nonnegative solutions of total probability $1$
  forming a convex set, so it is $\eprec$ itself, which is the evidence
  judgment of the statement.}
\end{proof}

\begin{lemma}[{\rev[fo1]{Validity of the routing evidence}}]\label{lemma:routing-validity}\mechnote{item~1, \leanname{TPLC.tagReorderD\_cell}, \leanname{TPLC.hvtag\_tagReorderTy}, \leanname{TPLC.goodTy\_reorderTy}, \leanname{TPLC.hvalidFor\_tagReorderD} and \leanname{TPLC.hvalid\_tagReorderD}; item~2, \leanname{TPLC.hemeetD\_entry}, \leanname{TPLC.goodTy\_emeetD\_entry} and \leanname{TPLC.hvalidFor\_dascD\_evidence}.}~
  \rev[fo1]{For any formula distribution types $\gtyap, \gtya, \gtyb \in
  \FGDType$ such that the entries of $\gtyap$ and $\gtya$ are well-formed,
  and any well-formed evidence $\evd$,
  \begin{enumerate}
  \item every entry $\kk$ of $\initReorder{\gtyap}{\gtya}$, which pairs an
    entry $\gtys[i][\prime]$ of $\gtyap$ with an equal entry $\gtys[j]$ of $\gtya$,
    carries an evidence $\ev[\kk]$ that is defined, well-formed, with
    $\vjtfl{\ev[\kk]}{\gtys[i][\prime]}$ and $\vjtfr{\ev[\kk]}{\gtys[j]}$; if
    moreover $\jtf{}{\gtyap}$, $\jtf{}{\gtya}$ and $\reoder{\gtyap}{\gtya}$,
    then $\vjtf{\initReorder{\gtyap}{\gtya}}{\gtyap \erel \gtya}$;
  \item if $\vjtf{\evd}{\gtya \erel \gtyb}$ and
    $(\initReorder{\gtyap}{\gtya}) \trans{} \evd$ is defined, then every
    entry of the composition carries a defined, well-formed evidence that is
    left valid for the entry of $\gtyap$ its left pointer names and right
    valid for the entry of $\gtyb$ its right pointer names, and if moreover
    $\jtf{}{\gtyap}$, $\jtf{}{\gtya}$ and $\reoder{\gtyap}{\gtya}$, then
    $\vjtf{(\initReorder{\gtyap}{\gtya}) \trans{} \evd}{\gtyap \erel \gtyb}$.
  \end{enumerate}}
\end{lemma}
\begin{proof}
  \rev[fo1]{For item~1, the entry at a pair of equal entries carries that
  entry with identity pointers at every depth, which is defined and
  well-formed when the entry is, and valid for the entry on both sides by
  induction on the entry; when $\reoder{\gtyap}{\gtya}$, the push-forward
  clauses of validity are the two marginal clauses of the operator's formula,
  satisfiable by Lemma~\ref{reorderevidence2}. For item~2, an entry of the
  composition pairs an entry $\kk$ of the reordering with an entry $\kk'$ of
  $\evd$ whose entries are consistent, so its evidence $\ev[\kk] \trans{} \ev[\kk']$
  is defined (Lemma~\ref{consistency-defined}) and well-formed, and
  Lemma~\ref{lemma:trans-validity} gives its two validities from those of
  $\ev[\kk]$ and of $\ev[\kk']$; the push-forward clauses follow from the
  marginal clauses of the meet, as in Lemma~\ref{lemma:trans-validity},
  whose well-formedness hypothesis the reordering satisfies by
  Lemma~\ref{reorderevidence2}.}
\end{proof}

\begin{lemma}[{\rev[fo1]{Refinement through the type operators}}]\label{lemma:refinement-operators}\mechnote{item~1, \leanname{TPLC.refDist\_chooseSem\_cong} and \leanname{TPLC.refDist\_chooseSemU\_cong}; item~2, \leanname{TPLC.refDist\_hull\_left} and \leanname{TPLC.refDist\_hull\_right}; item~3, \leanname{TPLC.refDist\_wsum\_letSem}; item~4, \leanname{TPLC.goodD\_confF\_wsum}.}~
  \rev[fo1]{For any formula distribution types
  $\gtya[1], \gtya[2], \gtya[1][\prime], \gtya[2][\prime], \gtyap \in \FGDType$,
  any formula distribution type
  $\gtya = \phty{\phi}{\d{\gtys[j][\p[j]]}[j \in \jSet]} \in \FGDType$, any
  family $(\gtya[j])_{j \in \jSet}$ of well-formed formula distribution
  types, any finite index set $\kSet$, any map $\kappa : \kSet \to \jSet$,
  any formula $\psi$ over $(\cww[\kk])_{\kk \in \kSet}$ and any family
  $(\rctx[\phi[][\kk]] \V[\kk])_{\kk \in \kSet}$ of configurations,
  \begin{enumerate}
  \item if $\gtya[1][\prime]$ refines $\gtya[1]$ and $\gtya[2][\prime]$ refines $\gtya[2]$,
    then the type that rule (G$\oplus$) assigns from $\gtya[1][\prime]$ and
    $\gtya[2][\prime]$ refines the type it assigns from $\gtya[1]$ and $\gtya[2]$,
    for the same lifted probability;
  \item if $\gtyap$ refines $\gtya[1]$ and $\gtya[2]$ is well-formed, then
    $\gtyap$ refines $\gtya[1] \mathbin{\oplus_{\?}} \gtya[2]$, and
    symmetrically with the roles of $\gtya[1]$ and $\gtya[2]$ exchanged;
  \item if the solutions of $\psi$ are nonnegative, their push-forward along
    $\kappa$ is a solution of $\phi$, and the runtime type of
    $\rctx[\phi[][\kk]] \V[\kk]$ refines $\gtya[\kappa(\kk)]$ for every
    $\kk \in \kSet$ to which some solution of $\psi$ gives a positive
    $\cww[\kk]$, then the runtime
    type of $\rctx[\psi \land \bigwedge_{\kk} (\cww[\kk] > 0 \Rightarrow \phi[][\kk])]
    \ssum[\kk \in \kSet] \cww[\kk] \cdot \V[\kk]$ refines
    $\sugarconj{\phi}{\sum_{j \in \jSet} \p[j] \cdot \gtya[j]}$;
  \item if the hypotheses of item~3 hold, the formula of the runtime type
    of item~3 is satisfiable, the solutions of $\psi$ form a
    convex set and have total probability $1$, the configuration
    $\rctx[\phi[][\kk]] \V[\kk]$ has a well-formed runtime type for every
    $\kk \in \kSet$ to which some solution of $\psi$ gives a positive
    $\cww[\kk]$, and all entries are well-formed, then the runtime type of item~3 is well-formed.
  \end{enumerate}}
\end{lemma}
\begin{proof}
  \rev[fo1]{Each item builds the coupling of the conclusion from those of
  the hypotheses. For item~1, a solution of the result is
  $(a \cdot p, (1-a) \cdot q)$ with $a$ a solution of the lifted probability
  and $p$, $q$ solutions of the operands; the two couplings $u$ and $v$ given
  by the refinements at $p$ and $q$ combine into the matrix with blocks
  $a \cdot u$ and $(1-a) \cdot v$, whose target is the solution
  $(a \cdot p', (1-a) \cdot q')$ with the same $a$. For item~2, given a
  solution of $\gtyap$ with its coupling $u$ onto a solution $p'$ of
  $\gtya[1]$, take any solution $q'$ of $\gtya[2]$, which exists since
  $\gtya[2]$ is well-formed, and the hull parameter $t = 1$. The solution
  $(1 \cdot p', 0 \cdot q')$ of the hull and the coupling $u$ extended by
  zero columns on the entries of $\gtya[2]$ are the required pair; the
  symmetric case takes $t = 0$. For item~3, a solution of the weighted sum
  is a solution $\omega$ of $\psi$ together with a solution $b_{\kk}$ of
  $\phi[][\kk]$ for each $\kk$ with $\omega_{\kk} > 0$. Let
  $\rho_j = \sum_{\kk : \kappa(\kk) = j} \omega_{\kk}$, a solution of $\phi$.
  For each positive $\kk$, the refinement of its runtime type at $b_{\kk}$
  gives a solution $f_{\kk}$ of $\gtya[\kappa(\kk)]$ and a coupling $w^{\kk}$.
  For each $j$ with $\rho_j > 0$ let
  $g_j = \rho_j^{-1} \sum_{\kk : \kappa(\kk) = j} \omega_{\kk} \cdot f_{\kk}$,
  a convex combination of solutions of $\gtya[j]$, hence a solution since
  $\gtya[j]$ is well-formed; for $\rho_j = 0$ take any solution of
  $\gtya[j]$. The assignment $(\rho_j \cdot g_j)_j$ is a solution of the
  weighted union, and the matrix with block $\omega_{\kk} \cdot w^{\kk}$ at
  $(\kk, \kappa(\kk))$ and zero elsewhere couples the two. Its row sums are
  $\omega_{\kk} \cdot b_{\kk}$ and its column sums
  $\sum_{\kk : \kappa(\kk) = j} \omega_{\kk} \cdot f_{\kk} = \rho_j \cdot g_j$,
  and it has positive weight only between equal entries because the $w^{\kk}$
  do. For item~4, the formula of the result is satisfiable by hypothesis;
  its solutions are nonnegative since $\omega$ and the $b_{\kk}$ are; their
  total probability is $\sum_{\kk} \omega_{\kk} \cdot 1 = 1$; and given two
  solutions $(\omega, b)$ and $(\omega', b')$ and $t \in [0,1]$, the mixture
  $t \omega + (1-t) \omega'$ is a solution of $\psi$ by convexity, and at
  each $\kk \in \kSet$ of positive mixed weight the branch solution is the mixture of
  $b_{\kk}$ and $b'_{\kk}$ with weights proportional to $t \omega_{\kk}$ and
  $(1-t) \omega'_{\kk}$, a solution by the convexity of the branch's runtime
  type; the entries are well-formed by hypothesis.}
\end{proof}

\begin{lemma}[{\rev[fo1]{Preservation for \tlang}}]\label{lemma:tplc-preservation}\mechnote{\leanname{TPLC.type\_safety}, from \leanname{TPLC.type\_safety\_vals}, \leanname{TPLC.type\_safety\_gr} and \leanname{TPLC.type\_safety\_sat}, and \leanname{TPLC.entriesIn\_red}.}~
  \rev[fo1]{For any term $\tm$ of \tlang{}, any formula distribution type
  $\gtyb \in \FGDType$, any configuration $\rctx[\phi] \V$ and any index
  $\kk$, if $\cdot |-d \tm : \gtyb$ and
  $\rrctx[\phii[]] \tm \nreds{}{\kk} \rctx[\phi] \V$, then
  $|- \rctx[\phi] \V : \gtya$, where $\gtya$ is the runtime type of
  $\rctx[\phi] \V$, the type $\gtya$ refines $\gtyb$, and every entry of
  $\gtya$ is an entry of $\gtyb$; in particular, $\reoder{\gtya}{\gtyb}$.}
\end{lemma}
\begin{proof}
  \begin{revblock}[fo1]
  By induction on the derivation of
  $\rrctx[\phii[]] \tm \nreds{}{\kk} \rctx[\phi] \V$, in four separate
  inductions, in this order: (a) every value of $\V$ is well typed at the
  type it displays; (b) if $\phi$ is satisfiable, then $\gtya$ is
  well-formed and refines $\gtyb$; (c) $\phi$ is satisfiable; (d) every
  entry of $\gtya$ is an entry of $\gtyb$. Clauses (a) and (b), with the
  well-formedness of $\gtya$, are the premises of rule (GV); (b) and (c)
  give the refinement, and the refinement of a well-formed type gives
  $\reoder{\gtya}{\gtyb}$. In each case the typing derivation of $\tm$ ends
  in the rule for its top constructor, and the premise of
  $\trg{(\mathit{Dlet})}$ and of $\trg{(\mathit{D}\mathord{::}\gtya)}$ that
  types the bound term agrees with that derivation by
  Lemma~\ref{lemma:target-det}. We give the cases of (b), where the
  hypothesis that the formula of the result is satisfiable is passed to the
  premises as indicated; the cases of (a), (c) and (d) are noted where they
  differ from a direct appeal to the induction hypothesis.

  \begin{case}[values, addition and errors]
    A value $\tv$ reduces to $\dt{\pt{\tv}}$; rule (Gv) gives
    $\cdot |-t \tv : \gtys$ and $\gtyb = \ds{\gtys[][1]}$, which is the
    runtime type, well-formed by
    Lemma~\ref{lemma:welltyped-wellform-target}, and a type refines itself.
    In $\trg{(\mathit{D+})}$, rule (G$+$) gives
    $\cdot |-t \asc{\ev[i] r_i}{\rtype} : \rtype$, so by
    (G$\tu\mathord{::}\gtys$) and the rules of right validity for $\rtype$,
    $\ev[1] = \ev[2] = \rtype$, hence $\ev[3] = \rtype$ and the result is
    typed at $\ds{\rtype^1}$. In $\trg{(\mathit{Derr})}$ the result places
    at each entry $\gtys[i]$ of $\gtyb$ the error $\errort{\gtys[i]}$, typed
    at $\gtys[i]$ by (Gerr$_{\gtys}$), so the runtime type is $\gtyb$
    itself, well-formed by the premise of (Gerr$_{\gtya}$). The four
    error-propagation rules and the error case of
    $\trg{(\mathit{D}\mathord{::}\gtys)}$ and of
    $\trg{(\mathit{D}\mathord{::}\gtya)}$ are the same. Their result is the
    error at the type of the redex, spread over its entries, and that type
    is well-formed by the typing rule of the redex, the hull of
    $\trg{(\mathit{Derr\,if})}$ being the type that (Gif) assigns, by
    Lemma~\ref{lemma:target-det}. None of these cases uses an induction
    hypothesis, and for (c) their formula is satisfiable by
    well-formedness.
  \end{case}

  \begin{case}[probabilistic choice]
    $\tm = \tm[1] \tpsum \tm[2]$, the two branches reduce to
    $\rctx[\phi[][1]] \V[1]$ and $\rctx[\phi[][2]] \V[2]$, and the result is
    $\rctx[\phi[][1] \land \phi[][2] \land \phi[][\oplus]] \cww[1] \cdot \V[1] + \cww[2] \cdot \V[2]$.
    Rule (G$\oplus$) gives $\cdot |-d \tm[i] : \gtya[i]$ and $\gtyb$ the type
    it assigns from $\gtya[1]$ and $\gtya[2]$. A solution of the result
    formula is $(a \cdot p, (1-a) \cdot q)$ with $a$ a solution of
    $\phi[][\oplus]$, $p$ of $\phi[][1]$ and $q$ of $\phi[][2]$; the
    induction hypothesis, with $p$ and with $q$, gives the runtime types of
    the two branches well-formed and refining $\gtya[1]$ and $\gtya[2]$, and
    item~1 of Lemma~\ref{lemma:refinement-operators} concludes; the runtime
    type of the result is well-formed as the type of a choice of
    well-formed types. For (c), the induction hypothesis gives $p$ and $q$,
    and $a$ exists since $\phi[][\oplus]$ is the lifting of a probability.
  \end{case}

  \begin{case}[conditional]
    $\tm = \trg{\ite{\asc{\ev \ttt}{\btype}}{\tm[1]}{\tm[2]}}$ and
    $\tm[1] \nreds{}{\j} \rctx[\phi] \V$. Rule (Gif) gives
    $\cdot |-d \tm[1] : \gtya[1]$, $\cdot |-d \tm[2] : \gtya[2]$ and
    $\gtyb = \gtya[1] \mathbin{\oplus_{\?}} \gtya[2]$. The induction
    hypothesis, with the same solution, gives $\gtya$ well-formed and
    refining $\gtya[1]$; since $\tm[2]$ is well typed, $\gtya[2]$ is
    well-formed by Lemma~\ref{lemma:welltyped-wellform-target}, and item~2
    of Lemma~\ref{lemma:refinement-operators} gives that $\gtya$ refines the
    hull. Rule $\trg{(\mathit{Dif})}$ is symmetric.
  \end{case}

  \begin{case}[value ascription]
    $\tm = \asc{\ev[2] (\asc{\ev[1] \tu}{\gtys})}{\gtysp}$ and the result is
    $\dt{\pt{\tw}}$. Rule (G$\mathord{::}\gtys$) gives
    $\cdot |-t \asc{\ev[1] \tu}{\gtys} : \gtys$,
    $\vjtf{\ev[2]}{\gtys \erel \gtysp}$, $\jtf{}{\ev[2]}$ and
    $\jtf{}{\gtysp}$, with $\gtyb = \ds{\gtysp[][1]}$. By
    Lemma~\ref{lemma:value-coercion}, $\cdot |-t \tw : \gtysp$, so the
    runtime type is $\ds{\gtysp[][1]} = \gtyb$. The case
    $\trg{(\mathit{Derr}\mathord{::}\gtys)}$ is included.
  \end{case}

  \begin{case}[application]
    $\tm = (\asc{\ev (\lambda \tx{:}\gtysp. \tm[0])}{\gtys -> \gtya[r]})\; \tv$,
    with $\asc{\cdom(\ev) \tv}{\gtysp} \nreds{}{1} \rctx[\cdot] \dt{\pt{\tw}}$
    and $\ssub{(\asc{\ccod(\ev) \tm[0]}{\gtya[r]})}{\tw}{\tx} \nreds{}{\j} \rctx[\phi] \V$.
    Rule (Gapp) gives
    $\cdot |-t \asc{\ev (\lambda \tx{:}\gtysp. \tm[0])}{\gtys -> \gtya[r]} : \gtys -> \gtya[r]$
    and $\cdot |-t \tv : \gtys$, with $\gtyb = \gtya[r]$; the typing of the
    function ends in (G$\tu\mathord{::}\gtys$), whose premises give
    $\tx{:}\gtysp |-d \tm[0] : \gtya[b]$ with $\jtf{}{\gtysp}$,
    $\vjtf{\ev}{(\gtysp -> \gtya[b]) \erel (\gtys -> \gtya[r])}$ and
    $\jtf{}{\ev}$. The rules of validity between two arrows force
    $\ev = \ev[s] -> \evd[d]$ with $\vjtfl{\ev[s]}{\gtysp}$,
    $\vjtfr{\ev[s]}{\gtys}$, $\vjtfl{\evd[d]}{\gtya[b]}$ and
    $\vjtfr{\evd[d]}{\gtya[r]}$. The argument is coerced with
    $\cdom(\ev)$, which is $\ev[s]$ with its two pointers exchanged, so
    $\vjtf{\cdom(\ev)}{\gtys \erel \gtysp}$, and
    Lemma~\ref{lemma:value-coercion} gives $\cdot |-t \tw : \gtysp$. The
    body $\asc{\ccod(\ev) \tm[0]}{\gtya[r]}$ is typed at $\gtya[r]$ under
    $\tx{:}\gtysp$ by (G$\mathord{::}\gtya$), with
    $\vjtf{\evd[d]}{\gtya[b] \erel \gtya[r]}$, so item~2 of
    Lemma~\ref{subtyp} types $\ssub{(\asc{\ccod(\ev) \tm[0]}{\gtya[r]})}{\tw}{\tx}$
    at $\gtya[r]$, and the induction hypothesis on its reduction, with the
    same solution, concludes.
  \end{case}

  \begin{case}[let]
    $\tm = \lett{\tx}{\tm[0]}{\overline{\tn[j]}^{\,j \in \jSet}}$. Rule
    (Glet) gives $\cdot |-d \tm[0] : \gtya[0]$ with
    $\gtya[0] = \phty{\phi[][0]}{\d{\gtysp[j][\p[j]]}[j \in \jSet]}$,
    $\tx{:}\gtysp[j] |-d \tn[j] : \gtya[j]$ for every $j$, and
    $\gtyb = \sugarconj{\phi[][0]}{\sum_{j} \p[j] \cdot \gtya[j]}$. The
    premises of the rule give $\tm[0] \nreds{}{\j[1]} \rctx[\phii[']] \V[0]$
    with $\V[0] = \d{\tv[i][\p[i]]}[i \in \iSet]$ and runtime type
    $\gtyap[0] = \phty{\phii[']}{\d{\gtys[i][\prime\p[i]]}[i \in \iSet]}$, the
    reordering $\initReorder{\gtyap[0]}{\gtya[0]} = \phty{\pphi[2]}{\d{\ev[\kk][\cww[\kk]]}[\kk \in \kSet]}$,
    and for every entry $\kk$, with $i = \projl{\cww[\kk]}$ and
    $j = \projr{\cww[\kk]}$, $\asc{\ev[\kk] \tv[i]}{\gtysp[j]} \nreds{}{1} \rctx[\cdot] \dt{\pt{\tw[\kk]}}$
    and $\ssub{\tn[j]}{\tw[\kk]}{\tx} \nreds{}{\j[2]} \rctx[\phi[][\kk]] \V[\kk]$;
    the result is
    $\rctx[\pphi[2] \land \bigwedge_{\kk} (\cww[\kk] > 0 \Rightarrow \phi[][\kk])] \ssum[\kk \in \kSet] \cww[\kk] \cdot \V[\kk]$.
    Three facts hold without any satisfiability hypothesis. By clause (a)
    on the bound term, $\cdot |-t \tv[i] : \gtys[i][\prime]$ for every $i$. For an
    entry $\kk$, the entries $\gtys[i][\prime] = \gtysp[j]$ are well-formed, so by
    item~1 of Lemma~\ref{lemma:routing-validity} the evidence $\ev[\kk]$
    is well-formed with $\vjtf{\ev[\kk]}{\gtys[i][\prime] \erel \gtysp[j]}$, and
    Lemma~\ref{lemma:value-coercion} gives $\cdot |-t \tw[\kk] : \gtysp[j]$.
    And by item~2 of Lemma~\ref{subtyp},
    $\cdot |-d \ssub{\tn[j]}{\tw[\kk]}{\tx} : \gtya[j]$. Clause (a) for the
    result follows from clause (a) on every branch, and clause (d) from
    clause (d) on every branch, since the values of the result are those of
    the $\V[\kk]$ and the entries of $\gtya[j]$ are entries of $\gtyb$.

    For clause (c), by clause (c) on the bound term, $\phii[']$ is
    satisfiable, so by clause (b) on the bound term $\gtyap[0]$ refines
    $\gtya[0]$; a solution of $\phii[']$ with its coupling onto a solution
    of $\gtya[0]$, supported on pairs of equal entries, is a solution of
    $\pphi[2]$ (Lemma~\ref{reorderevidence2}); with a solution of each
    $\phi[][\kk]$, given by clause (c) on the branches, it is a solution of
    the result formula.

    For clause (b), let $(\omega, b)$ be a solution of the result formula,
    $\omega$ a solution of $\pphi[2]$ and $b_{\kk}$ a solution of
    $\phi[][\kk]$ for each $\kk$ with $\omega_{\kk} > 0$. The left marginal
    of $\omega$ is a solution of $\phii[']$, by the marginal clause of the
    reordering, and the induction hypothesis on the bound term with it gives
    $\gtyap[0]$ well-formed and refining $\gtya[0]$; hence
    $\initReorder{\gtyap[0]}{\gtya[0]}$ is well-formed
    (Lemma~\ref{reorderevidence2}), so its solutions are nonnegative, form a
    convex set and have, by the right marginal clause, a push-forward along
    $\projr{\cdot}$ that is a solution of $\phi[][0]$. For each entry $\kk$
    with $\omega_{\kk} > 0$, the induction hypothesis on the branch with
    $b_{\kk}$ gives its runtime type well-formed and refining
    $\gtya[\projr{\cww[\kk]}]$. The $\gtya[j]$ are well-formed by
    Lemma~\ref{lemma:welltyped-wellform-target}, so items~3 and~4 of
    Lemma~\ref{lemma:refinement-operators}, with
    $\kappa(\kk) = \projr{\cww[\kk]}$ and $\psi = \pphi[2]$, give that the
    runtime type of the result is well-formed and refines $\gtyb$; the
    satisfiability that item~4 asks for is witnessed by $(\omega, b)$.
  \end{case}

  \begin{case}[distribution ascription]
    $\tm = \asc{\evd \tm[0]}{\gtyb}$ with
    $\gtyb = \phty{\pphi[3]}{\d{\gtysp[j][\p[j]]}[j \in \jSet]}$, in the
    first case of $\trg{(\mathit{D}\mathord{::}\gtya)}$. Rule
    (G$\mathord{::}\gtya$) gives $\cdot |-d \tm[0] : \gtya[0]$,
    $\vjtf{\evd}{\gtya[0] \erel \gtyb}$, $\jtf{}{\evd}$ and $\jtf{}{\gtyb}$.
    The premises give $\tm[0] \nreds{}{\kp} \rctx[\phi[][1]] \V[0]$ with
    runtime type $\gtyap[0]$, the routing evidence
    $(\initReorder{\gtyap[0]}{\gtya[0]}) \trans{} \evd = \phty{\pphi[2]}{\d{\ev[\kk][\cww[\kk]]}[\kk \in \kSet]}$,
    defined, and for every entry $\kk$ the reduction
    $\asc{\ev[\kk] \tv[\projl{\cww[\kk]}]}{\gtysp[\projr{\cww[\kk]}]} \nreds{}{1} \rctx[\cdot] \dt{\pt{\tw[\kk]}}$;
    the result is $\rctx[\pphi[2]] \d{\tw[\kk][\cww[\kk]]}[\kk \in \kSet]$.
    By clause (a) on the ascribed term and item~2 of
    Lemma~\ref{lemma:routing-validity}, each $\ev[\kk]$ is well-formed and
    valid for the judgment between the type of $\tv[\projl{\cww[\kk]}]$ and
    $\gtysp[\projr{\cww[\kk]}]$, so Lemma~\ref{lemma:value-coercion} gives
    $\cdot |-t \tw[\kk] : \gtysp[\projr{\cww[\kk]}]$. This is clause (a),
    and clause (d), since every $\tw[\kk]$ displays an entry of $\gtyb$.
    Clause (c) is the condition under which the rule takes its first case.
    For clause (b), a solution $\omega$ of $\pphi[2]$ restricts, by the
    marginal clauses of the meet and of the reordering, to a solution of
    $\phi[][1]$, and the induction hypothesis on the ascribed term with it gives
    $\gtyap[0]$ well-formed; its refinement is not needed, since the
    definedness of the routing evidence is given. The runtime type of the
    result is $\phty{\pphi[2]}{\d{\gtysp[\projr{\cww[\kk]}][\cww[\kk]]}[\kk \in \kSet]}$,
    well-formed because $\pphi[2]$ is the satisfiable formula of a meet of
    well-formed operands, $\initReorder{\gtyap[0]}{\gtya[0]}$ being
    well-formed by Lemma~\ref{reorderevidence2} once its formula is known
    to be satisfiable, and because its entries are entries of $\gtyb$. It
    refines $\gtyb$. Given $\omega$, the push-forward
    $q_j = \sum_{\kk : \projr{\cww[\kk]} = j} \omega_{\kk}$ is a solution of
    $\pphi[3]$, by the right marginal clause of the meet and the
    push-forward clause of the right validity of $\evd$, and the coupling
    that sends the whole weight $\omega_{\kk}$ of each entry $\kk$ to the entry
    $\projr{\cww[\kk]}$ has these marginals and pairs equal entries.
  \end{case}

  \begin{case}[monotonicity]
    Rule $\trg{(\mathit{Dmon})}$ is the induction hypothesis, in the four
    clauses. \qedhere
  \end{case}
  \end{revblock}
\end{proof}

\begin{theoremproof}{Theorem~\ref{lm:typesafety} (Type safety for \glang)}
  \rev[fo1]{By Lemma~\ref{theorem:ep}, there is an elaboration $\tm$ of
  $\sm$ with $\cdot |-d \tm : \gtyb$. If there is no $\kk$ such that
  $\rrctx[\cdot] \tm \nreds{}{\kk} \rctx[\phi] \V$, then $\diverge[\sm]$ and
  case~2 holds. Otherwise Lemma~\ref{lemma:tplc-preservation} gives
  $|- \rctx[\phi] \V : \gtya$ with $\reoder{\gtya}{\gtyb}$ and every entry
  of $\gtya$ an entry of $\gtyb$, which is case~1.}
\end{theoremproof}

\subsection{Dynamic Gradual Guarantee for \tlang}\label{app:target-gg}

\targetdgg*

\begin{revblock}[fo1]
\proofheading{Structure of the proof.}
The proof is by induction on the derivation of the reduction of the more
precise term, which raises no error, for all less precise terms, reductions
and derivations at once. In $\trg{(\mathit{Dapp})}$ and
$\trg{(\mathit{Dlet})}$ the induction hypothesis is applied to the
reductions of substituted bodies, which are not subterms of the less precise
term, and in $\trg{(\mathit{Dlet})}$ to several reductions of the less
precise term, one per entry of the less precise routing evidence coupled
with a given entry of the more precise one. In each case the typing
derivation of the more precise term ends in the rule for its top constructor,
the derivation of $\tm \gprec \tn$ ends in the rule for the same constructor,
which fixes the shape of $\tn$, and the reduction of $\tn$ ends, after the
applications of $\trg{(\mathit{Dmon})}$, in a rule that the shape of $\tn$
determines; the coupling required by rule $\trg{(\mathord{\gprec}\V)}$ is
then constructed. The rules excluded by the hypothesis on the derivation,
$\trg{(\mathit{Derr})}$, $\trg{(\mathit{Derr}\mathord{::}\gtys)}$, the four
error-propagation rules, the second case of
$\trg{(\mathit{D}\mathord{::}\gtys)}$ and of
$\trg{(\mathit{D}\mathord{::}\gtya)}$, and $\trg{(\mathit{Dv})}$ on an
error, are not cases.

Every coercion of a value is handled by one lemma
(Lemma~\ref{lemma:cell-simulation}). Two related values coerced with two
related evidences to two related types give two related values, and the
less precise coercion does not fail because the more precise one did not.
Its proof is the monotonicity of evidence combination
(Lemma~\ref{lemma:mono}), strengthened to the two tag-aware precisions that
the ascription rules of Figure~\ref{fig:term-precision} demand
(Lemma~\ref{lemma:tag-mono-trans}). The routed rules compare two routing
evidences computed on the two sides, entry by entry, and the coupling between
their entries must remember which value each entry routes and to which entry
of the target type.
This is done by the greatest-lower-bound property of the reordering
(Lemma~\ref{lemma:reorder-glb}), which needs the refinement of the less
precise runtime type given by Lemma~\ref{lemma:tplc-preservation}, and by
the tag-aware forms of reductivity and of the greatest lower bound of the
meet (Lemma~\ref{lemma:tag-meet}). The precision of configurations is a
congruence for the operations that the rules perform on them
(Lemma~\ref{lemma:confprec-congruence}); in the weighted sums of
$\trg{(\mathit{Dlet})}$ each entry of the less precise routing evidence
receives a convex mixture of the solutions transported from the entries
coupled with it, which is a
solution because the runtime types of the branches are well-formed. The
complete proof is given at the end of this subsection, after the lemmas it
uses.
\end{revblock}

\rev[mt1]{The complete precision rules are given in
Figure~\ref{fig:term-precision} in the article body.}

\begin{lemma}[Substitution preserves precision]\label{subep}\mechnote{\leanname{TPLC.subep\_noerr} and \leanname{TPLC.prec\_subst0}.}
  \rev[fo1]{For any variable $\tx \in \text{Var}$, any terms $\tm, \tn$ of
  \tlang{}, any formula simple types $\gtys, \gtysp \in \FGType$ and any
  closed values $\tv, \tvp$ of \tlang{} that are not errors, if
  $\tx{:}\gtys \gprec \tx{:}\gtysp |- \tm \gprec \tn$,
  $\cdot |-t \tv : \gtys$, $\cdot |-t \tvp : \gtysp$ and $\tv \gprec \tvp$,
  then $\ssub{\tm}{\tv}{\tx} \gprec \ssub{\tn}{\tvp}{\tx}$.}
\end{lemma}
\begin{proof}
  \rev[fo1]{Since neither value is an error, $\mathit{sub}$ is the
  substitution. By induction on the derivation of $\tm \gprec \tn$. The
  variable $\tx$ is replaced by the two related values, and every other
  rule is a congruence; the clauses of the ascription rules and of the
  \src{let} rule that mention the types of the subterms transport because
  the types of a substituted term and of the original one coincide
  (Lemma~\ref{subtyp}), and the bodies of a \src{let} are handled under the
  extended environments. The hypothesis on errors is needed, since substituting
  $\errort{\gtys}$ for $\tx$ in $\tx$ gives an error term, which is below
  no value.}
\end{proof}

\begin{lemma}[{\rev[fo1]{Weakening preserves precision}}]\label{lemma:prec-weakening}\mechnote{\leanname{TPLC.prec\_rename\_val} and \leanname{TPLC.prec\_rename\_tm}.}~
  \rev[fo1]{For any environments $\Gamma, \Gamma'$, any variable
  $\tx \in \text{Var}$ that does not occur in $\Gamma$, any terms
  $\tm, \tn$ of \tlang{} and any formula simple types
  $\gtys, \gtysp \in \FGType$, if $\Gamma \gprec \Gamma' |- \tm \gprec \tn$
  and $\gtys \gprec \gtysp$, then
  $\Gamma, \tx{:}\gtys \gprec \Gamma', \tx{:}\gtysp |- \tm \gprec \tn$;
  likewise for values.}
\end{lemma}
\begin{proof}
  \rev[fo1]{By induction on the derivation; the clauses that mention the
  types of the subterms transport by weakening of typing
  (Lemma~\ref{lemma:target-weakening}).}
\end{proof}

\mono*
\begin{proof}
  \rev[fo1]{By induction on the evidences, the two items together. For
  definedness, $\ev[1] \trans{} \ev[3]$ defined means $\ev[1] \erel \ev[3]$
  (Lemma~\ref{consistency-defined}), and consistency transports along
  precision by the runtime form of Lemma~\ref{conp}, so $\ev[2] \erel \ev[4]$
  and $\ev[2] \trans{} \ev[4]$ is defined. For the precision of the results,
  on simple types, if $\ev[2]$ or $\ev[4]$ is $\?$ the meet absorbs it and
  the claim is the reductivity of the meet (Lemma~\ref{meet-more-precise2})
  composed with the remaining hypothesis; otherwise the four types share
  their top constructor, the base cases are reflexive and the arrow case is
  componentwise. On distribution types, let $W$ be a solution of the formula
  of $\evd[1] \trans{} \evd[3]$, a coupling between solutions $p$ of
  $\evd[1]$ and $q$ of $\evd[3]$ with positive weight only between
  consistent entries, and let $u$
  and $v$ be the couplings that $\evd[1] \eprec \evd[2]$ at $p$ and
  $\evd[3] \eprec \evd[4]$ at $q$ provide, onto solutions $p'$ and $q'$. The
  composition (Lemma~\ref{lemma:coupling-composition}) of the transpose of
  $u$, of $W$ and of $v$ is a coupling between $p'$ and $q'$ whose positive
  weights come from a positive weight $W(\kk[1], \kk[3])$ through positive
  weights $u(\kk[1], \kk[2])$ and $v(\kk[3], \kk[4])$; by the simple-type
  statement, the entries $\kk[2]$ and $\kk[4]$ are consistent and their meet
  is above the meet carried by the entry $(\kk[1], \kk[3])$ of
  $\evd[1] \trans{} \evd[3]$. The composition is therefore a solution of the
  formula of $\evd[2] \trans{} \evd[4]$, and placing the weight
  $W(\kk[1], \kk[3])$ on the pairs $(\kk[2], \kk[4])$ it was routed to gives
  the coupling
  of $\evd[1] \trans{} \evd[3] \eprec \evd[2] \trans{} \evd[4]$.}
\end{proof}

\begin{lemma}[Reordering consistency]\label{reorder-consistency}\mechnote{\leanname{TPLC.econsTy\_of\_ereordTy} and \leanname{TPLC.econsD\_of\_ereordD}.}~
  \rev[fo1]{For any well-formed formula simple types
  $\gtys, \ggtys \in \FGType$ and any well-formed formula distribution types
  $\gtya, \gtyb \in \FGDType$, if $\reoder{\gtys}{\ggtys}$, then
  $\gtys \erel \ggtys$, and if $\reoder{\gtya}{\gtyb}$, then
  $\gtya \erel \gtyb$.}
\end{lemma}
\begin{proof}
  \rev[fo1]{By Lemma~\ref{reorderevidence2}, $\initReorder{\gtya}{\gtyb}$
  is defined, well-formed and $\eprec$ both types, and a well-formed common
  lower bound witnesses consistency by item~3 of
  Lemma~\ref{meet-more-precise2}.}
\end{proof}

\begin{lemma}[{\rev[fo1]{Tag-aware monotonicity of evidence combination}}]\label{lemma:tag-mono-trans}\mechnote{\leanname{TPLC.tagPrecTy\_emeetTy} and \leanname{TPLC.tagPrecD\_emeetD}; the definedness of the less precise composition, \leanname{TPLC.meetTy\_mono} and \leanname{TPLC.meetD\_mono}.}~
  \rev[fo1]{For any well-formed evidences $\ev[1], \ev[2], \evp[1], \evp[2]$
  and any formula simple types $\gtys, \gtysp \in \FGType$, if
  $\gtys \gprec \gtysp$ and $\ev[1] \trans{} \ev[2]$ is defined, then
  \begin{enumerate}
  \item if $\vjtfr{\ev[2]}{\gtys}$, $\vjtfr{\evp[2]}{\gtysp}$,
    $\ev[1] \eprec \evp[1]$ and $\tagprecr{\gtys}{\gtysp}{\ev[2]}{\evp[2]}$,
    then $\evp[1] \trans{} \evp[2]$ is defined and
    $\tagprecr{\gtys}{\gtysp}{\ev[1] \trans{} \ev[2]}{\evp[1] \trans{} \evp[2]}$;
  \item if $\vjtfl{\ev[1]}{\gtys}$, $\vjtfl{\evp[1]}{\gtysp}$,
    $\tagprecl{\gtys}{\gtysp}{\ev[1]}{\evp[1]}$ and $\ev[2] \eprec \evp[2]$,
    then $\evp[1] \trans{} \evp[2]$ is defined and
    $\tagprecl{\gtys}{\gtysp}{\ev[1] \trans{} \ev[2]}{\evp[1] \trans{} \evp[2]}$;
  \end{enumerate}
  and likewise for distribution evidences, with formula distribution types
  in place of $\gtys$ and $\gtysp$.}
\end{lemma}
\begin{proof}
  \rev[fo1]{Definedness and the runtime precision of the two compositions
  are Lemma~\ref{lemma:mono}, since the tag-aware precisions imply $\eprec$.
  When the judged types or the evidences are not all arrows, rule ($\pi$B)
  asks for nothing more. When they are, the claim is componentwise, and at
  the level of distributions the coupling built in the proof of
  Lemma~\ref{lemma:mono} is used. Its positive weights come from positive
  weights of the two precision couplings, and since the right pointer of an
  entry of a composition is that of its second operand, the entries that the
  right pointers of two coupled entries name are related by the hypothesis on
  $\ev[2]$ and $\evp[2]$; the left side is symmetric with the first
  operand. The evidence of the other operand only needs $\eprec$ because the
  composition takes no pointer from it.}
\end{proof}

\begin{lemma}[{\rev[fo1]{Tag-aware precision and the meet}}]\label{lemma:tag-meet}\mechnote{item~1, \leanname{TPLC.tagPrecTy\_emeet\_compose} and \leanname{TPLC.tagPrecD\_emeet\_compose}; item~2, \leanname{TPLC.tagPrecTy\_glb} and \leanname{TPLC.tagPrecD\_glb}; item~3, \leanname{TPLC.tagPrecTy\_reorderTy\_glb}.}~
  \rev[fo1]{For any well-formed evidences $\ev, \ev[1], \ev[2]$, any
  well-formed distribution evidences $\evd[1], \evd[2], \evdp$, any
  well-formed formula simple types $\gtys, \gtysp \in \FGType$ and any
  well-formed formula distribution types
  $\gtya, \gtyb, \gtya[0], \gtyb[0] \in \FGDType$,
  \begin{enumerate}
  \item if $\tagprecr{\gtya}{\gtyb}{\evd[2]}{\evdp}$ and
    $\evd[1] \trans{} \evd[2]$ is defined, then
    $\tagprecr{\gtya}{\gtyb}{\evd[1] \trans{} \evd[2]}{\evdp}$; likewise for
    simple evidences;
  \item if $\ev \eprec \ev[1]$, $\tagprecr{\gtys}{\gtysp}{\ev}{\ev[2]}$,
    $\vjtfr{\ev[2]}{\gtysp}$ and $\ev[1] \trans{} \ev[2]$ is defined, then
    $\tagprecr{\gtys}{\gtysp}{\ev}{\ev[1] \trans{} \ev[2]}$; likewise for
    distribution evidences;
  \item if $\gtys \gprec \gtysp$, then the entry of $\initReorder{\gtya}{\gtya[0]}$
    at a pair of entries equal to $\gtys$ and the entry of
    $\initReorder{\gtyb}{\gtyb[0]}$ at a pair equal to $\gtysp$ satisfy
    $\tagprecr{\gtys}{\gtysp}{\gtys}{\gtysp}$, the entries being those types
    with identity pointers.
  \end{enumerate}}
\end{lemma}
\begin{proof}
  \rev[fo1]{Items~1 and~2 hold because the right pointer of an entry of a
  composition is that of its second operand and the judgment
  $\vdash_{\rightvar}$ reads only that pointer. Outside the arrow case they
  are the reductivity and the greatest-lower-bound property of the meet
  (Lemma~\ref{meet-more-precise2}) for $\eprec$, and in the arrow case they
  are componentwise, the distribution level composing the couplings of the
  hypotheses along the right pointers. Item~3 is by induction on
  $\gtys \gprec \gtysp$, the entry being the type itself with identity
  pointers at every depth.}
\end{proof}

\begin{lemma}[{\rev[fo1]{Greatest lower bound of reordering}}]\label{lemma:reorder-glb}\mechnote{\leanname{TPLC.eprec\_reorderD\_glb} and \leanname{TPLC.eprec\_reorderD\_glb\_tags}; the instance for the proof of Theorem~\ref{theorem:dgg}, \leanname{TPLC.eprec\_reorderD\_glb\_align}.}~
  \rev[fo1]{For any well-formed formula distribution type
  $\gtya[1] \in \FGDType$ and any formula distribution types
  $\gtya[2], \gtya[3] \in \FGDType$, if $\gtya[3] \eprec \gtya[1]$
  and $\gtya[1]$ refines $\gtya[2]$, then $\initReorder{\gtya[1]}{\gtya[2]}$
  is defined and $\gtya[3] \eprec \initReorder{\gtya[1]}{\gtya[2]}$.
  Moreover, for every solution $x$, the coupling that witnesses
  $\gtya[3] \eprec \initReorder{\gtya[1]}{\gtya[2]}$ at $x$ can be chosen so
  that, for every entry $i$ of $\gtya[3]$ and every entry $\kk$ of the
  reordering to which it gives positive weight, the coupling that witnesses
  $\gtya[3] \eprec \gtya[1]$ at $x$ gives positive weight to $i$ and
  $\projl{\cww[\kk]}$.}
\end{lemma}
\begin{proof}
  \rev[fo1]{Let $u$ be the coupling of $\gtya[3] \eprec \gtya[1]$ at $x$,
  onto a solution $p$ of $\gtya[1]$, and $v$ the coupling that the
  refinement gives at $p$, onto a solution $q$ of $\gtya[2]$. Since $v$ is
  supported on pairs of equal entries, it is a solution of the formula of
  the reordering, which is therefore defined. The gluing
  $T(i, (a, b)) = u(i, a) \, v(a, b) / p(a)$ of $u$ and $v$
  (Lemma~\ref{lemma:coupling-composition}, as in item~3 of
  Lemma~\ref{meet-more-precise2}; an entry $a$ with $p(a) = 0$ contributes $0$)
  couples $x$ with $v$, and a positive weight $T(i, (a,b))$ has $u(i,a) > 0$,
  whence the entry $i$ of $\gtya[3]$ is $\eprec$ the entry $a$ of
  $\gtya[1]$, which is the type carried by the entry $(a,b)$ of the
  reordering.}
\end{proof}

\rev[fo1]{A value that is more precise than an ascribed value is itself an
ascribed value, so the next lemma states its hypotheses on two ascribed
values.}

\begin{lemma}[{\rev[fo1]{Coercion of related values}}]\label{lemma:cell-simulation}\mechnote{\leanname{TPLC.dgg\_cell\_ef}.}~
  \rev[fo1]{For any evidences $\ev[1], \evp[1]$, any raw values $\tu, \tup$,
  any formula simple types $\gtys[v], \gtys[v][\prime], \gtys[t],
  \gtys[t][\prime] \in \FGType$ and any well-formed evidences
  $\ev[2], \evp[2]$, if the values
  $\tv = \asc{\ev[1] \tu}{\gtys[v]}$ and
  $\tvp = \asc{\evp[1] \tup}{\gtys[v][\prime]}$ are closed and well typed,
  $\tv \gprec \tvp$, $\gtys[t] \gprec \gtys[t][\prime]$,
  $\vjtfr{\ev[2]}{\gtys[t]}$,
  $\vjtf{\evp[2]}{\gtys[v][\prime] \erel \gtys[t][\prime]}$ and
  $\tagprecr{\gtys[t]}{\gtys[t][\prime]}{\ev[2]}{\evp[2]}$, then, for every
  value $\tw$ such that $\asc{\ev[2] \tv}{\gtys[t]}$ reduces to
  $\dt{\pt{\tw}}$ by the first case of rule
  $\trg{(\mathit{D}\mathord{::}\gtys)}$, the term
  $\asc{\evp[2] \tvp}{\gtys[t][\prime]}$ reduces to $\dt{\pt{\twp}}$ by the
  first case of that rule as well, for some value $\twp$ with
  $\tw \gprec \twp$.}
\end{lemma}
\begin{proof}
  \rev[fo1]{The derivation of $\tv \gprec \tvp$ ends in the rule for
  ascribed values, whose premises give $\tu \gprec \tup$,
  $\gtys[v] \gprec \gtys[v][\prime]$, $\tagprecr{\gtys[v]}{\gtys[v][\prime]}{\ev[1]}{\evp[1]}$
  and $\tagprecl{\gtys[u]}{\gtys[u][\prime]}{\ev[1]}{\evp[1]}$ for the types
  $\gtys[u]$ and $\gtys[u][\prime]$ of $\tu$ and $\tup$. Both tag-aware precisions
  imply $\ev[1] \eprec \evp[1]$ and $\ev[2] \eprec \evp[2]$, so by
  Lemma~\ref{lemma:mono} the composition $\evp[1] \trans{} \evp[2]$ is
  defined, and the less precise ascription takes the first case of
  $\trg{(\mathit{D}\mathord{::}\gtys)}$, with
  $\twp = \asc{\evp[1] \trans{} \evp[2]\;\tup}{\gtys[t][\prime]}$. The two results
  are related by the rule for ascribed values. The clause
  $\tagprecr{\gtys[t]}{\gtys[t][\prime]}{\ev[1] \trans{} \ev[2]}{\evp[1] \trans{} \evp[2]}$
  is item~1 of Lemma~\ref{lemma:tag-mono-trans}, the clause
  $\tagprecl{\gtys[u]}{\gtys[u][\prime]}{\ev[1] \trans{} \ev[2]}{\evp[1] \trans{} \evp[2]}$
  is its item~2, with the left validities of $\ev[1]$ and $\evp[1]$ given by
  the typing of $\tv$ and $\tvp$, and $\tu \gprec \tup$ and
  $\gtys[t] \gprec \gtys[t][\prime]$ are given.}
\end{proof}

\rev[fo1]{Item~3 of the next lemma covers the weighted sums that rule
$\trg{(\mathit{Dlet})}$ builds; the result of
$\trg{(\mathit{D}\mathord{::}\gtya)}$ is the case of Dirac branches.}

\begin{lemma}[{\rev[fo1]{Precision of configurations is a congruence}}]\label{lemma:confprec-congruence}\mechnote{item~1, \leanname{TPLC.dconfprec\_point}; item~2, \leanname{TPLC.dconfprec\_choose\_cond}, \leanname{TPLC.dconfprec\_choose\_unk\_cond} and \leanname{TPLC.dconfprec\_chooseU}; item~3, \leanname{TPLC.dconfprec\_wsum}; the instance for the proof of Theorem~\ref{theorem:dgg}, \leanname{TPLC.dgg\_cells\_assemble}.}~
  \rev[fo1]{For any values $\tv, \tvp$ of \tlang{}, any configurations
  $\rctx[\phi[][1]] \V[1]$, $\rctx[\phi[\prime][1]] \V[1]'$,
  $\rctx[\phi[][2]] \V[2]$ and $\rctx[\phi[\prime][2]] \V[2]'$, any gradual
  probabilities $\pty, \ptyp \in \GProb$, any finite index sets $\kSet, \kSet'$,
  any formulas $\psi$ and $\psi'$ over $(\cww[\kk])_{\kk \in \kSet}$
  and $(\cww[\kk'])_{\kk' \in \kSet'}$, and any families
  $(\rctx[\phi[][\kk]] \V[\kk])_{\kk \in \kSet}$ and
  $(\rctx[\phi[\prime][\kk']] \V[\kk']')_{\kk' \in \kSet'}$ of configurations,
  \begin{enumerate}
  \item if $\tv \gprec \tvp$, then $\rctx[\cdot] \dt{\pt{\tv}} \gprec \rctx[\cdot] \dt{\pt{\tvp}}$;
  \item if the runtime types of the four configurations are well-formed,
    $\pty \gprec \ptyp$,
    $\rctx[\phi[][1]] \V[1] \gprec \rctx[\phi[\prime][1]] \V[1]'$ whenever
    $\crp{\pty}$ contains a positive probability, and
    $\rctx[\phi[][2]] \V[2] \gprec \rctx[\phi[\prime][2]] \V[2]'$ whenever
    $\crp{\pty}$ contains a probability below $1$, then the configuration
    that rule $\trg{(\mathit{D}\oplus)}$ builds from $\rctx[\phi[][1]] \V[1]$
    and $\rctx[\phi[][2]] \V[2]$ with probability $\pty$ is $\gprec$ the one
    it builds from $\rctx[\phi[\prime][1]] \V[1]'$ and
    $\rctx[\phi[\prime][2]] \V[2]'$ with probability $\ptyp$;
  \item if the branch formulas $\phi[\prime][\kk']$ have non-empty convex
    solution sets and, for every solution $\omega$ of $\psi$, there are a
    solution $\omega'$ of $\psi'$ and a coupling $t$ between them such that
    $t(\kk, \kk') > 0$ implies
    $\rctx[\phi[][\kk]] \V[\kk] \gprec \rctx[\phi[\prime][\kk']] \V[\kk']'$,
    then
    $\rctx[\psi \land \bigwedge_{\kk} (\cww[\kk] > 0 \Rightarrow \phi[][\kk])] \ssum[\kk \in \kSet] \cww[\kk] \cdot \V[\kk]
    \gprec
    \rctx[\psi' \land \bigwedge_{\kk'} (\cww[\kk'] > 0 \Rightarrow \phi[\prime][\kk'])] \ssum[\kk' \in \kSet'] \cww[\kk'] \cdot \V[\kk']'$.
  \end{enumerate}}
\end{lemma}
\begin{proof}
  \rev[fo1]{Item~1 is the coupling with the single weight $1$. For item~2, a solution of the
  more precise configuration is $(a \cdot p, (1-a) \cdot q)$ with
  $a \in \crp{\pty} \subseteq \crp{\ptyp}$; the hypotheses give couplings
  $u$ and $v$ at $p$ and $q$, and the matrix with blocks $a \cdot u$ and
  $(1-a) \cdot v$ couples it with $(a \cdot p', (1-a) \cdot q')$, a
  solution of the less precise configuration with the same $a$; when
  $a = 0$ or $a = 1$ the hypothesis on the pair of weight zero may be
  missing, and $p'$ or $q'$ is then any solution of its formula, which
  exists because its runtime type is well-formed, with a zero block. For item~3,
  a solution of the first sum is $\omega$ with a solution $b_{\kk}$ of
  $\phi[][\kk]$ for each positive $\kk$; take $\omega'$ and $t$. For each
  $\kk'$ with $\omega'_{\kk'} > 0$ and each $\kk$ with $t(\kk, \kk') > 0$, the
  hypothesis at $b_{\kk}$ gives a solution $r_{\kk\kk'}$ of $\phi[\prime][\kk']$
  and a coupling $u^{\kk\kk'}$; the solution of $\phi[\prime][\kk']$ is the
  convex mixture $\sum_{\kk} (t(\kk,\kk') / \omega'_{\kk'}) \cdot r_{\kk\kk'}$,
  which is a solution by convexity, and any solution when
  $\omega'_{\kk'} = 0$. The coupling places $t(\kk, \kk') \cdot u^{\kk\kk'}$
  in block $(\kk, \kk')$. Its row sums are $\omega_{\kk} \cdot b_{\kk}$, its
  column sums are $\omega'_{\kk'}$ times the mixture, and its positive weights
  lie inside a block, hence pair related values.}
\end{proof}

\begin{theoremproof}{Theorem~\ref{theorem:dgg} (Dynamic gradual guarantee for \tlang)}
  \begin{revblock}[fo1]
  We proceed by induction on the derivation of
  $\tm \nreds{}{k_1} \rctx[\phi[][1]] \dset[1]$, which raises no error, for
  all $\tn$, $\gtya$, $\gtyb$ and all reductions
  $\tn \nreds{}{k_2} \rctx[\phi[][2]] \dset[2]$ at once. Throughout, on both sides, the result of the term that a routed rule
  reduces first is typed by Lemma~\ref{lemma:tplc-preservation}. Its values
  are typed at the types they display, and its runtime type is well-formed
  and refines the static type of that term.

  \begin{case}[values, conditional and monotonicity]
    For $\trg{(\mathit{Dv})}$, $\tm$ is an ascribed value, $\tn$ is a
    related value, the reduction of $\tn$ ends in $\trg{(\mathit{Dv})}$, and
    item~1 of Lemma~\ref{lemma:confprec-congruence} concludes. For
    $\trg{(\mathit{Dit})}$, the guard of $\tm$ is $\asc{\ev \ttt}{\btype}$;
    the rule for ascribed values gives the guard of $\tn$ as
    $\asc{\evp \ttt}{\gtysp}$ with the same raw value, and (Gif) forces
    $\gtysp = \btype$, so the reduction of $\tn$ ends in
    $\trg{(\mathit{Dit})}$ as well, and the induction hypothesis on the
    branches concludes; $\trg{(\mathit{Dif})}$ is symmetric.
    $\trg{(\mathit{Dmon})}$ is the induction hypothesis.
  \end{case}

  \begin{case}[probabilistic choice]
    $\tm = \tm[1] \tpsum \tm[2]$ and $\tn = \tn[1] \tpsum[\ptyp] \tn[2]$
    with $\pty \gprec \ptyp$ and $\tm[i] \gprec \tn[i]$, by the rule
    $\trg{(\mathord{\gprec}\oplus)}$. The reduction of $\tn$ ends in
    $\trg{(\mathit{D}\oplus)}$; the induction hypothesis relates the
    results of the branches, their runtime types are well-formed by
    Lemma~\ref{lemma:tplc-preservation}, and item~2 of
    Lemma~\ref{lemma:confprec-congruence} concludes.
  \end{case}

  \begin{case}[addition]
    $\tm = \add{\asc{\ev[1] r_1}{\rtype}}{\asc{\ev[2] r_2}{\rtype}}$. By
    (G$+$), (G$\tu\mathord{::}\gtys$) and the rules of right validity for
    $\rtype$, $\ev[1] = \ev[2] = \rtype$, so $\ev[3] = \rtype$; the rule for
    ascribed values gives $\tn = \add{\asc{\evp[1] r_1}{\gtys[1][\prime]}}{\asc{\evp[2] r_2}{\gtys[2][\prime]}}$
    with the same reals, and the typing of $\tn$ forces
    $\gtys[i][\prime] = \rtype$ and $\evp[i] = \rtype$ in the same way. The
    reduction of $\tn$ ends in $\trg{(\mathit{D+})}$, both results are
    $\dt{\pt{\asc{\rtype\,(r_1 + r_2)}{\rtype}}}$, and they are related by
    reflexivity.
  \end{case}

  \begin{case}[value ascription]
    $\tm = \asc{\ev[2] (\asc{\ev[1] \tu}{\gtys[v]})}{\gtys[t]}$ in the first
    case of $\trg{(\mathit{D}\mathord{::}\gtys)}$. The rule for value
    ascriptions gives $\tn = \asc{\evp[2] \tvp}{\gtys[t][\prime]}$ with
    $\asc{\ev[1] \tu}{\gtys[v]} \gprec \tvp$, $\gtys[t] \gprec \gtys[t][\prime]$ and
    $\tagprecr{\gtys[t]}{\gtys[t][\prime]}{\ev[2]}{\evp[2]}$; the typing of $\tm$
    and of $\tn$ gives the validities that
    Lemma~\ref{lemma:cell-simulation} requires, and that lemma gives that
    the reduction of $\tn$ ends in the first case of the rule, with related
    results; item~1 of Lemma~\ref{lemma:confprec-congruence} concludes.
  \end{case}

  \begin{case}[application]
    $\tm = (\asc{\ev (\lambda \tx{:}\gtysp. \tm[b])}{\gtys[a] -> \gtya[r]})\; \tv$
    with $\ev = \ev[s] -> \evd[d]$, the coercion
    $\asc{\cdom(\ev) \tv}{\gtysp}$ reducing to $\dt{\pt{\tw}}$ by the first
    case of $\trg{(\mathit{D}\mathord{::}\gtys)}$, so that $\tw$ is an
    ascribed value, and the body
    $\asc{\ccod(\ev) (\tm[b][\tw/\tx])}{\gtya[r]}$ reducing to
    $\rctx[\phi[][1]] \dset[1]$. The rule for applications and the rule for
    ascribed values give $\tn = (\asc{\evp (\lambda \tx{:}\gtysp[][\prime]. \tn[b])}{\gtys[f][\prime]})\; \tvp$
    with $\tv \gprec \tvp$, $\gtysp \gprec \gtysp[][\prime]$,
    $\Gamma, \tx{:}\gtysp \gprec \Gamma', \tx{:}\gtysp[][\prime] |- \tm[b] \gprec \tn[b]$,
    $\tagprecr{\gtys[a] -> \gtya[r]}{\gtys[f][\prime]}{\ev}{\evp}$ and
    $\tagprecl{\gtysp -> \gtya[b]}{\gtysp[][\prime] -> \gtya[b][\prime]}{\ev}{\evp}$ for the
    types $\gtya[b]$ and $\gtya[b][\prime]$ of the two bodies. The typing of $\tn$
    by (Gapp) makes $\gtys[f][\prime]$ an arrow $\gtys[a][\prime] -> \gtya[r][\prime]$, so
    $\gtys[a] \gprec \gtys[a][\prime]$ and $\gtya[r] \gprec \gtya[r][\prime]$, and the left
    validity of $\evp$ against the arrow type of the $\lambda$ makes
    $\evp = \evp[s] -> \evdp[d]$; rule ($\pi{->}$) then gives
    $\tagprecr{\gtys[a]}{\gtys[a][\prime]}{\ev[s]}{\evp[s]}$,
    $\tagprecr{\gtya[r]}{\gtya[r][\prime]}{\evd[d]}{\evdp[d]}$,
    $\tagprecl{\gtysp}{\gtysp[][\prime]}{\ev[s]}{\evp[s]}$ and
    $\tagprecl{\gtya[b]}{\gtya[b][\prime]}{\evd[d]}{\evdp[d]}$. The reduction of
    $\tn$ ends in $\trg{(\mathit{Dapp})}$, since its function is an
    ascription and not an error. The coercions of the arguments use
    $\cdom(\ev)$ and $\cdom(\evp)$, which exchange the pointers of $\ev[s]$
    and $\evp[s]$, so the clause $\tagprecl{\gtysp}{\gtysp[][\prime]}{\ev[s]}{\evp[s]}$
    reads as $\tagprecr{\gtysp}{\gtysp[][\prime]}{\cdom(\ev)}{\cdom(\evp)}$, and
    Lemma~\ref{lemma:cell-simulation} gives $\tw \gprec \twp$, both ascribed
    values, typed at $\gtysp$ and $\gtysp[][\prime]$ by
    Lemma~\ref{lemma:value-coercion}. The two bodies are related by the rule
    for distribution ascriptions, since
    $\tagprecr{\gtya[r]}{\gtya[r][\prime]}{\evd[d]}{\evdp[d]}$ and
    $\tagprecl{\gtya[b]}{\gtya[b][\prime]}{\evd[d]}{\evdp[d]}$ are the two clauses
    above, the types of the substituted bodies being $\gtya[b]$ and
    $\gtya[b][\prime]$ by Lemmas~\ref{subtyp} and~\ref{lemma:target-det}, and
    $\tm[b][\tw/\tx] \gprec \tn[b][\twp/\tx]$ by Lemma~\ref{subep}. The
    induction hypothesis on the reduction of the body concludes.
  \end{case}

  \begin{case}[distribution ascription]
    $\tm = \asc{\evd \tm[0]}{\gtyb[t]}$ with
    $\gtyb[t] = \phty{\phi[][b]}{\d{\gtysp[j][\p[j]]}[j \in \jSet]}$, in the
    first case of $\trg{(\mathit{D}\mathord{::}\gtya)}$. Here $\tm[0]$ reduces to
    $\rctx[\phi[][0]] \V[0]$ with runtime type $\gtya[0]$, the routing
    evidence $E = (\initReorder{\gtya[0]}{\gtya[s]}) \trans{} \evd$ with
    $\cdot |-d \tm[0] : \gtya[s]$ is defined, each entry $\kk$ coerces
    $\tv[\projl{\cww[\kk]}]$ to $\gtysp[\projr{\cww[\kk]}]$ by the first case
    of $\trg{(\mathit{D}\mathord{::}\gtys)}$, and the result is
    $\rctx[\phi[][E]] \d{\tw[\kk][\cww[\kk]]}[\kk]$. The rule for
    distribution ascriptions gives $\tn = \asc{\evdp \tn[0]}{\gtyb[t][\prime]}$ with
    $\tm[0] \gprec \tn[0]$ and $\tagprecr{\gtyb[t]}{\gtyb[t][\prime]}{\evd}{\evdp}$;
    the typing of $\tn$ gives $\cdot |-d \tn[0] : \gtya[s][\prime]$ and
    $\vjtf{\evdp}{\gtya[s][\prime] \erel \gtyb[t][\prime]}$. The reduction of $\tn$ ends in
    $\trg{(\mathit{D}\mathord{::}\gtya)}$, with $\tn[0]$ reducing to
    $\rctx[\phi[\prime][0]] \V[0]'$ of runtime type $\gtya[0][\prime]$ and routing
    evidence $E' = (\initReorder{\gtya[0][\prime]}{\gtya[s][\prime]}) \trans{} \evdp$. By
    the induction hypothesis, $\rctx[\phi[][0]] \V[0] \gprec \rctx[\phi[\prime][0]] \V[0]'$.
    The proof has four steps.

    First, $E \eprec E'$. Related values display related types, by the
    premise $\gtys \gprec \gtysp$ of the rule for ascribed values, so the
    coupling of the induction hypothesis read on the displayed types gives
    $\gtya[0] \eprec \gtya[0][\prime]$. By the reductivity of the meet and of the
    reordering (Lemmas~\ref{meet-more-precise2} and~\ref{reorderevidence2}),
    $E \eprec \initReorder{\gtya[0]}{\gtya[s]} \eprec \gtya[0]$, hence
    $E \eprec \gtya[0][\prime]$ by transitivity, and since $\gtya[0][\prime]$ refines
    $\gtya[s][\prime]$, Lemma~\ref{lemma:reorder-glb} gives
    $E \eprec \initReorder{\gtya[0][\prime]}{\gtya[s][\prime]}$; also $E \eprec \evd \eprec \evdp$,
    and item~3 of Lemma~\ref{meet-more-precise2} gives $E \eprec E'$.

    Second, the reduction of $\tn$ takes the first case. Indeed, $E$ is defined, so
    its formula has a solution, which $E \eprec E'$ transports to a solution
    of the formula of $E'$.

    Third, the coupling between the entries of $E$ and of $E'$ is chosen so
    that coupled entries route related values to related entries of the
    target types with related evidences. For a solution $\omega$ of the formula of $E$, its
    push-forward along the left pointers is a solution of $\phi[][0]$, by
    the marginal clauses, and composes with the coupling of the induction
    hypothesis into a coupling from the entries of $E$ to the values of
    $\V[0]'$, supported on pairs where the routed value $\tv[\projl{\cww[\kk]}]$
    is below the value; Lemma~\ref{lemma:reorder-glb}, in its second part,
    turns it into a coupling from the entries of $E$ to those of
    $\initReorder{\gtya[0][\prime]}{\gtya[s][\prime]}$ with the same support property on
    the routed values. On the other side, item~1 of
    Lemma~\ref{lemma:tag-meet} turns $\tagprecr{\gtyb[t]}{\gtyb[t][\prime]}{\evd}{\evdp}$
    into $\tagprecr{\gtyb[t]}{\gtyb[t][\prime]}{E}{\evdp}$, a coupling from the entries
    of $E$ to those of $\evdp$ whose support has related target entries and
    related evidences along the right pointers. The construction of item~3
    of Lemma~\ref{meet-more-precise2} combines the two into a coupling $t$
    from the entries of $E$ to those of $E'$ in which a positive pair
    $(\kk, \kk')$ has related routed values, related target entries
    $\gtysp[\projr{\cww[\kk]}] \gprec \gtysp[\projr{\cww[\kk']}][\prime]$, and, by
    item~2 of Lemma~\ref{lemma:tag-meet},
    $\tagprecr{\gtysp[\projr{\cww[\kk]}]}{\gtysp[\projr{\cww[\kk']}][\prime]}{\ev[\kk]}{\evp[\kk']}$.

    Fourth, for every positive pair $(\kk, \kk')$,
    Lemma~\ref{lemma:cell-simulation} gives $\tw[\kk] \gprec \twp[\kk']$,
    the validities of $\ev[\kk]$ and $\evp[\kk']$ coming from
    Lemma~\ref{lemma:routing-validity}; since the branches are Diracs,
    item~3 of Lemma~\ref{lemma:confprec-congruence} with the coupling $t$
    concludes.
  \end{case}

  \begin{case}[let]
    $\tm = \lett{\tx}{\tm[0]}{\overline{\tn[j]}^{\,j \in \jSet}}$ with
    $\cdot |-d \tm[0] : \gtya[s] = \phty{\phi[][s]}{\d{\gtysp[j][\p[j]]}[j \in \jSet]}$.
    Here $\tm[0]$ reduces to $\rctx[\phi[][0]] \V[0]$ of runtime type
    $\gtya[0]$, $R = \initReorder{\gtya[0]}{\gtya[s]}$, each entry $\kk$
    coerces $\tv[\projl{\cww[\kk]}]$ to $\gtysp[\projr{\cww[\kk]}]$ by the
    first case of $\trg{(\mathit{D}\mathord{::}\gtys)}$, so that $\tw[\kk]$ is
    an ascribed value, and $\tn[\projr{\cww[\kk]}][\tw[\kk]/\tx]$ reduces to
    $\rctx[\phi[][\kk]] \V[\kk]$. The rule for \src{let} gives
    $\tn = \lett{\tx}{\tn[0]}{\overline{\tnp[j']}^{\,j' \in \jSet'}}$ with
    $\tm[0] \gprec \tn[0]$ and, for every pair of entries
    $\gtysp[j] \gprec \gtysp[j'][\prime]$ of $\gtya[s]$ and of the type $\gtya[s][\prime]$
    of $\tn[0]$, $\Gamma, \tx{:}\gtysp[j] \gprec \Gamma', \tx{:}\gtysp[j'][\prime] |- \tn[j] \gprec \tnp[j']$.
    The reduction of $\tn$ ends in $\trg{(\mathit{Dlet})}$. Here $\tn[0]$ reduces
    to $\rctx[\phi[\prime][0]] \V[0]'$ of runtime type $\gtya[0][\prime]$,
    $R' = \initReorder{\gtya[0][\prime]}{\gtya[s][\prime]}$, and each entry $\kk'$ gives
    $\twp[\kk']$ and a reduction of $\ssub{\tn[\projr{\cww[\kk']}]'}{\twp[\kk']}{\tx}$
    to $\rctx[\phi[\prime][\kk']] \V[\kk']'$. The induction hypothesis on the
    bound term gives $\rctx[\phi[][0]] \V[0] \gprec \rctx[\phi[\prime][0]] \V[0]'$.

    First, we build the coupling between the entries of $R$ and of $R'$. As in the
    previous case, for every solution $\omega$ of the formula of $R$,
    Lemma~\ref{lemma:reorder-glb} gives a solution $\omega'$ of the formula
    of $R'$ and a coupling $t$ in which a positive pair $(\kk, \kk')$ routes
    related values; an entry of a reordering pairs equal entries, so the two
    target entries are the displayed types of the two routed values, which
    are related, and item~3 of Lemma~\ref{lemma:tag-meet} relates the
    evidences that the two entries carry along the right pointers.

    Second, for every positive pair $(\kk, \kk')$, with
    $j = \projr{\cww[\kk]}$ and $j' = \projr{\cww[\kk']}$,
    Lemma~\ref{lemma:cell-simulation} gives $\tw[\kk] \gprec \twp[\kk']$,
    both ascribed values, so $\mathit{sub}$ is the substitution on both
    sides, and they are typed at $\gtysp[j]$ and $\gtysp[j'][\prime]$ by
    Lemma~\ref{lemma:value-coercion}; the premise of the \src{let} rule at
    the related entries $\gtysp[j] \gprec \gtysp[j'][\prime]$ relates the bodies,
    Lemma~\ref{subep} relates the substituted bodies, both typed by
    Lemma~\ref{subtyp}, and the induction hypothesis on the reduction of
    $\tn[j]{}[\tw[\kk]/\tx]$, applied to the reduction of
    $\tnp[j'][\twp[\kk']/\tx]$, gives
    $\rctx[\phi[][\kk]] \V[\kk] \gprec \rctx[\phi[\prime][\kk']] \V[\kk']'$.

    Third, item~3 of Lemma~\ref{lemma:confprec-congruence}, with the
    couplings $t$, relates the two weighted sums; its hypothesis on the
    branch formulas holds because the runtime types of the less precise
    branches are well-formed. \qedhere
  \end{case}
  \end{revblock}
\end{theoremproof}

\subsection{Dynamic Gradual Guarantee for \glang}\label{app:source-dgg}

\sourcegg*

\begin{revblock}[fo1]
\proofheading{Structure of the proof.}
The semantics of a source term is to elaborate it and reduce the
elaboration, so the theorem reduces to the dynamic gradual guarantee of
\tlang{} (Theorem~\ref{theorem:dgg}) once the two elaborations are related
by the term precision of Figure~\ref{fig:term-precision}. The proof
composes three results and has no induction of its own. Lemma~\ref{epp}
exhibits the elaboration of the less precise term and relates every
elaboration of the more precise one to it; the preservation half of
Lemma~\ref{theorem:ep} types both elaborations, which are closed, at the
source types; Theorem~\ref{theorem:dgg} then relates the two results.

The content is in Lemma~\ref{epp}, proved by induction on the typing
derivation of the less precise term, for a statement that quantifies over
every more precise term, environment and elaboration. The \src{let} case
needs that form. The elaboration turns the single source body into one body
per entry of the type of the bound term, and the body elaborated at an entry of the
less precise type is compared with the bodies elaborated at every entry of
the more precise type below it. The ascription rules of term precision
compare the evidences that the two elaborations insert, with the two
tag-aware judgments of Figure~\ref{fig:tag-precision}. Every such evidence
is the meet of two types that the elaboration judged, and the judged types
of the two sides are related, namely the lifted annotations by the precision of
the source terms, and the inferred types by the static gradual guarantee
(Lemma~\ref{lemma:elab-prec-types}). The two meets are then compared by a
form of the monotonicity of evidence combination (Lemma~\ref{lemma:mono})
that keeps track of the pointers of the entries
(Lemma~\ref{lemma:tag-meet-mono}). The typing clauses of the ascription
rules name the types that the subjects have in the environments of the
judgment, which are the elaboration types by the preservation half of
Lemma~\ref{theorem:ep} and Lemma~\ref{lemma:target-det}; the left reading of
Lemma~\ref{lemma:tag-meet-mono} is a statement about those types. The fresh
\src{let}s that the elaboration inserts have bound terms of single-entry type, for
which the \src{let} rule of precision has a single pair of entries and a
trivial coupling (Lemma~\ref{lemma:prec-letin1}). The complete proof is
given at the end of this subsection, after the lemmas it uses.
\end{revblock}

\begin{lemma}[{\rev[fo1]{Elaboration implies typing}}]\label{lemma:elab-typing}\mechnote{\leanname{TPLC.elab\_sound\_val} and \leanname{TPLC.elab\_sound\_tm}.}~
  \rev[fo1]{For any environment $\Gamma$, any value $\sv$ and any term $\sm$
  of \glang{}, any value $\tv$ and any term $\tm$ of \tlang{}, any formula
  simple type $\gtys \in \FGType$ and any formula distribution type
  $\gtya \in \FGDType$,
  \begin{enumerate}
  \item if $\ela{\Gamma |-t \sv}{\tv}{\gtys}$, then $\Gamma |-t \sv : \gtys$;
  \item if $\ela{\Gamma |-d \sm}{\tm}{\gtya}$, then $\Gamma |-d \sm : \gtya$.
  \end{enumerate}}
\end{lemma}
\begin{proof}
  \rev[fo1]{By induction on the elaboration derivation. Every rule of
  Figure~\ref{fig:elaboration} has the premises of the typing rule of
  Figure~\ref{fig:source-type-system} for the same constructor, with the
  same result type, and the equations that fix its evidences.}
\end{proof}

\begin{lemma}[{\rev[fo1]{Precision of elaborated types}}]\label{lemma:elab-prec-types}\mechnote{\leanname{TPLC.precTy\_elab} and \leanname{TPLC.precD\_elab}.}~
  \rev[fo1]{For any well-formed environments $\Gamma$ and $\Gamma'$, any
  values $\sv, \sw$ and any terms $\sm, \srcn$ of \glang{}, any values
  $\tv, \tw$ and any terms $\tm, \tn$ of \tlang{}, any formula simple types
  $\gtys, \gtysp \in \FGType$ and any formula distribution types
  $\gtya, \gtyb \in \FGDType$, if $\Gamma \gprec \Gamma'$, then
  \begin{enumerate}
  \item if $\sv \gprec \sw$, $\ela{\Gamma |-t \sv}{\tv}{\gtys}$ and
    $\ela{\Gamma' |-t \sw}{\tw}{\gtysp}$, then $\gtys \gprec \gtysp$;
  \item if $\sm \gprec \srcn$, $\ela{\Gamma |-d \sm}{\tm}{\gtya}$ and
    $\ela{\Gamma' |-d \srcn}{\tn}{\gtyb}$, then $\gtya \gprec \gtyb$.
  \end{enumerate}}
\end{lemma}
\begin{proof}
  \rev[fo1]{By Lemma~\ref{lemma:elab-typing}, $\Gamma |-d \sm : \gtya$ and
  $\Gamma' |-d \srcn : \gtyb$. The annotations of $\srcn$ are well-formed,
  by induction on its typing derivation. The $\lambda$ rule and the
  ascription rules check their annotation, and the body of a \src{let} is
  typed under at least one entry of the type of the bound term, which has an entry
  because it is well-formed. The open statement proved in
  Appendix~\ref{app:sgg} for Theorem~\ref{theorem:sgg} gives
  $\Gamma' |-d \srcn : \gtyb[1]$ with $\gtya \gprec \gtyb[1]$, and
  $\gtyb[1] = \gtyb$ by Lemma~\ref{lemma:source-det}. Values are the same.}
\end{proof}

\rev[fo1]{In the next lemma, a type in evidence position is read with
identity pointers.}

\begin{lemma}[{\rev[fo1]{Tag-aware monotonicity of the meet}}]\label{lemma:tag-meet-mono}\mechnote{item~1, \leanname{TPLC.tagPrecTy\_toTag} and \leanname{TPLC.tagPrecD\_toTagD}; items~2 and~3, \leanname{TPLC.tagPrecTy\_tagMeetTy}, \leanname{TPLC.tagPrecD\_tagMeetD} and \leanname{TPLC.tagPrecD\_tagMeetD\_toTag}; the definedness of the less precise meet, \leanname{TPLC.meetTy\_mono} and \leanname{TPLC.meetD\_mono}.}~
  \rev[fo1]{For any well-formed formula simple types
  $\gtys, \gtysp, \gtys[1], \gtys[2], \gtysp[1], \gtysp[2] \in \FGType$,
  \begin{enumerate}
  \item if $\gtys \gprec \gtysp$, then
    $\tagprecr{\gtys}{\gtysp}{\gtys}{\gtysp}$ and
    $\tagprecl{\gtys}{\gtysp}{\gtys}{\gtysp}$;
  \item if $\gtys[1] \gprec \gtysp[1]$, $\gtys[2] \gprec \gtysp[2]$ and
    $\gtys[1] \meet \gtys[2]$ is defined, then $\gtysp[1] \meet \gtysp[2]$
    is defined and
    $\tagprecr{\gtys[2]}{\gtysp[2]}{\gtys[1] \meet \gtys[2]}{\gtysp[1] \meet \gtysp[2]}$;
  \item if $\gtys[1] \gprec \gtysp[1]$, $\gtys[2] \gprec \gtysp[2]$ and
    $\gtys[1] \meet \gtys[2]$ is defined, then
    $\tagprecl{\gtys[1]}{\gtysp[1]}{\gtys[1] \meet \gtys[2]}{\gtysp[1] \meet \gtysp[2]}$;
  \end{enumerate}
  and likewise for well-formed formula distribution types.}
\end{lemma}
\begin{proof}
  \rev[fo1]{Item~1 is by induction on $\gtys \gprec \gtysp$; at the
  distribution level the coupling is the one of the precision itself, and
  since the pointers are identities, two coupled entries name themselves, and
  they are related by hypothesis. For items~2 and~3, the definedness
  of the less precise meet and the runtime precision of the two meets are
  Lemma~\ref{lemma:mono}, since $\gprec$ implies $\eprec$. On simple types,
  the proof of item~2 is by induction on $\gtys[2]$, the type on the side of
  the pointer. Unless the two judged types and the two meets are all
  arrows, rule ($\pi$B) asks for nothing more. In the arrow case, if
  $\gtys[1] = \?$, then $\gtysp[1] = \?$, the two meets are $\gtys[2]$ and
  $\gtysp[2]$ themselves, and item~1 concludes; if $\gtys[1]$ is an arrow
  and $\gtysp[1] = \?$, the less precise meet is $\gtysp[2]$ with identity
  pointers, and the claim splits into the domains, by the induction
  hypothesis, and the codomains, at the distribution level; if both are
  arrows, the claim is componentwise.}

  \rev[fo1]{At the distribution level, suppose first that both sides are
  meets. The coupling is the one built in the proof of
  Lemma~\ref{lemma:mono}. Each of its positive weights comes from a positive
  weight of the more precise meet through positive weights of the two precision
  couplings, so the entries of $\gtya[2]$ and $\gtyb[2]$ that the right
  pointers of two coupled entries name are related by the coupling of
  $\gtya[2] \gprec \gtyb[2]$, and the coupled entries themselves are pointwise
  meets of related entries, related by the simple-type statement. When the less
  precise side is $\gtyb[2]$ with identity pointers, the coupling composes
  (Lemma~\ref{lemma:coupling-composition}) the push-forward of each solution
  of $\gtya[1] \meet \gtya[2]$ along its right pointers, a solution of
  $\gtya[2]$ by the marginal clause of the meet, with the coupling of
  $\gtya[2] \gprec \gtyb[2]$; it has positive weight only between an entry
  of the meet and an entry of $\gtyb[2]$ above the entry of $\gtya[2]$ that
  the former names, and the entry of the meet compares with that entry of
  $\gtyb[2]$ by the simple-type statement with
  $\gtysp[1] = \?$. Item~3 is symmetric, with the left pointers and
  $\gtys[1]$ in place of $\gtys[2]$.}
\end{proof}

\begin{lemma}[{\rev[fo1]{Precision of single-entry lets}}]\label{lemma:prec-letin1}\mechnote{\leanname{TPLC.prec\_letin1}.}~
  \rev[fo1]{For any environments $\Gamma$ and $\Gamma'$, any variable
  $\tx \in \text{Var}$, any terms $\tm, \tn, \tm[1], \tn[1]$ of \tlang{}
  and any formula simple types $\gtys, \gtysp \in \FGType$, if
  $\Gamma \gprec \Gamma'$, $\gtys \gprec \gtysp$,
  $\Gamma \gprec \Gamma' |- \tm \gprec \tn$,
  $\Gamma |-d \tm : \ds{\gtys[][1]}$, $\Gamma' |-d \tn : \ds{\gtysp[][1]}$ and
  $\Gamma, \tx{:}\gtys \gprec \Gamma', \tx{:}\gtysp |- \tm[1] \gprec \tn[1]$,
  then
  $\Gamma \gprec \Gamma' |- \lett{\tx}{\tm}{\tm[1]} \gprec \lett{\tx}{\tn}{\tn[1]}$.}
\end{lemma}
\begin{proof}
  \rev[fo1]{By the \src{let} rule of Figure~\ref{fig:term-precision} with
  one entry on each side. The types of the bound terms in its premises are those of
  the hypotheses, by Lemma~\ref{lemma:target-det}; the only pair of entries
  is related, and its bodies are related by hypothesis; the coupling gives
  weight $1$ to that pair.}
\end{proof}

\begin{lemma}[{\rev[fo1]{Monotone canonical elaboration}}]\label{epp}\mechnote{\leanname{TPLC.elab\_canon\_val} and \leanname{TPLC.elab\_canon\_tm}.}~
  \rev[fo1]{For any well-formed environment $\Gamma'$, any term $\srcn$ of
  \glang{} and any formula distribution type $\gtyb \in \FGDType$, if
  $\Gamma' |-d \srcn : \gtyb$, then there is an elaboration
  $\ela{\Gamma' |-d \srcn}{\tn}{\gtyb}$ such that, for any well-formed
  environment $\Gamma$, any term $\sm$ of \glang{}, any term $\tm$ of
  \tlang{} and any formula distribution type $\gtya \in \FGDType$, if
  $\Gamma \gprec \Gamma'$, $\sm \gprec \srcn$ and
  $\ela{\Gamma |-d \sm}{\tm}{\gtya}$, then
  $\Gamma \gprec \Gamma' |- \tm \gprec \tn$; likewise for values.}
\end{lemma}
\begin{proof}
  \begin{revblock}[fo1]
  By induction on the typing derivation of $\srcn$, the statements for
  values and for terms together, for all $\Gamma$, $\sm$ and $\tm$ at once.
  In each case the derivation of $\sm \gprec \srcn$ ends in the rule of
  Figure~\ref{source-term-precision} for the top constructor of $\srcn$,
  which fixes the shape of $\sm$, and the elaboration of $\sm$ ends in the
  rule of Figure~\ref{fig:elaboration} for that constructor, which fixes
  the shape of $\tm$ and its evidences. The elaboration $\tn$ is built with
  the same rule from the elaborations that the induction hypotheses give
  for the subterms of $\srcn$; its meets are defined because their operands
  are consistent by the premises of the typing rule, as in the existence
  half of Lemma~\ref{theorem:ep}. Throughout, $\Gamma |-d \sm : \gtya$ by
  Lemma~\ref{lemma:elab-typing}, the types of the elaborated subterms are
  well-formed (Lemmas~\ref{lemma:welltyped-wellform-source}
  and~\ref{lemma:wellformed-lifting}), every evidence is well-formed and
  valid (Lemma~\ref{lemma:meet-validity}), and the types in the typing
  clauses of the ascription rules of Figure~\ref{fig:term-precision} are
  the elaboration types of the subjects, by the preservation half of
  Lemma~\ref{theorem:ep} and Lemma~\ref{lemma:target-det}; we do not repeat
  these facts in the cases.

  \begin{case}[variables and literals]
    $\srcn = \sx$ forces $\sm = \sx$, both elaborate to $\tx$, and the
    variable rule concludes. A literal $\ssr$ is related only to itself, and
    both sides elaborate to $\asc{\rtype\,\ttr}{\rtype}$; the ascription
    rule for raw values asks for $\tagprecr{\rtype}{\rtype}{\rtype}{\rtype}$,
    for the same judgment with $\vdash_{\leftvar}$ against the types of the
    two literals, which are $\rtype$ by (Gr), for $\ttr \gprec \ttr$ and for
    $\rtype \gprec \rtype$; all hold by rule ($\pi$B) and reflexivity.
    Boolean literals are the same.
  \end{case}

  \begin{case}[abstraction]
    $\srcn = \lambda \sx{:}\sgtysp.\,\srcn[0]$ and
    $\sm = \lambda \sx{:}\sgtys.\,\sm[0]$ with
    $\liftT{\sgtys} \gprec \liftT{\sgtysp}$ and $\sm[0] \gprec \srcn[0]$. The
    typing derivation of $\srcn$ ends in the $\lambda$ rule, whose premises
    give $\Gamma', \sx{:}\liftT{\sgtysp} |-d \srcn[0] : \gtyb[0]$ and
    $\jtf{}{\sgtysp}$. The elaboration of $\sm$ ends in $\src{(E\lambda)}$,
    with $\tm = \asc{\ev \lambda \tx{:}\liftT{\sgtys}.\,\tm[0]}{\liftT{\sgtys} -> \gtya[0]}$,
    $\ev = \liftT{\sgtys} -> \gtya[0]$ and
    $\ela{\Gamma, \sx{:}\liftT{\sgtys} |-d \sm[0]}{\tm[0]}{\gtya[0]}$. The
    induction hypothesis on $\srcn[0]$, under the extended environments,
    which are well-formed (Lemma~\ref{lemma:wellformed-lifting}) and
    related, gives $\tn[0]$ with
    $\Gamma, \sx{:}\liftT{\sgtys} \gprec \Gamma', \sx{:}\liftT{\sgtysp} |- \tm[0] \gprec \tn[0]$,
    and $\tn$ is the ascription of $\lambda \tx{:}\liftT{\sgtysp}.\,\tn[0]$
    with the evidence $\evp = \liftT{\sgtysp} -> \gtyb[0]$. By
    Lemma~\ref{lemma:elab-prec-types}, $\gtya[0] \gprec \gtyb[0]$, hence
    $\liftT{\sgtys} -> \gtya[0] \gprec \liftT{\sgtysp} -> \gtyb[0]$. The
    ascription rule for raw values then holds. Its two evidence clauses
    compare $\ev$ with $\evp$ against the annotations and against the types
    of the two abstractions, which are the same two types, and both are
    item~1 of Lemma~\ref{lemma:tag-meet-mono}; the two abstractions are
    related by the $\lambda$ rule of precision with the induction
    hypothesis, and the annotations are related.
  \end{case}

  \begin{case}[values as terms]
    $\srcn = \sw$ and $\sm = \sv$ with $\sv \gprec \sw$, both elaborated by
    $\src{(E\sv)}$ to the elaborations of the values. The induction
    hypothesis on $\sw$ gives the elaboration of $\sw$ and the precision of
    the two elaborated values, and the rule for values of
    Figure~\ref{fig:term-precision} concludes.
  \end{case}

  \begin{case}[value ascription]
    $\srcn = \asc{\sw}{\sgtysp}$ and $\sm = \asc{\sv}{\sgtys}$ with
    $\sv \gprec \sw$ and $\liftT{\sgtys} \gprec \liftT{\sgtysp}$. The typing
    derivation of $\srcn$ ends in the ascription rule, whose premises give
    $\Gamma' |-t \sw : \gtysp[1]$, $\gtysp[1] \rel \liftT{\sgtysp}$ and
    $\jtf{}{\sgtysp}$. The elaboration of $\sm$ ends in
    $\src{(E\mathord{::}\sgtys)}$, with $\tm = \asc{\ev \tv}{\liftT{\sgtys}}$,
    $\ela{\Gamma |-t \sv}{\tv}{\gtys[1]}$ and
    $\ev = \gtys[1] \meet \liftT{\sgtys}$. The induction hypothesis on $\sw$
    gives $\tw$ with $\Gamma \gprec \Gamma' |- \tv \gprec \tw$, and
    $\tn = \asc{\evp \tw}{\liftT{\sgtysp}}$ with
    $\evp = \gtysp[1] \meet \liftT{\sgtysp}$. By
    Lemma~\ref{lemma:elab-prec-types}, $\gtys[1] \gprec \gtysp[1]$. The
    ascription rule for values asks for four clauses. The comparison
    against the annotations,
    $\tagprecr{\liftT{\sgtys}}{\liftT{\sgtysp}}{\ev}{\evp}$, is item~2 of
    Lemma~\ref{lemma:tag-meet-mono}, with the lifted annotations as second
    operands. The comparison against the types of the subjects, which are
    $\gtys[1]$ and $\gtysp[1]$, namely
    $\tagprecl{\gtys[1]}{\gtysp[1]}{\ev}{\evp}$, is item~3 of the same lemma.
    The subjects are related by the induction hypothesis, and the
    annotations by hypothesis.
  \end{case}

  \begin{case}[term ascription]
    $\srcn = \asc{\srcn[0]}{\sgtyb[2]}$ and $\sm = \asc{\sm[0]}{\sgtyb[1]}$
    with $\sm[0] \gprec \srcn[0]$ and $\liftD{\sgtyb[1]} \gprec \liftD{\sgtyb[2]}$.
    The case is the previous one at the distribution level. The evidences
    are $\evd = \gtya[0] \meet \liftD{\sgtyb[1]}$ and
    $\evdp = \gtyb[0] \meet \liftD{\sgtyb[2]}$ for the types $\gtya[0]$ and
    $\gtyb[0]$ of the elaborated subterms, related by
    Lemma~\ref{lemma:elab-prec-types}, and the two evidence clauses of the
    ascription rule for terms are items~2 and~3 of
    Lemma~\ref{lemma:tag-meet-mono} for distribution types.
  \end{case}

  \begin{case}[addition]
    $\srcn = \add{\svp}{\swp}$ and $\sm = \add{\sv}{\sw}$ with $\sv \gprec \svp$
    and $\sw \gprec \swp$. The typing derivation of $\srcn$ ends in the
    addition rule, whose premises give $\Gamma' |-t \svp : \gtysp[1]$,
    $\Gamma' |-t \swp : \gtysp[2]$, and both types consistent with $\rtype$.
    The elaboration of $\sm$ ends in $\src{(E+)}$, with
    $\tm = \lett{\tx}{\asc{\ev[1] \tv}{\rtype}}{\lett{\ty}{\asc{\ev[2] \tw}{\rtype}}{\add{\tx}{\ty}}}$,
    $\ev[1] = \gtys[1] \meet \rtype$ and $\ev[2] = \gtys[2] \meet \rtype$ for
    the elaboration types $\gtys[1]$, $\gtys[2]$ of $\sv$ and $\sw$. The
    induction hypotheses on $\svp$ and $\swp$ give $\tvp$ and $\twp$, related
    to $\tv$ and $\tw$, and $\tn$ is the same term with primed components
    and the evidences $\evp[1] = \gtysp[1] \meet \rtype$ and
    $\evp[2] = \gtysp[2] \meet \rtype$. By Lemma~\ref{lemma:elab-prec-types},
    $\gtys[i] \gprec \gtysp[i]$. The two coercions are related by the
    ascription rule for values. The clause against the annotation is item~2
    of Lemma~\ref{lemma:tag-meet-mono} with $\rtype \gprec \rtype$, the
    clause against the types of the subjects is item~3, and the subjects are
    related by the induction hypotheses, the second pair under the binder
    $\tx$ by Lemma~\ref{lemma:prec-weakening}. Each coercion is typed at
    $\ds{\rtype^{1}}$ by (G$\mathord{::}\gtys$), the second under the binder
    by weakening (Lemma~\ref{lemma:target-weakening}). The innermost bodies
    are related by the addition rule and the variable rule. Two applications
    of Lemma~\ref{lemma:prec-letin1}, with $\rtype \gprec \rtype$ for the
    bound types, conclude.
  \end{case}

  \begin{case}[application]
    $\srcn = \svp\;\swp$ and $\sm = \sv\;\sw$ with $\sv \gprec \svp$ and
    $\sw \gprec \swp$. The typing derivation of $\srcn$ ends in the
    application rule, whose premises give $\Gamma' |-t \svp : \gtysp[1]$,
    $\Gamma' |-t \swp : \gtysp[2]$ and $\gtysp[2] \rel \cdom(\gtysp[1])$,
    with result $\ccod(\gtysp[1])$; the premises of the elaboration of $\sm$,
    which ends in $\src{(Eapp)}$, give the same for the elaboration types
    $\gtys[1]$ of $\sv$ and $\gtys[2]$ of $\sw$. Write
    $\gtys[3] = \cdom(\gtys[1])$, $\gtya[3] = \ccod(\gtys[1])$, and
    $\gtysp[3]$, $\gtyb[3]$ for the less precise side. Then
    $\tm = \lett{\tx}{\asc{\ev[1] \tw}{\gtys[3]}}{\lett{\ty}{\asc{\ev[2] \tv}{\gtys[3] -> \gtya[3]}}{\ty\;\tx}}$
    with $\ev[1] = \gtys[2] \meet \gtys[3]$ and
    $\ev[2] = \gtys[1] \meet (\gtys[3] -> \gtya[3])$, and $\tn$ is the same
    term with primed components, built from the elaborations that the
    induction hypotheses on $\svp$ and $\swp$ give. By
    Lemma~\ref{lemma:elab-prec-types}, $\gtys[1] \gprec \gtysp[1]$ and
    $\gtys[2] \gprec \gtysp[2]$; by item~2 of
    Lemma~\ref{lemma:prec-operators}, $\gtys[3] \gprec \gtysp[3]$ and
    $\gtya[3] \gprec \gtyb[3]$, hence
    $\gtys[3] -> \gtya[3] \gprec \gtysp[3] -> \gtyb[3]$. The second evidence
    is the meet of the type of the function with a type that is not an
    annotation of the program but is inferred from it, and the inferred
    types of the two sides are related as just shown; the clauses of the
    ascription rule are then items~2 and~3 of
    Lemma~\ref{lemma:tag-meet-mono}, for $\ev[1]$ with
    $\gtys[2] \gprec \gtysp[2]$ and $\gtys[3] \gprec \gtysp[3]$, and for
    $\ev[2]$ with $\gtys[1] \gprec \gtysp[1]$ and
    $\gtys[3] -> \gtya[3] \gprec \gtysp[3] -> \gtyb[3]$, the subjects being
    related by the induction hypotheses, $\tv$ and $\tvp$ under the binder
    $\tx$ by Lemma~\ref{lemma:prec-weakening}. The coercions are typed as in
    the proof of Lemma~\ref{theorem:ep}, the bodies $\ty\;\tx$ are related
    by the application rule and the variable rule, and two applications of
    Lemma~\ref{lemma:prec-letin1}, with the bound types
    $\gtys[3] \gprec \gtysp[3]$ and
    $\gtys[3] -> \gtya[3] \gprec \gtysp[3] -> \gtyb[3]$, conclude.
  \end{case}

  \begin{case}[conditional]
    $\srcn = \ite{\svp}{\srcn[1]}{\srcn[2]}$ and
    $\sm = \ite{\sv}{\sm[1]}{\sm[2]}$ with $\sv \gprec \svp$ and
    $\sm[i] \gprec \srcn[i]$. The typing derivation of $\srcn$ ends in the
    conditional rule, whose premises give $\Gamma' |-t \svp : \gtysp[1]$
    with $\gtysp[1] \rel \btype$ and type the branches. The elaboration of
    $\sm$ ends in $\src{(Eif)}$, with
    $\tm = \lett{\tx}{\asc{\ev \tv}{\btype}}{\ite{\tx}{\tm[1]}{\tm[2]}}$ and
    $\ev = \gtys[1] \meet \btype$ for the elaboration type $\gtys[1]$ of
    $\sv$; the branches carry no coercion. The induction hypotheses on
    $\svp$, $\srcn[1]$ and $\srcn[2]$ give the components of $\tn$, which
    has the same shape with $\evp = \gtysp[1] \meet \btype$, and relate them
    to those of $\tm$. By Lemma~\ref{lemma:elab-prec-types},
    $\gtys[1] \gprec \gtysp[1]$. The coercions of the guards are related by
    the ascription rule for values, with items~2 and~3 of
    Lemma~\ref{lemma:tag-meet-mono} and $\btype \gprec \btype$; the two
    conditionals under the binder are related by the conditional rule from
    $\tx \gprec \tx$ and the induction hypotheses on the branches, weakened
    under $\tx$ by Lemma~\ref{lemma:prec-weakening}; and
    Lemma~\ref{lemma:prec-letin1} concludes.
  \end{case}

  \begin{case}[probabilistic choice]
    $\srcn = \srcn[1] \pssum[\ptyp] \srcn[2]$ and
    $\sm = \sm[1] \pssum \sm[2]$ with $\sm[i] \gprec \srcn[i]$ and
    $\pty \gprec \ptyp$. The typing derivation of $\srcn$ ends in the choice
    rule for $\ptyp$, concrete or $\?$, and both sides elaborate by
    $\src{(E\oplus)}$ to the choice of the elaborated branches with the same
    probability annotations. The induction hypotheses on the branches and
    $\pty \gprec \ptyp$ give $\tm \gprec \tn$ by rule
    $\trg{(\mathord{\gprec}\oplus)}$.
  \end{case}

  \begin{case}[let]
    $\srcn = \lett{\sx}{\srcn[1]}{\srcn[2]}$ and
    $\sm = \lett{\sx}{\sm[1]}{\sm[2]}$ with $\sm[1] \gprec \srcn[1]$ and
    $\sm[2] \gprec \srcn[2]$. The typing derivation of $\srcn$ ends in the
    \src{let} rule, whose premises give
    $\Gamma' |-d \srcn[1] : \gtyb[1] = \db{\phi[][2]}{\gtysp[j][\p[j]]}[j \in \jSet]$
    and $\Gamma', \sx{:}\gtysp[j] |-d \srcn[2] : \gtyb[j]$ for every
    $j \in \jSet$. The elaboration of $\sm$ ends in $\src{(Elet)}$, with
    $\tm = \lett{\tx}{\tm[1]}{\overline{\tm[i]}^{\,i \in \iSet}}$,
    $\ela{\Gamma |-d \sm[1]}{\tm[1]}{\gtya[1]}$ for
    $\gtya[1] = \db{\phi[][1]}{\gtys[i][\p[i]]}[i \in \iSet]$, and
    $\ela{\Gamma, \sx{:}\gtys[i] |-d \sm[2]}{\tm[i]}{\gtya[i]}$ for every
    $i \in \iSet$. The types of both bound terms are well-formed, so the extended
    environments are. The induction hypothesis on $\srcn[1]$ gives $\tn[1]$
    with $\Gamma \gprec \Gamma' |- \tm[1] \gprec \tn[1]$, and the induction
    hypothesis on $\srcn[2]$ at each entry $j$, under
    $\Gamma', \sx{:}\gtysp[j]$, gives $\tn[j]$; then
    $\tn = \lett{\tx}{\tn[1]}{\overline{\tn[j]}^{\,j \in \jSet}}$ by
    $\src{(Elet)}$. The \src{let} rule of Figure~\ref{fig:term-precision}
    asks for three clauses. The bound terms are related, as just shown. For
    the bodies, the rule names the types $\gtya[1]$ and $\gtyb[1]$ of the
    two bound terms and asks, for every pair $(i, j)$ with
    $\gtys[i] \gprec \gtysp[j]$, that
    $\Gamma, \sx{:}\gtys[i] \gprec \Gamma', \sx{:}\gtysp[j] |- \tm[i] \gprec \tn[j]$.
    This is the induction hypothesis on $\srcn[2]$ at the entry $j$, applied
    to the elaboration $\tm[i]$ of $\sm[2]$ at the entry $i$, under
    environments that are related because $\gtys[i] \gprec \gtysp[j]$; one
    body $\tn[j]$ is thus compared with as many bodies $\tm[i]$ as there
    are entries of $\gtya[1]$ below $\gtysp[j]$, which is why the statement
    quantifies over the more precise elaboration. The coupling of the types of the two
    bound terms is the one of $\gtya[1] \gprec \gtyb[1]$, which
    Lemma~\ref{lemma:elab-prec-types} gives from $\sm[1] \gprec \srcn[1]$. \qedhere
  \end{case}
  \end{revblock}
\end{proof}

\begin{theoremproof}{Theorem~\ref{theorem:gg} (Dynamic gradual guarantee for \glang)}
  \rev[fo1]{Lemma~\ref{epp}, with $\Gamma' = \cdot$, gives an elaboration
  $\ela{|-d \srcn}{\tn}{\gtyb}$ such that
  $\cdot \gprec \cdot |- \tm \gprec \tn$ for every elaboration
  $\ela{|-d \sm}{\tm}{\gtya}$. Both $\tm$ and $\tn$ are closed, and
  $|-d \tm : \gtya$ and $|-d \tn : \gtyb$ by the preservation half of
  Lemma~\ref{theorem:ep}. If $\tm \nreds{}{k_1} \rctx[\phi[][1]] \dset[1]$ by
  a reduction that raises no error and
  $\tn \nreds{}{k_2} \rctx[\phi[][2]] \dset[2]$, then
  Theorem~\ref{theorem:dgg} gives
  $\phty{\phi[][1]}{\dset[1]} \gprec \phty{\phi[][2]}{\dset[2]}$.}
\end{theoremproof}

\subsection{Dynamic Conservative Extension}\label{app:conservative}

\dyneq*

\begin{revblock}[fo1]
\proofheading{Structure of the proof.}
The theorem compares the measure that the run of the elaborated program
induces on the values of \slang{}, through the erasure of
Figure~\ref{fig:erasure}, with the measure of the run of the source program.
The proof composes the lemmas of this subsection and has no induction of
its own. The elaboration $\tmp$ is well typed (Lemma~\ref{theorem:ep}), the
types that its typing infers realize static types
(Lemma~\ref{staticrealize}), and it erases to the normalized source
$\anf{\sm}$ (Lemma~\ref{lemma:elaberase}). From the first two facts, the run
of $\tmp$ takes neither error case of the two ascription rules
(Lemma~\ref{lemma:nofailure}), and the other error rules do not apply to a
term that has an erasure. The run of $\sm$ yields a run of $\anf{\sm}$ with the measure that item~1 describes
(Lemma~\ref{lemma:anfmeas}), and Theorem~\ref{erasuremeas} equates the
erased measure of the run of $\tmp$ with the measure of any run of its
erasure. Item~2 follows because the normalization rewrites only
abstractions.

The content is in Theorem~\ref{erasuremeas} and in
Lemma~\ref{lemma:nofailure}. The theorem is proved by induction on an
error-free reduction derivation, with no hypothesis on types, for every
erasure of the term and every run of that erasure, and the two measures are
compared value by value. The exhaustive routing of $\trg{(\mathit{Dlet})}$
produces entries of weight zero, and the exhaustive \src{let} of \slang{} runs
its body on outcomes of probability zero; neither has a counterpart on the
other side, and both contribute nothing to their measure. The regrouping of
the weighted sums of the two \src{let}s along the erasure rests on the
marginals of the routing evidences (Lemma~\ref{lemma:marginals}) and on a
regrouping identity for weighted sums (Lemma~\ref{lemma:fibers}), and the
determinism of \slang{} enters through the induction hypothesis, which holds
for every run of the body. Lemma~\ref{lemma:nofailure} rests on weak
realization (Definition~\ref{def:weak-realization}). Every evidence that
reaches a composition during the run sits below an annotation that realizes
a static type, and two such evidences compose
(Lemma~\ref{static-composition-defined2}). The complete proof is given at
the end of this subsection, after the lemmas it uses.
\end{revblock}

\rev[mt1]{The erasure $\erase{\cdot}$ is Figure~\ref{fig:erasure} in the article
body; the distribution semantics of \slang{} is
Figure~\ref{fig:static-distribution-reduction}. Throughout this
subsection the \slang{} \src{let}, like the gradual one, is exhaustive.}

\rev[mt1]{The normalization $\anf{\cdot}$ announced in
Section~\ref{sec:target} discards every type annotation, writing the placeholder
$\ph$ in its place, and rewrites an application, an addition and a conditional
into the nested $\<let>$s that the \rev[fo1]{monadic} elaboration\rev[fo1]{~\citep{Moggi1991NotionsOC}} builds:
\[
  \begin{array}{r@{~{}={}~}l}
    \anf{\static{x}} & \static{x}
    \qquad \anf{\static{r}} ~=~ \static{r}
    \qquad \anf{\static{b}} ~=~ \static{b}
    \\[0.4ex]
    \anf{\static{\lambda x{:}\tys.\ \sm}} & \static{\lambda x{:}\ph.\ \anf{\sm}}
    \\[0.4ex]
    \anf{\asc{\sv}{\tys}} & \asc{\anf{\sv}}{\ph}
    \\[0.4ex]
    \anf{\asc{\sm}{\tya}} & \asc{\anf{\sm}}{\ph}
    \\[0.4ex]
    \anf{\sm \spsum \sn} & \anf{\sm} \spsum \anf{\sn}
    \\[0.4ex]
    \anf{\lett{x}{\sm}{\sn}} & \lett{x}{\anf{\sm}}{\anf{\sn}}
    \\[0.4ex]
    \anf{\sv\;\sw} &
      \lett{x}{\asc{\anf{\sw}}{\ph}}{\lett{y}{\asc{\anf{\sv}}{\ph}}{y\;x}}
    \\[0.4ex]
    \anf{\add{\sv}{\sw}} &
      \lett{x}{\asc{\anf{\sv}}{\ph}}{\lett{y}{\asc{\anf{\sw}}{\ph}}{\add{x}{y}}}
    \\[0.4ex]
    \anf{\ite{\sv}{\sm}{\sn}} &
      \lett{x}{\asc{\anf{\sv}}{\ph}}{\ite{x}{\anf{\sm}}{\anf{\sn}}}
  \end{array}
\]}

\begin{revblock}[fo2]
The evidences that reach a composition during the run of an elaborated static
program sit below an annotation that realizes a static type, and a type below
a realizing one may carry entries of probability zero with no partner. Such
evidences satisfy a weaker relation, without the coverage clauses of
Definition~\ref{def:realization}.

\begin{definition}[Weak realization]\label{def:weak-realization}\mechnote{\leanname{TPLC.WeakRealizesTy} and \leanname{TPLC.WeakRealizesD}.}
Weak realization of a static type of \slang{} by a formula type,
$\wrealizes{\gtys}{\tys}$ and $\wrealizes{\gtya}{\tya}$, is defined by the
following rules:
\begin{mathpar}
  \inference{}{\wrealizes{\rtype}{\rtype}} \and
  \inference{}{\wrealizes{\btype}{\btype}} \and
  \inference{\wrealizes{\gtys}{\tys} & \wrealizes{\gtya}{\tya}}
    {\wrealizes{\gtys -> \gtya}{\tys -> \tya}} \and
  \inference{
    \satisfiablee{\phi} &
    R \subseteq \iSet \times \jSet \\
    \forall (i,j) \in R.~ \exists \tys.~ \wrealizes{\gtys[i]}{\tys} \land \tys =_{s} \tys[j] \\
    \forall \nu \models \phi.~ \exists (w_{ij} \geq 0)_{i \in \iSet, j \in \jSet}.~
      (\forall i.~ \textstyle\sum_{j} w_{ij} = \nu(\p[i])) \land
      (\forall j.~ \textstyle\sum_{i} w_{ij} = \ps[j]) \\
    \land~ (\forall i, j.~ w_{ij} > 0 \Rightarrow (i,j) \in R)
   }
   {\wrealizes{\db{\phi}{\gtys[i][\p[i]]}[i \in \iSet]}{\d{\tys[j][\ps[j]]}[j \in \jSet]}}
\end{mathpar}
\end{definition}
\end{revblock}

\begin{lemma}[{\rev[fo1]{Weak realization}}]\label{lemma:weak-realization}\mechnote{item~1, \leanname{TPLC.weakRealizesTy\_of\_realizes} and \leanname{TPLC.weakRealizesD\_of\_realizes}; item~2, \leanname{TPLC.weakRealizesTy\_of\_prec} and \leanname{TPLC.weakRealizesD\_of\_prec}.}~
  \rev[fo1]{For any formula simple types $\gtys, \gtysp \in \FGType$, any
  formula distribution types $\gtya, \gtyb \in \FGDType$, any simple type
  $\tys \in \Type$ and any distribution type $\tya \in \DType$,
  \begin{enumerate}
  \item if $\realizes{\gtys}{\tys}$, then $\wrealizes{\gtys}{\tys}$, and if
    $\realizes{\gtya}{\tya}$, then $\wrealizes{\gtya}{\tya}$;
  \item if $\jtf{}{\gtys}$, $\gtys \eprec \gtysp$ and
    $\wrealizes{\gtysp}{\tys}$, then $\wrealizes{\gtys}{\tys}$; and if
    $\jtf{}{\gtya}$, $\gtya \eprec \gtyb$ and $\wrealizes{\gtyb}{\tya}$, then
    $\wrealizes{\gtya}{\tya}$.
  \end{enumerate}}
\end{lemma}
\begin{proof}
  \rev[fo1]{Item~1 is by induction on the realization. At a distribution
  type, the relation of the weak realization pairs an entry $i$ with every
  entry $j$ whose type is $=_{s}$ the type of an entry paired with $i$ by the
  realization, and the couplings are unchanged. Item~2 is by induction on
  the types. At a distribution type, the coupling that $\gtya \eprec \gtyb$
  gives at a solution of $\gtya$ and the coupling of the weak realization at
  its image compose (Lemma~\ref{lemma:coupling-composition}) into a coupling
  onto the probabilities of the static type, supported on pairs related through an entry of
  $\gtyb$, whose types weakly realize by the induction hypothesis; the
  formula of $\gtya$ is satisfiable because $\gtya$ is well-formed.}
\end{proof}

\begin{lemma}[Static composition defined]\label{static-composition-defined2}\mechnote{item~1, \leanname{TPLC.econs\_of\_weakRealizes\_ty} and \leanname{TPLC.red\_ascV\_static\_ok}; item~2, \leanname{TPLC.econs\_of\_weakRealizes\_d}.}~
  \rev[fo1]{For any well-formed formula simple types
  $\ev[1], \ev[2] \in \GFType$, any simple types $\tys[1], \tys[2] \in \Type$,
  any well-formed formula distribution types $\evd[1], \evd[2] \in \FGDType$
  and any distribution types $\tya[1], \tya[2] \in \DType$,
  \begin{enumerate}
  \item if $\tys[1] =_{s} \tys[2]$, $\wrealizes{\ev[1]}{\tys[1]}$ and
    $\wrealizes{\ev[2]}{\tys[2]}$, then $\ev[1] \trans{} \ev[2]$ is defined,
    $\jtf{}{\ev[1] \trans{} \ev[2]}$, and
    $\wrealizes{\ev[1] \trans{} \ev[2]}{\tys[1]}$;
  \item if $\tya[1] =_{s} \tya[2]$, $\wrealizes{\evd[1]}{\tya[1]}$ and
    $\wrealizes{\evd[2]}{\tya[2]}$, then $\evd[1] \erel \evd[2]$ and
    $\evd[1] \trans{} \evd[2]$ is defined.
  \end{enumerate}}
\end{lemma}
\begin{proof}
  \rev[fo1]{By induction on the two types, the two items together. At a
  distribution type, the coupling of $\evd[1]$ onto the probabilities of $\tya[1]$,
  the identity coupling of $\tya[1] =_{s} \tya[2]$ and the transpose of the
  coupling of $\evd[2]$ onto the probabilities of $\tya[2]$ compose
  (Lemma~\ref{lemma:coupling-composition}) into a coupling of a solution of
  $\evd[1]$ with a solution of $\evd[2]$ that has positive weight only
  between entries that weakly realize $=_{s}$-equal static types, consistent by the
  induction hypothesis. This is $\evd[1] \erel \evd[2]$, and
  Lemma~\ref{consistency-defined} gives the definedness. On simple types the
  argument is the same without couplings, the composition is well-formed
  because its operands are (Lemma~\ref{lemma:meet-validity}), and it weakly
  realizes $\tys[1]$ because it is $\eprec \ev[1]$
  (Lemma~\ref{meet-more-precise2}) and weak realization descends along
  precision (item~2 of Lemma~\ref{lemma:weak-realization}).}
\end{proof}

\begin{lemma}[{\rev[mt1]{Coercion invariance}}]\label{lemma:erasecoerce}\mechnote{\leanname{TPLC.ascV\_rawOf}, \leanname{TPLC.erVal\_of\_raw\_eq}, \leanname{TPLC.erVal\_ascV} and \leanname{TPLC.erVal\_ascV\_redSt}.}~
  \rev[fo1]{For any evidences $\ev$, $\ev[1]$, $\ev[2]$, any raw value $\tu$,
  any values $\tv$, $\tw$ of \tlang{} and any formula simple types
  $\gtys[1], \gtys[2], \gtysp \in \FGType$,
  \begin{enumerate}
  \item if one of the values $\asc{\ev[1] \tu}{\gtys[1]}$ and
    $\asc{\ev[2] \tu}{\gtys[2]}$ has an erasure, then both have the same
    erasure;
  \item if $\tv$ and $\tw$ are not errors and $\asc{\ev \tv}{\gtysp}$ reduces
    in one step to $\dt{\pt{\tw}}$, then $\erase{\tw} = \erase{\tv}$.
  \end{enumerate}}
\end{lemma}
\begin{proof}
  \rev[fo1]{The clauses of Figure~\ref{fig:erasure} for values read the raw
  value and nothing else. A value ascription of a value that is not an
  error reduces by $\trg{(\mathit{D}\mathord{::}\gtys)}$, whose first case
  returns the raw value of $\tv$ under the composed evidence and the new
  annotation, and whose second case returns an error.}
\end{proof}

\begin{lemma}[{\rev[fo1]{Erasure and substitution}}]\label{lemma:erasure-subst}\mechnote{\leanname{TPLC.erVal\_subst} and \leanname{TPLC.erTm\_subst}.}~
  \rev[fo1]{For any variable $\tx$, any term $\tm$ and any value $\tv$ of
  \tlang{}, and any term $\sn$ and any value $\v$ of \slang{}, if
  $\erase{\tm} = \sn$ and $\erase{\tv} = \v$, then
  $\erase{\tm{}[\tv/\tx]} = \sn{}[\v/x]$; likewise for a value of \tlang{}
  in place of $\tm$.}
\end{lemma}
\begin{proof}
  \rev[fo1]{By induction on $\tm$. The clauses of the erasure copy a
  variable and are congruences elsewhere, and the \src{let} clause applies
  to the substituted family because its members still share an erasure.}
\end{proof}

\begin{lemma}[{\rev[mt1]{Elaboration erases to the normalized source}}]\label{lemma:elaberase}\mechnote{\leanname{TPLC.elab\_erVal} and \leanname{TPLC.elab\_erTm}.}~
  \rev[fo1]{For any well-formed environment $\Gamma$, any term $\sm$ of
  \slang{}, any term $\tm$ of \tlang{} and any formula distribution type
  $\gtya \in \FGDType$, if $\ela{\Gamma |-d \sm}{\tm}{\gtya}$ is an
  elaboration of the embedding of $\sm$, then $\erase{\tm} = \anf{\sm}$;
  likewise for values.}
\end{lemma}
\begin{proof}
  \rev[fo1]{By induction on $\sm$, with the elaboration derivation. The
  clauses of $\erase{\cdot}$ are those of Figure~\ref{fig:elaboration} read
  backwards with the types and the evidences dropped. The application, the
  addition and the conditional elaborate to the nested \src{let}s that
  $\anf{\cdot}$ builds, whose bound operands erase under the new binder, and
  the annotation of a $\lambda$, which the lifting does not determine, is
  discarded on both sides. In the \src{let} case the family of elaborated
  bodies has one member per entry of the type of the bound term, which has an entry
  because it is well-formed (Lemmas~\ref{theorem:ep}
  and~\ref{lemma:welltyped-wellform-target}), and all its members elaborate
  the same source body, so they share an erasure and the \src{let} clause of
  the erasure applies.}
\end{proof}

\begin{revblock}[fo2]
The second step rests on a property stronger than the realization of the
annotations of a term. The \tlang{} conditional infers the convex hull of its
branch types, which realizes a static type only when the branches are
$=_{s}$, and the \tlang{} typing does not require it. The property is
therefore about the types that the typing infers.

\begin{definition}[Realizing types]\label{def:realizing-types}\mechnote{\leanname{TPLC.SRRaw}, \leanname{TPLC.SRV} and \leanname{TPLC.SRT}.}
The judgments $\Gamma \rvdash \tu : \gtys$, $\Gamma \rvdash \tv : \gtys$ and
$\Gamma \rvdash \tm : \gtya$, read ``the inferred types realize static
types'', are defined by the following rules, where a premise
$\realizes{\gtys}{\tys}$ or $\realizes{\gtya}{\tya}$ stands for some static
type $\tys$ or $\tya$ of \slang{} (Definition~\ref{def:realization}):
\begin{mathpar}
\inference[(Rr)]{}{\Gamma \rvdash \ttr : \rtype} \and
\inference[(Rb)]{}{\Gamma \rvdash \ttb : \btype} \and
\inference[(Rx)]{\Gamma(\tx) = \gtys & \realizes{\gtys}{\tys}}
  {\Gamma \rvdash \tx : \gtys} \and
\inference[(Rerr$_{\gtys}$)]{\realizes{\gtys}{\tys}}
  {\Gamma \rvdash \errort{\gtys} : \gtys} \and
\inference[(R$\lambda$)]{\Gamma, \tx : \gtys \rvdash \tm : \gtya & \realizes{\gtys}{\tys}}
  {\Gamma \rvdash \trg{\lambda \tx:\gtys. \tm} : \trg{\gtys -> \gtya}} \and
\inference[(R$\tu\mathord{::}\gtys$)]{\Gamma \rvdash \tu : \gtys & \realizes{\gtysp}{\tys}}
  {\Gamma \rvdash \trg{\asc{\ev \tu}{\gtysp}} : \gtysp}
\end{mathpar}
\begin{mathpar}
\inference[(Rv)]{\Gamma \rvdash \tv : \gtys}
  {\Gamma \rvdash \tv : \ds{\gtys[][1]}} \and
\inference[(R$\mathord{::}\gtys$)]{\Gamma \rvdash \tv : \gtys & \realizes{\gtysp}{\tys}}
  {\Gamma \rvdash \trg{\asc{\ev \tv}{\gtysp}} : \ds{\gtysp[][1]}} \and
\inference[(Rapp)]{\Gamma \rvdash \tv : \gtys -> \gtya & \Gamma \rvdash \tw : \gtys}
  {\Gamma \rvdash \tv\;\tw : \gtya} \and
\inference[(R$+$)]{\Gamma \rvdash \tv : \rtype & \Gamma \rvdash \tw : \rtype}
  {\Gamma \rvdash \trg{\add{\tv}{\tw}} : \ds{\rtype^1}}
\end{mathpar}
\begin{mathpar}
\inference[(Rlet)]{
  \Gamma \rvdash \tm : \phty{\phi}{\d{\gtys[i][\p[i]]}[i \in \iSet]} &
  \forall i \in \iSet.~\Gamma, \tx : \gtys[i] \rvdash \tn[i] : \gtya[i] \\
  \realizes{\sugarconj{\phi}{\sum_{i \in \iSet} \p[i] \cdot \gtya[i]}}{\tya}}
  {\Gamma \rvdash \trg{\lett{\tx}{\tm}{\overline{\tn[i]}^{\,i \in \iSet}}} :
   \sugarconj{\phi}{\sum_{i \in \iSet} \p[i] \cdot \gtya[i]}} \and
\inference[(R$\oplus$)]{
  \Gamma \rvdash \tm : \gtya & \Gamma \rvdash \tn : \gtyb \\
  \liftP{\pty}[\cww[1]] = \phi[][1] & \liftP{(1{-}\pty)}[\cww[2]] = \phi[][2] &
  \phi = \phi[][1] \land \phi[][2] \land (\cww[1] + \cww[2] = 1) \\
  \realizes{\sugarconj{\phi}{\cww[1] \cdot \gtya + \cww[2] \cdot \gtyb}}{\tya}}
  {\Gamma \rvdash \trg{{\tm\:} \tpsum {\: \tn}} :
   \sugarconj{\phi}{\cww[1] \cdot \gtya + \cww[2] \cdot \gtyb}} \and
\inference[(Rif)]{
  \Gamma \rvdash \tv : \btype & \Gamma \rvdash \tm : \gtya &
  \Gamma \rvdash \tn : \gtyb \\
  \realizes{\gtya \mathbin{\oplus_{\?}} \gtyb}{\tya}}
  {\Gamma \rvdash \trg{\ite{\tv}{\tm}{\tn}} : \gtya \mathbin{\oplus_{\?}} \gtyb} \and
\inference[(R$\mathord{::}\gtya$)]{\Gamma \rvdash \tm : \gtya & \realizes{\gtyb}{\tya}}
  {\Gamma \rvdash \trg{\asc{\evd \tm}{\gtyb}} : \gtyb} \and
\inference[(Rerr$_{\gtya}$)]{\realizes{\gtya}{\tya}}
  {\Gamma \rvdash \errort{\gtya} : \gtya}
\end{mathpar}
\end{definition}

The rules are those of Figure~\ref{fig:target-type-system} without the
premises on evidences. They ask a type to realize exactly where it is not
determined by the types of the subterms, that is, at variables and $\lambda$
annotations, at ascriptions and errors, and at the rules whose result
combines the types of several subterms, namely $\<let>$, $\oplus$ and the
conditional. Elsewhere realization propagates. A Dirac of a realizing type
realizes, and so do the domain and the codomain of a realizing arrow. The
judgment reads no evidence, so the results below assume both
$\cdot |-d \tm : \gtya$ and $\cdot \rvdash \tm : \gtya$.
\end{revblock}
\rev[fo1]{The judgment assigns a term the type that the typing assigns it,
since both are directed by the syntax; it is preserved by weakening and by
the substitution of a closed value that satisfies it, since the environment
is read only at variables; and the type it assigns realizes a static type,
as do the annotations of the term. These facts are immediate from the rules
and are used below without further mention.}

\rev[fo1]{The next lemma states that elaborating a static program
introduces no type that a static program could not have written.}

\begin{lemma}[{\rev[mt1]{Static realization}}]\label{staticrealize}\mechnote{\leanname{TPLC.srT\_of\_elab} and \leanname{TPLC.srT\_closed}.}~
  \rev[fo1]{For any closed well-typed term $\sm$ of \slang{}, any term $\tm$
  of \tlang{} and any formula distribution type $\gtya \in \FGDType$, if
  $\sm \leadsto \tm : \gtya$ is the elaboration of the embedding of $\sm$,
  then $\cdot \rvdash \tm : \gtya$.}
\end{lemma}
\begin{proof}
  \rev[fo1]{By induction on the typing derivation of \slang{}, with the
  elaboration alongside, for the open statement in which an environment of
  \glang{} that realizes the static environment replaces the empty one. The
  realization premises of (Rlet), (R$\oplus$) and (Rif) are the forward
  direction of Theorem~\ref{theorem:staticeq} at the corresponding source
  node, the elaborated type being the type that \glang{} infers
  (Lemmas~\ref{lemma:elab-typing} and~\ref{lemma:source-det}). In the
  \src{let} case, the coverage clause of the realization of the type of the bound
  term pairs each of its entries with an entry of the static type, which
  gives the static environment under which the induction hypothesis types
  the corresponding body. The application, the addition and the conditional
  are typed through the \src{let}s that their elaboration introduces and by
  weakening. The
  remaining rules propagate. A Dirac of a realizing type realizes, and so do
  the domain and the codomain of a realizing arrow, since $\?$ realizes no
  static type.}
\end{proof}

\begin{lemma}[{\rev[mt1]{No failure}}]\label{lemma:nofailure}\mechnote{\leanname{TPLC.redSt\_of\_red}, \leanname{TPLC.red\_ascV\_below\_ok}, \leanname{TPLC.not\_dascErr\_below}, \leanname{TPLC.econsD\_below} and \leanname{TPLC.routing\_sat\_below}.}~
  \rev[fo1]{For any closed term $\tm$ of \tlang{} that erases to a term of
  \slang{} and any formula distribution type $\gtya \in \FGDType$, if
  $\cdot |-d \tm : \gtya$ and $\cdot \rvdash \tm : \gtya$, then no reduction
  of $\tm$ takes $\trg{(\mathit{D}{::}\gtys)}$ in its error case or
  $\trg{(\mathit{D}{::}\gtya)}$ in its error case, and, for any formula simple type
  $\gtys \in \FGType$, every value $\tv$ with type $\gtys$ in the
  distribution value that $\tm$ reduces to satisfies
  $\cdot \rvdash \tv : \gtys$.}
\end{lemma}
\begin{proof}
  \rev[fo1]{By induction on the reduction derivation, proving at the same
  time that every value of the result satisfies $\rvdash$ at its annotation.
  In each case the three hypotheses pass to the premises. The subterms erase
  and satisfy $\rvdash$, the substituted bodies of $\trg{(\mathit{Dapp})}$ and
  $\trg{(\mathit{Dlet})}$ erase by Lemma~\ref{lemma:erasure-subst}, and the
  typing of the components of a redex is as in the proof of
  Lemma~\ref{lemma:tplc-preservation}. The error case of
  $\trg{(\mathit{D}\mathord{::}\gtys)}$ does not arise. The typing of the
  redex places the two evidences $\eprec$ the annotation $\gtys$ of the
  value, which realizes a static type, so both weakly realize it
  (Lemma~\ref{lemma:weak-realization}) and their composition is defined
  (item~1 of Lemma~\ref{static-composition-defined2}). The error case of
  $\trg{(\mathit{D}\mathord{::}\gtya)}$ does not arise either. By
  Lemma~\ref{lemma:tplc-preservation} the runtime type $\gtyap$ of the
  result of the bound term is well-formed and reorders into $\gtya$, so
  $\initReorder{\gtyap}{\gtya}$ is well-formed and $\eprec \gtya$
  (Lemma~\ref{reorderevidence2}), $\evd \eprec \gtya$ by its validity, and
  $\gtya$ realizes a static type; by Lemma~\ref{lemma:weak-realization} and
  item~2 of Lemma~\ref{static-composition-defined2} the two are consistent,
  and by Lemma~\ref{consistency-defined} the routing formula is satisfiable.
  In $\trg{(\mathit{Dlet})}$ each entry of the routing evidence coerces a
  value to the type it already has, since the entries of
  $\initReorder{\gtyap}{\gtya}$ pair equal entries,
  and in the first case of $\trg{(\mathit{D}\mathord{::}\gtya)}$ to an entry
  of $\gtyb$, which realizes a static type by the coverage clause of
  realization; the coerced values satisfy $\rvdash$ and erase, and the
  induction hypothesis applies to the coercions and to the substituted
  bodies. Since no error is produced, the error-propagation rules are not
  reached.}
\end{proof}

\begin{lemma}[{\rev[fo1]{Regrouping by fibers}}]\label{lemma:fibers}\mechnote{\leanname{TPLC.sum\_fiber\_transfer}.}~
  \rev[fo1]{For any finite families $(a_i, u_i)_{i \in \iSet}$ and
  $(b_j, u'_j)_{j \in \jSet}$ of real weights and values, any function $F$
  from values to reals, any finite family $(\cww[k])_{k \in K}$ of reals,
  any function $g$ from $K$ to $\iSet$, any family $(q_i)_{i \in \iSet}$ of
  reals and any function $G$ from $\iSet$ to reals,
  \begin{enumerate}
  \item if $\sum_{i \mid u_i = u} a_i = \sum_{j \mid u'_j = u} b_j$ for every
    value $u$, then
    $\sum_{i \in \iSet} a_i \cdot F(u_i) = \sum_{j \in \jSet} b_j \cdot F(u'_j)$;
  \item if $\sum_{k \mid g(k) = i} \cww[k] = q_i$ for every $i \in \iSet$,
    then $\sum_{k \in K} \cww[k] \cdot G(g(k)) = \sum_{i \in \iSet} q_i \cdot G(i)$.
  \end{enumerate}}
\end{lemma}
\begin{proof}
  \rev[fo1]{Both sums regroup by the value, and each equals
  $\sum_{u} \bigl(\sum_{i \mid u_i = u} a_i\bigr) \cdot F(u)$ over the
  finitely many values that occur in the two families.}
\end{proof}

\begin{lemma}[{\rev[mt1]{Source determinism modulo measure}}]\label{lemma:sdetmeas}\mechnote{\leanname{TPLC.red\_srcMass\_det}.}~
  \rev[mt1]{Two reductions of the same closed \slang{} term yield distribution
  values with the same measure.}
\end{lemma}
\begin{proof}
  \rev[fo1]{By induction on the first derivation; the second ends in the
  same rule, since the rules of Figure~\ref{fig:static-distribution-reduction}
  are directed by the syntax. In the \src{let} case the induction
  hypothesis on the bound term gives the two outcome distributions the same
  measure, the induction hypothesis on the body makes the measure of a run
  of the body at an outcome independent of the run, and
  Lemma~\ref{lemma:fibers} regroups the two weighted unions by the value of
  the outcome.}
\end{proof}

\begin{lemma}[{\rev[mt1]{Normalization preserves the measure}}]\label{lemma:anfmeas}\mechnote{\leanname{TPLC.red\_anf}.}~
  \rev[fo1]{For any closed term $\sm$ of \slang{} and any distribution value
  $\sV$, if $\sm \snreds{*} \sV$, then $\anf{\sm} \snreds{*} \sV[']$ for
  some distribution value $\sV[']$ that gives each value $\static{v}$ the
  total probability that $\sV$ gives to the outcomes normalizing to
  $\static{v}$.}
\end{lemma}
\begin{proof}
  \rev[fo1]{By induction on the reduction, proving also that every outcome
  of the run of $\anf{\sm}$ is the normalization of an outcome of the run of
  $\sm$. The three constructs that $\anf{\cdot}$ rewrites acquire \src{let}s
  whose bound terms are ascriptions of values, which reduce to Diracs, so one
  \src{let} step per binding re-creates the original redex; the ascriptions
  that $\anf{\cdot}$ inserts leave a value unchanged, and the remaining
  constructs commute with $\anf{\cdot}$. The application and the \src{let}
  use closedness, since $\anf{\cdot}$ commutes with the substitution of a closed
  value, and the outcomes of a closed term are closed. In the \src{let} case
  the run of $\anf{\sm}$ needs a run of the normalized body at every outcome
  of the normalized bound term, including the outcomes of probability zero;
  each is the normalization of an original outcome, whose run of the body
  the induction hypothesis normalizes, and since two runs of the same body
  have the same measure (Lemma~\ref{lemma:sdetmeas}), the measure of the
  result regroups by Lemma~\ref{lemma:fibers}. On values, $\anf{\cdot}$
  rewrites only the annotation and the body of a $\lambda$, which is why the
  two measures agree after normalizing the source outcomes.}
\end{proof}

\rev[fo1]{The two routing evidences of the next lemma are those of rules
$\trg{(\mathit{Dlet})}$ and $\trg{(\mathit{D}\mathord{::}\gtya)}$, where
$\gtyap$ is the runtime type of the result of the bound term.}

\begin{lemma}[{\rev[mt1]{Routing preserves marginals}}]\label{lemma:marginals}\mechnote{\leanname{TPLC.meetD\_left\_marginal}, \leanname{TPLC.reorderD\_left\_marginal}, \leanname{TPLC.routing\_left\_marginal}, \leanname{TPLC.reorderD\_left\_marginal\_tag} and \leanname{TPLC.routing\_left\_marginal\_tag}.}~
  \rev[fo1]{For any formula distribution types
  $\gtyap = \phty{\phi[][1]}{\d{\gtys[i][\p[i]]}[i \in \iSet]}$ and $\gtya$
  in $\FGDType$, any evidence $\evd$ and any distribution
  evidence $\phty{\pphi[2]}{\d{\ev[k][\cww[k]]}[k \in \kSet]}$ that is either
  $\initReorder{\gtyap}{\gtya}$ or
  $(\initReorder{\gtyap}{\gtya}) \trans{} \evd$, for every solution $\nu$ of
  $\pphi[2]$ there is a solution $\nu'$ of $\phi[][1]$ such that
  $\sum_{k \mid \projl{\cww[k]} = i} \nu(\cww[k]) = \nu'(\p[i])$ for every
  $i \in \iSet$, and $\nu(\cww[k]) \geq 0$ for every $k \in \kSet$.}
\end{lemma}
\begin{proof}
  \rev[fo1]{The formula of $\initReorder{\gtyap}{\gtya}$ states its two
  marginal clauses and the nonnegativity of its weights
  (Lemma~\ref{reorderevidence2}), and the formula of a composition states
  the marginal clauses with respect to its two operands. A solution of the
  composition aggregates, along the left pointers of its entries, to a
  solution of $\initReorder{\gtyap}{\gtya}$, which aggregates along its own
  left pointers to a solution of $\phi[][1]$; the left pointer of an entry of
  the composition is that of its first operand, so the two aggregations
  compose.}
\end{proof}

\begin{lemma}[{\rev[fo1]{Coverage}}]\label{lemma:erasecover}\mechnote{\leanname{TPLC.erasure\_cover\_red} and \leanname{TPLC.erasure\_cover}.}~
  \rev[fo1]{For any closed term $\tm$ of \tlang{} that erases to a term of
  \slang{} and any formula distribution type $\gtya \in \FGDType$, if
  $\cdot |-d \tm : \gtya$ and $\cdot \rvdash \tm : \gtya$, then every outcome of
  a distribution value reachable from $\tm$ erases to a value of \slang{}.}
\end{lemma}
\begin{proof}
  \rev[fo1]{Lemma~\ref{lemma:nofailure} excludes the error cases of the two
  ascription rules, and the other error rules do not apply to a term that
  has an erasure, since the erasure has no clause for an error; by induction
  on the resulting derivation, for every erasure of the term. An outcome is the value of a $\trg{(\mathit{Dv})}$ step, which erases
  because the term does, a literal built by $\trg{(\mathit{D+})}$, or a
  coercion of an outcome of the result of the bound term, which erases by the
  induction hypothesis and Lemma~\ref{lemma:erasecoerce}; the substituted
  bodies of $\trg{(\mathit{Dapp})}$ and $\trg{(\mathit{Dlet})}$ erase by
  Lemma~\ref{lemma:erasure-subst}.}
\end{proof}

\rev[fo1]{The comparison in the following theorem is literal. The erasure
of the target term is the term that runs on the source side.}

\begin{theorem}[{\rev[mt1]{Erasure preserves measure}}]\label{erasuremeas}\mechnote{\leanname{TPLC.erasure\_meas\_red} and \leanname{TPLC.erasure\_meas}.}~
  \rev[fo1]{For any closed term $\tm$ of \tlang{} that erases to a term of
  \slang{}, any formula distribution type $\gtya \in \FGDType$ and any
  distribution value $\V$, if $\cdot |-d \tm : \gtya$,
  $\cdot \rvdash \tm : \gtya$ and $\tm \jreds[*][] \V$, then for every
  solution $\cww$ of the formula of $\V$, every value $\static{v}$ of
  \slang{} and every distribution value $\sV$ with
  $\erase{\tm} \snreds{*} \sV$,
  $\erase{\V}_{\cww}(\static{v}) = \sV(\static{v})$.}
\end{theorem}
\begin{proof}
  \begin{revblock}[fo1]
  Lemma~\ref{lemma:nofailure} excludes the error cases of
  $\trg{(\mathit{D}\mathord{::}\gtys)}$ and
  $\trg{(\mathit{D}\mathord{::}\gtya)}$ from the derivation of
  $\tm \jreds[*][] \V$, and the remaining error rules cannot occur because
  $\tm$ has an erasure, which has no clause for an error term or value. We
  prove, by induction on such a derivation and with no hypothesis on types, that for every term $\sn$ with
  $\erase{\tm} = \sn$ and every run $\sn \snreds{*} \sV$,
  $\erase{\V}_{\cww}(\static{v}) = \sV(\static{v})$ for every solution
  $\cww$ of the formula of $\V$ and every value $\static{v}$. The
  subderivations of such a derivation take no error case either, so the
  induction hypothesis applies to them directly. Recall that
  $\erase{\V}_{\cww}(\static{v})$ sums $\cww[j]$ over the outcomes $\tv[j]$ of
  $\V$ with $\erase{\tv[j]} = \static{v}$. The probabilities of a distribution
  value of \slang{} are nonnegative, since those of a program lie in
  $[0,1]$.

  \begin{case}[rules that do not apply]
    An error has no erasure, so $\trg{(\mathit{Derr})}$,
    $\trg{(\mathit{Derr}\mathord{::}\gtys)}$ and the four error-propagation
    rules, whose redex contains an error, are not cases; nor is
    $\trg{(\mathit{D}\oplus)}$ with probability $\?$, since the erasure of a
    choice requires a concrete probability; and the error cases of
    $\trg{(\mathit{D}\mathord{::}\gtys)}$ and
    $\trg{(\mathit{D}\mathord{::}\gtya)}$ are excluded by hypothesis.
  \end{case}

  \begin{case}[values, addition, conditional, monotonicity and value ascription]
    In $\trg{(\mathit{Dv})}$, $\V = \dt{\pt{\tv}}$ with $\erase{\tv} = \v$,
    and the run of $\sn = \v$ ends in the value rule, with
    $\sV = \dt{\pt{\v}}$; both measures are $1$ at $\v$ and $0$ elsewhere,
    the erasure being a function. In $\trg{(\mathit{D+})}$ both results are
    Diracs on $r_1 + r_2$. In $\trg{(\mathit{Dit})}$ and
    $\trg{(\mathit{Dif})}$, $\sn$ is a conditional on the same literal, the
    run of $\sn$ ends in the rule for that literal with a run of the
    corresponding branch, and the induction hypothesis on the branch
    concludes. In $\trg{(\mathit{Dmon})}$ the induction hypothesis concludes.
    In the first case of $\trg{(\mathit{D}\mathord{::}\gtys)}$,
    $\tm = \asc{\ev[2] (\asc{\ev[1] \tu}{\gtys})}{\gtysp}$ erases to
    $\asc{\v}{\ph}$ with $\v = \erase{\asc{\ev[1] \tu}{\gtys}}$, the result is
    the Dirac on $\asc{\ev[3] \tu}{\gtysp}$, which erases to $\v$ by
    Lemma~\ref{lemma:erasecoerce}, and the run of $\asc{\v}{\ph}$ ends in the
    rule for value ascriptions, with the Dirac on $\v$.
  \end{case}

  \begin{case}[probabilistic choice]
    $\tm = \tm[1] \tpsum[\ps] \tm[2]$ with a concrete $\ps$ and
    $\sn = \sn[1] \spsum \sn[2]$ with $\erase{\tm[i]} = \sn[i]$. The
    premises give $\tm[i] \jreds[*][] \rctx[\phi[][i]] \V[i]$, and the result
    is $\cww[1] \cdot \V[1] + \cww[2] \cdot \V[2]$ with formula
    $\phi[][1] \land \phi[][2] \land \phi[][\oplus]$, where
    $\phi[][\oplus]$ fixes $\cww[1] = \ps$ and $\cww[2] = 1 - \ps$. A solution
    of it restricts to solutions $\nu_1$ of $\phi[][1]$ and $\nu_2$ of
    $\phi[][2]$. The run of $\sn$ ends in the choice rule, with runs
    $\sn[i] \snreds{*} \sV[i]$ and $\sV = \ps \cdot \sV[1] + (1 - \ps) \cdot \sV[2]$.
    Then
    $\erase{\V}_{\cww}(\static{v}) = \ps \cdot \erase{\V[1]}_{\nu_1}(\static{v}) + (1 - \ps) \cdot \erase{\V[2]}_{\nu_2}(\static{v})$,
    the induction hypotheses on the two branches, at $\nu_1$ and $\nu_2$ and
    with the two source runs, replace the two erased measures by
    $\sV[1](\static{v})$ and $\sV[2](\static{v})$, and the result is
    $\sV(\static{v})$.
  \end{case}

  \begin{case}[application]
    $\tm = \trg{(\asc{\ev (\lambda \tx{:}\gtysp.\,\tm[0])}{\gtys -> \gtya})\;\tv}$
    and $\sn = (\lambda x{:}\ph.\,\sn[0])\;\v$ with $\erase{\tm[0]} = \sn[0]$
    and $\erase{\tv} = \v$. The premises give
    $\asc{\cdom(\ev)\,\tv}{\gtysp} \nreds{}{1} \dt{\pt{\tw}}$ and a reduction
    of $\ssub{\asc{\ccod(\ev)\,\tm[0]}{\gtya}}{\tw}{\tx}$ to $\V$. Since
    $\tv$ erases, it is not an error, so the first premise takes the first
    case of $\trg{(\mathit{D}\mathord{::}\gtys)}$. Thus $\tw$ is not an error and
    $\erase{\tw} = \v$ by Lemma~\ref{lemma:erasecoerce}. Hence
    $\mathit{sub}$ is the substitution, and
    $\asc{\ccod(\ev)\,\tm[0]{}[\tw/\tx]}{\gtya}$ erases to
    $\asc{\sn[0]{}[\v/x]}{\ph}$ by Lemma~\ref{lemma:erasure-subst}. The run of
    $\sn$ ends in the application rule, with a run
    $\sn[0]{}[\v/x] \snreds{*} \sV$, which the rule for term ascriptions
    extends to a run of $\asc{\sn[0]{}[\v/x]}{\ph}$ with the same result; the
    induction hypothesis on the second premise, with that run, concludes.
  \end{case}

  \begin{case}[term ascription]
    $\tm = \asc{\evd \tm[0]}{\gtyb}$, in the first case of
    $\trg{(\mathit{D}\mathord{::}\gtya)}$, and $\sn = \asc{\sn[0]}{\ph}$ with
    $\erase{\tm[0]} = \sn[0]$. The premises give
    $\tm[0] \jreds[*][] \rctx[\phi[][1]] \d{\tv[i][\p[i]]}[i \in \iSet]$, of
    runtime type $\gtyap$, the routing evidence
    $(\initReorder{\gtyap}{\gtya}) \trans{} \evd = \phty{\pphi[2]}{\d{\ev[k][\cww[k]]}[k \in \kSet]}$
    and, for each entry $k$ with $i = \projl{\cww[k]}$ and
    $j = \projr{\cww[k]}$, a coercion
    $\asc{\ev[k] \tv[i]}{\gtysp[j]} \nreds{}{1} \dt{\pt{\tw[k]}}$; the result
    is $\rctx[\pphi[2]] \d{\tw[k][\cww[k]]}[k \in \kSet]$. The run of $\sn$
    ends in the rule for term ascriptions, with a run
    $\sn[0] \snreds{*} \sV$. By Lemma~\ref{lemma:erasecover} applied to the
    first premise, each $\tv[i]$ erases to some $\v[i]$, and by
    Lemma~\ref{lemma:erasecoerce}, $\erase{\tw[k]} = \v[\projl{\cww[k]}]$. Let
    $\nu$ be a solution of $\pphi[2]$; Lemma~\ref{lemma:marginals} gives a
    solution $\nu'$ of $\phi[][1]$ with
    $\sum_{k \mid \projl{\cww[k]} = i} \nu(\cww[k]) = \nu'(\p[i])$ for every
    $i$. Then
    \[
      \erase{\V}_{\nu}(\static{v})
      = \sum_{k \mid \erase{\tw[k]} = \static{v}} \nu(\cww[k])
      = \sum_{i \mid \v[i] = \static{v}}\ \sum_{k \mid \projl{\cww[k]} = i} \nu(\cww[k])
      = \sum_{i \mid \v[i] = \static{v}} \nu'(\p[i]),
    \]
    which is the erased measure, at $\nu'$, of the result of the bound term, and
    the induction hypothesis on the first premise, at $\nu'$ and with the run
    of $\sn[0]$, equates it with $\sV(\static{v})$.
  \end{case}

  \begin{case}[let]
    $\tm = \trg{\lett{\tx}{\tm[0]}{\overline{\tn[j]}^{\,j \in \jSet}}}$ and
    $\sn = \lett{x}{\sn[0]}{\sn[1]}$ with $\erase{\tm[0]} = \sn[0]$,
    $\jSet \neq \emptyset$ and $\erase{\tn[j]} = \sn[1]$ for every $j$. The
    premises give
    $\tm[0] \jreds[*][] \rctx[\phi[][1]] \d{\tv[i][\p[i]]}[i \in \iSet]$, of
    runtime type $\gtyap$, with $\tm[0]$ of type
    $\gtya = \phty{\phi}{\d{\gtysp[j][\p[j]]}[j \in \jSet]}$, the reordering
    $\initReorder{\gtyap}{\gtya} = \phty{\pphi[2]}{\d{\ev[k][\cww[k]]}[k \in \kSet]}$
    and, for each entry $k$ with $i = \projl{\cww[k]}$ and
    $j = \projr{\cww[k]}$, a coercion
    $\asc{\ev[k] \tv[i]}{\gtysp[j]} \nreds{}{1} \dt{\pt{\tw[k]}}$ and a
    reduction of $\ssub{\tn[j]}{\tw[k]}{\tx}$ to $\rctx[\phi[][k]] \V[k]$;
    the result is
    $\rctx[\pphi[2] \land \bigwedge_{k} (\cww[k] > 0 \Rightarrow \phi[][k])] \ssum[k \in \kSet] \cww[k] \cdot \V[k]$.
    A solution $\cww$ of its formula is a solution $\nu$ of $\pphi[2]$ together
    with a solution $\nu_k$ of $\phi[][k]$ for each $k$ with
    $\nu(\cww[k]) > 0$, and it gives the outcome $c$ of $\V[k]$ the weight
    $\nu(\cww[k]) \cdot \nu_k(c)$. The run of $\sn$ ends in the \src{let}
    rule, with a run
    $\sn[0] \snreds{*} \sV[0] = \dt{\pt{u_{i'}}[\ps[i']] \mid i' \in \iSet'}$,
    a run $\sn[1]{}[u_{i'}/x] \snreds{*} \mathscr{W}_{i'}$ for every
    $i' \in \iSet'$, including the outcomes of probability $0$, and
    $\sV = \ssum[i' \in \iSet'] \ps[i'] \cdot \mathscr{W}_{i'}$. Fix a value
    $\static{v}$.

    By Lemma~\ref{lemma:erasecover}, each $\tv[i]$ erases to some $\v[i]$;
    since $\tv[i]$ is not an error, $\tw[k]$ is not one,
    $\erase{\tw[k]} = \v[\projl{\cww[k]}]$ (Lemma~\ref{lemma:erasecoerce}),
    $\mathit{sub}$ is the substitution, and $\tn[j]{}[\tw[k]/\tx]$ erases to
    $\sn[1]{}[\v[\projl{\cww[k]}]/x]$ (Lemma~\ref{lemma:erasure-subst}).
    Lemma~\ref{lemma:marginals} gives a solution $\nu'$ of $\phi[][1]$ with
    $\sum_{k \mid \projl{\cww[k]} = i} \nu(\cww[k]) = \nu'(\p[i])$ for every
    $i$, and $\nu(\cww[k]) \geq 0$, so $\nu'(\p[i]) \geq 0$. The induction
    hypothesis on the bound term, at $\nu'$ and with the run of $\sn[0]$,
    gives $\sum_{i \mid \v[i] = u} \nu'(\p[i]) = \sV[0](u)$ for every value
    $u$. The induction hypothesis on the body at an entry $k$ with
    $\nu(\cww[k]) > 0$ holds for every run of
    $\sn[1]{}[\v[\projl{\cww[k]}]/x]$. That is, for every $i'$ with
    $u_{i'} = \v[\projl{\cww[k]}]$,
    $\erase{\V[k]}_{\nu_k}(\static{v}) = \mathscr{W}_{i'}(\static{v})$.

    For a value $u$, let $F(u)$ be $\mathscr{W}_{i'}(\static{v})$ for some
    $i'$ with $u_{i'} = u$, and $0$ if there is none. If $\sV[0](u) > 0$,
    then $F(u) = \mathscr{W}_{i'}(\static{v})$ for every $i'$ with
    $u_{i'} = u$. By the identity above some $i$ has $\v[i] = u$ and
    $\nu'(\p[i]) > 0$, hence some entry $k$ has $\projl{\cww[k]} = i$ and
    $\nu(\cww[k]) > 0$, and the induction hypothesis at that entry equates
    every such $\mathscr{W}_{i'}(\static{v})$ with
    $\erase{\V[k]}_{\nu_k}(\static{v})$. This is where the determinism of
    \slang{} enters. The induction hypothesis at one entry identifies the
    measures of all the runs of the body at the same outcome.

    On the target side,
    $\erase{\V}_{\cww}(\static{v}) = \sum_{k} \nu(\cww[k]) \cdot \erase{\V[k]}_{\nu_k}(\static{v})$,
    where an entry $k$ with $\nu(\cww[k]) = 0$ contributes $0$ whatever $\nu_k$
    is. For $\nu(\cww[k]) > 0$ the identities above give
    $\sV[0](\v[\projl{\cww[k]}]) \geq \nu'(\p[\projl{\cww[k]}]) \geq \nu(\cww[k]) > 0$,
    so $\erase{\V[k]}_{\nu_k}(\static{v}) = F(\v[\projl{\cww[k]}])$;
    regrouping the entries by the outcome they route, with
    Lemma~\ref{lemma:marginals} and the second part of
    Lemma~\ref{lemma:fibers}, the sum equals
    $\sum_{i} \nu'(\p[i]) \cdot F(\v[i])$. On the source side,
    $\sV(\static{v}) = \sum_{i'} \ps[i'] \cdot \mathscr{W}_{i'}(\static{v}) = \sum_{i'} \ps[i'] \cdot F(u_{i'})$,
    since a term with $\ps[i'] = 0$ contributes $0$, and $\ps[i'] > 0$ gives
    $\sV[0](u_{i'}) \geq \ps[i'] > 0$. The two sums agree by
    Lemma~\ref{lemma:fibers}, whose hypothesis is the identity
    $\sum_{i \mid \v[i] = u} \nu'(\p[i]) = \sV[0](u) = \sum_{i' \mid u_{i'} = u} \ps[i']$. \qedhere
  \end{case}
  \end{revblock}
\end{proof}

\begin{theoremproof}{Theorem~\ref{tm:dynamic-cons-ext} (Dynamic conservative extension of \tlang \wrt \slang)}
  \rev[fo1]{The elaboration $\tmp$ of the embedding of $\sm$ is closed, with
  $\cdot |-d \tmp : \gtya$ by Lemma~\ref{theorem:ep} and
  $\cdot \rvdash \tmp : \gtya$ by Lemma~\ref{staticrealize}, and
  $\erase{\tmp} = \anf{\sm}$ by Lemma~\ref{lemma:elaberase}. By
  Lemma~\ref{lemma:anfmeas}, the run $\sm \snreds{*} \sV[1]$ yields a run
  $\anf{\sm} \snreds{*} \sV[']$ such that $\sV['](\static{v})$ is the sum of
  the probabilities that $\sV[1]$ gives to the outcomes $\static{v_i}$ with
  $\anf{\static{v_i}} = \static{v}$. Theorem~\ref{erasuremeas}, applied to
  the run $\tmp \jreds[*][] \V[2]$, to its erasure $\anf{\sm}$ and to that
  run, gives $\erase{\V[2]}_{\cww}(\static{v}) = \sV['](\static{v})$ for
  every solution $\cww$ of the formula of $\V[2]$, which is item~1. For
  item~2, $\anf{\cdot}$ rewrites only the annotation and the body of a
  $\lambda$, so $\anf{\static{v_i}}$ is a real number or a boolean exactly
  when $\static{v_i}$ is that same value, and the sum of item~1 at such a
  value is $\sV[1](\static{v})$.}
\end{theoremproof}

\label{lastpage01}


\begin{thebibliography}{84}

\ifx \showCODEN    \undefined \def \showCODEN     #1{\unskip}     \fi
\ifx \showDOI      \undefined \def \showDOI       #1{#1}\fi
\ifx \showISBNx    \undefined \def \showISBNx     #1{\unskip}     \fi
\ifx \showISBNxiii \undefined \def \showISBNxiii  #1{\unskip}     \fi
\ifx \showISSN     \undefined \def \showISSN      #1{\unskip}     \fi
\ifx \showLCCN     \undefined \def \showLCCN      #1{\unskip}     \fi
\ifx \shownote     \undefined \def \shownote      #1{#1}          \fi
\ifx \showarticletitle \undefined \def \showarticletitle #1{#1}   \fi
\ifx \showURL      \undefined \def \showURL       {\relax}        \fi
\providecommand\bibfield[2]{#2}
\providecommand\bibinfo[2]{#2}
\providecommand\natexlab[1]{#1}
\providecommand\showeprint[2][]{arXiv:#2}

\bibitem[\protect\citeauthoryear{Avanzini, {Dal Lago}, and Ghyselen}{Avanzini
  et~al\mbox{.}}{2019}]%
        {10.5555/3470152.3470193}
\bibfield{author}{\bibinfo{person}{Martin Avanzini}, \bibinfo{person}{Ugo {Dal
  Lago}}, {and} \bibinfo{person}{Alexis Ghyselen}.}
  \bibinfo{year}{2019}\natexlab{}.
\newblock \showarticletitle{Type-Based Complexity Analysis of Probabilistic
  Functional Programs}. In \bibinfo{booktitle}{\emph{Proceedings of the
  ACM/IEEE Symposium on Logic in Computer Science (LICS 2019)}}.
  \bibinfo{publisher}{IEEE}, \bibinfo{pages}{1--13}.
\newblock

\bibitem[\protect\citeauthoryear{{Azevedo de Amorim}, Fredrikson, and
  Jia}{{Azevedo de Amorim} et~al\mbox{.}}{2020}]%
        {azevedo:lics2020}
\bibfield{author}{\bibinfo{person}{Arthur {Azevedo de Amorim}},
  \bibinfo{person}{Matt Fredrikson}, {and} \bibinfo{person}{Limin Jia}.}
  \bibinfo{year}{2020}\natexlab{}.
\newblock \showarticletitle{Reconciling Noninterference and Gradual Typing}. In
  \bibinfo{booktitle}{\emph{Proceedings of the ACM/IEEE Symposium on Logic in
  Computer Science (LICS 2020)}}. \bibinfo{publisher}{ACM},
  \bibinfo{pages}{116--129}.
\newblock

\bibitem[\protect\citeauthoryear{{Ba{\~n}ados Schwerter}, Garcia, and
  Tanter}{{Ba{\~n}ados Schwerter} et~al\mbox{.}}{2016}]%
        {banadosAl:jfp2016}
\bibfield{author}{\bibinfo{person}{Felipe {Ba{\~n}ados Schwerter}},
  \bibinfo{person}{Ronald Garcia}, {and} \bibinfo{person}{{\'E}ric Tanter}.}
  \bibinfo{year}{2016}\natexlab{}.
\newblock \showarticletitle{Gradual Type-and-Effect Systems}.
\newblock \bibinfo{journal}{\emph{Journal of Functional Programming}}
  \bibinfo{volume}{26} (\bibinfo{year}{2016}), \bibinfo{pages}{e19}.
\newblock

\bibitem[\protect\citeauthoryear{Barthe, Gr{\'e}goire, and
  {Zanella-B{\'e}guelin}}{Barthe et~al\mbox{.}}{2009}]%
        {Barthe:2009}
\bibfield{author}{\bibinfo{person}{Gilles Barthe}, \bibinfo{person}{Benjamin
  Gr{\'e}goire}, {and} \bibinfo{person}{Santiago {Zanella-B{\'e}guelin}}.}
  \bibinfo{year}{2009}\natexlab{}.
\newblock \showarticletitle{Formal Certification of Code-Based Cryptographic
  Proofs}. In \bibinfo{booktitle}{\emph{Proceedings of the ACM SIGPLAN-SIGACT
  Symposium on Principles of Programming Languages (POPL 2009)}}.
  \bibinfo{publisher}{ACM}, \bibinfo{pages}{90--101}.
\newblock

\bibitem[\protect\citeauthoryear{Bertsimas and Tsitsiklis}{Bertsimas and
  Tsitsiklis}{1997}]%
        {bertsimasTsitsiklis:1997}
\bibfield{author}{\bibinfo{person}{Dimitris Bertsimas} {and}
  \bibinfo{person}{John~N. Tsitsiklis}.} \bibinfo{year}{1997}\natexlab{}.
\newblock \bibinfo{booktitle}{\emph{Introduction to Linear Optimization}}.
\newblock \bibinfo{publisher}{Athena Scientific}.
\newblock

\bibitem[\protect\citeauthoryear{Borgstr{\"o}m, {Dal Lago}, Gordon, and
  Szymczak}{Borgstr{\"o}m et~al\mbox{.}}{2016}]%
        {DBLP:conf/icfp/BorgstromLGS16}
\bibfield{author}{\bibinfo{person}{Johannes Borgstr{\"o}m},
  \bibinfo{person}{Ugo {Dal Lago}}, \bibinfo{person}{Andrew~D. Gordon}, {and}
  \bibinfo{person}{Marcin Szymczak}.} \bibinfo{year}{2016}\natexlab{}.
\newblock \showarticletitle{A Lambda-Calculus Foundation for Universal
  Probabilistic Programming}. In \bibinfo{booktitle}{\emph{Proceedings of the
  ACM SIGPLAN International Conference on Functional Programming (ICFP 2016)}}.
  \bibinfo{publisher}{ACM}, \bibinfo{pages}{33--46}.
\newblock

\bibitem[\protect\citeauthoryear{Boston, Sampson, Grossman, and Ceze}{Boston
  et~al\mbox{.}}{2015}]%
        {DBLP:conf/oopsla/BostonSGC15}
\bibfield{author}{\bibinfo{person}{Brett Boston}, \bibinfo{person}{Adrian
  Sampson}, \bibinfo{person}{Dan Grossman}, {and} \bibinfo{person}{Luis Ceze}.}
  \bibinfo{year}{2015}\natexlab{}.
\newblock \showarticletitle{Probability Type Inference for Flexible Approximate
  Programming}. In \bibinfo{booktitle}{\emph{Proceedings of the ACM SIGPLAN
  International Conference on Object-Oriented Programming, Systems, Languages,
  and Applications (OOPSLA 2015)}}. \bibinfo{publisher}{ACM},
  \bibinfo{pages}{470--487}.
\newblock

\bibitem[\protect\citeauthoryear{Breuvart and {Dal Lago}}{Breuvart and {Dal
  Lago}}{2018}]%
        {breuvartDalLago:ppdp2018}
\bibfield{author}{\bibinfo{person}{Flavien Breuvart} {and} \bibinfo{person}{Ugo
  {Dal Lago}}.} \bibinfo{year}{2018}\natexlab{}.
\newblock \showarticletitle{On Intersection Types and Probabilistic Lambda
  Calculi}. In \bibinfo{booktitle}{\emph{Proceedings of the International
  Symposium on Principles and Practice of Declarative Programming (PPDP
  2018)}}. \bibinfo{publisher}{ACM}, \bibinfo{pages}{8:1--8:13}.
\newblock

\bibitem[\protect\citeauthoryear{Carbin, Misailovic, and Rinard}{Carbin
  et~al\mbox{.}}{2016}]%
        {DBLP:journals/cacm/CarbinMR16}
\bibfield{author}{\bibinfo{person}{Michael Carbin}, \bibinfo{person}{Sasa
  Misailovic}, {and} \bibinfo{person}{Martin~C. Rinard}.}
  \bibinfo{year}{2016}\natexlab{}.
\newblock \showarticletitle{Verifying Quantitative Reliability for Programs
  That Execute on Unreliable Hardware}.
\newblock \bibinfo{journal}{\emph{{Communications of the ACM}}}
  \bibinfo{volume}{59}, \bibinfo{number}{8} (\bibinfo{year}{2016}),
  \bibinfo{pages}{83--91}.
\newblock

\bibitem[\protect\citeauthoryear{Castagna and Lanvin}{Castagna and
  Lanvin}{2017}]%
        {castagnaLanvin:icfp2017}
\bibfield{author}{\bibinfo{person}{Giuseppe Castagna} {and}
  \bibinfo{person}{Victor Lanvin}.} \bibinfo{year}{2017}\natexlab{}.
\newblock \showarticletitle{Gradual Typing with Union and Intersection Types}.
\newblock \bibinfo{journal}{\emph{Proceedings of the ACM on Programming
  Languages}} \bibinfo{volume}{1}, \bibinfo{number}{ICFP}
  (\bibinfo{year}{2017}), \bibinfo{pages}{41:1--41:28}.
\newblock

\bibitem[\protect\citeauthoryear{Castagna, Lanvin, Petrucciani, and
  Siek}{Castagna et~al\mbox{.}}{2019}]%
        {castagnaAl:popl2019}
\bibfield{author}{\bibinfo{person}{Giuseppe Castagna}, \bibinfo{person}{Victor
  Lanvin}, \bibinfo{person}{Tommaso Petrucciani}, {and}
  \bibinfo{person}{Jeremy~G. Siek}.} \bibinfo{year}{2019}\natexlab{}.
\newblock \showarticletitle{Gradual Typing: A New Perspective}.
\newblock \bibinfo{journal}{\emph{Proceedings of the ACM on Programming
  Languages}} \bibinfo{volume}{3}, \bibinfo{number}{POPL}
  (\bibinfo{year}{2019}), \bibinfo{pages}{16:1--16:32}.
\newblock

\bibitem[\protect\citeauthoryear{Claret, Rajamani, Nori, Gordon, and
  Borgstr{\"o}m}{Claret et~al\mbox{.}}{2013}]%
        {DBLP:conf/sigsoft/ClaretRNGB13}
\bibfield{author}{\bibinfo{person}{Guillaume Claret},
  \bibinfo{person}{Sriram~K. Rajamani}, \bibinfo{person}{Aditya~V. Nori},
  \bibinfo{person}{Andrew~D. Gordon}, {and} \bibinfo{person}{Johannes
  Borgstr{\"o}m}.} \bibinfo{year}{2013}\natexlab{}.
\newblock \showarticletitle{{Bayesian} Inference Using Data Flow Analysis}. In
  \bibinfo{booktitle}{\emph{Proceedings of the Joint Meeting of the European
  Software Engineering Conference and the ACM SIGSOFT Symposium on the
  Foundations of Software Engineering (ESEC/FSE 2013)}}.
  \bibinfo{publisher}{ACM}, \bibinfo{pages}{92--102}.
\newblock

\bibitem[\protect\citeauthoryear{{Dal Lago}, Faggian, and {Ronchi Della
  Rocca}}{{Dal Lago} et~al\mbox{.}}{2021}]%
        {dalLagoAl:popl2021}
\bibfield{author}{\bibinfo{person}{Ugo {Dal Lago}}, \bibinfo{person}{Claudia
  Faggian}, {and} \bibinfo{person}{Simona {Ronchi Della Rocca}}.}
  \bibinfo{year}{2021}\natexlab{}.
\newblock \showarticletitle{Intersection Types and (Positive) Almost-Sure
  Termination}.
\newblock \bibinfo{journal}{\emph{Proceedings of the ACM on Programming
  Languages}} \bibinfo{volume}{5}, \bibinfo{number}{POPL}
  (\bibinfo{year}{2021}), \bibinfo{pages}{1--32}.
\newblock

\bibitem[\protect\citeauthoryear{{Dal Lago} and Grellois}{{Dal Lago} and
  Grellois}{2019}]%
        {Lago2017ProbabilisticTB}
\bibfield{author}{\bibinfo{person}{Ugo {Dal Lago}} {and}
  \bibinfo{person}{Charles Grellois}.} \bibinfo{year}{2019}\natexlab{}.
\newblock \showarticletitle{Probabilistic Termination by Monadic Affine Sized
  Typing}.
\newblock \bibinfo{journal}{\emph{ACM Transactions on Programming Languages and
  Systems}} \bibinfo{volume}{41}, \bibinfo{number}{2} (\bibinfo{year}{2019}),
  \bibinfo{pages}{10:1--10:65}.
\newblock

\bibitem[\protect\citeauthoryear{{Dal Lago} and Zorzi}{{Dal Lago} and
  Zorzi}{2012}]%
        {DBLP:journals/ita/LagoZ12}
\bibfield{author}{\bibinfo{person}{Ugo {Dal Lago}} {and}
  \bibinfo{person}{Margherita Zorzi}.} \bibinfo{year}{2012}\natexlab{}.
\newblock \showarticletitle{Probabilistic Operational Semantics for the Lambda
  Calculus}.
\newblock \bibinfo{journal}{\emph{RAIRO -- Theoretical Informatics and
  Applications}} \bibinfo{volume}{46}, \bibinfo{number}{3}
  (\bibinfo{year}{2012}), \bibinfo{pages}{413--450}.
\newblock

\bibitem[\protect\citeauthoryear{Danos and Ehrhard}{Danos and Ehrhard}{2011}]%
        {DANOS2011966}
\bibfield{author}{\bibinfo{person}{Vincent Danos} {and} \bibinfo{person}{Thomas
  Ehrhard}.} \bibinfo{year}{2011}\natexlab{}.
\newblock \showarticletitle{Probabilistic Coherence Spaces as a Model of
  Higher-Order Probabilistic Computation}.
\newblock \bibinfo{journal}{\emph{Information and Computation}}
  \bibinfo{volume}{209}, \bibinfo{number}{6} (\bibinfo{year}{2011}),
  \bibinfo{pages}{966--991}.
\newblock

\bibitem[\protect\citeauthoryear{Deng and Du}{Deng and Du}{2011}]%
        {DBLP:journals/corr/abs-1103-4577}
\bibfield{author}{\bibinfo{person}{Yuxin Deng} {and} \bibinfo{person}{Wenjie
  Du}.} \bibinfo{year}{2011}\natexlab{}.
\newblock \bibinfo{title}{Logical, Metric, and Algorithmic Characterisations of
  Probabilistic Bisimulation}.
\newblock \bibinfo{howpublished}{arXiv:1103.4577}.
\newblock

\bibitem[\protect\citeauthoryear{Disney and Flanagan}{Disney and
  Flanagan}{2011}]%
        {disneyFlanagan:stop2011}
\bibfield{author}{\bibinfo{person}{Tim Disney} {and} \bibinfo{person}{Cormac
  Flanagan}.} \bibinfo{year}{2011}\natexlab{}.
\newblock \showarticletitle{Gradual Information Flow Typing}. In
  \bibinfo{booktitle}{\emph{Proceedings of the Workshop on Script to Program
  Evolution (STOP 2011)}}.
\newblock

\bibitem[\protect\citeauthoryear{Dwork and Roth}{Dwork and Roth}{2014}]%
        {DBLP:journals/fttcs/DworkR14}
\bibfield{author}{\bibinfo{person}{Cynthia Dwork} {and} \bibinfo{person}{Aaron
  Roth}.} \bibinfo{year}{2014}\natexlab{}.
\newblock \showarticletitle{The Algorithmic Foundations of Differential
  Privacy}.
\newblock \bibinfo{journal}{\emph{Foundations and Trends in Theoretical
  Computer Science}} \bibinfo{volume}{9}, \bibinfo{number}{3--4}
  (\bibinfo{year}{2014}), \bibinfo{pages}{211--407}.
\newblock

\bibitem[\protect\citeauthoryear{Ehrhard, Pagani, and Tasson}{Ehrhard
  et~al\mbox{.}}{2018}]%
        {10.1145/3158147}
\bibfield{author}{\bibinfo{person}{Thomas Ehrhard}, \bibinfo{person}{Michele
  Pagani}, {and} \bibinfo{person}{Christine Tasson}.}
  \bibinfo{year}{2018}\natexlab{}.
\newblock \showarticletitle{Measurable Cones and Stable, Measurable Functions:
  A Model for Probabilistic Higher-Order Programming}.
\newblock \bibinfo{journal}{\emph{Proceedings of the ACM on Programming
  Languages}} \bibinfo{volume}{2}, \bibinfo{number}{POPL}
  (\bibinfo{year}{2018}), \bibinfo{pages}{59:1--59:28}.
\newblock

\bibitem[\protect\citeauthoryear{Erwig and Kollmansberger}{Erwig and
  Kollmansberger}{2006}]%
        {Erwig:FPP:06}
\bibfield{author}{\bibinfo{person}{Martin Erwig} {and} \bibinfo{person}{Steve
  Kollmansberger}.} \bibinfo{year}{2006}\natexlab{}.
\newblock \showarticletitle{Functional Pearls: Probabilistic Functional
  Programming in {Haskell}}.
\newblock \bibinfo{journal}{\emph{Journal of Functional Programming}}
  \bibinfo{volume}{16}, \bibinfo{number}{1} (\bibinfo{year}{2006}),
  \bibinfo{pages}{21--34}.
\newblock

\bibitem[\protect\citeauthoryear{Fennell and Thiemann}{Fennell and
  Thiemann}{2013}]%
        {fennellThiemann:csf2013}
\bibfield{author}{\bibinfo{person}{Luminous Fennell} {and}
  \bibinfo{person}{Peter Thiemann}.} \bibinfo{year}{2013}\natexlab{}.
\newblock \showarticletitle{Gradual Security Typing with References}. In
  \bibinfo{booktitle}{\emph{Proceedings of the IEEE Computer Security
  Foundations Symposium (CSF 2013)}}. \bibinfo{publisher}{IEEE Computer
  Society}, \bibinfo{pages}{224--239}.
\newblock

\bibitem[\protect\citeauthoryear{Garcia}{Garcia}{2013}]%
        {garcia2013calculating}
\bibfield{author}{\bibinfo{person}{Ronald Garcia}.}
  \bibinfo{year}{2013}\natexlab{}.
\newblock \showarticletitle{Calculating Threesomes, with Blame}. In
  \bibinfo{booktitle}{\emph{Proceedings of the ACM SIGPLAN International
  Conference on Functional Programming (ICFP 2013)}}. \bibinfo{publisher}{ACM},
  \bibinfo{pages}{417--428}.
\newblock

\bibitem[\protect\citeauthoryear{Garcia and Cimini}{Garcia and Cimini}{2015}]%
        {garciaCimini:popl2015}
\bibfield{author}{\bibinfo{person}{Ronald Garcia} {and} \bibinfo{person}{Matteo
  Cimini}.} \bibinfo{year}{2015}\natexlab{}.
\newblock \showarticletitle{Principal Type Schemes for Gradual Programs}. In
  \bibinfo{booktitle}{\emph{Proceedings of the ACM SIGPLAN-SIGACT Symposium on
  Principles of Programming Languages (POPL 2015)}}. \bibinfo{publisher}{ACM},
  \bibinfo{pages}{303--315}.
\newblock

\bibitem[\protect\citeauthoryear{Garcia, Clark, and Tanter}{Garcia
  et~al\mbox{.}}{2016}]%
        {garciaAl:popl2016}
\bibfield{author}{\bibinfo{person}{Ronald Garcia}, \bibinfo{person}{Alison~M.
  Clark}, {and} \bibinfo{person}{{\'E}ric Tanter}.}
  \bibinfo{year}{2016}\natexlab{}.
\newblock \showarticletitle{Abstracting Gradual Typing}. In
  \bibinfo{booktitle}{\emph{Proceedings of the ACM SIGPLAN-SIGACT Symposium on
  Principles of Programming Languages (POPL 2016)}}. \bibinfo{publisher}{ACM},
  \bibinfo{pages}{429--442}.
\newblock
\newblock
\shownote{See erratum: \url{https://www.cs.ubc.ca/~rxg/agt-erratum.pdf}.}

\bibitem[\protect\citeauthoryear{Gehr, Misailovic, and Vechev}{Gehr
  et~al\mbox{.}}{2016}]%
        {gehrAl:cav2016}
\bibfield{author}{\bibinfo{person}{Timon Gehr}, \bibinfo{person}{Sasa
  Misailovic}, {and} \bibinfo{person}{Martin~T. Vechev}.}
  \bibinfo{year}{2016}\natexlab{}.
\newblock \showarticletitle{{PSI}: Exact Symbolic Inference for Probabilistic
  Programs}. In \bibinfo{booktitle}{\emph{Proceedings of the International
  Conference on Computer Aided Verification (CAV 2016)}}. Lecture Notes in
  Computer Science, Vol.~9779, \bibinfo{publisher}{Springer},
  \bibinfo{pages}{62--83}.
\newblock

\bibitem[\protect\citeauthoryear{Ghahramani}{Ghahramani}{2015}]%
        {Ghahramani:2015}
\bibfield{author}{\bibinfo{person}{Zoubin Ghahramani}.}
  \bibinfo{year}{2015}\natexlab{}.
\newblock \showarticletitle{Probabilistic Machine Learning and Artificial
  Intelligence}.
\newblock \bibinfo{journal}{\emph{Nature}} \bibinfo{volume}{521},
  \bibinfo{number}{7553} (\bibinfo{year}{2015}), \bibinfo{pages}{452--459}.
\newblock

\bibitem[\protect\citeauthoryear{Goldwasser and Micali}{Goldwasser and
  Micali}{1984}]%
        {Goldwasser:1984}
\bibfield{author}{\bibinfo{person}{Shafi Goldwasser} {and}
  \bibinfo{person}{Silvio Micali}.} \bibinfo{year}{1984}\natexlab{}.
\newblock \showarticletitle{Probabilistic Encryption}.
\newblock \bibinfo{journal}{\emph{{Journal of Computer and System Sciences}}}
  \bibinfo{volume}{28}, \bibinfo{number}{2} (\bibinfo{year}{1984}),
  \bibinfo{pages}{270--299}.
\newblock

\bibitem[\protect\citeauthoryear{Goodman, Mansinghka, Roy, Bonawitz, and
  Tenenbaum}{Goodman et~al\mbox{.}}{2008}]%
        {10.5555/3023476.3023503}
\bibfield{author}{\bibinfo{person}{Noah~D. Goodman}, \bibinfo{person}{Vikash~K.
  Mansinghka}, \bibinfo{person}{Daniel~M. Roy}, \bibinfo{person}{Keith
  Bonawitz}, {and} \bibinfo{person}{Joshua~B. Tenenbaum}.}
  \bibinfo{year}{2008}\natexlab{}.
\newblock \showarticletitle{Church: A Language for Generative Models}. In
  \bibinfo{booktitle}{\emph{Proceedings of the Conference on Uncertainty in
  Artificial Intelligence (UAI 2008)}}. \bibinfo{publisher}{AUAI Press},
  \bibinfo{pages}{220--229}.
\newblock

\bibitem[\protect\citeauthoryear{Goodman and Stuhlm{\"u}ller}{Goodman and
  Stuhlm{\"u}ller}{2014}]%
        {dippl}
\bibfield{author}{\bibinfo{person}{Noah~D. Goodman} {and}
  \bibinfo{person}{Andreas Stuhlm{\"u}ller}.} \bibinfo{year}{2014}\natexlab{}.
\newblock \bibinfo{title}{The Design and Implementation of Probabilistic
  Programming Languages}.
\newblock \bibinfo{howpublished}{\url{http://dippl.org}}.
\newblock
\newblock
\shownote{Accessed: 2026-10-08.}

\bibitem[\protect\citeauthoryear{Gordon, Henzinger, Nori, and Rajamani}{Gordon
  et~al\mbox{.}}{2014}]%
        {DBLP:conf/icse/GordonHNR14}
\bibfield{author}{\bibinfo{person}{Andrew~D. Gordon},
  \bibinfo{person}{Thomas~A. Henzinger}, \bibinfo{person}{Aditya~V. Nori},
  {and} \bibinfo{person}{Sriram~K. Rajamani}.} \bibinfo{year}{2014}\natexlab{}.
\newblock \showarticletitle{Probabilistic Programming}. In
  \bibinfo{booktitle}{\emph{Proceedings of the Future of Software Engineering
  (FOSE 2014)}}. \bibinfo{publisher}{ACM}, \bibinfo{pages}{167--181}.
\newblock

\bibitem[\protect\citeauthoryear{Hattori, Kobayashi, and Sato}{Hattori
  et~al\mbox{.}}{2023}]%
        {Hattori23}
\bibfield{author}{\bibinfo{person}{Momoko Hattori}, \bibinfo{person}{Naoki
  Kobayashi}, {and} \bibinfo{person}{Ryosuke Sato}.}
  \bibinfo{year}{2023}\natexlab{}.
\newblock \showarticletitle{Gradual Tensor Shape Checking}. In
  \bibinfo{booktitle}{\emph{Proceedings of the European Symposium on
  Programming (ESOP 2023)}}. Lecture Notes in Computer Science, Vol.~13990,
  \bibinfo{publisher}{Springer}, \bibinfo{pages}{197--224}.
\newblock

\bibitem[\protect\citeauthoryear{Heijltjes and Majury}{Heijltjes and
  Majury}{2025}]%
        {DBLP:conf/csl/HeijltjesM25}
\bibfield{author}{\bibinfo{person}{Willem Heijltjes} {and}
  \bibinfo{person}{Georgina Majury}.} \bibinfo{year}{2025}\natexlab{}.
\newblock \showarticletitle{Simple Types for Probabilistic Termination}. In
  \bibinfo{booktitle}{\emph{Proceedings of the EACSL Annual Conference on
  Computer Science Logic (CSL 2025)}}. Leibniz International Proceedings in
  Informatics (LIPIcs), Vol.~326, \bibinfo{publisher}{Schloss Dagstuhl --
  Leibniz-Zentrum f{\"u}r Informatik}, \bibinfo{pages}{31:1--31:21}.
\newblock

\bibitem[\protect\citeauthoryear{Herman, Tomb, and Flanagan}{Herman
  et~al\mbox{.}}{2007}]%
        {herman07}
\bibfield{author}{\bibinfo{person}{David Herman}, \bibinfo{person}{Aaron Tomb},
  {and} \bibinfo{person}{Cormac Flanagan}.} \bibinfo{year}{2007}\natexlab{}.
\newblock \showarticletitle{Space-Efficient Gradual Typing}. In
  \bibinfo{booktitle}{\emph{Proceedings of the Symposium on Trends in
  Functional Programming (TFP 2007)}}. Trends in Functional Programming,
  Vol.~8, \bibinfo{publisher}{Intellect}, \bibinfo{pages}{1--18}.
\newblock

\bibitem[\protect\citeauthoryear{Herman, Tomb, and Flanagan}{Herman
  et~al\mbox{.}}{2010}]%
        {hermanAl:hosc10}
\bibfield{author}{\bibinfo{person}{David Herman}, \bibinfo{person}{Aaron Tomb},
  {and} \bibinfo{person}{Cormac Flanagan}.} \bibinfo{year}{2010}\natexlab{}.
\newblock \showarticletitle{Space-Efficient Gradual Typing}.
\newblock \bibinfo{journal}{\emph{Higher-Order and Symbolic Computation}}
  \bibinfo{volume}{23}, \bibinfo{number}{2} (\bibinfo{year}{2010}),
  \bibinfo{pages}{167--189}.
\newblock

\bibitem[\protect\citeauthoryear{Heunen, Kammar, Staton, and Yang}{Heunen
  et~al\mbox{.}}{2017}]%
        {DBLP:conf/lics/HeunenKSY17}
\bibfield{author}{\bibinfo{person}{Chris Heunen}, \bibinfo{person}{Ohad
  Kammar}, \bibinfo{person}{Sam Staton}, {and} \bibinfo{person}{Hongseok
  Yang}.} \bibinfo{year}{2017}\natexlab{}.
\newblock \showarticletitle{A Convenient Category for Higher-Order Probability
  Theory}. In \bibinfo{booktitle}{\emph{Proceedings of the ACM/IEEE Symposium
  on Logic in Computer Science (LICS 2017)}}. \bibinfo{publisher}{IEEE Computer
  Society}, \bibinfo{pages}{1--12}.
\newblock

\bibitem[\protect\citeauthoryear{Holtzen, {Van den Broeck}, and
  Millstein}{Holtzen et~al\mbox{.}}{2020}]%
        {holtzenAl:oopsla2020}
\bibfield{author}{\bibinfo{person}{Steven Holtzen}, \bibinfo{person}{Guy {Van
  den Broeck}}, {and} \bibinfo{person}{Todd~D. Millstein}.}
  \bibinfo{year}{2020}\natexlab{}.
\newblock \showarticletitle{Scaling Exact Inference for Discrete Probabilistic
  Programs}.
\newblock \bibinfo{journal}{\emph{Proceedings of the ACM on Programming
  Languages}} \bibinfo{volume}{4}, \bibinfo{number}{OOPSLA}
  (\bibinfo{year}{2020}), \bibinfo{pages}{140:1--140:31}.
\newblock

\bibitem[\protect\citeauthoryear{Huang and Oliveira}{Huang and
  Oliveira}{2020}]%
        {Huang2020ATO}
\bibfield{author}{\bibinfo{person}{Xuejing Huang} {and} \bibinfo{person}{Bruno
  C. d.~S. Oliveira}.} \bibinfo{year}{2020}\natexlab{}.
\newblock \showarticletitle{A Type-Directed Operational Semantics for a
  Calculus with a Merge Operator}. In \bibinfo{booktitle}{\emph{Proceedings of
  the European Conference on Object-Oriented Programming (ECOOP 2020)}}.
  Leibniz International Proceedings in Informatics (LIPIcs), Vol.~166,
  \bibinfo{publisher}{Schloss Dagstuhl -- Leibniz-Zentrum f{\"u}r Informatik},
  \bibinfo{pages}{26:1--26:32}.
\newblock

\bibitem[\protect\citeauthoryear{Jafery and Dunfield}{Jafery and
  Dunfield}{2017}]%
        {jaferyDunfield:popl2017}
\bibfield{author}{\bibinfo{person}{Khurram~A. Jafery} {and}
  \bibinfo{person}{Jana Dunfield}.} \bibinfo{year}{2017}\natexlab{}.
\newblock \showarticletitle{Sums of Uncertainty: Refinements Go Gradual}. In
  \bibinfo{booktitle}{\emph{Proceedings of the ACM SIGPLAN-SIGACT Symposium on
  Principles of Programming Languages (POPL 2017)}}. \bibinfo{publisher}{ACM},
  \bibinfo{pages}{804--817}.
\newblock

\bibitem[\protect\citeauthoryear{Jones and Plotkin}{Jones and Plotkin}{1989}]%
        {Jones1989APP}
\bibfield{author}{\bibinfo{person}{C. Jones} {and} \bibinfo{person}{Gordon~D.
  Plotkin}.} \bibinfo{year}{1989}\natexlab{}.
\newblock \showarticletitle{A Probabilistic Powerdomain of Evaluations}. In
  \bibinfo{booktitle}{\emph{Proceedings of the ACM/IEEE Symposium on Logic in
  Computer Science (LICS 1989)}}. \bibinfo{publisher}{IEEE Computer Society},
  \bibinfo{pages}{186--195}.
\newblock

\bibitem[\protect\citeauthoryear{Kaminski}{Kaminski}{2019}]%
        {kaminski2019advanced}
\bibfield{author}{\bibinfo{person}{Benjamin~Lucien Kaminski}.}
  \bibinfo{year}{2019}\natexlab{}.
\newblock \emph{\bibinfo{title}{Advanced Weakest Precondition Calculi for
  Probabilistic Programs}}.
\newblock \bibinfo{thesistype}{Ph.D. Dissertation}. \bibinfo{school}{RWTH
  Aachen University}.
\newblock

\bibitem[\protect\citeauthoryear{Kaminski, Katoen, and Matheja}{Kaminski
  et~al\mbox{.}}{2019}]%
        {Kaminski:AI}
\bibfield{author}{\bibinfo{person}{Benjamin~Lucien Kaminski},
  \bibinfo{person}{Joost{-}Pieter Katoen}, {and} \bibinfo{person}{Christoph
  Matheja}.} \bibinfo{year}{2019}\natexlab{}.
\newblock \showarticletitle{On the Hardness of Analyzing Probabilistic
  Programs}.
\newblock \bibinfo{journal}{\emph{Acta Informatica}} \bibinfo{volume}{56},
  \bibinfo{number}{3} (\bibinfo{year}{2019}), \bibinfo{pages}{255--285}.
\newblock

\bibitem[\protect\citeauthoryear{Kiselyov}{Kiselyov}{2016}]%
        {DBLP:conf/aplas/Kiselyov16}
\bibfield{author}{\bibinfo{person}{Oleg Kiselyov}.}
  \bibinfo{year}{2016}\natexlab{}.
\newblock \showarticletitle{Probabilistic Programming Language and Its
  Incremental Evaluation}. In \bibinfo{booktitle}{\emph{Proceedings of the
  Asian Symposium on Programming Languages and Systems (APLAS 2016)}}. Lecture
  Notes in Computer Science, Vol.~10017, \bibinfo{publisher}{Springer},
  \bibinfo{pages}{357--376}.
\newblock

\bibitem[\protect\citeauthoryear{Le, Baydin, and Wood}{Le
  et~al\mbox{.}}{2017}]%
        {DBLP:conf/aistats/LeBW17}
\bibfield{author}{\bibinfo{person}{Tuan~Anh Le}, \bibinfo{person}{Atilim~Gunes
  Baydin}, {and} \bibinfo{person}{Frank~D. Wood}.}
  \bibinfo{year}{2017}\natexlab{}.
\newblock \showarticletitle{Inference Compilation and Universal Probabilistic
  Programming}. In \bibinfo{booktitle}{\emph{Proceedings of the International
  Conference on Artificial Intelligence and Statistics (AISTATS 2017)}}.
  Proceedings of Machine Learning Research, Vol.~54, \bibinfo{publisher}{PMLR},
  \bibinfo{pages}{1338--1348}.
\newblock

\bibitem[\protect\citeauthoryear{Lehmann and Tanter}{Lehmann and
  Tanter}{2017}]%
        {lehmannTanter:popl2017}
\bibfield{author}{\bibinfo{person}{Nico Lehmann} {and}
  \bibinfo{person}{{\'E}ric Tanter}.} \bibinfo{year}{2017}\natexlab{}.
\newblock \showarticletitle{Gradual Refinement Types}. In
  \bibinfo{booktitle}{\emph{Proceedings of the ACM SIGPLAN-SIGACT Symposium on
  Principles of Programming Languages (POPL 2017)}}. \bibinfo{publisher}{ACM},
  \bibinfo{pages}{775--788}.
\newblock

\bibitem[\protect\citeauthoryear{Lennon{-}Bertrand, Maillard, Tabareau, and
  Tanter}{Lennon{-}Bertrand et~al\mbox{.}}{2022}]%
        {lennonAl:toplas2022}
\bibfield{author}{\bibinfo{person}{Meven Lennon{-}Bertrand},
  \bibinfo{person}{Kenji Maillard}, \bibinfo{person}{Nicolas Tabareau}, {and}
  \bibinfo{person}{{\'E}ric Tanter}.} \bibinfo{year}{2022}\natexlab{}.
\newblock \showarticletitle{Gradualizing the Calculus of Inductive
  Constructions}.
\newblock \bibinfo{journal}{\emph{ACM Transactions on Programming Languages and
  Systems}} \bibinfo{volume}{44}, \bibinfo{number}{2} (\bibinfo{year}{2022}),
  \bibinfo{pages}{7:1--7:82}.
\newblock

\bibitem[\protect\citeauthoryear{Mak, Ong, Paquet, and Wagner}{Mak
  et~al\mbox{.}}{2021}]%
        {DBLP:conf/esop/MakOPW21}
\bibfield{author}{\bibinfo{person}{Carol Mak}, \bibinfo{person}{C.{-}H.~Luke
  Ong}, \bibinfo{person}{Hugo Paquet}, {and} \bibinfo{person}{Dominik Wagner}.}
  \bibinfo{year}{2021}\natexlab{}.
\newblock \showarticletitle{Densities of Almost Surely Terminating
  Probabilistic Programs Are Differentiable Almost Everywhere}. In
  \bibinfo{booktitle}{\emph{Proceedings of the European Symposium on
  Programming (ESOP 2021)}}. Lecture Notes in Computer Science, Vol.~12648,
  \bibinfo{publisher}{Springer}, \bibinfo{pages}{432--461}.
\newblock

\bibitem[\protect\citeauthoryear{Malewski, Greenberg, and Tanter}{Malewski
  et~al\mbox{.}}{2021}]%
        {malewskiAl:oopsla2021}
\bibfield{author}{\bibinfo{person}{Stefan Malewski}, \bibinfo{person}{Michael
  Greenberg}, {and} \bibinfo{person}{{\'E}ric Tanter}.}
  \bibinfo{year}{2021}\natexlab{}.
\newblock \showarticletitle{Gradually Structured Data}.
\newblock \bibinfo{journal}{\emph{Proceedings of the ACM on Programming
  Languages}} \bibinfo{volume}{5}, \bibinfo{number}{OOPSLA}
  (\bibinfo{year}{2021}), \bibinfo{pages}{1--29}.
\newblock

\bibitem[\protect\citeauthoryear{Migeed, Reed, Ansel, and Palsberg}{Migeed
  et~al\mbox{.}}{2024}]%
        {migeed24}
\bibfield{author}{\bibinfo{person}{Zeina Migeed}, \bibinfo{person}{James Reed},
  \bibinfo{person}{Jason Ansel}, {and} \bibinfo{person}{Jens Palsberg}.}
  \bibinfo{year}{2024}\natexlab{}.
\newblock \showarticletitle{Generalizing Shape Analysis with Gradual Types}. In
  \bibinfo{booktitle}{\emph{Proceedings of the European Conference on
  Object-Oriented Programming (ECOOP 2024)}}. Leibniz International Proceedings
  in Informatics (LIPIcs), Vol.~313, \bibinfo{publisher}{Schloss Dagstuhl --
  Leibniz-Zentrum f{\"u}r Informatik}, \bibinfo{pages}{29:1--29:28}.
\newblock

\bibitem[\protect\citeauthoryear{Mittal}{Mittal}{2016}]%
        {DBLP:journals/csur/Mittal16b}
\bibfield{author}{\bibinfo{person}{Sparsh Mittal}.}
  \bibinfo{year}{2016}\natexlab{}.
\newblock \showarticletitle{A Survey of Techniques for Approximate Computing}.
\newblock \bibinfo{journal}{\emph{{ACM Computing Surveys}}}
  \bibinfo{volume}{48}, \bibinfo{number}{4} (\bibinfo{year}{2016}),
  \bibinfo{pages}{62:1--62:33}.
\newblock

\bibitem[\protect\citeauthoryear{Miyazaki, Sekiyama, and Igarashi}{Miyazaki
  et~al\mbox{.}}{2019}]%
        {miyazakiAl:popl2019}
\bibfield{author}{\bibinfo{person}{Yusuke Miyazaki}, \bibinfo{person}{Taro
  Sekiyama}, {and} \bibinfo{person}{Atsushi Igarashi}.}
  \bibinfo{year}{2019}\natexlab{}.
\newblock \showarticletitle{Dynamic Type Inference for Gradual {Hindley-Milner}
  Typing}.
\newblock \bibinfo{journal}{\emph{Proceedings of the ACM on Programming
  Languages}} \bibinfo{volume}{3}, \bibinfo{number}{POPL}
  (\bibinfo{year}{2019}), \bibinfo{pages}{18:1--18:29}.
\newblock

\bibitem[\protect\citeauthoryear{Moggi}{Moggi}{1991}]%
        {Moggi1991NotionsOC}
\bibfield{author}{\bibinfo{person}{Eugenio Moggi}.}
  \bibinfo{year}{1991}\natexlab{}.
\newblock \showarticletitle{Notions of Computation and Monads}.
\newblock \bibinfo{journal}{\emph{Information and Computation}}
  \bibinfo{volume}{93}, \bibinfo{number}{1} (\bibinfo{year}{1991}),
  \bibinfo{pages}{55--92}.
\newblock

\bibitem[\protect\citeauthoryear{Motwani and Raghavan}{Motwani and
  Raghavan}{1995}]%
        {Motwani:1995}
\bibfield{author}{\bibinfo{person}{Rajeev Motwani} {and}
  \bibinfo{person}{Prabhakar Raghavan}.} \bibinfo{year}{1995}\natexlab{}.
\newblock \bibinfo{booktitle}{\emph{Randomized Algorithms}}.
\newblock \bibinfo{publisher}{Cambridge University Press}.
\newblock

\bibitem[\protect\citeauthoryear{New, Licata, and Ahmed}{New
  et~al\mbox{.}}{2019}]%
        {newAl:popl2019}
\bibfield{author}{\bibinfo{person}{Max~S. New}, \bibinfo{person}{Daniel~R.
  Licata}, {and} \bibinfo{person}{Amal Ahmed}.}
  \bibinfo{year}{2019}\natexlab{}.
\newblock \showarticletitle{Gradual Type Theory}.
\newblock \bibinfo{journal}{\emph{Proceedings of the ACM on Programming
  Languages}} \bibinfo{volume}{3}, \bibinfo{number}{POPL}
  (\bibinfo{year}{2019}), \bibinfo{pages}{15:1--15:31}.
\newblock

\bibitem[\protect\citeauthoryear{Pfeffer}{Pfeffer}{2011}]%
        {DBLP:conf/ilp/Pfeffer10}
\bibfield{author}{\bibinfo{person}{Avi Pfeffer}.}
  \bibinfo{year}{2011}\natexlab{}.
\newblock \showarticletitle{Practical Probabilistic Programming}. In
  \bibinfo{booktitle}{\emph{Proceedings of the International Conference on
  Inductive Logic Programming (ILP 2010)}}. Lecture Notes in Computer Science,
  Vol.~6489, \bibinfo{publisher}{Springer}, \bibinfo{pages}{2--3}.
\newblock

\bibitem[\protect\citeauthoryear{Phipps{-}Costin, Anderson, Greenberg, and
  Guha}{Phipps{-}Costin et~al\mbox{.}}{2021}]%
        {phippsEtAl:oopsla2021}
\bibfield{author}{\bibinfo{person}{Luna Phipps{-}Costin},
  \bibinfo{person}{Carolyn~Jane Anderson}, \bibinfo{person}{Michael Greenberg},
  {and} \bibinfo{person}{Arjun Guha}.} \bibinfo{year}{2021}\natexlab{}.
\newblock \showarticletitle{Solver-Based Gradual Type Migration}.
\newblock \bibinfo{journal}{\emph{Proceedings of the ACM on Programming
  Languages}} \bibinfo{volume}{5}, \bibinfo{number}{OOPSLA}
  (\bibinfo{year}{2021}), \bibinfo{pages}{1--27}.
\newblock

\bibitem[\protect\citeauthoryear{Pierce}{Pierce}{2002}]%
        {pierce:tapl}
\bibfield{author}{\bibinfo{person}{Benjamin~C. Pierce}.}
  \bibinfo{year}{2002}\natexlab{}.
\newblock \bibinfo{booktitle}{\emph{Types and Programming Languages}}.
\newblock \bibinfo{publisher}{MIT Press}.
\newblock

\bibitem[\protect\citeauthoryear{Ramsey and Pfeffer}{Ramsey and
  Pfeffer}{2002}]%
        {Ramsey2002StochasticLC}
\bibfield{author}{\bibinfo{person}{Norman Ramsey} {and} \bibinfo{person}{Avi
  Pfeffer}.} \bibinfo{year}{2002}\natexlab{}.
\newblock \showarticletitle{Stochastic Lambda Calculus and Monads of
  Probability Distributions}. In \bibinfo{booktitle}{\emph{Proceedings of the
  ACM SIGPLAN-SIGACT Symposium on Principles of Programming Languages (POPL
  2002)}}. \bibinfo{publisher}{ACM}, \bibinfo{pages}{154--165}.
\newblock

\bibitem[\protect\citeauthoryear{Reed and Pierce}{Reed and Pierce}{2010}]%
        {Reed:2010}
\bibfield{author}{\bibinfo{person}{Jason Reed} {and}
  \bibinfo{person}{Benjamin~C. Pierce}.} \bibinfo{year}{2010}\natexlab{}.
\newblock \showarticletitle{Distance Makes the Types Grow Stronger: A Calculus
  for Differential Privacy}. In \bibinfo{booktitle}{\emph{Proceedings of the
  ACM SIGPLAN International Conference on Functional Programming (ICFP 2010)}}.
  \bibinfo{publisher}{ACM}, \bibinfo{pages}{157--168}.
\newblock

\bibitem[\protect\citeauthoryear{Saad, Rinard, and Mansinghka}{Saad
  et~al\mbox{.}}{2021}]%
        {saadAl:pldi2021}
\bibfield{author}{\bibinfo{person}{Feras~A. Saad}, \bibinfo{person}{Martin~C.
  Rinard}, {and} \bibinfo{person}{Vikash~K. Mansinghka}.}
  \bibinfo{year}{2021}\natexlab{}.
\newblock \showarticletitle{{SPPL}: Probabilistic Programming with Fast Exact
  Symbolic Inference}. In \bibinfo{booktitle}{\emph{Proceedings of the ACM
  SIGPLAN Conference on Programming Language Design and Implementation (PLDI
  2021)}}. \bibinfo{publisher}{ACM}, \bibinfo{pages}{804--819}.
\newblock

\bibitem[\protect\citeauthoryear{Saheb{-}Djahromi}{Saheb{-}Djahromi}{1978}]%
        {SahebDjahromi1978ProbabilisticL}
\bibfield{author}{\bibinfo{person}{Nasser Saheb{-}Djahromi}.}
  \bibinfo{year}{1978}\natexlab{}.
\newblock \showarticletitle{Probabilistic {LCF}}. In
  \bibinfo{booktitle}{\emph{Proceedings of the Symposium on Mathematical
  Foundations of Computer Science (MFCS 1978)}}. Lecture Notes in Computer
  Science, Vol.~64, \bibinfo{publisher}{Springer}, \bibinfo{pages}{442--451}.
\newblock

\bibitem[\protect\citeauthoryear{Saheb{-}Djahromi}{Saheb{-}Djahromi}{1980}]%
        {SahebDjahromi1980CPOsOM}
\bibfield{author}{\bibinfo{person}{Nasser Saheb{-}Djahromi}.}
  \bibinfo{year}{1980}\natexlab{}.
\newblock \showarticletitle{{CPO}'s of Measures for Nondeterminism}.
\newblock \bibinfo{journal}{\emph{Theoretical Computer Science}}
  \bibinfo{volume}{12}, \bibinfo{number}{1} (\bibinfo{year}{1980}),
  \bibinfo{pages}{19--37}.
\newblock

\bibitem[\protect\citeauthoryear{{\'S}cibior, Kammar, V{\'a}k{\'a}r, Staton,
  Yang, Cai, Ostermann, Moss, Heunen, and Ghahramani}{{\'S}cibior
  et~al\mbox{.}}{2018}]%
        {DBLP:journals/pacmpl/ScibiorKVSYCOMH18}
\bibfield{author}{\bibinfo{person}{Adam {\'S}cibior}, \bibinfo{person}{Ohad
  Kammar}, \bibinfo{person}{Matthijs V{\'a}k{\'a}r}, \bibinfo{person}{Sam
  Staton}, \bibinfo{person}{Hongseok Yang}, \bibinfo{person}{Yufei Cai},
  \bibinfo{person}{Klaus Ostermann}, \bibinfo{person}{Sean~K. Moss},
  \bibinfo{person}{Chris Heunen}, {and} \bibinfo{person}{Zoubin Ghahramani}.}
  \bibinfo{year}{2018}\natexlab{}.
\newblock \showarticletitle{Denotational Validation of Higher-Order {Bayesian}
  Inference}.
\newblock \bibinfo{journal}{\emph{Proceedings of the ACM on Programming
  Languages}} \bibinfo{volume}{2}, \bibinfo{number}{POPL}
  (\bibinfo{year}{2018}), \bibinfo{pages}{60:1--60:29}.
\newblock

\bibitem[\protect\citeauthoryear{Segala and Lynch}{Segala and Lynch}{1995}]%
        {DBLP:journals/njc/SegalaL95}
\bibfield{author}{\bibinfo{person}{Roberto Segala} {and}
  \bibinfo{person}{Nancy~A. Lynch}.} \bibinfo{year}{1995}\natexlab{}.
\newblock \showarticletitle{Probabilistic Simulations for Probabilistic
  Processes}.
\newblock \bibinfo{journal}{\emph{Nordic Journal of Computing}}
  \bibinfo{volume}{2}, \bibinfo{number}{2} (\bibinfo{year}{1995}),
  \bibinfo{pages}{250--273}.
\newblock

\bibitem[\protect\citeauthoryear{Siek, Garcia, and Taha}{Siek
  et~al\mbox{.}}{2009}]%
        {Siek2009ExploringTD}
\bibfield{author}{\bibinfo{person}{Jeremy~G. Siek}, \bibinfo{person}{Ronald
  Garcia}, {and} \bibinfo{person}{Walid Taha}.}
  \bibinfo{year}{2009}\natexlab{}.
\newblock \showarticletitle{Exploring the Design Space of Higher-Order Casts}.
  In \bibinfo{booktitle}{\emph{Proceedings of the European Symposium on
  Programming (ESOP 2009)}}. Lecture Notes in Computer Science, Vol.~5502,
  \bibinfo{publisher}{Springer}, \bibinfo{pages}{17--31}.
\newblock

\bibitem[\protect\citeauthoryear{Siek and Taha}{Siek and Taha}{2006}]%
        {siekTaha:sfp2006}
\bibfield{author}{\bibinfo{person}{Jeremy~G. Siek} {and} \bibinfo{person}{Walid
  Taha}.} \bibinfo{year}{2006}\natexlab{}.
\newblock \showarticletitle{Gradual Typing for Functional Languages}. In
  \bibinfo{booktitle}{\emph{Proceedings of the Scheme and Functional
  Programming Workshop 2006}}. \bibinfo{pages}{81--92}.
\newblock

\bibitem[\protect\citeauthoryear{Siek, Thiemann, and Wadler}{Siek
  et~al\mbox{.}}{2015a}]%
        {siekAl:pldi2015}
\bibfield{author}{\bibinfo{person}{Jeremy~G. Siek}, \bibinfo{person}{Peter
  Thiemann}, {and} \bibinfo{person}{Philip Wadler}.}
  \bibinfo{year}{2015}\natexlab{a}.
\newblock \showarticletitle{Blame and Coercion: Together Again for the First
  Time}. In \bibinfo{booktitle}{\emph{Proceedings of the ACM SIGPLAN Conference
  on Programming Language Design and Implementation (PLDI 2015)}}.
  \bibinfo{publisher}{ACM}, \bibinfo{pages}{425--435}.
\newblock

\bibitem[\protect\citeauthoryear{Siek and Vachharajani}{Siek and
  Vachharajani}{2008}]%
        {siekVachharajani:dls2008}
\bibfield{author}{\bibinfo{person}{Jeremy~G. Siek} {and}
  \bibinfo{person}{Manish Vachharajani}.} \bibinfo{year}{2008}\natexlab{}.
\newblock \showarticletitle{Gradual Typing with Unification-Based Inference}.
  In \bibinfo{booktitle}{\emph{Proceedings of the Symposium on Dynamic
  Languages (DLS 2008)}}. \bibinfo{publisher}{ACM}, \bibinfo{pages}{7}.
\newblock

\bibitem[\protect\citeauthoryear{Siek, Vitousek, Cimini, and Boyland}{Siek
  et~al\mbox{.}}{2015b}]%
        {siekAl:snapl2015}
\bibfield{author}{\bibinfo{person}{Jeremy~G. Siek}, \bibinfo{person}{Michael~M.
  Vitousek}, \bibinfo{person}{Matteo Cimini}, {and} \bibinfo{person}{John~Tang
  Boyland}.} \bibinfo{year}{2015}\natexlab{b}.
\newblock \showarticletitle{Refined Criteria for Gradual Typing}. In
  \bibinfo{booktitle}{\emph{Proceedings of the Summit on Advances in
  Programming Languages (SNAPL 2015)}}. Leibniz International Proceedings in
  Informatics (LIPIcs), Vol.~32, \bibinfo{publisher}{Schloss Dagstuhl --
  Leibniz-Zentrum f{\"u}r Informatik}, \bibinfo{pages}{274--293}.
\newblock

\bibitem[\protect\citeauthoryear{Siek, Vitousek, Cimini, Tobin{-}Hochstadt, and
  Garcia}{Siek et~al\mbox{.}}{2015c}]%
        {Siek2015MonotonicRF}
\bibfield{author}{\bibinfo{person}{Jeremy~G. Siek}, \bibinfo{person}{Michael~M.
  Vitousek}, \bibinfo{person}{Matteo Cimini}, \bibinfo{person}{Sam
  Tobin{-}Hochstadt}, {and} \bibinfo{person}{Ronald Garcia}.}
  \bibinfo{year}{2015}\natexlab{c}.
\newblock \showarticletitle{Monotonic References for Efficient Gradual Typing}.
  In \bibinfo{booktitle}{\emph{Proceedings of the European Symposium on
  Programming (ESOP 2015)}}. Lecture Notes in Computer Science, Vol.~9032,
  \bibinfo{publisher}{Springer}, \bibinfo{pages}{432--456}.
\newblock

\bibitem[\protect\citeauthoryear{Siek and Wadler}{Siek and Wadler}{2010}]%
        {siekAl:popl10}
\bibfield{author}{\bibinfo{person}{Jeremy~G. Siek} {and}
  \bibinfo{person}{Philip Wadler}.} \bibinfo{year}{2010}\natexlab{}.
\newblock \showarticletitle{Threesomes, With and Without Blame}. In
  \bibinfo{booktitle}{\emph{Proceedings of the ACM SIGPLAN-SIGACT Symposium on
  Principles of Programming Languages (POPL 2010)}}. \bibinfo{publisher}{ACM},
  \bibinfo{pages}{365--376}.
\newblock

\bibitem[\protect\citeauthoryear{Takikawa, Strickland, Dimoulas,
  Tobin{-}Hochstadt, and Felleisen}{Takikawa et~al\mbox{.}}{2012}]%
        {takikawaAl:oopsla2012}
\bibfield{author}{\bibinfo{person}{Asumu Takikawa}, \bibinfo{person}{T.~Stephen
  Strickland}, \bibinfo{person}{Christos Dimoulas}, \bibinfo{person}{Sam
  Tobin{-}Hochstadt}, {and} \bibinfo{person}{Matthias Felleisen}.}
  \bibinfo{year}{2012}\natexlab{}.
\newblock \showarticletitle{Gradual Typing for First-Class Classes}. In
  \bibinfo{booktitle}{\emph{Proceedings of the ACM SIGPLAN International
  Conference on Object-Oriented Programming, Systems, Languages, and
  Applications (OOPSLA 2012)}}. \bibinfo{publisher}{ACM},
  \bibinfo{pages}{793--810}.
\newblock

\bibitem[\protect\citeauthoryear{Toro, Garcia, and Tanter}{Toro
  et~al\mbox{.}}{2018}]%
        {toroAl:toplas2018}
\bibfield{author}{\bibinfo{person}{Mat{\'\i}as Toro}, \bibinfo{person}{Ronald
  Garcia}, {and} \bibinfo{person}{{\'E}ric Tanter}.}
  \bibinfo{year}{2018}\natexlab{}.
\newblock \showarticletitle{Type-Driven Gradual Security with References}.
\newblock \bibinfo{journal}{\emph{ACM Transactions on Programming Languages and
  Systems}} \bibinfo{volume}{40}, \bibinfo{number}{4} (\bibinfo{year}{2018}),
  \bibinfo{pages}{16:1--16:55}.
\newblock

\bibitem[\protect\citeauthoryear{Toro, Olmedo, and Ye}{Toro
  et~al\mbox{.}}{2026}]%
        {gplc-lean}
\bibfield{author}{\bibinfo{person}{Mat{\'\i}as Toro}, \bibinfo{person}{Federico
  Olmedo}, {and} \bibinfo{person}{Wenjia Ye}.} \bibinfo{year}{2026}\natexlab{}.
\newblock \bibinfo{title}{The Metatheory of {A Gradual Probabilistic Lambda
  Calculus}, in {Lean}~4}.
\newblock \bibinfo{howpublished}{\url{https://github.com/pleiad/gplc-lean}}.
\newblock
\newblock
\shownote{Accessed: 2026-10-08.}

\bibitem[\protect\citeauthoryear{Toro and Tanter}{Toro and Tanter}{2017}]%
        {toroTanter:sas2017}
\bibfield{author}{\bibinfo{person}{Mat{\'\i}as Toro} {and}
  \bibinfo{person}{{\'E}ric Tanter}.} \bibinfo{year}{2017}\natexlab{}.
\newblock \showarticletitle{A Gradual Interpretation of Union Types}. In
  \bibinfo{booktitle}{\emph{Proceedings of the International Static Analysis
  Symposium (SAS 2017)}}. Lecture Notes in Computer Science, Vol.~10422,
  \bibinfo{publisher}{Springer}, \bibinfo{pages}{382--404}.
\newblock

\bibitem[\protect\citeauthoryear{Toro and Tanter}{Toro and Tanter}{2020}]%
        {toroTanter:scp2020}
\bibfield{author}{\bibinfo{person}{Mat{\'\i}as Toro} {and}
  \bibinfo{person}{{\'E}ric Tanter}.} \bibinfo{year}{2020}\natexlab{}.
\newblock \showarticletitle{Abstracting Gradual References}.
\newblock \bibinfo{journal}{\emph{Science of Computer Programming}}
  \bibinfo{volume}{197} (\bibinfo{year}{2020}), \bibinfo{pages}{102496}.
\newblock

\bibitem[\protect\citeauthoryear{Tran, Hoffman, Saurous, Brevdo, Murphy, and
  Blei}{Tran et~al\mbox{.}}{2017}]%
        {DBLP:conf/iclr/TranHSB0B17}
\bibfield{author}{\bibinfo{person}{Dustin Tran}, \bibinfo{person}{Matthew~D.
  Hoffman}, \bibinfo{person}{Rif~A. Saurous}, \bibinfo{person}{Eugene Brevdo},
  \bibinfo{person}{Kevin Murphy}, {and} \bibinfo{person}{David~M. Blei}.}
  \bibinfo{year}{2017}\natexlab{}.
\newblock \showarticletitle{Deep Probabilistic Programming}. In
  \bibinfo{booktitle}{\emph{Proceedings of the International Conference on
  Learning Representations (ICLR 2017)}}. \bibinfo{publisher}{OpenReview.net}.
\newblock

\bibitem[\protect\citeauthoryear{van~de Meent, Paige, Yang, and Wood}{van~de
  Meent et~al\mbox{.}}{2018}]%
        {DBLP:journals/corr/abs-1809-10756}
\bibfield{author}{\bibinfo{person}{Jan{-}Willem van~de Meent},
  \bibinfo{person}{Brooks Paige}, \bibinfo{person}{Hongseok Yang}, {and}
  \bibinfo{person}{Frank Wood}.} \bibinfo{year}{2018}\natexlab{}.
\newblock \bibinfo{title}{An Introduction to Probabilistic Programming}.
\newblock \bibinfo{howpublished}{arXiv:1809.10756}.
\newblock

\bibitem[\protect\citeauthoryear{Villani}{Villani}{2009}]%
        {villani:2009}
\bibfield{author}{\bibinfo{person}{C{\'e}dric Villani}.}
  \bibinfo{year}{2009}\natexlab{}.
\newblock \bibinfo{booktitle}{\emph{Optimal Transport: Old and New}}.
  \bibinfo{series}{Grundlehren der mathematischen Wissenschaften},
  Vol.~\bibinfo{volume}{338}.
\newblock \bibinfo{publisher}{Springer}.
\newblock

\bibitem[\protect\citeauthoryear{Wadler and Findler}{Wadler and
  Findler}{2009}]%
        {wadlerFindler:esop2009}
\bibfield{author}{\bibinfo{person}{Philip Wadler} {and}
  \bibinfo{person}{Robert~Bruce Findler}.} \bibinfo{year}{2009}\natexlab{}.
\newblock \showarticletitle{Well-Typed Programs Can't Be Blamed}. In
  \bibinfo{booktitle}{\emph{Proceedings of the European Symposium on
  Programming (ESOP 2009)}}. Lecture Notes in Computer Science, Vol.~5502,
  \bibinfo{publisher}{Springer}, \bibinfo{pages}{1--16}.
\newblock

\bibitem[\protect\citeauthoryear{Ye and Oliveira}{Ye and Oliveira}{2024}]%
        {yeOliveira:jfp2024}
\bibfield{author}{\bibinfo{person}{Wenjia Ye} {and} \bibinfo{person}{Bruno C.
  d.~S. Oliveira}.} \bibinfo{year}{2024}\natexlab{}.
\newblock \showarticletitle{Type-Directed Operational Semantics for Gradual
  Typing}.
\newblock \bibinfo{journal}{\emph{Journal of Functional Programming}}
  \bibinfo{volume}{34} (\bibinfo{year}{2024}), \bibinfo{pages}{e9}.
\newblock

\bibitem[\protect\citeauthoryear{Ye, Oliveira, and Huang}{Ye
  et~al\mbox{.}}{2021}]%
        {Ye2021TypeDirectedOS}
\bibfield{author}{\bibinfo{person}{Wenjia Ye}, \bibinfo{person}{Bruno C. d.~S.
  Oliveira}, {and} \bibinfo{person}{Xuejing Huang}.}
  \bibinfo{year}{2021}\natexlab{}.
\newblock \showarticletitle{Type-Directed Operational Semantics for Gradual
  Typing}. In \bibinfo{booktitle}{\emph{Proceedings of the European Conference
  on Object-Oriented Programming (ECOOP 2021)}}. Leibniz International
  Proceedings in Informatics (LIPIcs), Vol.~194, \bibinfo{publisher}{Schloss
  Dagstuhl -- Leibniz-Zentrum f{\"u}r Informatik},
  \bibinfo{pages}{12:1--12:30}.
\newblock

\bibitem[\protect\citeauthoryear{Ye, Oliveira, and Toro}{Ye
  et~al\mbox{.}}{2024}]%
        {Ye24merge}
\bibfield{author}{\bibinfo{person}{Wenjia Ye}, \bibinfo{person}{Bruno C. d.~S.
  Oliveira}, {and} \bibinfo{person}{Mat{\'\i}as Toro}.}
  \bibinfo{year}{2024}\natexlab{}.
\newblock \showarticletitle{Merging Gradual Typing}.
\newblock \bibinfo{journal}{\emph{Proceedings of the ACM on Programming
  Languages}} \bibinfo{volume}{8}, \bibinfo{number}{OOPSLA2}
  (\bibinfo{year}{2024}), \bibinfo{pages}{648--676}.
\newblock

\bibitem[\protect\citeauthoryear{Ye, Toro, and Olmedo}{Ye
  et~al\mbox{.}}{2023}]%
        {yeAl:oopsla2023}
\bibfield{author}{\bibinfo{person}{Wenjia Ye}, \bibinfo{person}{Mat{\'\i}as
  Toro}, {and} \bibinfo{person}{Federico Olmedo}.}
  \bibinfo{year}{2023}\natexlab{}.
\newblock \showarticletitle{A Gradual Probabilistic Lambda Calculus}.
\newblock \bibinfo{journal}{\emph{Proceedings of the ACM on Programming
  Languages}} \bibinfo{volume}{7}, \bibinfo{number}{OOPSLA1}
  (\bibinfo{year}{2023}), \bibinfo{pages}{256--285}.
\newblock

\end{thebibliography}
\end{document}